%% file: main.tex
\documentclass[a4paper,11pt]{article}
\pdfoutput=1

\usepackage[utf8]{inputenc}
\usepackage[T1]{fontenc}

\usepackage{xparse}
\usepackage{rotating}

\usepackage[sorting=nyt, 
style=alphabetic, 
backref, 
maxbibnames=20, 
doi=false,
url=true]{biblatex}
\usepackage[margin=0.8in]{geometry}

\usepackage{amsmath,amssymb,amsfonts,mathrsfs,ulem,bm,url}
\usepackage{longtable}

\expandafter\def\expandafter\normalsize\expandafter{%
    \normalsize%
    \setlength\abovedisplayskip{1pt}%
    \setlength\belowdisplayskip{3pt}%
    \setlength\abovedisplayshortskip{-8pt}%
    \setlength\belowdisplayshortskip{2pt}%
}

\usepackage{enumitem}
\usepackage{graphicx}
\usepackage[rightcaption]{sidecap}
\usepackage{wrapfig}
\usepackage{soul}
\usepackage{verbatim}
\usepackage{pdflscape}
\usepackage{mathtools}
\usepackage{multirow}
\usepackage{pgfplots}
\usepgfplotslibrary{fillbetween}
\pgfplotsset{compat=newest}
\usepgfplotslibrary{groupplots}
\usetikzlibrary{matrix}
\usetikzlibrary{quantikz2}
\usepackage{stmaryrd}
\usepackage{tikz}
\usepackage{cancel}
\usepackage{appendix}
\usepackage{caption}
\usepackage{booktabs}

\definecolor{coral}{RGB}{254,125,106}

\usepackage{soul}

\newcommand{\T}{T}
\usepackage{braket}
\newcommand{\ketbra}[1]{\ket{#1}\bra{#1}}

\usepackage{amsthm}
\usepackage[most]{tcolorbox}
\usepackage{keytheorems}

\newtheorem{theorem}{Theorem}[section]
\newtheorem{example}[theorem]{Example}
\newtheorem{remark}[theorem]{Remark}
\newtheorem{corollary}[theorem]{Corollary}
\newtheorem{definition}[theorem]{Definition}
\newtheorem{lemma}[theorem]{Lemma}
\newtheorem{proposition}[theorem]{Proposition}
\newtheorem{claim}[theorem]{Claim}

\newkeytheorem{construction}[sibling=theorem, tcolorbox-no-titlebar={colback=red!10,breakable}]

\input{macrosetup}

\usepackage{hyperref}

\usepackage{array}
\usepackage{tabularx}
\usepackage{makecell}

\newcolumntype{P}[1]{>{\centering\arraybackslash}p{#1}}
\newcolumntype{Y}{>{\centering\arraybackslash}X}

\usepackage{authblk}
\makeatletter
\renewcommand\AB@affilsepx{\quad\quad\protect\Affilfont}
\makeatother
\begin{document}

\title{Magic State Distillation via Codes over Binary Extension Fields}
\author[1,3]{Anqi Gong}
\author[2]{Christopher A. Pattison}
\author[3]{Patrick Rall}
\author[4,3]{Adam Wills}
\affil[1]{ETH Zürich}
\affil[2]{UC Berkeley}
\affil[3]{IBM Quantum}
\affil[4]{MIT}
\date{\today}

\maketitle

\begin{abstract}
Fault-tolerant quantum computation architectures are frequently bottlenecked by the overhead of producing high-fidelity magic states. 
In this work, we use algebraic geometric techniques to construct codes over binary extension fields $\mathbb{F}_{2^s}$, thus discovering new protocols for the distillation of qubit magic states, where our focus is on the regime of practical qubit-based quantum computing architectures. To do this, we show that multi-qubit gates of interest such as $\CS$, $\CCZ$, and $\TOF\# = \CCZ_{123}\CCZ_{345}$, can be packaged into simple gates over the larger fields, and we derive simple algebraic conditions in the extension fields allowing the distillation of these gates. Because they are derived from Galois qudits, the corresponding qubits codes naturally handle the correlated errors present on such multi-qubit states. Moreover, the protocols we discover are extremely compact; for example, we show that $4$ $\CS$ states can be distilled to $1$ $\CS$ state at distance $2$, using only $4$ logical qubits.

For a case study, we consider the distillation of $\CS$ and $\CCZ$ states from injected $\T$ and $\CS$ states. When optimized for magic state production per unit time, or logical spacetime volume, we find that our protocols outperform the state-of-the-art in almost every situation, both at input error rates $10^{-3}$ (direct injection), and $10^{-6}$ (allowing some cultivation pre-injection).

\end{abstract}

\tableofcontents
\section{Introduction}
\subsection{Overview}
Gate teleportation is a common way to achieve a universal quantum gateset.
Most commonly, the teleported gate is in the third level of the Clifford hierarchy (termed a ``third-level gate'') where the input resource state is known as a \emph{magic state}.
The usual recipe for the preparation of magic states is to first prepare encoded noisy magic states and \textit{distill} them to high-fidelity magic states, which requires only Clifford operations and measurements.
In the standard procedure for magic state distillation~\cite{MSD,bravyi2012magic}, a quantum error-detecting stabilizer code encoding \(k\) qubits into \(n\) qubits with a transversal third-level gate can be used to construct a magic state distillation protocol by 1) preparing a codestate, 2) applying the noisy transversal third-level gate \(U\) using input noisy magic states, 3) measuring the syndrome of the code, and 4) outputting the unencoded state if the syndrome is trivial.
If the action of the third-level gate is Clifford-equivalent to encoded \(V^{\otimes k}\) for some third-level gate \(V\), then this protocol turns \(n\) noisy magic states for \(U\) into \(k\) less-noisy magic states for \(V\).

In this paper, we construct magic state distillation protocols leveraging algebraic geometry codes over binary extension fields, with a focus on using novel theory to construct practical protocols. A binary extension field is a finite field of size $2^s$, denoted $\mathbb{F}_{2^s}$. A CSS code over $\mathbb{F}_{2^s}$ is also a CSS code over $\mathbb{F}_2$~\cite{wills2026review}. However, by working directly with the extension field, one is often able to construct qubit codes with desirable properties that would have been very difficult to construct by working directly over qubits. This intuition has been leveraged for magic state distillation before~\cite{wills2024,golowich2025asymptotically,nguyen2024,nguyen_pattison}, although this is the first time these techniques have been considered in a truly practical setting. The main messages of this paper are as follows:
\begin{enumerate}
    \item\label{message:gate_package} Multi-qubit gates of interest, including $\CS$ and $\CCZ$, may be neatly packaged into gates for single Galois qudits.\footnote{Galois qudits are the quantum systems that encode finite extension fields like $\mathbb{F}_{2^s}$. Such a Galois qudit is equivalent to a set of $s$ qubits~\cite{wills2026review}. Throughout this work, when we refer to a qudit, it is left implicit that we simply mean a Galois qudit, that is, a set of qubits. Moreover, whenever we refer to a quantum code for qudits, we mean a CSS code over $\mathbb{F}_{2^s}$.} It was previously unknown that such widely-considered multi-qubit gates admitted such neat algebraic descriptions in the Galois qudit picture. With these descriptions, we can prove conditions for Galois qudit codes to admit these gates transversally, thus allowing their distillation.
    \item Moving to codes over binary extension fields offers richer algebraic structures, allowing for better code parameters than are possible than by working with the binary field directly, including in the small-size regime of practicality, ultimately resulting in record-breaking protocols.
    \item\label{message:multi-qubit-errors} Working with Galois qudit codes allows us to naturally handle the multi-qubit errors present on the corresponding multi-qubit states. This is the same intuition that has led to the wide-spread use of codes over larger alphabets to handle burst errors in classical communication and storage.
    \item\label{message:exotic} Constructing codes with novel algebraic techniques often leads to the possibility of distilling exotic magic states. If such an exotic magic state is called for in a particular algorithm, it will often be more efficient to use a specialized magic state distillation protocol for that state, rather than using generic schemes that distill states such as $\ket{\T}$, and synthesizing the gate.
\end{enumerate}
Let us give some examples of the above points in our work. For point~\ref{message:gate_package}, in Section~\ref{sec:CS_distillation_protocols}, we show that the well-known two-qubit $\CS$ gate can be written into a natural form as a single-qudit gate over $\mathbb{F}_4$, and that its transversality can be achieved using a simple three-orthogonality condition. Using the Reed-Solomon code over $\mathbb{F}_4$, we use this to construct a protocol that distills $4$ $\ket{\CS}$ states into $1$ $\ket{\CS}$ state at distance $2$, that can be run on $4$ logical qubits. Moving to something slightly larger, by using the $\mathbb{F}_8$ Klein quartic algebraic curve, we show that $24$ $\ket{\CCZ}$ states can be distilled to $4$ $\ket{\CCZ}$ states at distance $3$. For some context here, prior to this work, we are not aware of any way to perform magic state distillation with rate $1/6$ at distance $3$, let alone at such small sizes (outside of very pathological and/or large-qudit magic states). There are several examples of such breakthroughs in our catalogue.

Note that, in the two examples mentioned above, the notion of \textit{distance} is the correct one for this situation --- that is --- the distance is the minimum number of input magic states that must have errors to cause an undetectable logical failure. Because multi-qubit states like $\ket{\CS}$ and $\ket{\CCZ}$ can have correlated errors, this is the correct notion of distance, rather than the usual qubit distance of a code; this is the content of point~\ref{message:multi-qubit-errors} above. For an example of a more exotic distillation protocol as mentioned in point~\ref{message:exotic} above, we will show using the $\mathbb{F}_4$ Hermitian algebraic curve that $8$ $\ket{\CS}$ states may be distilled into the $5$-qubit magic state for the $\TOF\#$ gate, which is $\CCZ_{123}\CCZ_{345}$. This gate is Clifford-equivalent to a shared-controlled-SWAP, which is used in fermionic Hamiltonian simulation~\cite{babbush2018encoding}, and other algorithmic subroutines~\cite{Childs_2003}. We show our full catalogue of new protocols in Section~\ref{subsec:catalogue}. Following that, in Section~\ref{subsec:footprint_results}, we show that our novel protocols turn out to be exceed the state-of-the-art in a very broad sense, when considering natural distillation tasks, as well as sizes, states, and error rates relevant to realistic quantum computation architectures.

\subsection{Relation to Prior Work}

As discussed above, quantum codes over Galois qudits~\cite{wills2026review}, meaning quantum codes over binary extension fields, were considered recently, to construct asymptotic families of protocols with constant distillation rate for qubits~\cite{wills2024,golowich2025asymptotically,nguyen2024}, following a long line of work improving the asymptotic rate of distillation~\cite{MSD,meier2012magic,   ,jones2013multilevel,haah2017magic,sublogarithmic,krishna2018towards}. Codes over binary extension fields were subsequently used to construct quantum error-correcting codes with desirable properties for qubits~\cite{nguyen_pattison,golowich2025quantum,he2025quantum,he2025asymptotically,golowich2025near,wills2026concatenating}. Finite-size magic state distillation protocols were constructed in the recent work~\cite{cervia2025magic} using codes over binary extension fields, although the resulting protocols are for qudits, not qubits, and the sizes considered are larger than would allow for immediate practical implementation. Nevertheless, our works on the distillation of qudit $\CCZ$ gates in binary extension fields --- see Section~\ref{sec:AG_codes} --- may be viewed as complementary to theirs. 

Further recent works have constructed, or surveyed, magic state distillation protocols in several parameter regimes~\cite{nezami2022classification,singh2026borrowed,jacinto2026exploring}, although our protocols lie outside the purvue of all of these. A further, related group of works~\cite{jacob2025single,li2026transversaldimensionjumpproduct,Menon_2026,tiew2026copy} constructs transversal non-Clifford gates in the finite length regime, where here the focus is on qLDPC codes and implementation at the \textit{physical} level. Complementarily, our protocols have stronger parameters, and generally distill more directly-useful gates, although are intended for use at the \textit{logical} level. One might consider if our techniques, allowing the construction of distillation protocols taking in exotic multi-qubit gates at the logical level, can be used in tandem with those works, which generate such exotic magic states from physical-level magic.

\subsection{Results}\label{sec:results}

\subsubsection{Catalogue of Protocols}\label{subsec:catalogue}

We begin by presenting a catalogue of our new protocols in Table~\ref{tab:protocol_catalogue}. 

\begin{table}[ht]
\centering
\begin{tabular}{cccc}
\textbf{Input} & \textbf{Output} & \textbf{Distance} & \textbf{Method}\\
\hline\hline
$(2k+2)\;\CS$ & $k\;\CS$ & $d=2$ & \hyperref[con:2k+2-to-k-CS]{$\F_4$ Reed-Solomon}
\\
$8 \; \CS$ & $1\;\TOF\#$ & $d=2$ & \hyperref[con:8cstoTOF]{$\F_4$ Hermitian code}
\\
$(6k+2)\;\CCZ$ & $2k\;\CCZ$ & $d=2$ & \hyperref[con:6k+2CCZto2kCCZ]{$\F_8$ Reed-Solomon}
\\
$2^s\;\text{C}^{s-1}\text{Z}$ & $2\;\text{C}^{s-1}\text{Z}$ & $d=2$ & \hyperref[con:multi-control-Z-d2]{$\F_{2^s}$ Reed-Solomon}
\\
$8\;\CCS$ & $1\;\CCS$ & $d=2$ & \hyperref[con:8ccsto1ccsd2]{$\F_8$ Reed-Solomon}
\\
$16\;\CS$ & $1\;\TOF\#$ & $d=3$ & \hyperref[con:16cs1TOF]{$\F_4$ affine monomial code}
\\
$9\;\CCZ$ & $1\;\CCZ$ & $d=3$ & \hyperref[con:9-to-1-CCZ]{$\F_8$ projective Reed-Solomon}
\\
$16\;\CCZ$ & $2\;\CCZ$ & $d=3$ & \hyperref[con:16to2CCZ]{$\F_8$ Reed-Solomon}
\\
$21\;\CCZ$ & $3\;\CCZ$ & $d=3$ & \hyperref[con:21CCZ-to-3CCZ]{$\F_8$ projective monomial code}
\\
$(2^s+1)\;\text{C}^{s-1}\text{Z}$ & $1\;\text{C}^{s-1}\text{Z}$ & $d=3$ & \hyperref[con:d3multicontrol]{$\mathbb{F}_{2^s}$ projective Reed-Solomon}
\\
$12\;\CS$ & $1\;\CS$ & $d=3$ & \hyperref[con:12cs1csd3]{$\mathbb{F}_4$ projective monomial code}
\\
$16\;\CS$ & $2\;\CS$ & $d=3$ & \hyperref[con:16cs2csd3]{$\mathbb{F}_4$ affine monomial code}
\\
$21\;\CS$ & $4\;\CS$ & $d=3$ & \hyperref[con:21CS-to-4CS]{$\mathbb{F}_4$ projective monomial code}
\\
$24\;\CCZ$ & $4\;\CCZ$ & $d=3$ & \hyperref[con:24CCZ-to-4CCZ-d3]{$\mathbb{F}_8$ Klein quartic}
\\
$64\;T$ & $2\;\TOF\#$ & $d=4$ & \hyperref[con:64T_to_2TOF]{$\F_2$ Reed-Muller}
\\
$64\;T$ & $1\;\F_4$-$\CCZ$ & $d=4$ & \hyperref[con:64ttof4qudit]{$\F_{16}$ Hermitian trace code}
\\
$64\;T$ & $2\;\CCZ$ & $d=4$ & \hyperref[con:64tto2ccz]{$\F_{16}$ Hermitian trace code}
\\
$64\;\CCZ$ & $8\;\CCZ$ & $d=4$ & \hyperref[con:64CCZ-to-8CCZ]{$\F_8$ affine monomial code}
\end{tabular}
\caption{Catalogue of new protocols in this paper. For each new protocol, we write the input, the output, the distance, and the method used to obtain the protocol. Each method is hyperlinked to the construction in the paper. Note that a protocol distilling $64$ $\T$ states to $2$ $\CCZ$ states at distance $4$ is already known~\cite{Haah2018}, although our novel protocol has a lower leading-order probability of any error in the factory: $2720p^4$ versus $2944p^4$. In the table, we emphasize that all states are regular qubit magic states. $\TOF\#$ is the $5$-qubit magic state $\CCZ_{123}\CCZ_{345}$. Similarly, $\mathbb{F}_{2^s}$-$\CCZ$ is a $3s$-qubit magic state formed from a certain product of $\CCZ$ gates; see Section~\ref{sec:qudit-CCZ}. There are further novel protocols taking $\mathbb{F}_{2^s}$-$\CCZ$ as input and output; see Sections~\ref{subsec:qudit-ccz-protocols} and~\ref{sec:AG_codes}.\label{tab:protocol_catalogue}}
\end{table}

\subsubsection{Spacetime Footprints for Distillation Tasks}\label{subsec:footprint_results}

In this section, we discuss how we benchmark the performance of our schemes on distillation tasks of interest, at sizes and error rates of interest, against metrics of interest. The resulting Pareto frontiers are shown in Figures~\ref{fig:pareto-grid-p1e-3} and~\ref{fig:pareto-grid-p1e-6}, where we will find that the optimal schemes use our protocols, in the overwhelming majority of situations. We are going to consider various distillation tasks, and examine which concatenation of protocols, both known and novel to this paper, performs best. The distillation tasks we consider are encoded in the following choices:
\begin{itemize}
    \item \textbf{Injected state}. This is the state at the start of the concatenated tower of distillation protocols. We focus on either injected $\ket{\T}$ states or $\ket{\CS}$ states;
    \item \textbf{Input error rate}. We consider input error rate $p = 10^{-3}$, modeling a situation of direct injection, or $p = 10^{-6}$, modeling a situation of some small amount of cultivation~\cite{gidney2024magic} before injection, as is becoming popular\footnote{While cultivation on its own can reach error rates lower than $10^{-6}$, doing so can be costly in terms of physical spacetime volume. To reach extremely low logical error rates, or to save physical spacetime volume, a desirable middle ground is to cultivate a small amount to error rates of (roughly) $10^{-6}$, before performing distillation. It is true, however, that if one is willing to pay the physical spacetime volume to cultivate down to very low error rates, or if one is not aiming for extremely low logical error rates, only one round of a simple distance-$2$ protocol~\cite{campbell2017unified,singh2026borrowed} is likely necessary post-cultivation.}~\cite{gidney2025factor2048bitrsa,lee2025low,xu2025batched,xu2026distillingmagicstatesbicycle,webster2026pinnacle,cain2026shor}. It is worth noting that the cultivation of $\ket{\CS}$ states has not been studied directly, as far as we know, but is entirely reasonable.
    \item \textbf{Outputted state}. We consider the distillation of either $\ket{\CS}$ states, or $\ket{\CCZ}$ states, of a given target error rate, as these are valuable states for the purposes of real logical algorithms.
\end{itemize}

\begin{figure*}[t]
  \centering
  \begin{minipage}[t]{0.495\textwidth}
    \centering
    \includegraphics[width=\linewidth]{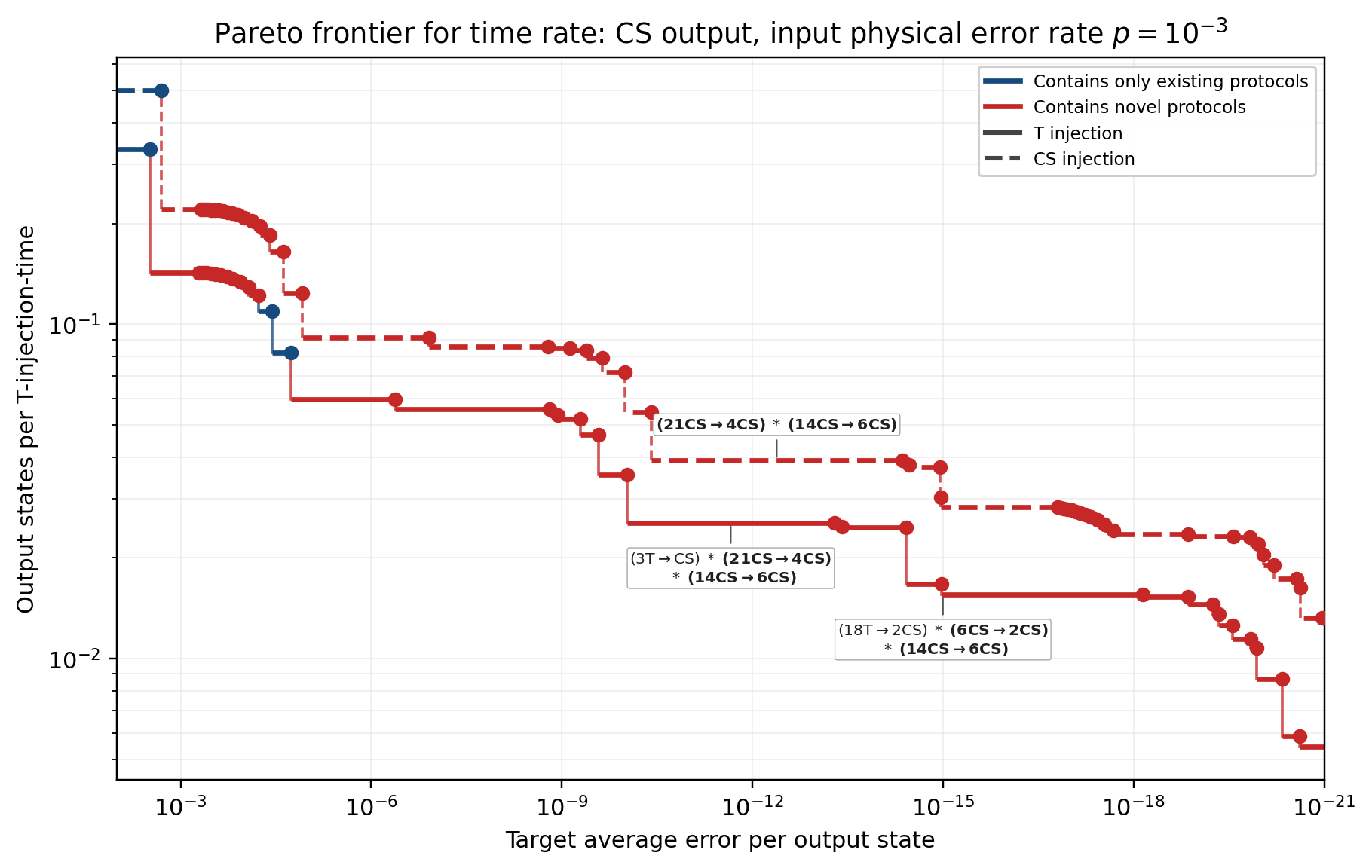}
    \par\smallskip
    \textup{(a) CS distillation: optimal for time rate}
  \end{minipage}%
  \hfill
  \begin{minipage}[t]{0.495\textwidth}
    \centering
    \includegraphics[width=\linewidth]{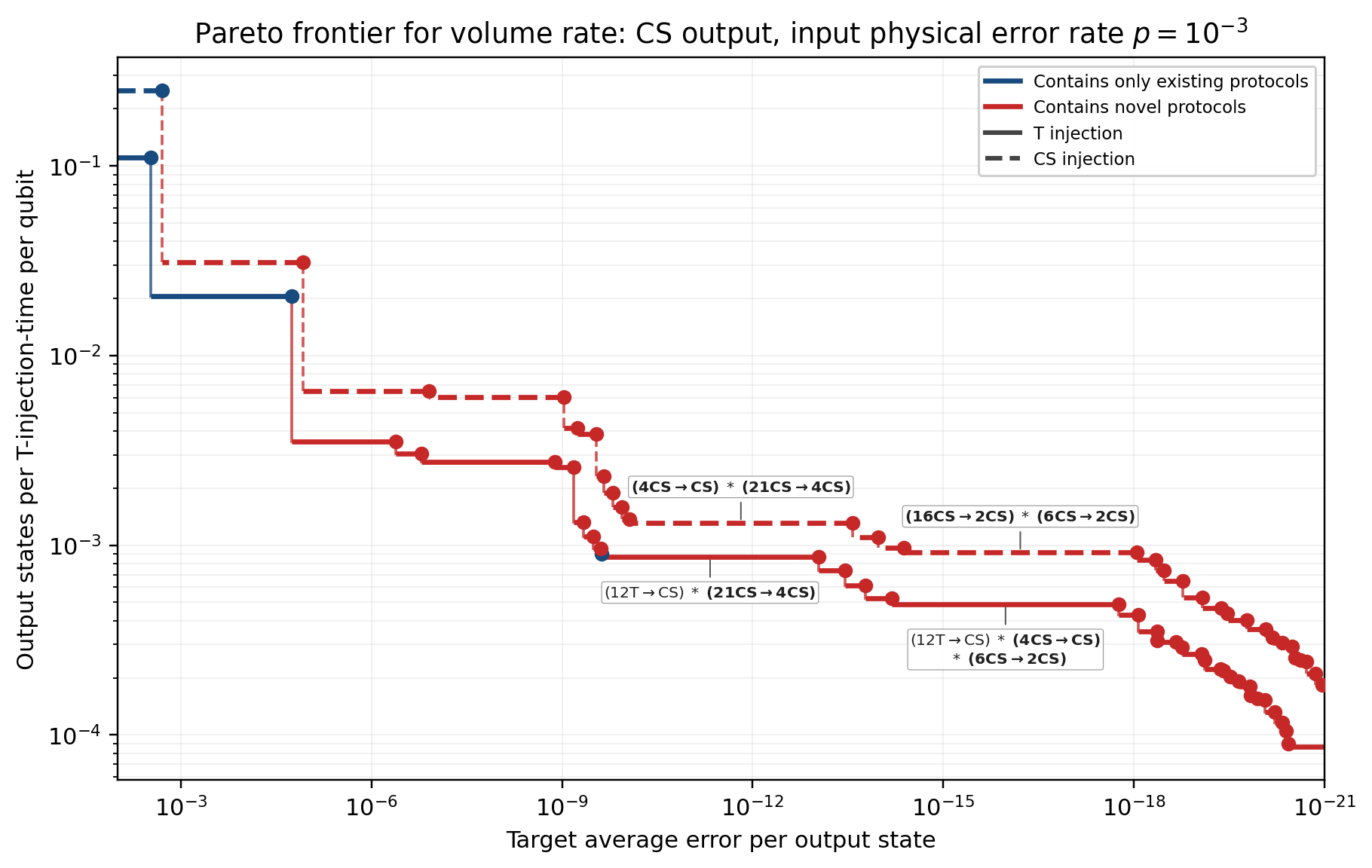}
    \par\smallskip
    \textup{(b) CS distillation: optimal for volumetric rate}
  \end{minipage}

  \par\medskip

  \begin{minipage}[t]{0.495\textwidth}
    \centering
    \includegraphics[width=\linewidth]{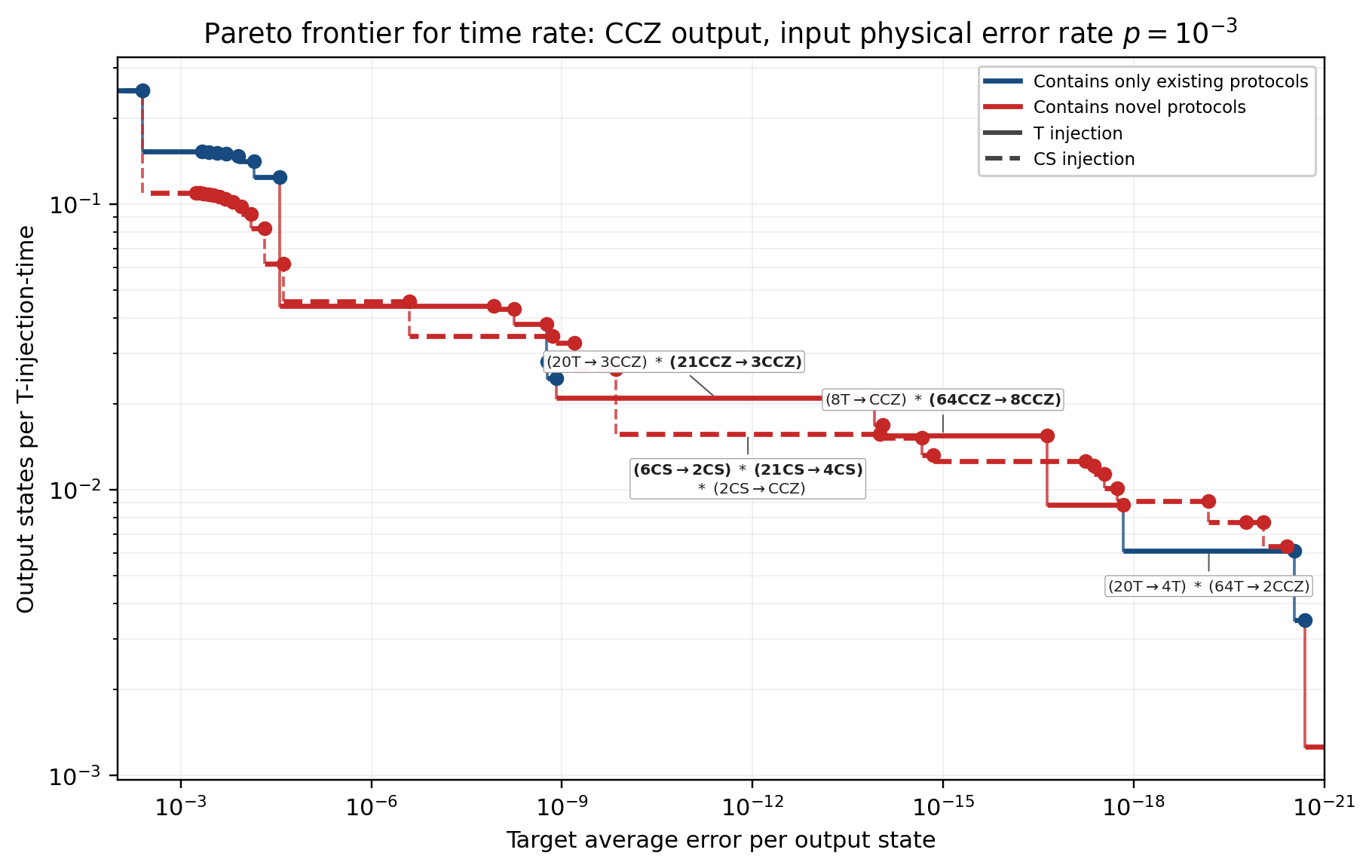}
    \par\smallskip
    \textup{(c) CCZ distillation: optimal for time rate}
  \end{minipage}%
  \hfill
  \begin{minipage}[t]{0.495\textwidth}
    \centering
    \includegraphics[width=\linewidth]{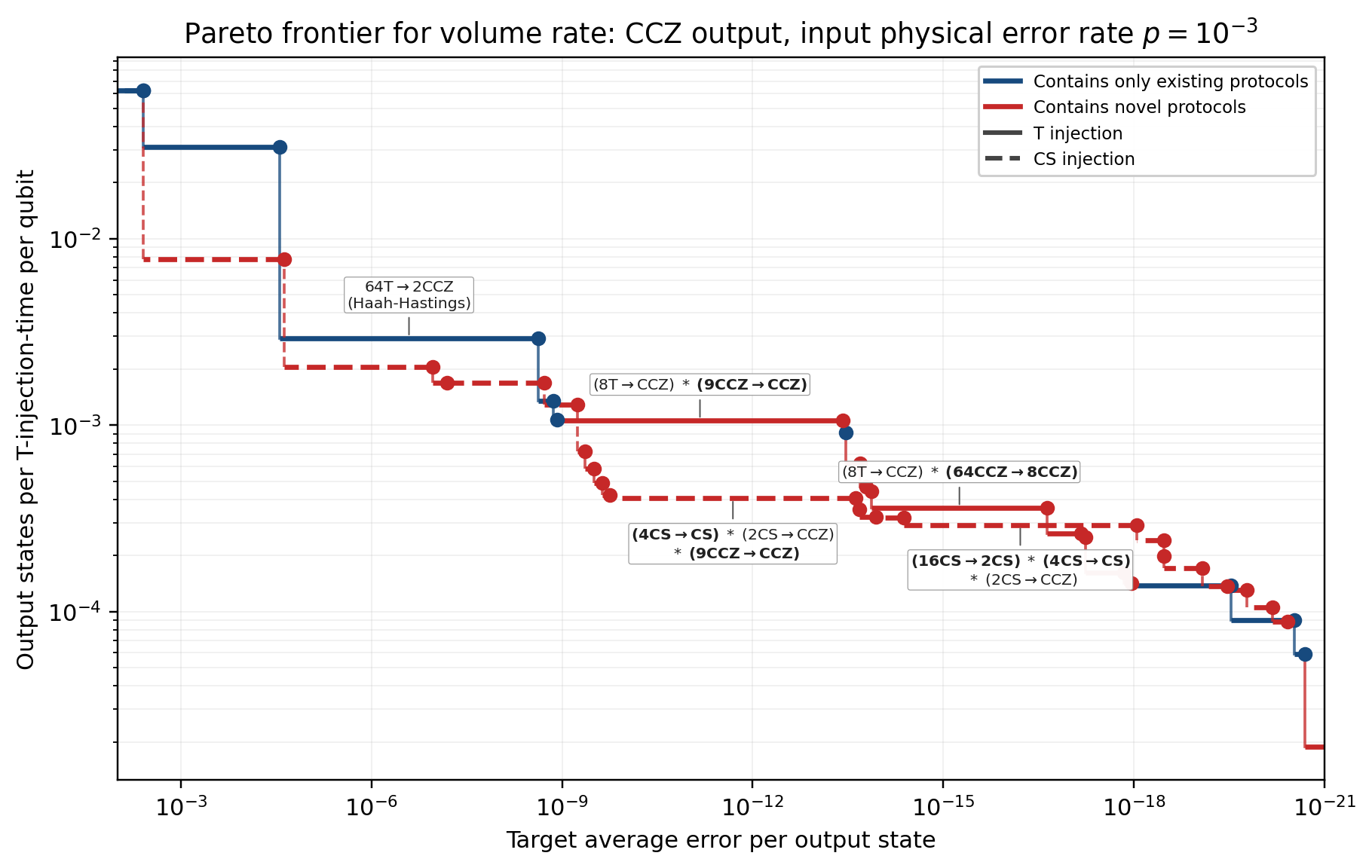}
    \par\smallskip
    \textup{(d) CCZ distillation: optimal for volumetric rate}
  \end{minipage}

  \caption{Pareto frontiers for physical input error
  $p=10^{-3}$. Solid lines use $\ket{\T}$ state injection and dashed lines use
  $\ket{\CS}$ state injection. Red lines indicate that our novel protocols are used somewhere in the competitive scheme; blue lines indicate that they are not. Our novel protocols are bolded in labels. Statistics for some of the competitive protocols are given in Appendix~\ref{subsec:tables_footprints_errors}.}
  \label{fig:pareto-grid-p1e-3}
\end{figure*}

\begin{figure*}[t]
  \centering
  \begin{minipage}[t]{0.495\textwidth}
    \centering
    \includegraphics[width=\linewidth]{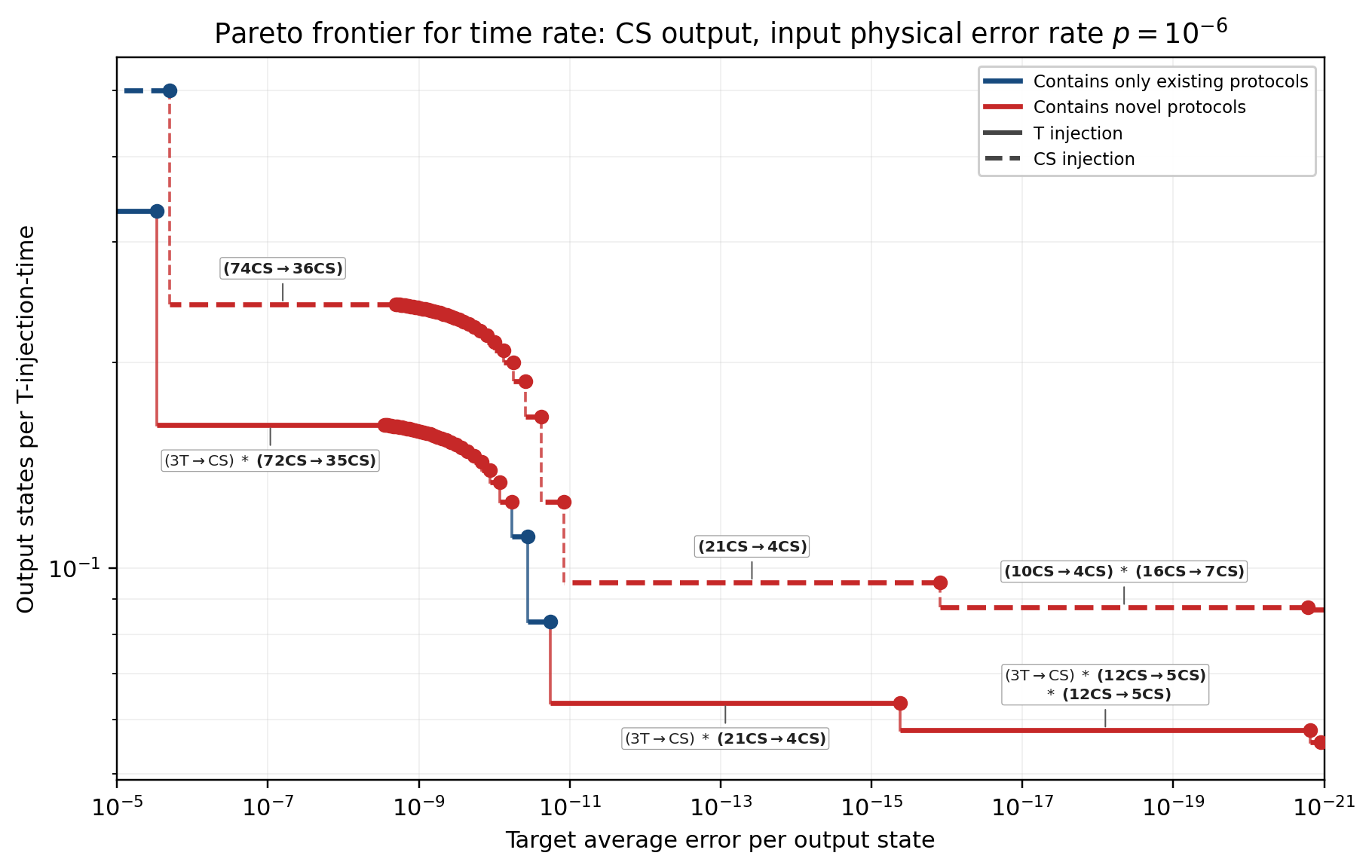}
    \par\smallskip
    \textup{(a) CS distillation: optimal for time rate}
  \end{minipage}%
  \hfill
  \begin{minipage}[t]{0.495\textwidth}
    \centering
    \includegraphics[width=\linewidth]{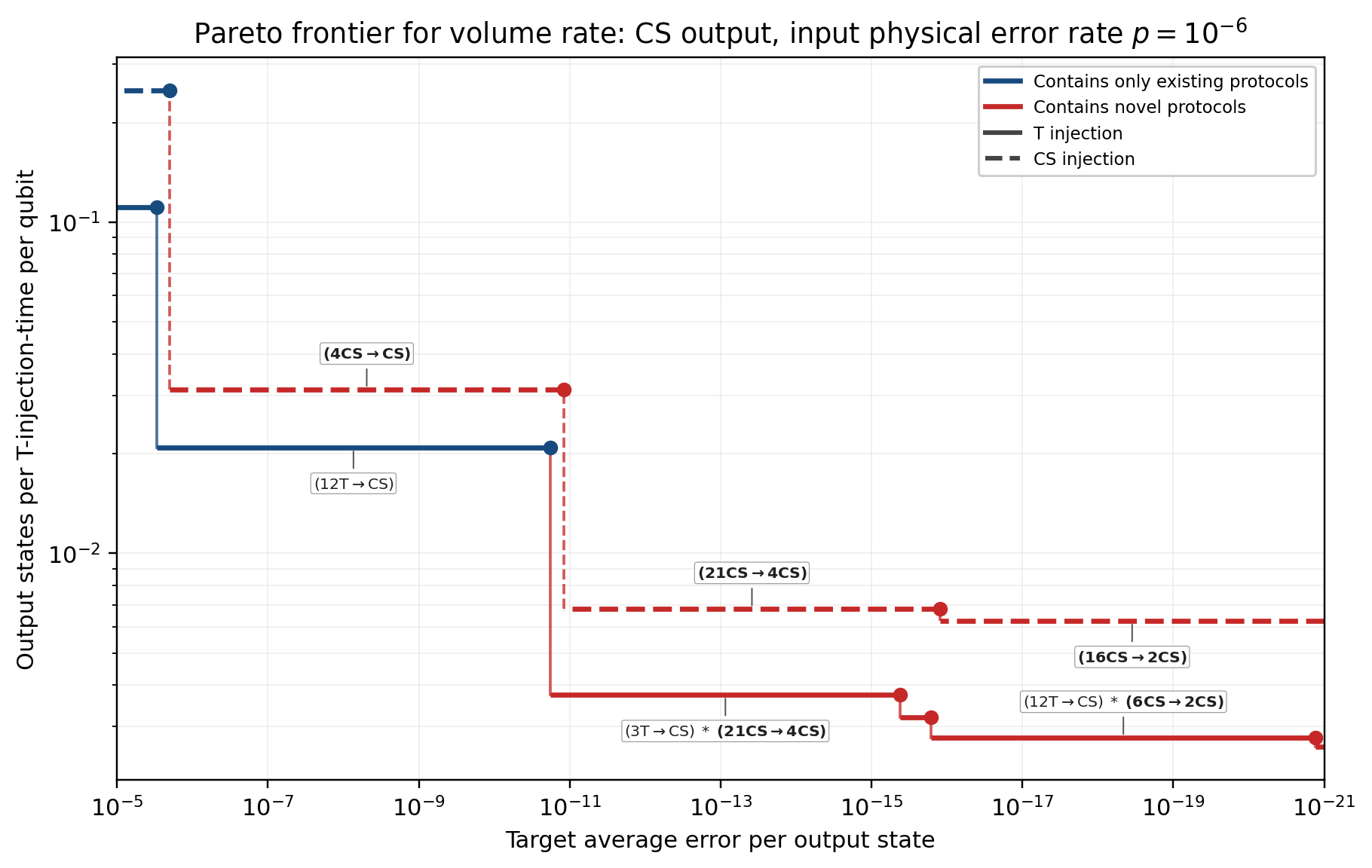}
    \par\smallskip
    \textup{(b) CS distillation: optimal for volumetric rate}
  \end{minipage}

  \par\medskip

  \begin{minipage}[t]{0.495\textwidth}
    \centering
    \includegraphics[width=\linewidth]{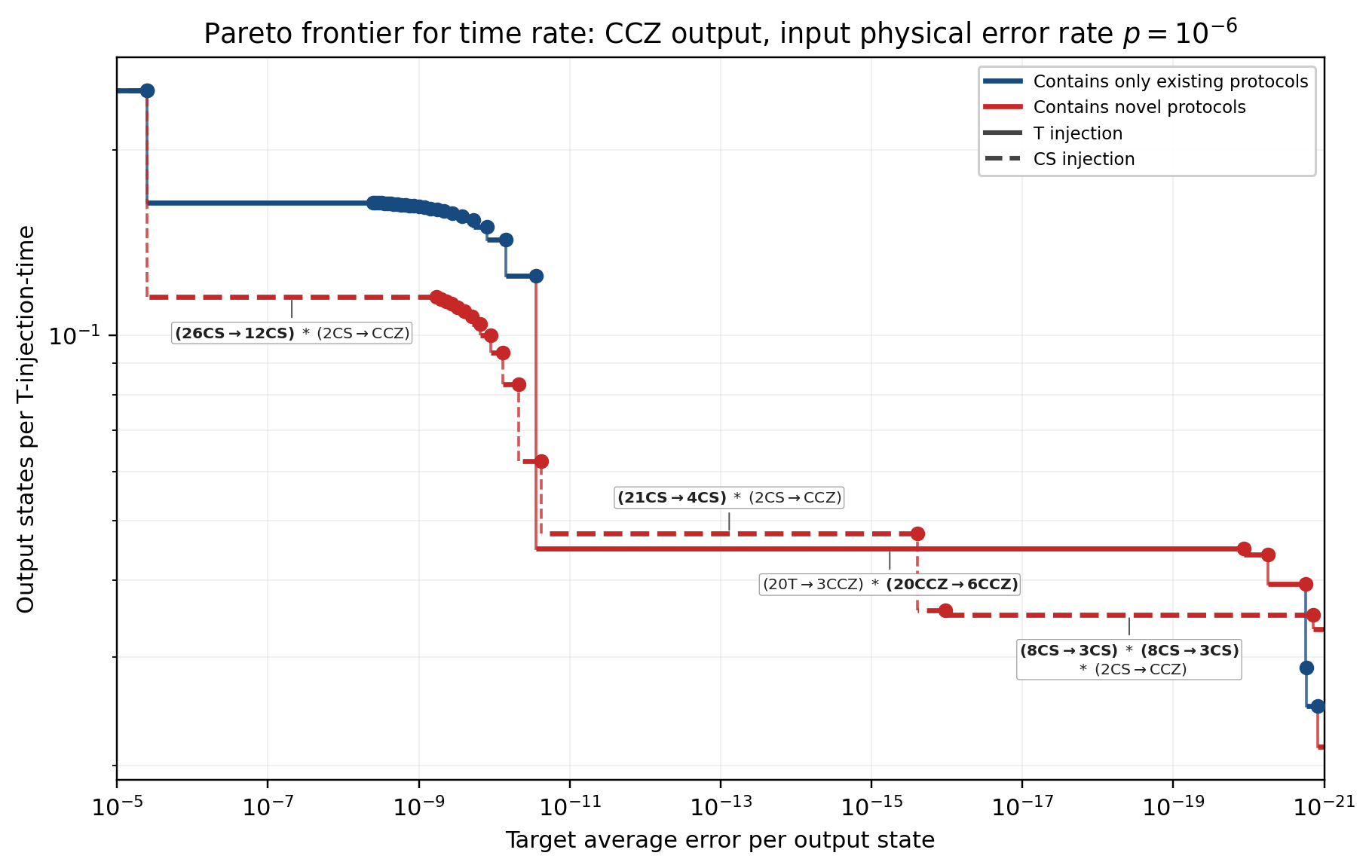}
    \par\smallskip
    \textup{(c) CCZ distillation: optimal for time rate}
  \end{minipage}%
  \hfill
  \begin{minipage}[t]{0.495\textwidth}
    \centering
    \includegraphics[width=\linewidth]{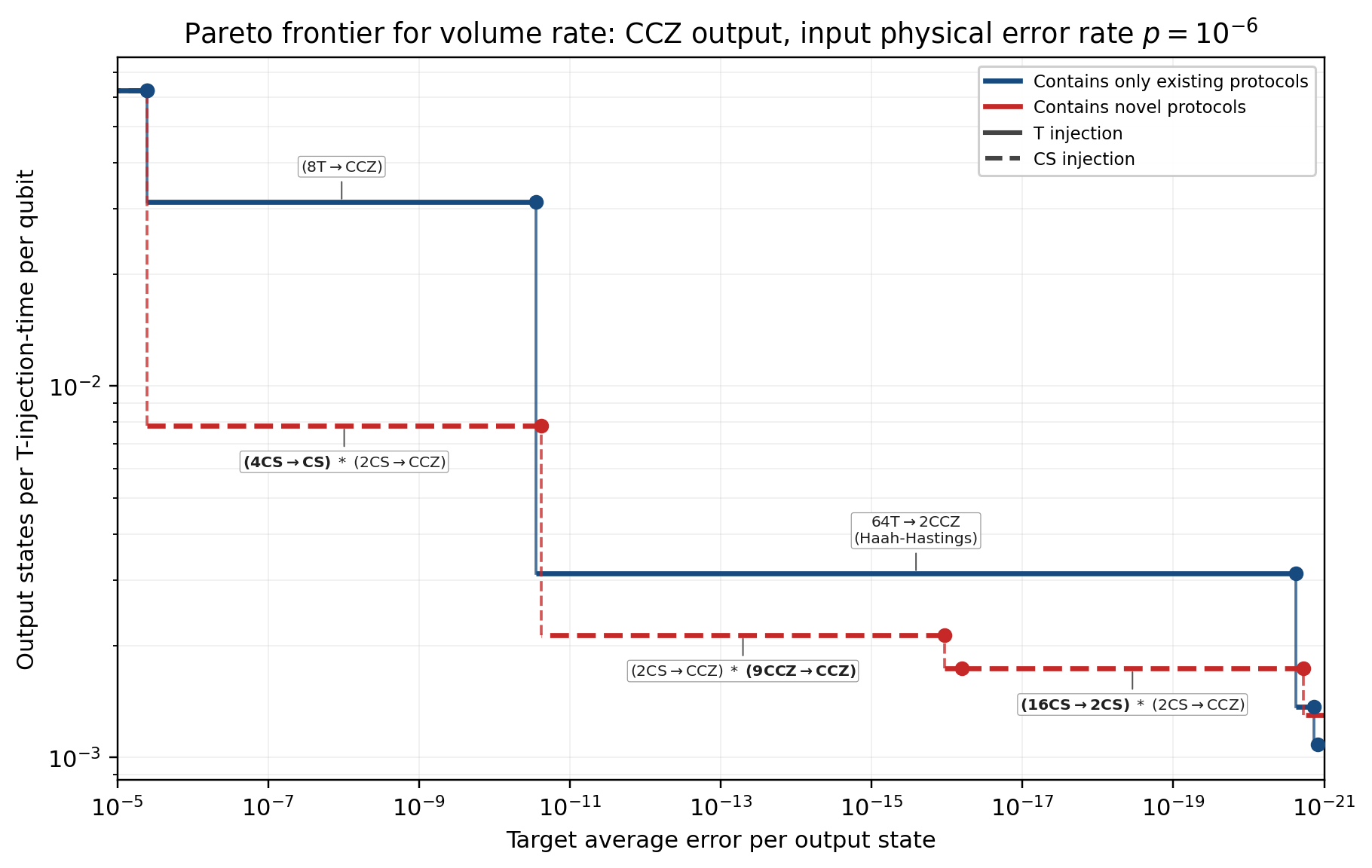}
    \par\smallskip
    \textup{(d) CCZ distillation: optimal for volumetric rate}
  \end{minipage}

  \caption{Pareto frontiers for physical input error
  $p=10^{-6}$ (modeling some small amount of cultivation pre-injection). Solid curves use $\ket{\T}$ state injection and dashed curves use
  $\ket{\CS}$ state injection. Red lines indicate that our novel protocols are used somewhere in the competitive scheme; blue lines indicate that they are not. Our novel protocols are bolded. Statistics for some of the competitive protocols are given in Appendix~\ref{subsec:tables_footprints_errors}.}
  \label{fig:pareto-grid-p1e-6}
\end{figure*}

For each of the $8$ distillation tasks encoded in the above $3$ choices, we aim to construct the best distillation protocols according to some metric, considering concatenated chains of both our novel protocols, and the existing protocols. We may make different choices for the success metric, as follows:
\begin{itemize}
    \item \textbf{Success metric}: Either \textit{time rate}: average retained number of output magic states of the target fidelity per unit time, or \textit{logical volumetric rate} (or simply \textit{volume rate}): average retained number of output magic states of the target fidelity per unit logical spacetime footprint. Our calculations of time and logical spacetime volume will be elucidated below, and explained in detail in App.~\ref{sec:time_footprint_calculations}.
\end{itemize}
Ultimately, we will find that our novel protocols will feature extensively in the optimal frontier of protocols, for almost every choice of injected state, input error rate, outputted state, and success metric.

We assume that magic state distillation protocols take place on logical qubits. Note that, when setting up magic state distillation protocols in an architecture, several design choices may be made, and we find that these different choices compound significantly when considering concatenated schemes. However, we aim to make choices that capture many of the popular architectures for realistic fault-tolerant quantum computation. 

To this end, we consider the implementation of magic state distillation protocols in an abstract Pauli-based computation (PBC) model as detailed in \cite{Litinski_2019}. This approach allows us to define a relatively architecture-agnostic cost model for implementation of magic state distilleries in logical qubits, be they in surface codes \cite{GOSC}, bivariate bicycle codes \cite{yoder2025tourgrossmodularquantum}, homological product codes \cite{Xu2025}, or general qLDPC architectures \cite{he2025extractorsqldpcarchitecturesefficient}. We assume magic state distillation protocols are executed on logical qubits, using a gate set of multi-qubit Pauli measurements between logical qubits and optionally undistilled $\ket{\T}$ states.
Techniques for protocol implementation using PBC \cite{Litinski_2019} readily generalize \cite{beverland2020lower_bounds} to non-$\ket{\T}$ multi-qubit input and output states using diagonal gate injection. In Appendix~\ref{sec:building_distillation_circuits}, we extend these techniques to the Galois qudit codes considered in this paper.

In large quantum computing architectures targeting extremely low error rates, magic state factories are expected to dominate the overall cost of the computation \cite{Lee_2021}. This motivates dedicating a large number of logical qubits in the computer to magic state production, meaning output states per unit spacetime volume is the natural cost metric. However, in early fault-tolerant quantum architectures, logical error rates are less stringent (e.g., \cite{gidney2025factor2048bitrsa}), but space constraints are more severe. In this setting, architectures \cite{yoder2025tourgrossmodularquantum} consider allocating a single magic state factory instance within a module. Then, footprint becomes more of a measure of feasibility (does the factory fit within a single module, yes or no?), and the relevant measure of cost is time. Our analysis considers both situations and can therefore be used to evaluate both large-scale and small-scale FTQC settings. To these ends, we elucidate the calculation of our cost metrics; full explanations are deferred to App.~\ref{sec:time_footprint_calculations}.
\begin{description}
\item[Time:] We consider a simplified timing model for which a unit of time is one $\ket{\T}$ state injection. We imagine that $\ket{\T}$ states are always consumed sequentially: an implementation choice. Our model neglects additional overhead from Clifford synthesis and routing deliberately, since these concerns are architecture-specific. For the concatenated distillation protocols involving $\ket{\CS}$ state injections, one state injection is deemed to take two units of time. The time metric we will calculate will be a rate: the average number of retained output magic states of the target fidelity per unit time. Ignoring the possibility of postselection, this is simply the number of magic states produced by a concatenated protocol divided by the number of inputs (divided by two in the case of $\ket{\CS}$ injection). We account for the overhead from repeating the protocol due to postselection in our analysis, although the corresponding effect on the time rate is generally small.

\item[Spatial Footprint:] In a small, early fault-tolerant quantum computer, or in a modular architecture with a limited number of qubits per module, it is preferable to implement protocols on a small number of logical qubits. Typically, this is the number of rows in the binary matrix describing the protocol~\cite{GOSC}.\footnote{The number of rows is equal to the number of qubit $X$ logicals and $X$ stabilizer generators describing the corresponding code.} For our $\mathbb{F}_4$ and $\mathbb{F}_8$-Galois qudit codes describing $\CS$-to-$\CS$ distillation and $\CCZ$-to-$\CCZ$ distillation, respectively, we show in App.~\ref{sec:time_footprint_calculations} that Litinski's ideas~\cite{GOSC} carry over naturally, and that the protocols may be run on a number of logical Galois qudits equal to the number of rows in the $\mathbb{F}_4$ or $\mathbb{F}_8$ matrix. The corresponding logical qubit footprint is then twice, or three times this, respectively. We further consider a compression technique~\cite{xu2026distillingmagicstatesbicycle}, which in certain special cases allows the protocol to be run in even smaller space, although this applies to a minority of the protocols considered.

\item[Spacetime Volume:] Simply defined as the product of time and compressed footprint, volumetric considerations are natural when scaling production across many parallel factory instances. The metric calculated will again be a rate: the average number of retained output magic states of the target fidelity per unit spacetime volume.
\end{description}
We add a final remark on the error rate calculations.
\begin{description}
\item[Output error rate:] Since our protocols generally give more than one magic state output, the natural error rate to calculate is average error rate per-output. This differs from the more commonly quoted error rate of a distillation protocol, which is the probability of any error. However, note that because a failing protocol typically has highly correlated errors, the average error rate per-output is usually much higher than the probability of any error divided by the number of outputs. We go into much more detail on these calculations in App.~\ref{sec:time_footprint_calculations}.
\end{description}
Now turning to our results, we define a concatenated protocol to be competitive if, for some distillation task, for some metric, for some target output error rate, it is optimal. For these distillation tasks, the protocols included in our survey are as follows; this is intended to be comprehensive over the protocols that could be competitive for the present tasks. We consider concatenations of the following:
\begin{itemize}
    \item Our novel protocols: Our $\CCZ$-to-$\CCZ$ distillation protocols --- see Section~\ref{sec:CCZ_distillation_protocols} --- based on $\mathbb{F}_8$-Galois qudit codes, our $\CS$-to-$\CS$ distillation protocols --- see Section~\ref{sec:CS_distillation_protocols} --- based on $\mathbb{F}_4$-Galois qudit codes, our novel $64\T$-to-$2\CCZ$ $(d=4)$ protocol\footnote{Note that our novel distance-$4$ $64\T$-to-$2\CCZ$ protocol has the same parameters as the existing Haah-Hastings protocol~\cite{Haah2018}. However, our novel protocol has a lower leading-order total error rate: $2720p^4$ versus $2944p^4$ (albeit an equal per-state average error rate of $2368p^4$). Unfortunately, our novel protocol also requires greater space footprint, and thus does not appear in any competitive protocols.} --- see Section~\ref{sec:trace_codes}. Note that this paper contains many more novel protocols than these; we are only considering the ones here that could be competitive for the present tasks.
    \item Existing protocols: Synthillation protocols~\cite{campbell2017unified}, including the $(7k+4 + (k \text{ mod } 2))\T \to k\CS$ and the $(6k+2)\T \to k \CCZ$ protocols, both at distance 2, the well-known $15\T\to\T$ protocol at distance $3$~\cite{MSD}, $(3k+8)\T\to k\T$ protocols, for even $k$, at distance $2$~\cite{bravyi2012magic}, the Haah-Hastings $64\T\to2\CCZ$ protocol at distance $4$~\cite{Haah2018}, and the known synthesis (non-distillation) conversions of $3\T\to\CS$~\cite{howard2017application}, and $2\CS\to\CCZ$~\cite{beverland2020lower_bounds}.
\end{itemize}

We further comment that, in finding the competitive protocols, we restrict ourselves to concatenated distance at most $8$, and spatial footprint at most $75$.\footnote{Competitive protocols for time rate tend to be larger, as better rates can be achieved at larger sizes, especially using our novel protocol families. We cap the allowed spatial footprint at $75$ to model a situation in which a quantum architecture sacrifices a moderate number, in this case $75$, of logical qubits, for the given distillation task.} The resulting Pareto frontiers are shown in Figures~\ref{fig:pareto-grid-p1e-3} and~\ref{fig:pareto-grid-p1e-6}. We find that our novel protocols feature heavily in the competitive schemes at the target error rates of interest; the only exception is when considering $\ket{\T}$ state injection at physical error rate $10^{-6}$ (assuming some small amount of cultivation), and where the metric for success is volumetric rate. In this case, either the $8\T\to\CCZ$ synthillation protocol~\cite{campbell2017unified}, or the $64\T\to2\CCZ$ Haah-Hastings protocol~\cite{Haah2018} are found to be the competitive protocols.

\subsection{Open Questions}

\begin{itemize}
    \item Can one design larger codes \emph{maximally} satisfying the twisted three-orthogonality condition of Eq.~\ref{eq:U7_orthogonal} for distilling the single-qudit $U_7$ gates~\cite{wills2024}? This is particularly relevant since the $U_7$ gate is equivalent to the widely-considered qubit $\CCZ$ gate over the field $\mathbb{F}_8$; see Section~\ref{sec:U7_gate}.
    Similar twisted products also appear when considering trace codes of AG codes. In this work, for this problem we develop a numerical approach, which is to build a forbidden hypergraph, and solve for its maximum independent set (MIS); see Construction~\ref{con:64CCZ-to-8CCZ} and~\ref{con:24CCZ-to-4CCZ-d3}. However, MIS is an NP-hard problem, and so this approach does not scale. 
    
    Even in the small-codes examples we developed in this work, we have to restrict ourselves to codewords coming from monomial evaluations for finding $\CCZ$-to-$\CCZ$ ($U_7$ in $\F_8$) distillation protocols. A better understanding, or more geometrical picture of the twisted three-orthogonality condition, might help in generalizing our constructions to larger lengths and alphabet sizes. Such an understanding might help in establishing the best asymptotic achievable rate and distance of $U_7$ distillation protocols. Moreover, once we have subspaces satisfying the twisted three-orthogonality condition over large alphabets, can one take the trace~\cite{cohen2025tracing}, subfield subcode~\cite{DELSARTE1970403}, or the subspace code~\cite{hattori1998}, to obtain codes in a smaller alphabet (since $U_7$ is simpler there)?
    
    \item Can two- and multipoint divisors
    ~\cite{matthews2001weierstrass,duursma2011improved,cervia2025magic}, and improved order bounds~\cite{BEELEN2007665} improve the parameters of our distillation protocols? Can other algebraic geometry objects, such as curves of positive-defect, or higher-dimensional varieties~\cite{HANSEN2001530} help with the code design?
    
    \item For real-phase qudit gates, we have only classified their Clifford hierarchy level, but not classified them up to Clifford equivalence; see Section~\ref{subsubsec:classification}. For example, in Sec.~\ref{sec:U7_gate}, we leave the question open whether the three possible $U_7^\beta$ gates in $\F_{2^s}$, where $3\mid s$, are Clifford inequivalent in general. For non-real-phase gates, due to the difference between $\F_{2^s}$ and $\Z_{2^s}$, we find it difficult to even classify the Clifford hierarchy levels of the corresponding gates. This is problematic, since if one wants to distill some multi-qubit magic gate, one naturally tries to pack such a multi-qubit gate into a Galois qudit gate, and derive the corresponding orthogonality condition for transversality. A better understanding of the non-real-phase magic gate would make this process easier.
    
    \item For our CCZ and CS-input protocols, we are mostly enforcing the logical action to be disjoint CCZ and CS (except Con.~\ref{con:8cstoTOF} and Con.~\ref{con:16cs1TOF}). However, hypergraph state outputs also provide an interesting design space. Can they be directly used for some algorithmic subroutines (for example, we show that qudit-CCZ gate can be used to perform multiplication in the corresponding binary fields)?
    If not, how can one extract relevant simple-enough hypergraph states (e.g. $\TOF\#$) from a complicated hypergraph magic state? We gave such examples in Con.~\ref{con:64T_to_2TOF} \& 
    \ref{con:64ttof4qudit}), but a systematic understanding is lacking, except for extracting disjoint CCZs from a CCZ-only circuit, which was shown to be equivalent to the subrank problem of trilinear forms~\cite{Menon_2026}.

    \item How can one implement the distillation factories we found most efficiently in a broader architecture? Especially, how can one compile these protocols at the logical level in relevant qLDPC codes? One might need to optimize the weight of our protocols, co-optimize $\F_2$-bases and qLDPC code logical bases to reduce surgery overhead.

    \item The distillation protocols constructed in this work are mostly algebraically motivated, and we know little about their optimality. Recent SAT-solver~\cite{jacinto2026exploring} or symmetry-based protocol discovery approaches~\cite{singh2026borrowed} might help either certifying the optimality, or find smaller factories.
\end{itemize}

\subsection{Organization}
We begin in Section~\ref{sec:prelims} by presenting preliminary material necessary for Sections~\ref{sec:galois} and~\ref{sec:protocols}. In Section~\ref{sec:galois}, we prove results on Galois qudits and their gates which will support our protocol constructions. A result in this section which lies outside our main contributions can be found in Section~\ref{subsubsec:classification}, where we give a classification of the diagonal gates in the Galois qudit Clifford hierarchy with real phases in a way that is native to those gates, resolving an open question in~\cite{he2025quantum}. Section~\ref{sec:protocols} then constructs our most important protocols. From there, Section~\ref{sec:AG_codes} constructs more exotic protocols that are outside our primary scope. These use more advanced algebraic geometric techniques, and we begin that section with the necessary preliminary material.

Appendix~\ref{sec:CS_transversality_condition} proves several facts around our distillation of $\ket{\CS}$ using $\mathbb{F}_4$ codes. Appendix~\ref{sec:algebraic_curves} then presents complete preliminaries on algebraic geometry codes in order to supplement Section~\ref{sec:AG_codes}. Appendix~\ref{sec:building_distillation_circuits} discusses how our Galois qudit magic state distillation protocols may be run most compactly in hardware, by showing how known tricks for the regular qubit factories transfer seamlessly to this case. Finally, Appendix~\ref{sec:time_footprint_calculations} goes into detail on our error rate and spacetime footprint calculations, finishing with tables of statistics on some of the competitive schemes.

The SageMath notebooks containing all the necessary tools to reproduce and verify our protocols are available online at \href{https://github.com/gongaa/MSD-binary-fields}{https://github.com/gongaa/MSD-binary-fields}.

\section{Preliminaries}\label{sec:prelims}
\subsection{Diagonal Gates in the Clifford hierarchy}\label{subsec:diagonal_gates_ch}

The Clifford hierarchy is a recursively defined hierarchy of quantum gates~\cite{gottesman1999demonstrating}, whose exact classification remains a major open problem. While the $k$-th level of the Clifford hierarchy does not form a group, the diagonal gates in the $k$-th level do form a group, and furthermore the diagonal gates in the Clifford hierarchy have been classified~\cite{cui2017diagonal}. We consider the following diagonal $n$-qubit gates (and their products):
\begin{equation}
    U_{m,\boldsymbol{a}} \coloneq \sum_{\bx \in \mathbb{F}_2^n}\exp\left(\frac{2\pi i}{2^m}x_1^{a_1}\ldots x_n^{a_n}\right)\ket{\bx}\bra{\bx},
\end{equation}
where $m$ is a positive integer and $\boldsymbol{a}=(a_1,\dots,a_n) \in \mathbb{F}_2^n$. One immediate example is $n = 2, m = 1$ and $\boldsymbol{a} = (1,1)$, which gives the $\CZ$ gate. $U_{m,\boldsymbol{a}}$ is in the Clifford hierarchy for $n$ qubits and its level is $(m-1)+\wt(\boldsymbol{a})$. Moreover, these gates can be used to generate the diagonal Clifford hierarchy. The magic state corresponding to $U_{m,\boldsymbol{a}}$, often denoted $\ket{U_{m,\boldsymbol{a}}}$, is defined as the gate $U_{m,\boldsymbol{a}}$ acting on $\ket{+}^{\otimes n}$. Using gate teleportation~\cite{gottesman1999demonstrating}\cite[Alg.~A.1]{beverland2020lower_bounds}, one may consume the magic state corresponding to $U_{m,\boldsymbol{a}}$ to execute the gate $U_{m,\boldsymbol{a}}$ on an arbitrary $n$-qubit state using only operations in lower levels of the Clifford hierarchy. In particular, a gate in the third level can be teleported using Clifford operations. 

The most famous gates (and corresponding states) of this kind are $T$, $\CS$ and $\CCZ$. More general $U_{m,\boldsymbol{a}}$ gates, and their products, can be described by ``phase polynomials''. For example, when $n=5$ and $m=1$, the phase polynomial $x_1x_2x_3 + x_3x_4x_5$ corresponds to the multiplication of two $\CCZ$ gates (overlapping on a single qubit), and is known as the $\TOF\#$ gate:
\begin{equation}
    \TOF\#_{12345} \coloneq \sum_{\bx \in \mathbb{F}_2^5}(-1)^{x_1x_2x_3+x_3x_4x_5}\ket{\bx}\bra{\bx}
\end{equation}
This gate is useful in many algorithmic subroutines, for example, in a log-depth adder~\cite{log_adder}. Additionally, $\TOF\#$ can also be used to execute a shared-controlled-SWAP; indeed, $\TOF\#_{12345} = \CCZ_{123}\CCZ_{145}$ is Clifford-equivalent to $C_1$-$(\text{SWAP}_{23}\text{SWAP}_{45})$ via
\begin{equation}
    C_1\text{-}(\text{SWAP}_{23}\text{SWAP}_{45}) = \text{CNOT}_{23}H_3\text{CNOT}_{45}H_5\cdot(\TOF\#_{12345})\cdot H_5\text{CNOT}_{45}H_3\text{CNOT}_{23}.  
\end{equation}
The shared-controlled swap is prevalent in fermionic Hamiltonian simulation~\cite{babbush2018encoding}, and other algorithmic subroutines~\cite{Childs_2003}.

More generally than the above, one can consider the notion of Clifford equivalence:

\begin{definition}[Clifford Equivalence]
\label{def:Clifford_equiv}
We say two unitary gates $U$ and $V$ are \emph{Clifford equivalent}, if there exist Clifford unitaries $C$ and $C'$, such that $U=CVC'$. 
\end{definition}
It is possible to exhaustively classify (at least at small sizes) classes of phase polynomial gates up to Clifford equivalence. For example, considering only gates formed from products of $\CCZ$ gates, one can check that there is only one such gate up to Clifford equivalence on four qubits, and it is Clifford equivalent to a single $\CCZ$ gate. Further, on five qubits, any product of $\CCZ$ gates is Clifford equivalent to either a single $\CCZ$ gate, or the $\TOF\#$ gate. Such examples are illustrated in Figure~\ref{fig:Clifford_equiv}.

\begin{figure}[ht]
\centering
\includegraphics[width=0.9\linewidth]{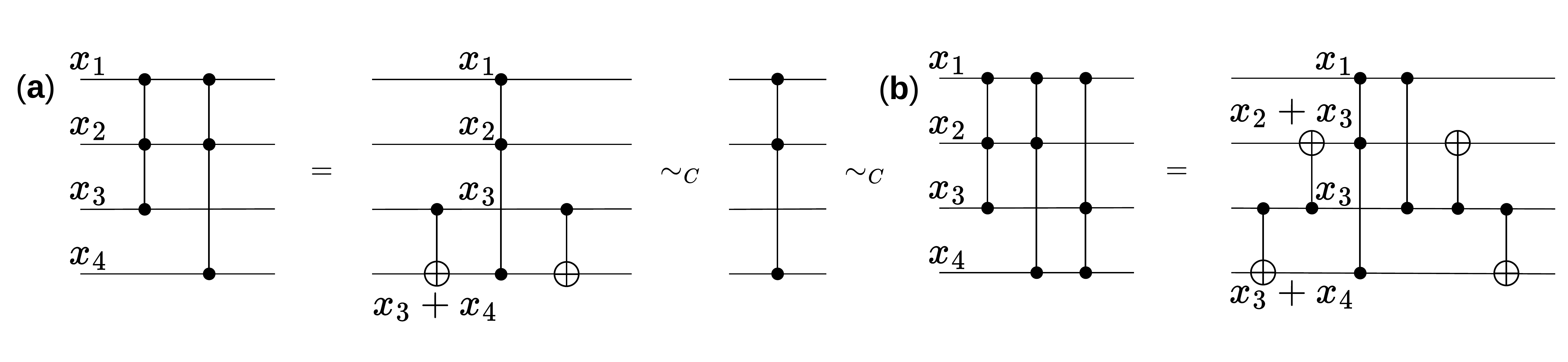}
\caption{Examples of Clifford equivalence. (a) and (b) are both Clifford equivalent to a single CCZ gate. (a) $x_1x_2x_3+x_1x_2x_4=x_1x_2(x_3+x_4)$. (b) $x_1x_2x_3+x_1x_2x_4+x_1x_3x_4=x_1(x_2+x_3)(x_3+x_4)+x_1x_3$.}
\label{fig:Clifford_equiv}
\end{figure}
Relatedly to Clifford equivalence, it is useful to establish the most efficient means to synthesize various gates from others. For example, the ancilla-free $T$ count for a diagonal gate in the third level of the Clifford hierarchy is defined in~\cite{campbell2017unified} to be the \emph{minimum} number of $T$ gates required to synthesize the gate unitarily with $T$ gates, $\CNOT$ gates, and $S$ gates, and using no ancillas. The ancilla-free $T$ count of $\CCZ$ is $7$. Such transformations can be discovered by manipulating phase polynomials. For example, considering variables $x_i \in \{0,1\}$, one verifies the identity
\begin{equation}
    4x_1x_2x_3\equiv x_1+x_2+x_3-(x_1\oplus x_2)-(x_1\oplus x_3)-(x_2\oplus x_3)+(x_1\oplus x_2 \oplus x_3),
\end{equation}
where $+$ denotes regular integer addition, and $\oplus$ denotes integer addition modulo $2$.\footnote{The minimality of the various ancilla-free $T$ counts is established by a relation to the decoding of a Reed-Muller code~\cite{T_count_RM}.} It follows from this identity that the $\CCZ$ state (or gate) can be unitarily synthesized using $7$ $T$ gates, as well as $\CNOT$ operations, by considering various $U_{m,\boldsymbol{a}}$ gates with $n=3$ and $m=3$. Because $m>1$, we must differentiate between regular integer addition and addition modulo $2$. Explicitly, we have
\begin{equation}
    \exp\left(\frac{2\pi i}{2^m}(x_1+x_2)\right) \neq \exp\left(\frac{2\pi i}{2^m}(x_1\oplus x_2)\right)\text{ for } m>1.
\end{equation}
This is unlike when we consider products of $\CCZ$ gates like the $\TOF\#$ gate, in which case addition in the phase polynomial is both regular integer addition and modulo $2$ addition. For example, for the $\TOF\#$ gate,
\begin{equation}
    \exp\left(\pi i(x_1x_2x_3+x_3x_4x_5)\right) = \exp\left(\pi i(x_1x_2x_3\oplus x_3x_4x_5)\right).
\end{equation}
It is also found in~\cite{campbell2017unified} that the ancilla-free $T$ count of the $\TOF\#$ gate mentioned above is $11$. We note that this is less than twice the ancilla-free $T$ count of twice that of $\CCZ$, showing that it is more efficient to directly synthesize $\TOF\#$ than to synthesize two $\CCZ$ gates separately, at least in this setting. Synthillation techniques~\cite{campbell2017unified} also allow one to distill $12$ $T$ gates to $1$ $\TOF\#$ gate at distance $2$ (see App.C for the explicit example).

Hadamard gates and measurement-and-feedback add more power over merely using CNOT+$T$. The most famous example is Jones' construction of a $\CCZ$ gate from four $T$ gates~\cite{jones2013low}. An interpretation~\cite{beverland2020lower_bounds} is that, applying $HSH=\sqrt{X}$ on the third qubit on the $|\CCZ\;\CS_{1,2}\ra=\CCZ\;\CS_{12}|\text{+++}\ra$ state (can be unitarily synthesized using CNOTs and four $T$) leads to the $|\CCZ\ra$ state~\cite[Fig.~12(b)]{beverland2020lower_bounds}, which can be further teleported to realize a $\CCZ$ \emph{gate}. As another example, while unitary synthesis of a $\CCZ$ gate (using CNOT+$\CS$) requires $3$ $\CS$ gates, only $2$ $\CS$ gates are required for the $|\CCZ\ra$ state~\cite[Fig.~13]{beverland2020lower_bounds}.

\subsection{Introduction to rings, finite fields and Reed-Muller codes}
\label{sec:ring_field}
We now give a brief introduction to rings and finite fields.
The following definitions and facts can be found in many textbooks such as \cite{Lidl_Niederreiter_1996}; we thus state the most relevant ones without proof.

A \emph{ring} is a set $R$ equipped with two binary operations $+$ (addition) and $\cdot$ (multiplication) such that: $R$ is an abelian group with respect to $+$; $\cdot$ is associative, that is, $(a\cdot b)\cdot c=a\cdot (b\cdot c),\;\forall a,b,c\in R$; multiplication is distributive over addition, $a\cdot (b+c)=a\cdot b+a\cdot c$ and $(b+c)\cdot a=b\cdot a+c\cdot a$.
In this work, we only work with commutative rings (those for which $a\cdot b = b\cdot a$ for all $a,b \in R$) with a (multiplicative) identity (those containing an element called $1$ satisfying $1 \cdot a = a$ for all $a \in R$).

Given a positive integer $n$ and $r \in R$, we write $nr$ for the element of $R$ obtained by summing $r$ $n$ times. If, for a ring $R$, there exists a positive integer $n$ such that $nr=0$ for every $r\in R$, then the least such positive integer $n$ is called the \emph{characteristic} of $R$.

In a ring, multiplicative inverses are not required to exist. A ring in which every nonzero element has a multiplicative inverse is called a \emph{field}. A subset of fields are integral domains, which are integral domains. These are rings in which there are no zero divisors, meaning that if $ab = 0$, then $a = 0$ or $b=0$. Given an integral domain $R$, one may construct a field from it called its fraction field, $\text{Frac}(R)$, which contains a copy of $R$. The elements of the fraction field are the equivalence classes of $R \times (R \setminus \{0\})$, under the equivalence relation $(a,b) \sim (c,d) \iff ad = bc$. Canonically, the equivalence class of $(a,b)$ is denoted $\frac{a}{b}$, and the operations of addition and multiplication are defined via $\frac{a}{b} + \frac{c}{d} = \frac{ad+bc}{bc}$, and $\frac{a}{b}\cdot \frac{c}{d} = \frac{ac}{bd}$.

An example of a field is $\F_p=\{0,1,\dots,p-1\}$ for a prime $p$, where the operations are achieved by arithmetic modulo $p$. More generally, a finite field (also known as a Galois field) $K$ must have $q=p^s$ elements for some $s\in\N$, where the prime $p:=\text{char}\; K$ is the characteristic of $K$, and $s$ is the \emph{extension degree} of $K$ over $\F_p$. For all primes $p$ and $s\in \N$, $\F_{p^s}$ exists and it is unique up to isomorphism. It is canonical to denote the corresponding $K$ as $\F_{q}$ or alternatively $\GF(q)$.
The \emph{multiplicative group} $\F_q^\times$ of nonzero elements of $\F_q$ is cyclic, meaning that there exists at least one \emph{primitive element} $\alpha$ such that $\F_q=\{0,\alpha,\alpha^2,\dots,\alpha^{q-1}=1\}$. 
Since any power of $1$ is $1$ itself, one has $\gamma^{q-1}=\begin{cases}1 & \gamma\in K^\times\\ 0 & \gamma=0\end{cases}$, and hence also $\gamma^q=\gamma,\;\forall \gamma\in \F_q$. $q$ is called the \emph{order} of the finite field $\F_q$.

A subset $F$ of $K$ that is itself a field is called a \emph{subfield} of $K$, and $K$ is called an \emph{extension field} of $F$.
It turns out that every \emph{subfield} of $\F_{p^{s'}}$ has order $p^{s}$, where $s$ is a positive integer dividing $s'$. Let $q=p^s$ and $m=s'/s$ in the following.

A field isomorphism is a bijective map from one field to another that preserves multiplication and addition. A field automorphism on a field $K$ is a field isomorphism from $K$ to itself. The group of automorphisms of $\F_{q^m}$ that fix every element of $\F_q$ ($q=p^s$) is also referred to as the \emph{Galois group} of $\F_{q^m}$ over $\F_q$, and it is generated cyclically by the \emph{Frobenius transform}
$\sigma_1:\alpha\mapsto \alpha^q$ ($\alpha\in \F_{q^m}$). To be more explicit, we can write the elements of the Galois group as the mappings $\sigma_0,\sigma_1,\dots,\sigma_{m-1}$ defined by $\sigma_j(\alpha)=\alpha^{q^j}$ for $\alpha\in\F_{q^m}$. The images $\sigma_j(\alpha)$ are called the Galois \emph{conjugates} of $\alpha$. Note that indeed $\sigma_j(\gamma)=\gamma$ for $\gamma\in \F_q$.

The \emph{trace} $\tr_{\F_{q^m}/\F_q}(\alpha)$ of $\alpha\in \F_{q^m}$ is the sum of the conjugates of $\alpha$, i.e.,
\begin{equation}
\label{eq:def_trace}
\tr_{\F_{q^m}/\F_q}(\alpha)=\alpha +\alpha^q+\alpha^{q^2}+\dots+\alpha^{q^{m-1}}.
\end{equation}
The trace maps into $\mathbb{F}_q$ surjectively.
The trace function $\tr_{\F_{q^m}/\F_q}$ is $\F_q$-linear,
\begin{equation}
\label{eq:trace}
\tr_{\F_{q^m}/\F_q}(\gamma_1 \alpha_1+\gamma_2\alpha_2)=\gamma_1\tr_{\F_{q^m}/\F_q}(\alpha_1) + \gamma_2\tr_{\F_{q^m}/\F_q}(\alpha_2)
,\;\forall \gamma_1,\gamma_2\in \F_q\text{ and } \forall \alpha_1,\alpha_2\in \F_{q^m}\;,
\end{equation}
where the additivity follows from the fact that $(a+b)^p=a^p+b^p$ (and hence $(a+b)^{p^i}=a^{p^i}+b^{p^i}$) for $a,b\in \F_{p^s}$ since $\binom{p}{i}=\frac{p!}{i!(p-i)!}\equiv 0\mod p$ for $i=1,2\dots,p-1$.
We have $\tr_{\F_{q^m}/\F_q}(\alpha^q)=\tr_{\F_{q^m}/\F_q}(\alpha)$. A further, less trivial fact, is that $\tr_{\F_{q^m}/\F_q}(\alpha)=0$ if and only if $\alpha=\beta^q-\beta$ for some $\beta\in F_{q^m}$.

The \emph{norm} $\Nm_{\F_{q^m}/\F_q}$ of $\alpha$ is the product of of all its conjugates:
\begin{equation}
\label{eq:norm}
\Nm_{\F_{q^m}/\F_q}(\alpha)=\alpha\cdot \alpha^q\cdot\cdots\cdot \alpha^{q^{m-1}}=\alpha^{(q^m-1)/(q-1)}\;.
\end{equation}
The norm function is multiplicative: $\Nm_{\F_{q^m}/\F_q}(\alpha\beta)=\Nm_{\F_{q^m}/\F_q}(\alpha)\Nm_{\F_{q^m}/\F_q}(\beta),\;\forall\alpha,\beta\in \F_{q^m}$. As for the trace, $\Nm_{\F_{q^m}/\F_q}(\alpha^q)=\Nm_{\F_{q^m}/\F_q}(\alpha)$, and the norm function maps into $\mathbb{F}_q$ surjectively.

For a general ring (not necessarily a field), an ideal $I$ of a ring $R$ is a subset of $R$ that is itself a ring and, moreover, $\forall a\in I$ and $r\in R$ we have $ar=ra\in I$. We will only be working with finitely generated ideals. Such ideals are formed by considering $\{a_1,\dots,a_k\}$ a subset of $R$, and the ideal generated by them is $(a_1,\dots,a_k):=\{a_1r_1+\dots a_kr_k\;|\;\forall r_1,\dots,r_k\in R\}$. We will sometimes use $\la a_1,\dots,a_k\ra$ to denote the ideal generated by $\{a_1, \ldots, a_k\}$.

$R/I$ is the \emph{quotient ring} of $R$ by the ideal $I$. Its addition and multiplication are given by
$(r+I) + (s+I) = (r+s) + I$, $(r+I) \cdot (s+I) = rs + I$. In the quotient ring when $I=\la a_1,\dots,a_k\ra$, it is convenient to write $R/I = R/(a_1, \ldots, a_k)$. In this quotient ring, one can identify $a_1=0,\dots,a_k=0$.

In coding theory, e.g., when defining the Reed-Muller codes, we will frequently deal with \emph{polynomial rings} and their quotient rings. Letting $x$ be an indeterminate variable, the polynomial ring $R[x]$ is the set of formal sums $a_n x^n +\cdots + a_1x+ a_0$ under usual addition and multiplication. 
Multivariate polynomial rings are defined inductively by $R[x_1,x_2,\dots,x_n]:=R[x_1,x_2,\dots,x_{n-1}][x_n]$. In the case that $R=K$ for some field $K$ of characteristic $p$, we have that $p$ times anything equals zero.

To define Reed-Muller codes, we consider the quotient ring
$$\F_q[x_1,\dots,x_m]/(x_1^q-x_1,\dots,x_m^q-x_m),\quad q=2^s.$$ The following definitions and facts about Reed-Muller codes may be found in~\cite{theoryEC}. The $q$-ary Reed-Muller code in $m$ variables of degree $r$ is denoted $RM_q(r,m)$, and is the length-$q^m$ code over $\mathbb{F}_q$ formed by evaluating all polynomials in the above ring of degree at most $r$ at all points $(x_1, x_2, \ldots, x_m) \in \mathbb{F}_q^m$. Because $\eta^q = \eta$ for all $\eta \in \mathbb{F}_q$, quotienting by the terms $(x_i^q-x_i)$ in the above ring allows us to identify the given polynomials with the corresponding codewords of the Reed-Muller code in a one-to-one fashion. For example, note that the degree of any individual variable in any polynomial in the quotient ring does not exceed $q-1$.

For Reed-Muller codes, a particularly important example for us will be the binary Reed-Muller code with $q=2$, defined over the quotient ring
$$\mathbb{F}_2[x_1, \ldots, x_m]/(x_1^2-x_1, \ldots, x_m^2-x_m).$$ The distance of $\RM_2(r,m)$ turns out to be $2^{m-r}$; indeed, one may quickly see that it must be at most this by considering the degree-$r$ polynomial $x_1x_2\ldots x_r$, whose evaluation has weight $2^{m-r}$. We may denote by $\ev(f)$ the codeword corresponding to the polynomial $f$. Then we have that the coordinate-wise multiplication of $\ev(f)$ and $\ev(g)$ is $\ev(fg)$. More generally, one can see that the weight of a codeword corresponding to the evaluation of a monomial formed from the product of $v$ variables is $2^{m-v}$. The code dual to $\RM_2(r,m)$ is $\RM_2(m-r-1)$.

Analogous facts hold more generally for $q$-ary Reed-Muller codes. For example, the distance of $\RM_q(r,m)$ is $(\mu+1)q^\nu$, if we write $(q-1)m -r = \nu(q-1)+\mu$ for $0 \leq \mu < q-1$. The coordinate-wise multiplication of $\ev(f)$ and $\ev(g)$ is $\ev(fg)$ in the general $q$-ary case, also. 

More general codes, called \textit{affine monomial codes}, can be considered by evaluating different sets of polynomials (than those of degree at most $r$) at the points of $\mathbb{F}_q^m$. A key example for us will be the hyperbolic codes. These can be obtained from Reed-Muller codes by deleting parity checks, increasing the code dimension while keeping the same minimum distance~\cite[Ch.~4.4]{AG_codes_handbook}. Concretely, a basis of the hyperbolic code of designed distance $d$ is formed by the monomials $X_1^{\alpha_1}\ldots X_m^{\alpha_m}$ where $\prod_{i=1}^m(\alpha_i+1)<d$. A more general class of affine monomial codes will be considered in App.~\ref{sec:decreasing_monomial_code}.

\subsection{Affine and projective space}
\label{sec:affine_projective_space}
The \emph{affine space} over the field $K$, is simply the vector space $K^n$ over $K$; this is usually denoted $\A_K^n$ or just $\A^n$ when the field $K$ is implicit. We write all its elements as $\A^n_K=\{(x_1,\dots,x_n)\;|\;x_i\in K\}$.
The \emph{projective space} over a field $K$, denoted as $\bbP_K^n$, or just $\bbP^n$ when $K$ is implicit, is defined by quotienting $K^{n+1}$ by scalar multiplication by elements of $K^\times = K\backslash\{0\}$. Explicitly, $$\bbP^n=(K^{n+1}\backslash\{0\})/\sim,$$ where the equivalence relation $\sim$ on $K^{n+1}$ is defined as $$(x_0:\dots:x_n)\sim(\lambda x_0:\dots:\lambda x_n)\quad\forall\;\lambda\in K^\times.$$

Take $\bbP^2_{\F_4}$ as an example, one can enumerate all its elements in the following way: $\{(1:y:z)\;|\;y,z\in \F_4\}\cup \{(0:1:z)\;|\;z\in \F_4\}\cup \{(0:0:1)\}$.

Let $U_i\subset \bbP^n$ be the subset of points $(X_0:\dots:X_n)$ with $X_i\neq 0$. Then on $U_i$ the ratios $x_j=X_j/X_i$ are well-defined and give a bijection $U_i\cong \A^n$. 
Notice that $U_0,\dots, U_n$ together cover $\bbP^n$.
We will later refer to $U_i$ as the $X_i=1$ affine chart covering $\bbP^n$.

A homogeneous polynomial is a polynomial for which all terms have a common degree. An example would be
\begin{equation*}
    f^h(x_0,x_1,x_2) = x_0^2x_1x_2 + x_1^3x_2,
\end{equation*}
which is a homogeneous polynomial of degree $4$. Notice that if $F$ is a homogeneous polynomial of degree $d$ then
\begin{equation}\label{eq:scalar_homogeneous}
    f^h(\lambda x_0, \lambda x_1, \ldots, \lambda x_n) = \lambda^d\cdot f^h(x_0, x_1, \ldots, x_n).
\end{equation}
Such functions are not in general well-defined functions on $\mathbb{P}^n$, but what is well-defined, because of Eq.~\eqref{eq:scalar_homogeneous}, is the zero locus of such functions. The zero locus of $F$ is
\begin{equation*}
    V(F) \coloneq \{(x_0:\dots:x_n) \in \mathbb{P}^n:f^h(x_0, x_1, \ldots, x_n) = 0\}.
\end{equation*}
Clearly, given any polynomial on the affine space $\mathbb{A}^n$, the polynomial is a genuinely well-defined function on $\mathbb{A}^n$, and it always makes sense to talk about the zero locus. Later, we will use capitalized variables like $X,Y,Z$ to emphasize that they are variables taking values in a projective space, see for example Construction~\ref{con:21CCZ-to-3CCZ}, but lower-case variables like $x,y$ to emphasize that they take values in affine space, see for example Construction~\ref{con:64CCZ-to-8CCZ}.

\section{\texorpdfstring{Galois Qudits over $\GF(2^s)$}{Galois Qudits over Binary Extension Fields}}\label{sec:galois}

\subsection{\texorpdfstring{Galois Qudits and their Implementation in Qubits}{Galois Qudits and their Implementation in Qubits}}\label{sec:intro_qudit_codes}

A Galois qudit over the finite field $\GF(2^s) = \F_{2^s}$ is a $2^s$-dimensional qudit with a choice of Pauli group deliberately made to encode the arithmetic of the finite field $\F_{2^s}$. In making this choice, the $2^s$-dimensional qudit is made equivalent to a set of $s$ qubits in its states, Pauli group and Clifford hierarchy~\cite{Gottesman,wills2026review}. Building quantum codes for Galois qudits over $\GF(2^s)$ with desirable properties often allows one to build qubit codes with similarly desirable properties~\cite{wills2024,golowich2025asymptotically,nguyen2024,nguyen_pattison,he2025quantum,he2025asymptotically}. In particular, by constructing qudit codes over $\GF(2^s)$ with transversal third-level gates, we can construct distillation protocols for qubit magic states that only require qubit Clifford operations~\cite{wills2024,nguyen_pattison} to execute on qubits.
Thus, we study qubit magic state distillation protocols by constructing and analyzing quantum codes over \(\F_{2^s}\).

The mapping between the $2^s$-dimensional Galois qudit and the set of $s$ qubits goes via a ``qudit-to-qubit mapping'', for which all the details are laid out in~\cite{wills2026review}. In particular, isomorphisms may be constructed for the states of the two systems, CSS codes on them, as well as the levels of the Clifford hierarchy, which specialises to diagonal gates in each level of the Clifford hierarchy. All of these isomorphisms are compatible with each other in the natural ways. Accordingly, by constructing distillation protocols for Galois qudits, we obtain distillation protocols for qubits. We make this part of the transformation explicit shortly. 

Throughout, we specify distillation protocols via CSS codes, where the CSS codes are specified by matrices in a way that is now standard~\cite{bravyi2012magic}, where the rows of the given matrix specify the $X$-logical operators in its first $k$ rows, and its remaining rows specify $X$ stabilizer generators. This is an important enough notion for us that we give it a name here. One may see~\cite{wills2026review} for the construction of Galois qudit CSS codes from spaces of $X$ and $Z$ stabilizers/logical operators. In what follows, two vectors \(x,y \in \F^n_q\) are said to be \textit{orthogonal} if \(\langle x,y\rangle \equiv \sum_i x_i y_i = 0\).

\begin{definition}[X-generator matrix]\label{def:X_gen_matrix}
    An $X$-generator matrix for a quantum CSS code is some matrix $G \in \mathbb{F}_q^{m \times n}$, with an associated positive integer $1 \leq k \leq m$. We require that $G$ is full rank over $\mathbb{F}_q$. This matrix defines a quantum CSS code where the first \(k\) rows are non-trivial logical \(X\) operators (referred to as \textit{logical rows}), and the remaining rows are \(X\) stabilizer generators. Formally, let $\mathcal{G}$ be the $\mathbb{F}_q$-vector space generated by the rows of $G$, and $\mathcal{G}_0$ be the $\mathbb{F}_q$-vector space generated by the latter $m-k$ rows of $G$. The CSS code has $X$ stabilizers given by the space $\mathcal{G}_0$ and $Z$ stabilizers given by the space $\mathcal{G}^\perp$.
\end{definition}

It is straightforward to compute code properties from the \(X\)-generator matrix.
\begin{proposition}\label{prop:matrix_to_code}
    Let $G \in \mathbb{F}_q^{m \times n}$ be an $X$-generator matrix for a quantum CSS code over $\mathbb{F}_q$ with $k$ logical rows. The quantum CSS code has $k$ logical qudits and $n$ physical qudits. Its $Z$-distance is equal to the minimum Hamming weight of a nonzero vector in $\mathbb{F}_q^n$ orthogonal to all of the latter $m-k$ rows of $G$, but not orthogonal to at least one of the first $k$ rows of $G$.
\end{proposition}
\begin{proof}
In the notation of Definition~\ref{def:X_gen_matrix}, the number of logical qudits of the corresponding CSS code is $\dim_{\mathbb{F}_q}\mathcal{G}-\dim_{\mathbb{F}_q}\mathcal{G}_0 = m-(m-k) = k$, where the first equality follows from the requirement that $G$ is full rank. The $Z$ distance is then the minimum Hamming weight of an element of $\mathcal{G}_0^\perp \setminus \mathcal{G}^\perp$. By linearity, a vector $z \in \mathbb{F}_q^n$ is in $\mathcal{G}_0^\perp\setminus\mathcal{G}^\perp$ if and only if it is orthogonal to every row of $G_0$, but not every row of $G$.
\end{proof}

For distilling magic states of diagonal gates, only the $Z$-distance is relevant, since one can twirl the magic states~\cite{MSD,bravyi2012magic}.
\begin{proposition}[Generalization of \cite{MSD,bravyi2012magic} twirling]\label{prop:diagonal_twirl}
    Let \(U\) be a diagonal third-level \(s\)-qubit operator.
    For \(a \in \F_2^s\), define the operators \(M_a \equiv UX^aU^\dagger\) and states \(\ket{A_a} \equiv Z^a U\ket{+}^s\).
    Then, for any state, \(\rho\), its twirl with respect to \((M_a)_{a \in \F_2}\) picked uniformly at random is \(\frac{1}{2^s}\sum_{a \in \F_2^s} M_a  \rho M_a ^\dagger = \sum_{a \in \F_2^s} \rho_a \ketbra{A_a}{A_a}\) where \(\rho_a \equiv \bra{A_a} \rho \ket{A_a}\).
\end{proposition}
\begin{proof}
    First, note that \(\bra{A_a}A_b\rangle = \bra{+}^s Z^{a+b} \ket{+}^s = \delta_{ab}\), so the states \((\ket{A_a})_{a \in \F_2^s}\) form an orthonormal basis. 
    Additionally, they are eigenvectors of the operators \((M_a)_{a \in \F_2}\) with eigenvalues \(M_a \ket{A_b} = UX^aU^\dagger Z^b U \ket{+}^s= (-1)^{a \cdot b}\ket{A_b}\).
    For an arbitrary density matrix, we can expand it in the basis \(\rho = \sum_{a,b\in\F_2} \rho_{ab} |A_a\rangle\langle A_b|\), where \(\rho_{ab} = \bra{A_a} \rho \ket{A_b}\).
    After passing this state through the twirling channel, we get 
    \begin{align*}
        \frac{1}{2^s}\sum_{a \in \F_2^s} M_a  \rho M_a ^\dagger &= \frac{1}{2^s}\sum_{c \in \F_2^s} \sum_{a,b\in\F_2^s} \rho_{ab} M_c |A_a\rangle\langle A_b|M_c^\dagger \\
        &= \frac{1}{2^s}\sum_{c \in \F_2^s} \sum_{a,b\in\F_2^s} \rho_{ab} (-1)^{c \cdot (a+b)}|A_a\rangle\langle A_b|\\
        &= \sum_{a\in\F_2^s} \rho_{aa} |A_a\rangle\langle A_a|
    \end{align*}
\end{proof}

Typically, we will refer to the rows of $G$ as $\bg_1,\dots,\bg_m$, and the $j$-th entry of the $i$-th row as $g_{ij}$. We will usually think of $G$ as a generator matrix for a classical code, i.e., the rows of $G$ form a basis for the code. We restrict ourselves to $\F_{2^s}$-linear codes, i.e., if $\bg_1,\bg_2\in \F_{2^s}^n$ are codewords, then so is $\gamma_1\bg_1+\gamma_2\bg_2$ for all $\gamma_1,\gamma_2\in\F_{2^s}$. 

If $G$ is a binary matrix, one immediately has a quantum CSS code for qubits. If $s > 1$, then we may binarize the matrix $G$ as follows to produce a qubit code.\footnote{It is equivalent to construct a qudit code from the matrix over $\mathbb{F}_{2^s}$, and binarize the code itself~\cite{wills2026review}.} Consider some basis $B = (\alpha_i)_{i=1}^s$ for $\mathbb{F}_{2^s}$ over $\mathbb{F}_2$, i.e., every $\gamma\in \F_{2^s}$ can be written uniquely as $\gamma=\sum_{i=1}^s b_i \alpha_i$ for $b_i\in\F_2$, $i=1,\dots,n$. The generator matrix $G\in\F_{2^s}^{m\times n}$ may be expanded into a binary matrix in $\F_2^{ms\times ns}$, simply by taking any entry $\gamma$ of $G$ and replacing it with an $\F_2^{s\times s}$ matrix whose $(i,j)$-th entry is $\tr(\gamma\alpha_i\alpha_j)$~\cite{theoryEC}\footnote{One can think of this $\gamma\mapsto (\tr(\gamma \alpha_i\alpha_j))_{i,j=1}^s$ mapping as multiplying each codeword of the $\F_{2^s}$-linear code by $\alpha_i$ and then subsequently binarizing each entry; see~\cite{wills2026review} for more details on mapping qudit codes to qubit codes.}. The resultant binary matrix specifies a qubit CSS code with its first $ks$ rows being logical $X$ operators and its latter $(m-k)s$ rows being $X$ stabilizer generators.

One may use different bases $B$ for the mapping of qudits to qubits. This yields different isomorphisms between the qudits and qubits. Changing the basis $B$ used in the mapping is equivalent to performing a $\CNOT$ circuit on the $s$ qubits making up the qudit. A particularly convenient choice of basis is a \textit{self-dual basis}, one for which $\tr(\alpha_i\alpha_j) = \delta_{ij}$. Such a basis exists for all $s$, and can be found in polynomial time using Lempel factorization~\cite{lempel_factorization}. It is particularly convenient to extract the components of a field element in a self-dual basis:
\begin{equation}
\label{eq:extract_bit}\F_{2^s}\ni \gamma=\sum_{i=1}^s b_i\alpha_i,\quad b_i=\tr(\gamma\alpha_i)\in\F_2\;.
\end{equation}
Using this expression and the phase polynomial language, we can build up the connection between qudit and qubit diagonal gates. In the rest of this section, we will study particular qudit gates, including the qudit-CCZ gates acting on three qudits, as well as single qudit gates that include $\CS$ (specifically over $\F_4$), norm gates and $U_7$ gates for arbitrary binary extension fields. These are all non-Clifford gates. The aim is therefore to construct qudit codes, encoded in matrices $G \in \mathbb{F}_{2^s}^{m \times n}$, having these gates transversal, in order to distill their magic states. This immediately gives us qubit protocols distilling the equivalent of these gates under the qudit-to-qubit mapping.

Before proceeding, however, let us state some simple observations about the Clifford hierarchy for Galois qudits of dimension $2^s$ (regardless of the extension degree $s$). We focus particularly on the operations that perform arithmetic over $\mathbb{F}_{2^s}$, for which we note that the level of the Clifford hierarchy are very different to arithmetic over $\Z$, originating from the fact that addition does not cause ``carrying'' in $\F_{2^s}$.

\begin{remark}
\label{remark:CH_arithmetic}
Clifford hierarchy of arithmetic operations in binary extension fields.
\begin{enumerate}
\item In-place addition over $\F_{2^s}$, that is, $\ket{x}\ket{y} \mapsto \ket{x}\ket{x+y}$ for $x,y \in \F_{2^s}$, is a Clifford operation.
\item Multiplying an element in $\F_{2^s}$ by a known constant $\beta\in\F^\times_{2^s}$ is a Clifford operation. This can be also seen as a change of basis, i.e., from $(\alpha_i)_{i=1}^s$ to $(\beta\alpha_i)_{i=1}^s$. There are many more changes of basis than these, however, and they are all Clifford.
\item Multiplication of two unknown elements $x,y\in \F_{2^s}$, that is, $\ket{x}\ket{y}\ket{0} \mapsto \ket{x}\ket{y}\ket{xy}$, is in the third level of the Clifford hierarchy. We will see this explicitly in subsection~\ref{sec:qudit-CCZ}.
\item The in-place Frobenius transform $|\gamma\ra\mapsto |\gamma^2\ra$ is a Clifford operation. A way to see this is to expand in the normal basis\footnote{A normal basis of $\F_{q^m}$ over $\F_q$, which is a basis of the form $\{\alpha,\alpha^q,\dots,\alpha^{q^{m-1}}\}$, always exists, see e.g. \cite[Thm.~2.35]{Lidl_Niederreiter_1996}. As we commented before, for $\F_{2^s}$ over $\F_2$, a self-dual basis always exists; however, a self-dual normal basis exists if and only if $s$ is not divisible by $4$.~\cite{mullen2013handbook}} $\alpha_i=\alpha^{2^i}$ for some primitive element $\alpha$ of $\F_{2^s}$. Writing $\gamma=\sum_{i=1}^s b_i\alpha_i$, we have $\gamma^2=(\sum_{i=1}^s b_i\alpha_i)^2=\sum_{i=1}^s b_i^2 \alpha_i^2=\sum_{i=1}^s b_i \alpha_{i+1\pmod{s+1}}$, and so the Frobenius transform is just cyclically shifting the $s$ qubits making up the qudit (with this basis $B$), which is Clifford.\footnote{Recall that picking a different basis $B$ for $\mathbb{F}_q$ over $\mathbb{F}_2$ yields different isomorphisms from qudits to qubits, however each give valid isomorphisms of the diagonal gates in each level of the Clifford hierarchy. Thus, we may choose a particular basis $B$, establish the level of the corresponding qubit gate, and conclude the level of the initial qudit gate.}
\end{enumerate} 
\end{remark}

\subsection{Expansion of gates}\label{subsec:qudit_gates}

\subsubsection{Qudit-CCZ gate}
\label{sec:qudit-CCZ}

As mentioned, gates on $\GF(2^s)$ qudits can be expanded into qubit gates as well, where diagonal gates are mapped to diagonal gates, and they retain their level in the Clifford hierarchy. Let us examine this further. Full details can be found in~\cite{wills2026review}, but we give the idea here as well. Let us start with the example of the three-qudit CCZ gate~\cite{golowich2025asymptotically,nguyen2024}, defined as follows for $x,y,z\in\GF(2^s)$:
\begin{equation}
\text{CCZ}:\;|x\ra|y\ra|z\ra\mapsto (-1)^{\tr(xyz)}|x\ra|y\ra|z\ra
\end{equation}
For simplicity, let $B = (\alpha_i)_{i=1}^s$ be a self-dual basis for $\mathbb{F}_q$ over $\mathbb{F}_2$. To expand the gate, first expand $x,y,z$ in the self-dual basis: $x=\sum_{i=1}^s x_i\alpha_i$, $y=\sum_{i=1}^s y_i\alpha_i$, and $z=\sum_{i=1}^s z_i\alpha_i$, where $x_i,y_i,z_i\in\{0,1\}$. The computational basis states $\ket{x}\ket{y}\ket{z}$ on $3$ qudits can be expanded to computational basis states on $3s$ qubits by writing $x,y,z$ out in their components. Next, using the fact that $\tr(\cdot)$ is $\F_2$-linear:
\begin{equation}
\label{eq:expand_qudit_CCZ_into_qubit_CCZ}
\tr(xyz)=\tr\left(\sum_{i,j,k=1}^s x_iy_jz_k\;\alpha_i\alpha_j\alpha_k\right)=\sum_{i,j,k=1}^s x_iy_jz_k\tr(\alpha_i\alpha_j\alpha_k).
\end{equation}
When we expand $\CCZ$ to a $3s$-qubit gate, it gives a phase $\tr(xyz) = \sum_{i,j,k=1}^sx_iy_jz_k\tr(\alpha_i\alpha_j\alpha_k)$ to the computational basis state $\ket{x}\ket{y}\ket{z}$.
Therefore, at the qubit level, this qudit CCZ gate applies a qubit CCZ to every triplet $(x_i,y_j,z_k)$ for which $\tr(\alpha_i\alpha_j\alpha_k)=1$.

Let us show an explicit example of the qudit $\CCZ$ gate for $\F_4=\{0,1,\omega,\omega^2\}$, where arithmetic in this field is specified by $\omega + \omega^2 = 1$. One can check that $(\omega, \omega^2)$ forms a self-dual basis. The left side of figure~\ref{fig:F4_qudit_CCZ} denotes the example of the $\mathbb{F}_4$-qudit $\CCZ$ gate expanded in this basis. Claim~\ref{claim:coherent_field_mult} shows that the qudit $\CCZ$ gate can be used to coherently multiply elements in the field $\mathbb{F}_q$ on a quantum computer.
\begin{figure}[ht]
\centering
\includegraphics[width=1.0\linewidth]{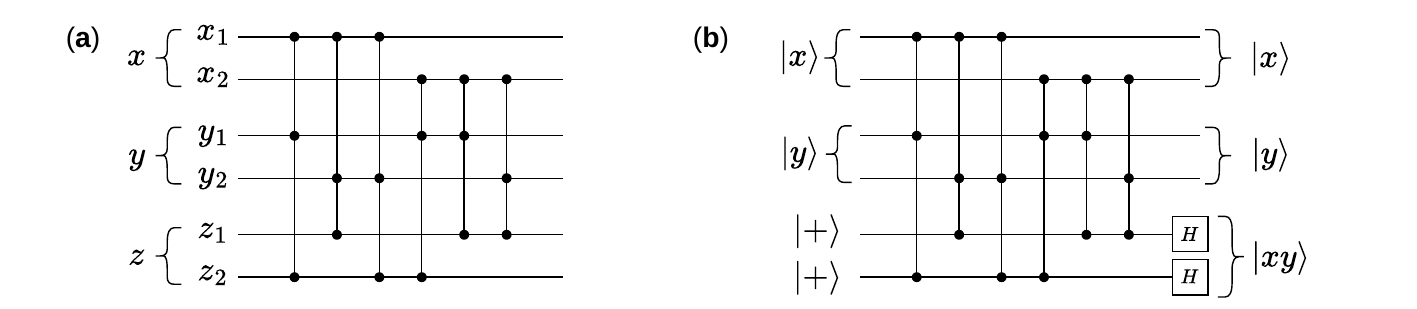}
\caption{(\textbf{a}) $\F_4$-qudit CCZ expanded as qubit CCZs. Expand $x=x_1\omega+x_2\omega^2$, where $x_1,x_2\in\F_2$, and similarly for $y,z$. (\textbf{b}) The magic state corresponding to $\mathbb{F}_4$-$\CCZ$ can be consumed to perform multiplication in $\F_4$ -- one just needs to conjugate the third pair of qubits by Hadamard gates.}
\label{fig:F4_qudit_CCZ}
\end{figure}

\begin{claim}\label{claim:coherent_field_mult}
The $\F_{2^s}$-qudit CCZ gate allows one to multiply two elements in that field on a quantum computer. That is, one use of a $\CCZ$ gate, and Clifford operations, is sufficient to perform the arithmetic operation $\ket{x}\ket{y}\ket{0} \mapsto \ket{x}\ket{y}\ket{xy}$, where $x,y \in \F_{2^s}$, and each ket is an $s$-qubit register (equivalently, a one $2^s$-dimensional qudit register. Fig.~\ref{fig:F4_qudit_CCZ}(b) shows an example for $\F_4$. 
\end{claim}
\begin{proof}
We will show that the following sequence of operations executes the desired operation on the $3$-qudit state $\ket{x}\ket{y}\ket{0}$:
\begin{enumerate}
    \item Perform Hadamard on every qubit making up the third register, which all begin in the $\ket{0}$ state;
    \item Perform the qudit $\CCZ$ gate on the three qudits ($3s$ qubits);\footnote{This step can also be performed by consuming a $\ket{\CCZ}$ state, and using Clifford operations.}
    \item Perform Hadamard again on all $s$ qubits making up the last register.
\end{enumerate}
A fact we will be using is that 
\begin{equation}
H^{\otimes s}\,|z\ra\;=\;\frac{1}{\sqrt{2^s}}\sum_{t\in \F_{2^s}} (-1)^{\tr(zt)}\,|t\rangle.
\end{equation}
Here, $\ket{z}$ denotes the qudit computational basis state corresponding to $z \in \mathbb{F}_q$, but $H^{\otimes s}$ corresponds to performing the qubit Hadamard operation on the $s$ qubits making up the qudit.
To see this equation, consider expanding $z\in\F_{2^s}$ in a self-dual basis $(\alpha_i)_{i=1}^s$ as $z=z_1\alpha_1+\dots+z_s\alpha_s$; one also expands $t\in\F_{2^s}$ as $t=t_1\alpha_1+\dots,t_s\alpha_s$: $z_i,t_i\in\{0,1\}$. Then, $H^{\otimes s}\ket{z}$ is a uniform superposition over $\ket{\T}$ with phase $(-1)^{z_1t_1+\dots+z_st_s}$. This exponent is exactly $$\tr(zt)=\tr\left((\sum_i z_i\alpha_i)\;(\sum_j t_j\alpha_j)\right)=\sum_{i,j}z_i t_j\tr(\alpha_i\alpha_j)=\sum_i z_it_i\;.$$

The desired circuit is then performed as follows:
\begin{equation}
|x\ra|y\ra|0^s\ra\xrightarrow{H^{\otimes s}\text{ on }z}\ |x\ra|y\ra\frac{1}{\sqrt{2^s}}\sum_{z\in \F_{2^s}}|z\ra \xrightarrow{\text{Qudit }\CCZ}
\ \frac{1}{\sqrt{2^s}}\sum_{z}(-1)^{\tr(xyz)}|x\ra|y\ra|z\ra\;.
\end{equation}
Apply $H^{\otimes s}$ again on the $z$-register:
\begin{align*}
H^{\otimes s}\left(\frac{1}{\sqrt{2^s}}\sum_{z}(-1)^{\tr(xyz)}|z\ra\right)
&=\frac{1}{\sqrt{2^s}}\sum_{z} (-1)^{\tr(xyz)}\left(\frac{1}{\sqrt{2^s}}\sum_{t} (-1)^{\tr(z t)} |t\ra\right)\\
&=\frac{1}{2^s}\sum_{t}\left(\sum_{z} (-1)^{\tr\big(z(xy+t)\big)}\right)|t\ra\\
&=|xy\ra\;,
\end{align*}
where in the last step, we have used
\begin{equation*}
\sum_{z\in \F_{2^s}} (-1)^{\operatorname{tr}(zu)} \;=\;
\begin{cases}
2^s, & \text{if } u=0,\\
0, & \text{if } u\neq 0.
\end{cases}
\end{equation*}
\end{proof}
This multiplication in the binary extension field finds usage in, for example, decoded quantum interferometry~\cite{jordan2025DQI}. In particular, \cite{khattar2025verifiable} provides a detailed resource analysis for solving the optimal polynomial intersection problem in binary extension fields. The binary extension field multiplication might also be useful in preparing certain catalyst states~\cite{kim2025catalytic}.

Later, we will only present qudit-CCZ distillation protocols (Table~\ref{table:maximal_curves_CCZ_only}) up to $\F_{64}$. For multiplication in larger binary extension fields, e.g., $\F_{2^{163}}$, one can reduce the problem into many multiplications in small binary extension fields by paying with some additional $\CNOT$ gates~\cite{chudnovsky1988algebraic, CENK2010}. We present an example of this reduction in App.~\ref{sec:binary_field_multiplication}, after introducing relevant concepts for algebraic curves in Sec.~\ref{sec:intro_to_AG_codes} and App.~\ref{sec:curves_formal_definitions}.

Note that, as a consequence of remark~\ref{remark:CH_arithmetic}, the qudit CCZ gates $\CCZ^{\beta_1}$ and $\CCZ^{\beta_2}$ are Clifford equivalent for any $\beta_1,\beta_2\in\F_{2^s}^\times$, where
\begin{equation}
\CCZ^{\beta}|\eta_1\ra |\eta_2\ra |\eta_3\ra = (-1)^{\tr(\beta \eta_1 \eta_2 \eta_3)}|\eta_1\ra |\eta_2\ra |\eta_3\ra\;.
\end{equation}
This is because an equivalent way of doing $\CCZ^{\beta_1}$ is to perform multiplication on the first qudit\footnote{In fact, any of the three qudits work.} $|\eta_1\ra\mapsto|\frac{\beta_2}{\beta_1}\eta_1\ra$ before applying $\CCZ^{\beta_2}$. This fact will later play a role in Sec.~\ref{sec:qudit-CCZ-one-point-codes} when we distill qudit-CCZ gates.

\subsubsection{\texorpdfstring{Control-$S$ gate in $\GF(4)$}{Control-S gate in GF(4)}}
Besides expanding a qudit diagonal gate into a qubit one, one can also do the reverse in some cases.
If the qubit gate only involves multi-control-$Z$ gates, i.e., it is a $\pm1$-phase gate, then it turns out to always be possible; this boils down to the fact that arithmetic in $\mathbb{F}_2$ is simply integer arithmetic modulo $2$, that is, $\mathbb{Z}_2 = \mathbb{F}_2$.
However, if we are considering phase polynomial gates with non-real phases, i.e., involving phases $e^{i\pi/2^t}$, for some $t\in\mathbb{N}$, to the power of some polynomial, then due to the difference between $\Z_{2^t}$ and $\F_{2^t}$, complications arise in general.

Nevertheless, let us show a simple viable example for $\F_4=\{0,1,\omega,\omega^2\}$, as it leads to interesting MSD protocols in Sec.~\ref{sec:CS_distillation_protocols}. In $\F_4$, we have that $\omega^2=\omega+1$ and $\omega^3=1$; also $\tr(1)=0=\tr(0)$ and $\tr(\omega)=\tr(\omega^2)=1$. We consider applying the $\CS$ gate to the $2$ qubits of an $\F_4$ qudit that has been expanded in the self-dual basis $(\omega, \omega^2)$.
\begin{equation}
\gamma=b_1\omega+b_2\omega^2,\quad\CS|\gamma\ra:=\CS|b_1\ra|b_2\ra=i^{b_1b_2}|\gamma\ra.
\end{equation}
Since the basis $(\omega, \omega^2)$ is self dual, one can extract $b_1$ and $b_2$ as $b_1=\tr(\gamma\omega)$ and $b_2=\tr(\gamma\omega^2)$ (see Equation~\eqref{eq:extract_bit}). Using arithmetic over $\mathbb{F}_4$, $b_1b_2=(\gamma\omega+\gamma^2\omega^2)(\gamma\omega^2+\gamma^2\omega^4)=\gamma^3+\tr(\gamma)$. Therefore,
\begin{equation}
\CS:|\gamma\ra\mapsto i^{\gamma^3 + \tr(\gamma)}|\gamma\ra,
\end{equation}
where we emphasise that the arithmetic in the exponent is over $\mathbb{F}_4$.

\subsubsection{\texorpdfstring{The norm gate}{Norm gate}}
Consider the multi-control-$Z$ gate $\text{C}^{s-1}\text{Z}$ acting on the $s$ qubits constituting a $\GF(2^s)$ qudit:
\begin{equation}
\text{C}^{s-1}\text{Z}|\gamma\ra := \text{C}^{s-1}\text{Z} |b_1\ra|b_2\ra\cdots |b_s\ra=(-1)^{b_1b_2\cdots b_n}|b_1\ra|b_2\ra\cdots|b_s\ra\;.
\end{equation}
Considering a self-dual basis $(\alpha_i)_{i=0}^{s-1}$ for $\mathbb{F}_{2^s}$ over $\mathbb{F}_2$, we can use $b_i=\tr(\gamma\alpha_i)$ and the explicit formula of trace in Eq.~\eqref{eq:trace} to write the above qubit phase polynomial into a qudit one. However, a better way is to directly observe that such a gate is Pauli equivalent to 
\begin{equation}
|\gamma\ra\mapsto (-1)^{\Nm(\gamma)}|\gamma\ra=(-1)^{\gamma^{2^s-1}}|\gamma\ra\;.
\end{equation}
This is a gate for a single $2^s$-dimensional qudit which we call the \emph{norm gate}.
Recall that besides $\tr(\cdot)$, the norm function $\Nm(\cdot)$ also maps $\F_{2^s}$ elements to $\{0,1\}$; it takes a particularly simple form in $\F_{2^s}$ (cf. Eq.~\eqref{eq:norm}): $\Nm(\gamma):=\prod_{i=0}^{s-1}\gamma^{2^i}=\gamma^{2^s-1}$. Note that $\Nm(\gamma)=0$ if and only if $\gamma=0$, and else $\Nm(\gamma)=1$; this is most easily seen from the fact that the multiplicative group of a finite field is cyclic.\footnote{Strictly speaking, this would be true irrespective of the structure of the multiplicative group given Lagrange's theorem for finite groups.}
That is to say, $|\gamma\ra$ gains a phase of $(-1)^0=1$ when $\gamma=0$ and $-1$ when $\gamma\neq 0$. 
The Pauli equivalence of the norm gate to the $\text{C}^{s-1}\text{Z}$ gate is now obvious: the norm gate can be implemented by first applying $X$ gates to each of the $s$ qubits making up the qudit, giving us $b_i\mapsto b_i\oplus 1$, then applying a $\text{C}^{s-1}\text{Z}$ gate, and finally undoing the $X$ gates. Then only in the case of $b_1=\cdots=b_s=0$ does the intermediate state gain a phase of $-1$.

\subsubsection{Classification of the diagonal Galois qudit Clifford hierarchy with real phases}\label{subsubsec:classification}

Aside from the norm gate, we can examine more general diagonal Galois qudit gates in the Clifford hierarchy with real phases, i.e., with phases $\pm 1$. The aim is to classify the level of such gates in the Clifford hierarchy. Note that in some sense this has already been achieved, because each level of the diagonal Clifford hierarchy for (multiple) qudits of prime dimension is isomorphic to that of a single larger Galois qudit, as discussed, and the former has been classified~\cite{cui2017diagonal}. Thus, for any given Clifford hierarchy gate on one Galois qudit, one can use the isomorphism to write it as a gate for multiple qudits of a smaller dimension, and establish its level using that classification. The aim here, however, is to establish the level of such a gate in language which is native to the Galois qudits. Aside from giving us a deeper understanding of the real diagonal Clifford hierarchy gates for Galois qudits, it is operationally useful as it allows us to derive transversality conditions for the corresponding magic states naturally. This treatment resolves an open question in~\cite{he2025quantum}.

All single-qudit diagonal real-phase gates are in the Clifford hierarchy, and can be written
\begin{equation}\label{eq:F_qtoF_2fn}
    \ket{\gamma} \mapsto (-1)^{f(\gamma)}\ket{\gamma}\;, \quad  f(\gamma)=\sum_{i=0}^{2^s-1}c_i\gamma^i\in\{0,1\}\;,\quad c_i \in \mathbb{F}_q\;.
\end{equation}
Writing $q=2^s$, we note that the polynomial only has terms of degree less than $q$, because $\gamma^q=\gamma,\;\forall \gamma\in\F_{q}$. 

Now, since $f:\F_q\to \F_2$, we have $f^2=f$, so the coefficients $c_i$ satisfy $c_{2i}=c_i^2$. Therefore, if $c_i$ is nonzero, then so is $c_j$ for any $j$ in the $2$-cyclotomic coset $C_i$ of $i$, which is defined by
\begin{equation}
\label{eq:cyclotomic_coset}
C_i=\{i,2i,4i,\dots,2^{\ell_i-1}i\}\!\!\mod{q-1}\;,\text{ where }\; \ell_i=\min\{j\ge 1\;|\;2^j i\equiv i\!\!\mod{q-1}\}.
\end{equation}
For example, if $s=4$ then $q = 2^4 = 16$, and considering $i=3$, we have $C_i = \{3, 6, 12, 9\}$, and $\ell_i = 4$. In this case, we notice that $\ell_i = s$, but more generally, we must always have $\ell_i \mid s$, following from the fact that $2^si \equiv i \mod q-1$.
Moreover, using $c_{2i}=c_i^2$, we have
\begin{equation}
    \sum_{j=0}^{\ell_i-1}c_{2^ji}\gamma^{2^ji} = \sum_{j=0}^{\ell_i-1}c_i^{2^j}\gamma^{2^ji} = \sum_{j=0}^{\ell_i-1}\left(c_i\gamma_i\right)^{2^j} = \tr_{\F_{2^{\ell_i}}/\F_2}(c_i\gamma^i).
\end{equation}
Note that the right-hand side makes sense, indeed, $c_i\gamma^i \in \mathbb{F}_{2^\ell_i}$ follows from 
\begin{equation}
    c_i = c_{2^{\ell_i}i} = c_i^{2^{\ell_i}} \implies c_i \in \mathbb{F}_{2^{\ell_i}}
\end{equation}
and
\begin{equation}
    2^{\ell_i}i \equiv i\mod q-1 \implies \gamma^i = (\gamma^i)^{2^{\ell_i}}.
\end{equation}
Now, notationally, we always let $i$, the smallest element in $C_i$, represent $C_i$, and we then let $\Gamma$ be the set of such representatives of the $2$-cyclotomic cosets modulo $q-1$.
In total, every Boolean function $f:\F_{2^s}\to \F_2$ can be written uniquely in the form of Equation~\eqref{eq:F_qtoF_2fn}, and in turn uniquely as
\begin{equation}
    f(\gamma)=c_0+\sum_{i\in \Gamma}\tr_{\F_{2^{\ell_i}}/\F_2}(c_i \gamma^i)\;.
\end{equation}
using this discussion.

In Theorem~\ref{thm:real_phase_monomial_level}, we will start by establishing the level of the Clifford hierarchy of each term $(-1)^{\tr_{\F_{2^{\ell_i}}/\F_2}(c_i \gamma^i)}$. Later, in Theorem~\ref{thm:clifford_level_full_polynomial}, we will argue that the level of their product is simply the maximum of each of their levels.
\begin{theorem}
\label{thm:real_phase_monomial_level}
Given $q=2^s$ and a non-negative integer $n<q-1$, let $\ell_n$ be the size of the $2$-cyclotomic coset $C_n$ of $n$ as in Eq.~\eqref{eq:cyclotomic_coset}. For $c_n\in \F_{2^{\ell_n}}$,
the single-qudit real-phase gate mapping $|\gamma\ra$ to $(-1)^{\tr_{\F_{2^{\ell_n}}/\F_2}(c_n \gamma^n)}|\gamma\ra$ is in exactly the $\wt(n)$-th level of Clifford hierarchy\footnote{That is, the gate is in the $\wt(n)$-th level, but not in a lower level.} if $c_n\neq 0$, where $\wt(n)$ stands for the Hamming weight of the binary expansion of $n$. 
If $c_n=0$, then the gate is trivial up to a global phase.
\end{theorem}
\begin{proof}
Say $(\alpha_i)_{i=1}^s$ is a basis for $\F_{2^s}$, and expand $\gamma=\sum_{i=1}^s b_i\alpha_i$ where $b_i\in\{0,1\}$. Then, the exponent is
\begin{align}
\tr_{\F_{2^{\ell_n}}/\F_2}(c_n \gamma^n) &=\tr_{\F_{2^{\ell_n}}/\F_2}\left(c_n (b_1 \alpha_1+\cdots+b_s\alpha_s)^n\right)\nonumber\\
&= \tr_{\F_{2^{\ell_n}}/\F_2}\left(c_n \sum_{e_1+\dots+e_s=n}\binom{n}{e_1,\dots,e_s}\prod_{i=1}^s (b_i\alpha_i)^{e_i}\right),\label{eq:single_qubit_gate_expansion}
\end{align}
where we use the \href{https://en.wikipedia.org/wiki/Multinomial_theorem}{multinomial expansion formula}, and the multinomial coefficients are $\binom{n}{e_1,e_2,\dots,e_s}:=\frac{n!}{e_1! e_2!\cdots e_s!}$. This expression is raised to the power of $-1$, and so it is relevant to establish the parity of the multinomial coefficients.

\begin{lemma}
\label{lemma:lucas_theorem}
\textbf{(Corollary of Lucas' theorem).} Consider the multinomial coefficient $\binom{n}{e_1,e_2,\dots,e_m}$, where $e_i$ are positive integers and $e_1+e_2+\dots,e_m=n$. This is odd if and only if there is no carrying when adding up the binary expansions of $e_1,e_2,\dots,e_m$.
That is, considering the expansions of each $e_i$ in binary, the coefficient is odd if and only if, in each bit position of the binary expansion, there is at most one $1$ across the $e_i$'s. An immediate corollary is that, if $m >\wt(n)$, then the multinomial coefficient is even.
\end{lemma}
\begin{proof}
\href{https://en.wikipedia.org/wiki/Lucas\%27s_theorem}{Lucas' theorem} states that for non-negative integers $A,B$ and a prime $p$, the binomial coefficient satisfies the following congruence relation:
\begin{equation}
    \binom{A}{B}\equiv \prod_{i=1}^r \binom{A_i}{B_i} \mod{p}\;,
\end{equation}
where $A=A_r p^r+A_{r-1}p^{r-1}+\cdots + A_1p+A_0$ and $B=B_r p^r+B_{r-1}p^{r-1}+\cdots+B_1p+B_0$ are the base-$p$ expansions of $A,B$. This uses the convention that $\binom{A}{B}=0$ if $A<B$. It follows that $\binom{A}{B}$ is divisible by $p$ if and only if $B_i > A_i$ for some $i$. We will only use the $p=2$ case. The binomial coefficient $\binom{A}{B}$ is odd if and only if $B_i \leq A_i$ for every $i$, which is if and only if there is no carrying when $B$ and $A-B$ are summed.

Notice that we can write the multinomial coefficient as a product of binomial coefficients:
\begin{equation*}
\binom{n}{e_1,\dots,e_m}=\binom{n}{e_1}\binom{n-e_1}{e_2}\binom{n-e_1-e_2}{e_3}\cdots\binom{e_{m-1}+e_m}{e_{m-1}}\;.
\end{equation*}
Equivalently, writing the suffix sum as $S_k=e_k+e_{k+1}+\cdots+e_m$, the above formula is $\binom{n}{e_1,\dots,e_m}=\prod_{k=1}^{m-1} \binom{S_k}{e_k} = \prod_{k=1}^{m-1}\binom{S_k}{S_{k+1}}$.
The multinomial coefficient on the left is odd if and only if all the binomial coefficients on the right-hand side are odd, and $\binom{S_k}{e_k}$ is odd if and only if there is no carrying when $e_k$ and $S_{k+1}$ are summed in binary to make $S_k$. Finally, there is no carrying when $e_1, \ldots, e_m$ are summed to make $n$ if and only if there is no carrying when $e_k$ and $S_{k+1}$ are summed to make $S_k$ for all $k$.
\end{proof}
Returning to the expression in Equation~\eqref{eq:single_qubit_gate_expansion}, we note that an $s$-term product of $b_i$'s corresponds to a $\text{C}^{(s-1)}\text{Z}$ gate. With this, and the above lemma, we find that the single-qudit phase gate $|\gamma\ra\mapsto (-1)^{\tr_{\F_{2^{\ell_n}}/\F_2}(c_n \gamma^n)}|\gamma\ra$ is in the $\wt(n)$-th level of the Clifford hierarchy. The remaining challenge is to establish that the gate lies precisely in that level. To this end, we consider the terms in Equation~\eqref{eq:single_qubit_gate_expansion} with a $\wt(n)$-product of $b_i$'s. We look at the coefficient ($\in\{0,1\}$) in front of all the $\wt(n)$-term products of $b_i$'s, and show that at least one of them is one. That is, we just want to consider terms in the sum of Equation~\eqref{eq:single_qubit_gate_expansion} for which exactly $\wt(n)$ of the $e_i$ are positive, and of these we only need to consider terms for which the $e_i$ have binary expansions of weight one, i.e., they are powers of two, since the multinomial coefficient is odd in exactly these cases.

Let $T\subseteq\{0,1,\dots,s-1\}$ be the support of the binary expansion of $n$, i.e., $n=\sum_{t\in T}2^t$. Let $D=|T|=\wt(n)$, and write $T=\{t_1,\dots,t_D\}$.
Then considering some distinct set of indices $\{i_1, \ldots, i_D\} \subseteq [s]$, the coefficient of $b_{i_1}b_{i_2}\ldots b_{i_D}$ is seen to be
\begin{equation*}
B_{c_n}(\alpha_{i_1},\dots,\alpha_{i_D}):=\tr_{\F_{2^{\ell_n}}/\F_2}\left(c_n M_T(\alpha_{i_1},\dots,\alpha_{i_D})\right),\quad M_T(\alpha_{i_1},\dots,\alpha_{i_D}):=\det\begin{pmatrix}
\alpha_{i_1}^{2^{t_1}} & \alpha_{i_1}^{2^{t_2}} & \cdots & \alpha_{i_1}^{2^{t_D}}\\
\vdots & \vdots & & \vdots\\
\alpha_{i_D}^{2^{t_1}} & \alpha_{i_D}^{2^{t_2}} & \cdots & \alpha_{i_D}^{2^{t_D}}
\end{pmatrix}.
\end{equation*}
We wish to show that there is some $\{i_1, \ldots, i_D\}$ such that $B_{c_n}(\alpha_{i_1}, \ldots, \alpha_{i_D}) = 1$. Note that the trace in its definition is well-defined because
\begin{equation*}
M_T(h_1,\dots,h_D)^{2^{\ell_n}}=\det\begin{pmatrix}
h_1^{2^{t_1+\ell_n}} & \cdots & h_1^{2^{t_D+\ell_n}}\\
\vdots &  & \vdots\\
h_D^{2^{t_1+\ell_n}} & \cdots & h_D^{2^{t_D+\ell_n}}
\end{pmatrix}=M_{\{t_1+\ell_n,\dots,t_D+\ell_n\}}(h_1,\dots,h_D)=M_T(h_1,\dots,h_D).
\end{equation*}
The last step follows from $\{t_1+\ell_n,\dots,t_D+\ell_n\}=T=\{t_1,\dots,t_D\}$. This fact follows from the definition of $\ell_n$, since $\{t_1, \ldots, t_D\}$ forms the support of the binary expansion of $n$, and multiplying $n$ modulo $q-1$ is equivalent to cyclically shifting its binary expansion.
The following properties of the determinant $M_T(h_1,\dots,h_D)$ are easy to verify:
\begin{enumerate}
    \item If $\exists i\neq j$ such that $h_i=h_j$ are equal, then the determinant is zero (since two rows are equal).
    \item It is an $\F_2$-multlinear form since $(h_i+h_i')^{2^{t_i}}=(h_i)^{2^{t_i}}+(h'_i)^{2^{t_i}}$, i.e., $$M_T(h_1,\dots,h_i+h'_i,\dots,h_D)=M_T(h_1,\dots,h_i,\dots,h_D)+M_T(h_1,\dots,h'_i,\dots,h_D).$$ Since the trace function is also $\F_2$-linear, $B_{c_n}(h_1,\dots,h_D)$ is also an $\F_2$-multilinear form.
\end{enumerate}
We claim that it suffices to show the existence of $h_1,\dots,h_D\in \F_q$ such that $B_{c_n}(h_1,\dots,h_D)=1$. Indeed, expanding $h_r=\sum_{i=1}^s \lambda_{r,i}\alpha_i$ with $\lambda_{r,i}\in\F_2$, and using the second property above, we have
\begin{equation}
\label{eq:F2_expansion_of_B}
1=B_{c_n}(h_1,\dots,h_D)=\sum_{i_1,\dots,i_D} \lambda_{1,i_1}\cdots\lambda_{D,i_D} B_{c_n}(\alpha_{i_1},\dots,\alpha_{i_D})\;.
\end{equation}
There must therefore then be some $\{i_1,\dots,i_D\}\subseteq[s]$ such that $B_{c_n}(\alpha_{i_1},\dots,\alpha_{i_D})=1$; clearly $i_1,\dots,i_D$ need to be pairwise disjoint for this to happen by the first property above.

To show that there exists some $h_1, \ldots, h_D \in \mathbb{F}_q$ such that $1 = B_{c_n}(h_1, \ldots, h_D)$, we first show that there exists $h_1,\dots,h_D\in \F_q$ such that $\F_{2^{\ell_n}}\ni M_T(h_1,\dots,h_D)\neq 0$. Indeed, suppose that $M_T(h_1, \ldots, h_D) = 0$ for all $h_1, \ldots, h_D \in \mathbb{F}_q$. For $h \in \mathbb{F}_q$, let $r(h)$ be the row vector $(h^{2^{t_1}},  h^{2^{t_2}}, \ldots , h^{2^{t_D}})$, and consider the vector space $W \coloneq \text{span}_{\mathbb{F}_q}\left\{r(h):h \in \mathbb{F}_q\right\} \subseteq \mathbb{F}_q^D$. Since $M_T(h_1, \ldots, h_D) = 0$ for all $h_1, \ldots, h_D$, we have that $r(h_1), \ldots, r(h_D)$ are linearly dependent over $\mathbb{F}_q$ for all $h_1, \ldots, h_D$. This means that $\dim_{\mathbb{F}_q}W < D$, since otherwise there are vectors $r(h_1), \ldots, r(h_D)$ spanning the whole space. This means that there is some fixed non-zero vector $(a_1, \ldots, a_D) \in \mathbb{F}_q^D$ orthogonal to every vector in $W$, i.e., for which $\sum_{j=1}^Da_jh^{2^{t_j}} = 0$ for all $h \in \mathbb{F}_q$. However, no non-zero polynomial of degree $<q$ over $\mathbb{F}_q$ can vanish at all elements of $\mathbb{F}_q$, yielding a contradiction.

Now given $h_1,\dots,h_D\in \F_q$ for which $v:=M_T(h_1,\dots,h_D)\neq 0$, consider scaling them simultaneously by $a\in\F_q$. Since $\sum_{t\in T}2^t=n$, we have $M_T(ah_1,\dots,ah_D)=a^n M_T(h_1,\dots,h_D)$. We show that there exists $a$ such that $B_{c_n}(a h_1,\dots,a h_D)=\tr_{\F_{2^{\ell_n}/\F_2}}(c_n a^n v)=1$. It is clear that $a^n\in \F_{2^{\ell_n}}$ since $(a^n)^{2^{\ell_n}}=a^n$. We now use the fact that the $\F_2$-linear span of $\{a^n\;|\;a\in\F_q\}$ is all of $\F_{2^{\ell_n}}$, which we show at the end. Using this fact, the surjectivity of $\tr(\cdot)$ means that there exists $a'\in \F_{2^{\ell_n}}$ that $\tr_{\F_{2^{\ell_n}/\F_2}}(c_n a' v)=1$. Expanding $a'$ into an $\F_2$-sum of elements in $\{a^n\;|\;a\in \F_q\}$, there must be some $a^n$ such that $\tr_{\F_{2^{\ell_n}/\F_2}}(c_n a^n v)=1$ since $\tr_{\F_{2^{\ell_n}/\F_2}}(c_n a' v)=1$.

Let us conclude by showing that \begin{equation}
    U \coloneq \text{span}_{\mathbb{F}_2}\{a^n\;|\;a \in \mathbb{F}_q\} = \mathbb{F}_{2^{\ell_n}}.
\end{equation}
First, we know that the left-hand side is contained in the right-hand side, and one readily checks that $U$ forms a field. $U$ is a subfields of $\mathbb{F}_{2^s}$, and must therefore be isomorphic to $\mathbb{F}_{2^m}$ for some $m$. Moreover, since $U$ is a subfield of $\mathbb{F}_{2^{\ell_n}}$, we must have $m \leq \ell_n$. We can consider some primitive element $p$ of $\mathbb{F}_q$, that is, some generator of the (cyclic) multiplicative group $\mathbb{F}_q^*$, and note that $p^n \in U$. This means that $(p^n)^{2^m} = p^n$, which implies that $2^mn \equiv n \mod q-1$. By the minimality of $\ell_n$ in its definition, and since $m \leq \ell_n$, we have $m = \ell_n$, as required.

\end{proof}

\begin{remark}
\label{remark:qudit_monomial_gate}
Usually we would just define phase gates via $|\gamma\ra\mapsto (-1)^{\tr(\beta\gamma^n)}|\gamma\ra$, where the trace is from $\F_q=\F_{2^s}$ down to $\F_2$. In this case, to make use of Thm.~\ref{thm:real_phase_monomial_level}, one first uses the following identity, which can be proved by plugging in the definition of trace, cf. Eq.~\eqref{eq:def_trace}
\begin{equation}
\tr_{\F_{2^s}/\F_2}(\beta \gamma^n)=\tr_{\F_{2^{\ell_n}}/\F_2}(\tr_{\F_{2^s}/\F_{2^{\ell_n}}}(\beta)\gamma^n),
\end{equation}
and then we may use $c_n=\tr_{\F_q/\F_{2^{\ell_n}}}(\beta)$ in the theorem.
\end{remark}

\begin{theorem}\label{thm:clifford_level_full_polynomial}
The single-qudit real-phase gate $|\gamma\ra\mapsto (-1)^{f(\gamma)}$, where $f(\gamma)=c_0+\sum_{i\in\Gamma}\tr_{\F_{2^{\ell_i}}/\F_2}(c_i \gamma^i)$, lies in exactly the $\max\limits_{i\in\Gamma,\; c_i\neq 0}\wt(i)$-th level of the Clifford hierarchy.
\end{theorem}
\begin{proof}
The proof is similar to that of Theorem~\ref{thm:real_phase_monomial_level}.

Let $D=\max\limits_{i\in\Gamma,\; c_i\neq 0}\wt(i)$, and let $\Gamma_D=\{i\in\Gamma\;|\;c_i\neq 0,\;\wt(i)=D\}$. For $i\in\{0,1,\dots,s-1\}$, we use $T(i)$ to denote the support of the binary expansion of $i$.
For $\{i_1,\dots,i_D\}\subseteq [s]$ pairwise distinct, the coefficient of $b_{i_1}b_{i_2}\cdots b_{i_D}$ is
\begin{equation}
\label{eq:def_B_f}
B_f(\alpha_{i_1},\dots,\alpha_{i_D})=\sum_{m\in\Gamma_D}\tr_{\F_{2^{\ell_m}}/\F_2}(c_m M_{T(m)}(\alpha_{i_1},\dots,\alpha_{i_D}))\;.
\end{equation}
The strategy is again to show that there exists $h_1,\dots,h_D\in\F_q$ such that $B_f(h_1,\dots,h_D)=1$, and then by the same reasoning as in Eq.~\eqref{eq:F2_expansion_of_B}, there exist $i_1,\dots,i_D$ pairwise distinct such that $B_f(\alpha_{i_1},\dots,\alpha_{i_D})=1$.

Pick an arbitrary $m\in \Gamma_D$ and choose $h_1,\dots,h_D\in\F_q$ such that $B_{c_m}(h_1,\dots,h_D)=1$. Fix this choice of $h_1,\dots,h_D$ and further expand Eq.~\eqref{eq:def_B_f} using the definition of trace,
\begin{align*}
\{0,1\}\ni B_f(h_1,\dots,h_D)&=\sum_{m\in\Gamma_D}\sum_{j=0}^{\ell_m-1}\left(c_m M_{T(m)}(h_1,\dots,h_D)\right)^{2^j}\\
&=\sum_{i\in \Lambda_D} c_i M_{T(i)}(h_1,\dots,h_D)\;,
\end{align*}
where $\Lambda_D=\bigsqcup\limits_{m\in \Gamma_D}C_m$, we use $M_{T(m)}(h_1,\dots,h_D)^{2^j}=M_{T(2^j m)}(h_1,\dots,h_D)$, and we define $c_{2^j m}=c_m^{2^j}$. We have $$B_f(a h_1,\dots, ah_D)=\sum_{i\in\Lambda_D}c_i M_{T(i)}(h_1,\dots,h_D)a^i.$$ 
Since the cyclotomic cosets are disjoint, the coefficient of $a^i$ only gets a contribution from $C_i$. Since $c_i \neq 0$ for all $i \in \Lambda_D$ and $M_{t(m)}(h_1, \ldots, h_D) \neq 0$, this is a non-zero polynomial in $a$. No non-zero polynomial over $\mathbb{F}_q$ of degrees at most $q-1$ can vanish on all $q$ elements of $\mathbb{F}_q$. Therefore, there exists $a \in \mathbb{F}_q$ such that $B_f(ah_1, \ldots, ah_D) \neq 0$, meaning that $B_f(ah_1, \ldots, ah_D) = 0$.

\end{proof}

\begin{remark}
The above results provide a classification of the diagonal gates in the Galois qudit Clifford hierarchy with real phases, that is, those with phases $\pm 1$. While the full diagonal Clifford hierarchy is classified for qubits, it is not clear how to extend the above proof to general single-Galois qudit phase gates,\footnote{Again, the diagonal Galois qudit Clifford hierarchy is already classified in the sense that the diagonal qubit Clifford hierarchy is classified, and the two are isomorphic. The challenge here is to classify the diagonal Clifford hierarchy for Galois qudits in a way that is native to the Galois qudits.} for example, those with a phase
$$(-1)^{f_1(\gamma)}\cdot i^{f_2(\gamma)}\cdot (e^{i\pi/4})^{f_3(\gamma)}\cdot \cdots,$$
where $f_1,f_2,f_3,\dots$ are $\mathbb{F}_2$-valued functions.
One might naively think that we can classify each $f_i(\gamma)$ individually and combine the results (adding $t-1$ to the classification result of $f_t(\gamma)$ to account for the phase $e^{2\pi i/2^t}$ and taking the max). However, it is difficult to even transfer the above classification to the case of phases $i^{f_2(\gamma)}$, and this is because $\F_4$ is distinct from $\Z_4$.
Think of an $\F_{16}$ qudit expanded as $\gamma=\sum_{i=1}^s b_i\alpha_i$ for $b_i\in\{0,1\}$. Consider the phase gate $|\gamma\ra\mapsto i^{b_1b_2\oplus b_3b_4}|\gamma\ra$, where $\oplus$ denotes modulo $2$ addition. This gate is $\CS_{12}\CS_{34}\text{CCCZ}_{1234}$ on the four qubits, and therefore is in exactly the fourth level of the Clifford hierarchy. However, 
$$b_1b_2\oplus b_3b_4=\tr(\gamma\alpha_1)\tr(\gamma\alpha_2)+\tr(\gamma\alpha_3)\tr(\gamma\alpha_4)$$
only involves terms of degree at most $2$ in $\gamma$. Even with a $+1$ contribution accounting for base $i$, it will be wrongly classified into the third level.

Having said this, there might be a way to constrain each $f_i(\gamma)$ into a certain form, so that one can classify each term separately and combine the results, but the form of the above proof for real phases will be lost.

\end{remark}

\subsubsection{\texorpdfstring{$U_7$ gate}{U7 gate}}
\label{sec:U7_gate}

A single-qudit non-Clifford introduced in~\cite{wills2024} is the so-called $U_7$ gate.
\begin{equation}
U_7: |\gamma\ra\mapsto (-1)^{\tr(\gamma^7)}|\gamma\ra\;.
\end{equation}
We will also consider the $U_7^\beta$ generalization of it~\cite{he2025quantum}:
\begin{equation}
U_7^\beta: |\gamma\ra\mapsto (-1)^{\tr(\beta\gamma^7)}|\gamma\ra\;.
\end{equation}
As per Remark~\ref{remark:qudit_monomial_gate}, if $\tr_{\F_{2^s}/\F_{2^{\ell_7}}}(\beta)=0$ then this gate is effectively trivial, otherwise $U_7^\beta$ lies in the $\wt(7)=\wt(111_2)=3$-rd level of the Clifford hierarchy. Recall from Eq.~\eqref{eq:cyclotomic_coset} that 
\begin{equation}
\ell_7=\min\{j\ge 1\;|\; (2^j-1)\cdot 7\equiv 0\!\!\mod{q-1}\}\;,\; q=2^s\;.
\end{equation}
Besides the level of $U_7^\beta$, we also want to consider $U_7^\beta$ for different $\beta$, and their Clifford equivalence.

If $s$ is not divisible by three, then $\gcd(2^s-1,2^3-1)=1$.
Thus, in this case, any $\beta\in\F_{2^s}^\times$ can be written as $\eta^7$ for some $\eta\in\F_{2^s}$. So $U_7^\beta$ is Clifford equivalent to $U_7$ because $U_7^\beta:|\gamma\ra\to |\gamma\eta\ra\xrightarrow{U_7}|\gamma^7\eta^7\ra$ where the first step is a change of basis and is Clifford.
Therefore, if $s$ is not divisible by $3$, all $U_7^\beta$ are Clifford equivalent to each other as long as $\beta\neq 0$.

If $s$ is divisible by three, let us first determine $\ell_7$.
It is clear that for $s=3$ we have $\ell_7=1$, and thus $U_7^\beta$ becomes trivial if $\tr(\beta)=0$. Also, $U_7^\beta|\gamma\ra=(-1)^{\tr(\beta)\gamma^7}|\gamma\ra$, and thus all $U_7^\beta$ where $\tr(\beta)=1$ are the same gate, namely a $\CCZ$ gate supported on the three qubits constituting the qudit.

For $3\mid s$ and $s>3$, let us first show that $\ell_7=s$. Of course, $(2^s-1)\cdot 7 \equiv 0 \mod q-1$. On the other hand, suppose we have $(2^j-1)\cdot 7\equiv 0\!\!\mod{2^s-1}$ for $1\le j < s$. Write $M=(2^s-1)/7$, then $M\mid 2^s-1$ and $M\mid 2^j-1$ and thus $M\mid \gcd(2^s-1,\;2^j-1)=2^{\gcd(s,j)}-1$. Since $\gcd(s,j)\le s/2$ (given that $j<s$), and $M=\frac{2^s-1}{7}>2^{s/2}-1$ (this may be checked for $3 \mid s$ and $s > 3$), we have a contradiction. 

Thus, whenever $3 \mid s$ and $s > 3$, $\ell_7=s$, $U_7^\beta$ is a non-trivial third-level gate for every $\beta \neq 0$. Let $\alpha$ be a primitive root of $\F_{2^s}$. First notice that if $\beta_1=\beta_2\cdot \alpha^{7i}$ for some $i\in\Z_{\ge 0}$, then $U_7^{\beta_1}$ and $U_7^{\beta_2}$ are Clifford equivalent. This is because $\tr(\beta_1\gamma^7)=\tr(\beta_2(\alpha^i\gamma)^7)$ and according to Remark~\ref{remark:CH_arithmetic}, multiplying by the known constant $\alpha^i$ is Clifford. We therefore focus on the multiplicative group $\mathbb{F}_{2^s}^\times = \{1, \alpha, \alpha^2, \alpha^3, \ldots, \alpha^{q-2}\}$, where we emphasise $7 \mid q-1$. Different $\beta$'s that are equal under multiplication by some power of $\alpha^7$ lead to Cliford equivalent $U_7^\beta$, and so one only needs to focus on $U_7^\beta$ with $\beta \in \{1, \alpha, \alpha^2, \alpha^3, \ldots, \alpha^6\}$. Moreover, again from Remark~\ref{remark:CH_arithmetic}, we know that the Frobenius automorphism is Clifford. This further tells us that the three $U_7^\beta$ with $\beta \in \{\alpha, \alpha^2, \alpha^4\}$ are all Clifford equivalent, and similarly the $U_7^\beta$ with $\beta \in \{\alpha^3, \alpha^5, \alpha^6\}$ are Clifford equivalent. We conclude that for $3 \mid s$ and $s > 3$, there are at most $3$ Clifford inequivalent $U_7^\beta$.

The first non-trivial case of this is the $\F_{64}$ case shown in Fig.~\ref{fig:U7_expansion}.
For the three cases in $\F_{64}$, we perform minimum-weight decoding of a length $63$ Reed-Muller code to determine their minimal ancilla-free $T$ count~\cite{T_count_RM} and found them to be unequal, and therefore the three cases in $\F_{64}$ must be Clifford inequivalent to each other.
It is worth mentioning that, through brute-force basis change\footnote{For small fields $\F_{2^s}$, we may attempt to establish Clifford equivalence of two $\CCZ$ circuits by randomly sampling a basis transform described by an $\GL(s,\F_2)$ matrix $(A_{ij})$. Each variable $x_i$ occurring in the phase polynomial transforms as $x_i\mapsto\sum_{i}A_{ij}x_j$. We can then obtain the transformed phase polynomial (using $x_i^2=x_i$, $2$ times anything equals zero), removing any degree $\le 2$ terms (since they correspond to CZ and Z gates). We obtain Clifford equivalence if what is left over is the target phase polynomial. The procedure is illustrated in Fig.~\ref{fig:Clifford_equiv}(b). }, we observe that the case (\textbf{d3}) in Fig.~\ref{fig:U7_expansion} is Clifford equivalent to the $\F_4$-qudit CCZ gate plotted in Fig.~\ref{fig:F4_qudit_CCZ}(a).
\begin{figure}[ht]
\centering
\includegraphics[width=0.8\linewidth]{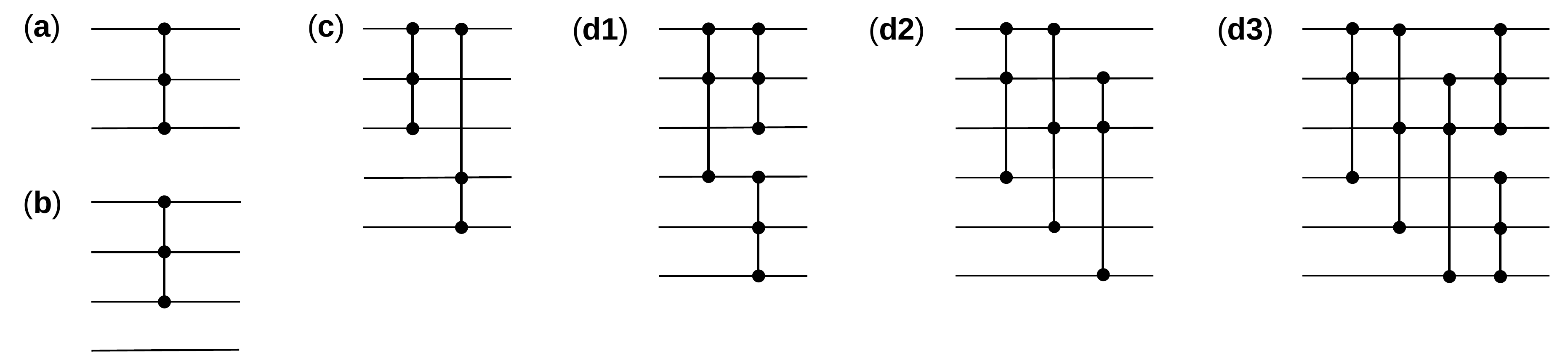}
\caption{$U_7^\beta$ gate expanded in qubit CCZ gates up to Clifford equivalence. (\textbf{a}) $U_7^\beta$ in $\F_8$ with $\beta$ satisfying $\tr(\beta)\neq 0$. If $\tr(\beta) = 0$, $U_7^\beta$ is the identity. (\textbf{b}) In $\F_{16}$, all $U_7^\beta$ for which $\beta \neq 0$ are Clifford equivalent to a single $\CCZ$. (\textbf{c}) In $\F_{32}$, all $U_7^\beta$ for which $\beta \neq 0$ are Clifford equivalent to a $\TOF\#$ gate. (\textbf{d1}-\textbf{d3}) In $\F_{64}$; (\textbf{d1}) For $\beta=\alpha^{7i}$; (\textbf{d2}) For $\beta=\alpha^{1+7i}$ or $\alpha^{2+7i}$ or $\alpha^{4+7i}$; (\textbf{d3}) For $\beta = \alpha^{3+7i}$ or $\alpha^{5+7i}$ or $\alpha^{6+7i}$. The three are Clifford inequivalent because their minimal ancilla-free $T$ count~\cite{campbell2017unified} is not equal. The $T$ count for (\textbf{d1}-\textbf{d3}) is $13,15,17$, respectively.}
\label{fig:U7_expansion}
\end{figure}

As a final remark, note that one can use a qudit-CCZ gate in $\F_{2^{3s}}$ to implement the $U_7$ gate in $\F_{2^s}$; copy $|\gamma\ra$ to the other two all-zero registers using CNOT gates and apply the Frobenius transform to raise the two copies to $|\gamma^2\ra$ and $|\gamma^4\ra$. Apply qudit-CCZ across the three registers, and finally undo the CNOT gates. Since the ancillae need to start in the zero state, this does not tell us that $U_7\otimes I\otimes I$ and qudit-CCZ are Clifford equivalent.

\section{Distillation: transversality conditions and protocol constructions}
\label{sec:protocols}

\subsection{Generalized triorthogonality on qubit codes and \texorpdfstring{$\TOF\#$}{Toffoli-Sharp} protocols}
\label{sec:generalized_triorthogonality}

The generalized triorthogonality framework~\cite[App.~D]{campbell2017unified} gives a flexible way to construct distillation protocols based on binary CSS codes. In particular, acting with $T$ gates transversally on the physical qubits, we can establish whether the codespace is preserved (possibly after Clifford corrections), and if so, what the induced logical action is.
Canonically, the binary CSS code is written into a matrix $G=\begin{pmatrix}G_1\\\hline G_0\end{pmatrix}\in\F_2^{m\times n}$ matrix, where rows of $G_1\in\F_2^{k\times n}$ and $G_0\in\F_2^{(m-k)\times n}$ describe the $X$-logicals and $X$ stabilizers of the $[[n,k,d]]$ code. From $G$, the procedure to check the logical action of the transversal $T$ gate on the physical qubits is efficient, only requiring checks on the weights of overlaps of various rows.

Let us now explain how this works, meanwhile slightly extending the framework with a metric $\bGamma=[\Gamma_1,\dots,\Gamma_n]\in\Z_8^n$. The physical meaning of the metric is that, instead of applying transversal $T$, we apply $T^{\Gamma_1}\otimes \cdots\otimes T^{\Gamma_n}$ to the physical qubits.

The binary matrix $G\in\F_2^{m\times n}$ is a generator matrix for some code. Its rows are denoted $\bg_1,\dots,\bg_k$, $\bg_{k+1},\dots,\bg_m$.
We associate binary variables $z_1,\dots,z_m$ with each row from $G$; codewords of the code for which $G$ is a generator matrix may be written $\sum_{i=1}^m z_i\bg_i$.

Applying the physical gate $T^{\Gamma_1}\otimes \cdots\otimes T^{\Gamma_n}$, the following transformation occurs:
\begin{equation}
|{\scriptstyle\sum}_{i=1}^m z_i\bg_i\ra \mapsto \exp\left[\frac{i\pi}{4}p(z_1,z_2,\cdots,z_m)\right]|{\scriptstyle\sum}_{i=1}^m z_i\bg_i\ra.
\end{equation}
The phase gained, as described by the phase polynomial $p(z_1,\cdots,z_m)\mod 8$ can be calculated using the general fact that \cite[Eq.~D7]{campbell2017unified} 
\begin{equation}
\label{eq:xor_to_sum}
\bigoplus_{i=1}^m a_i=\sum_{i=1}^m a_i - 2\sum_{1\le i<j\le m}a_ia_j + 4\sum_{1\le i<j<l\le m} a_ia_ja_k - 8\sum\cdots,
\end{equation}
where $\oplus$ denotes addition modulo $2$, and $a_i \in \{0,1\}$. The expression is up to $m$ terms, but the coefficient in front of the $m'$-th term is $2^{m'-1}$ ($m'<m$), and thus for us to determine the phase polynomial mod $8$, we can ignore all intersections involving more than three terms. We therefore have
\begin{align}
&p(z_1,\dots,z_m)
= \sum_{j=1}^n \Gamma_j \left(\bigoplus_{i=1}^m z_i g_{ij}\right)\nonumber\\
\equiv &\; \sum_{j=1}^n\Gamma_j\left[\sum_{i=1}^m z_i g_{ij} - 2\sum_{1\le i<l\le m}z_i z_l\; g_{ij}g_{lj} + 4\sum_{1\le i<l<r\le m}z_iz_lz_r\;g_{ij}g_{lj}g_{rj}\right]\mod 8\nonumber\\
\equiv &\; \sum_{i=1}^m \wt_{\bGamma} (\bg_i)z_i-2\sum_{1\le i<l\le m}\wt_{\bGamma}(\bg_i\wedge \bg_l)z_i z_l + 4\sum_{1\le i<l<r\le m}\wt_{\bGamma}(\bg_i\wedge \bg_l\wedge \bg_r)z_i z_l z_r\mod{8}\label{eq:phase_poly_expansion},
\end{align}
where $\bg_i\wedge \bg_j:=(g_{i1}g_{j1},\dots,g_{in}g_{jn})$ denotes entry-wise AND of the two binary vectors, and $\wt_{\bGamma}(\bu):=\sum_{j=1}^n \Gamma_j u_j$ denotes a Hamming weight in the metric $\bGamma$.

For a moment, think of the $X$ stabilizers as additional $X$ logicals fixed to the $|+\ra$ state and treat these $m$ logicals as $m$ wires. Then Eq.~\ref{eq:phase_poly_expansion} says that the logical circuit created on these $m$ wires is
\begin{enumerate}
\item $T^{\wt_{\bGamma}(\bg_i)}$ acting on logical qubit $i$, which becomes trivial if $\wt_{\bGamma}(\bg_i)\equiv0\mod 8$, and becomes Clifford if $\wt_{\bGamma}(\bg_i)\equiv0\mod 2$.
\item $\CS^{\wt_{\bGamma}(\bg_i\wedge\bg_l)}$ acting on logical qubits $i$ and $l$, which becomes trivial if $\wt_{\bGamma}(\bg_i\wedge\bg_l) \equiv 0\mod 4$, and Clifford if $\wt_{\bGamma}(\bg_i\wedge\bg_l) \equiv 0 \mod 2$.
\item $\CCZ^{\wt_{\bGamma}(\bg_i\wedge \bg_l\wedge \bg_r)}$ acting on logical qubits $i,l,r$, which becomes trivial if $\wt_{\bGamma}(\bg_i\wedge\bg_l) \equiv 0 \mod 2$.
\end{enumerate}
Therefore, the physical gate $T^{\Gamma_1}\otimes \cdots\otimes T^{\Gamma_n}$ is a valid logical gate if the non-trivial part of this circuit is only on the first $k$ wires (authentic logical qubits), which is equivalent to saying that $p(z_1,\dots,z_m)\mod 8$ depends only on $z_1,\dots,z_k$, but not $z_{k+1},\dots,z_m$ (the variables associated to the stabilizer rows).

However, as pointed out in~\cite{campbell2017unified}, when this binary code is used as a distillation protocol, it is enough to request only that the \emph{non-Clifford part} of the logical circuit lies on the first $k$ wires. This is because we can control the encoding circuit of this code. The encoding into logical all-plus state can be executed by, initializing wires $1$ to $m$ in $|+\ra$ and other $n-m$ wires in $|0\ra$ and then apply some $\CNOT$ circuit. The Clifford corrections (even in the case of a $\CZ$ gate between wire $1$ and $k+1$) can be done just before the $\CNOT$ circuit (but after the initialization).

We follow the same logic of generalized triorthogonality in deriving all our $\F_{2^s}$-code distillation protocols. Associate a variable with each row of the generator matrix, calculate the phase polynomials under some physical magic gates, and enforce that the \emph{non-Clifford part} of the phase polynomial be independent of the variables associated to the stabilizer rows.

We want to remark that it is straightforward to extend this framework for calculating the logical effect of applying $Z^{1/2^t}$ gates (or its power) to the physical qubits. One just needs to consider the phase polynomial modulo $2^{t+1}$, and thus change the cutoff in the expansion~\cite[Prop.~A.1]{webster2023transversal}
\begin{equation}
\label{eq:oplus_to_AND}
\bigoplus_{i=1}^m a_i=\sum_{s=1}^m 2^{s-1} \sum_{\substack{S\subseteq[m]\\|S|=s}}\prod_{j\in S}a_j.
\end{equation}
One immediately sees that only overlaps between every $\le (t+1)$ rows from $G$ contribute to the calculation.

\begin{construction}[$64T\to 2\TOF\#$ at $d=4$]\label{con:64T_to_2TOF} For this construction, we recall the preliminary material on Reed-Muller codes, see Section~\ref{sec:ring_field}, where we wish to consider binary Reed-Muller codes in the present case. In \cite{Haah2018}, several $T\to CCZ$ protocols based on Reed-Muller codes are developed. Now, we consider a length-$64$ Reed-Muller code, and letting $X$ logicals correspond to the $15$ degree-$2$ monomials $\{x_ix_j\}_{1\leq i < j \leq 6}$, and letting $X$ stabilizers be the degree $<2$ monomials, i.e., $(x_i)_{1 \leq i\leq 6}$, and the constant function $1$. This code has a $Z$-distance of $4$, and an $X$-distance of $16$.

To calculate the effect of the physical transversal $T$ gate on this code, we may consider the generalized triorthogonality framework, as well as the information about the overlaps and weights of Reed-Muller codewords given in the preliminaries. For this code, the logical action of the transversal gate turns out to be a $\CCZ$ circuit, with $\CCZ$ applied to every triplet of $X$-logicals (represented by monomials) multiplying to the full monomial $x_1x_2x_3x_4x_5x_6$, whose evaluation vector has Hamming weight one.\footnote{This logical \emph{hypergraph} magic state, as shown in Figure~\ref{fig:RM_2TOF_sharp}, has also been derived using other methods~\cite{optimality_CSS_T,barg2024rm}.} We can simplify this hypergraph state into $2\TOF\#$ gates, see Figure~\ref{fig:RM_2TOF_sharp}. 
One can set the logical qubits $3,4,7,8,12$ (wires associated to  $x_1x_4$, $x_1x_5$, $x_2x_4$, $x_2x_5$, $x_3x_6$) to $|0\ra$ prior to applying the physical $T$ gates.
$\CCZ$ gates that involve these qubits become trivial, and what is left is two decoupled $\TOF\#$ gates on logicals $\{1,10,11,14,15\}$ and $\{2,5,6,9,13\}$.
\end{construction}

\begin{figure}[htb]
    \centering
    \includegraphics[width=0.8\linewidth]{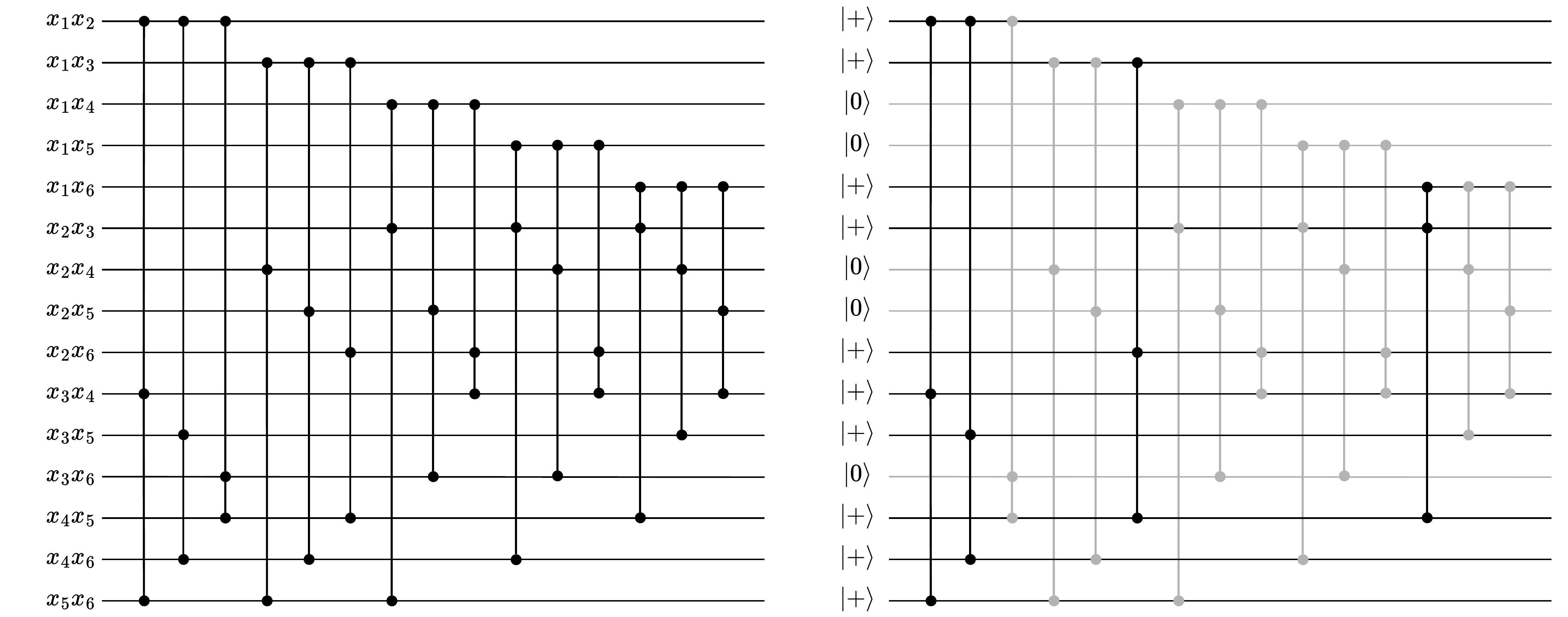}
    \caption{Left: the resulting logical circuit when applying transversal $T$ to the $[[64,15,4]]$ Reed-Muller code. The $15$ logicals are represented by the $\binom{6}{2}$ degree-two monomials in six variables. There is a $\CCZ$ for every three monomials that multiply to $x_1x_2x_3x_4x_5x_6$. Right: by initializing some of the logicals to the $|0\ra$ state, two disjoint $\TOF\#$ states can be obtained. The error suppression of this protocol is $10256p^4$.}
    \label{fig:RM_2TOF_sharp}
\end{figure}

\subsection{\texorpdfstring{Qudit $\CCZ$ gate protocols}{Qudit CCZ protocols}}\label{subsec:qudit-ccz-protocols}

Beyond $\mathbb{F}_2$, we may use matrices over higher fields, $G \in \mathbb{F}_q^{m \times n}$, to encode qudit quantum codes, where we only consider binary extension fields (the case $q = 2^s$) in our work, so that we can ultimately obtain qubit magic state distillation protocols. The first $k$ rows serve as logical $X$ operators, and the latter $m-k$ rows are $X$ stabilizer generators; see Definition~\ref{def:X_gen_matrix}. Generalizing the case for qubits, we denote the two sub-matrices given by the former $k$ rows and the latter $m-k$ rows as $G_1$ and $G_0$, and the $\mathbb{F}_q$-row span of $G_1,G_0,G$ as $\cG_1,\cG_0,\cG$ respectively.

\cite{nguyen2024,golowich2025asymptotically,nguyen_pattison} showed that if a qudit code encoded by a matrix $G$ satisfies the three-orthogonal property
\begin{equation}
\label{eq:three_orthogonal}
\sum_{i=1}^n (g_{a,i})(g_{b,i})(g_{c,i}) = \begin{cases}1 & \text{if }1\le a=b=c\le k\\0&\text{otherwise,}\end{cases}
\end{equation}
then applying transversal qudit CCZ gates to three patches of this qudit code leads to qudit CCZ gates on the logical level. A generalized version yielding a transversal qudit $\text{C}^{\ell-1}\text{Z}$ gate is found in Eq.~\eqref{eq:l_ortho_wrt_metric}~\cite{he2025asymptotically}; qudit $\CCZ$ is the special case of $\ell = 3$. 

We defer a discussion of our protocols distilling qudit CCZ states to Sec.~\ref{sec:AG_codes}, after we have introduced the necessary notions on algebraic geometry codes. The parameters of the corresponding finite-length constructions are summarized in Table~\ref{table:maximal_curves_CCZ_only}. Note that those protocols take qudit CCZ states as input and output, making them seem unwieldy, since synthesizing one input qudit CCZ state over $\mathbb{F}_{2^s}$ naively requires $O(s^3)$ qubit CCZ states. However, using the connection between CCZ states and multiplication in their corresponding field, it is in fact possible to synthesize such a qudit CCZ state with only $O(s)$ qubit CCZ states and Clifford operations; see the discussion in App.~\ref{sec:binary_field_multiplication}, as well as~\cite{ballet2011tensor}.

\subsection{\texorpdfstring{$U_7$ gate protocols}{U7 gate}}
\label{sec:U7_protocols}\label{sec:CCZ_distillation_protocols}
The $U_7^\beta$ gate, discussed in Section~\ref{sec:U7_gate}, is a non-Clifford gate~\cite{wills2024}. Moreover, they proved that if a qudit quantum code satisfies the following \emph{twisted three-orthogonal} criterion
\begin{equation}
\label{eq:U7_orthogonal}
\sum_{i=1}^n (g_{a,i})^4 (g_{b,i})^2 (g_{c,i}) = \begin{cases}1 & \text{if }1\le a=b=c\le k\\0&\text{otherwise,}\end{cases}
\end{equation}
then $U_7$ ($\beta=1$) is a transversal gate for this code. We consider here a more generalized framework that is useful for deriving protocols. If the code is twisted three-orthogonal with respect to $(\mathbf{\Gamma},\boldsymbol{\tau})$, that is,
\begin{equation}
\label{eq:U7_ortho_wrt_metric}
\sum_{i=1}^n \Gamma_i(g_{a,i})^4(g_{b,i})^2 (g_{c,i})=\begin{cases}\tau_{a} & \text{if } 1\le a=b=c\le k\\ 0&\text{otherwise}\end{cases},
\end{equation}
then $\forall\beta\in\F_q$, $U_7^{\beta\mathbf{\Gamma}}:=\otimes_{i=1}^nU_7^{\beta\Gamma_i}$ implements $\overline{U_7^{\beta\boldsymbol{\tau}}}:=\overline{\otimes_{i=1}^k U_7^{\beta\tau_i}}$.
The proof is almost identical to \cite{wills2024}, but before going into that, we remark that introducing the metric gamma $\mathbf{\Gamma}$ is a useful generalization when dealing with elliptic curve codes that are in general not self-orthogonal. Moreover, considering $\beta$ is important in $\F_{2^s}$ where $s$ is a multiple of $3$, see Section~\ref{sec:U7_gate}. This generalization is similar to the $\ell$-orthogonality framework developed in \cite{he2025quantum}:
\begin{equation}
\label{eq:l_ortho_wrt_metric}
\sum_{i=1}^n \Gamma_i(g_{a_1,i})(g_{a_2,i})\cdots (g_{a_\ell,i})=\begin{cases}\tau_{a_1} & \text{if } 1\le a_1=a_2=\cdots=a_\ell\le k\\ 0&\text{otherwise}\end{cases}.
\end{equation}
Notice that a matrix $G$ satisfying Eq.~\ref{eq:l_ortho_wrt_metric} satisfies Eq.~\ref{eq:U7_ortho_wrt_metric}. In~\cite{wills2024}, distillation protocols for $U_7^\beta$ gates were built by simply defining matrices $G$ that satisfied the stronger statement of Eq.~\ref{eq:l_ortho_wrt_metric} with $\ell=7$. However, this is inherently wasteful, because Eq.~\ref{eq:U7_ortho_wrt_metric} is sufficient for $U_7^\beta$ transversality, and requires the matrix to satisfy fewer conditions. For example, one does not need to check whether $\sum_{i=1}^n \Gamma_i (g_i^1)^4(g_i^2)(g_i^3)(g_i^4)$ is zero\footnote{This is because by lemma~\ref{lemma:lucas_theorem}, the multinomial coefficient $\binom{7}{4,1,1,1}$ is even and hence disappear in $\F_2$ arithmetic. In fact, there is no need to calculate the overlap of more than three rows.}. 
We will exploit this later (in Sec.~\ref{sec:U7_in_F8_protocols} and Con.~\ref{con:24CCZ-to-4CCZ-d3}), only aiming to satisfy Eq.~\ref{eq:U7_ortho_wrt_metric} in order to construct improved protocols for the $U_7^\beta$ gate.

\begin{lemma}
\label{lemma:U7_ortho_wrt_metric}
    On a quantum code over $\mathbb{F}_8$ constructed from a matrix satisfying Eq.~\eqref{eq:U7_ortho_wrt_metric}, the physical $U_7^{\beta\mathbf{\Gamma}}$ gate implements the logical gate $\overline{U_7^{\beta\boldsymbol{\tau}}}$.
\end{lemma}
\begin{proof}
Write a classical codeword of $\cG$ as a linear combination of its rows: $\boldsymbol{f}=\sum_{a=1}^m u_a\mathbf{g}_a\in\F_q^n$ for $u_a\in \F_q$, $a=1,\dots,m$. The $i$-th entry of this codeword is $f_i=\sum_{a=1}^m u_a g_{a,i}$ and so $f_i^7=(\sum_{a=1}^m u_a g_{a,i})^7$. Applying Eq.~\eqref{eq:U7_ortho_wrt_metric} and the multinomial formula (where note that all even multinomial coefficients vanish, and all odd multinomial coefficients reduce to $1$, since we work in a binary extension field), one has 
\begin{multline}
\sum_{i=1}^n \Gamma_i f_i^7=\sum_{a=1}^m\sum_{i=1}^n \Gamma_i(u_a g_{a,i})^7+\sum_{\substack{1\le a,b\le m\\ a\neq b}}\sum_{i=1}^n \Gamma_i [(u_ag_{a,i})^6(u_bg_{b,i})+(u_ag_{a,i})^5(u_bg_{b,i})^2+(u_ag_{a,i})^4(u_bg_{b,i})^3]\\ 
+\sum_{\substack{1\le a,b,c\le m \\ a,b,c\text{ pairwise distinct}}}\sum_{i=1}^n\Gamma_i (u_ag_{a,i})^4(u_bg_{b,i})^2(u_cg_{c,i})=\sum_{a=1}^k \tau_a u_a^7.
\label{eq:sum_seventh_power}
\end{multline}
Notice that many of the terms in the expression vanish by applying Eq.~\eqref{eq:U7_orthogonal} and contribute trivially. 
The effect of applying $U_7^{\beta\mathbf{\Gamma}}$ on the product state $|\boldsymbol{f}\ra=\otimes_{i=1}^n |f_i\ra$ is thus
\begin{equation*}
\otimes_{i=1}^n U_7^{\beta\Gamma_i}|\boldsymbol{f}\ra=(-1)^{\sum_{i=1}^n\tr(\beta \Gamma_i f_i^7)}|\boldsymbol{f}\ra=(-1)^{\tr(\beta\sum_{i=1}^n \Gamma_i f_i^7)}|\boldsymbol{f}\ra=(-1)^{\tr(\beta\sum_{a=1}^k \tau_a u_a^7)}|\boldsymbol{f}\ra
\end{equation*}
The logical action on $\overline{|\bu\ra}$ is thus
\begin{align*}
\otimes_{i=1}^n U_7^{\beta\Gamma_i}\overline{|\bu\ra}&=\sum\limits_{\bg\in\cG_0}\otimes_{i=1}^n U_7^{\beta\Gamma_i}\left|\sum\limits_{a=1}^k u_a\bg_a+\bg\right\ra\\
&=\sum\limits_{\bg\in\cG_0}(-1)^{\tr\left(\beta\sum_{a=1}^k \tau_a u_a^7\right)}\left|\sum\limits_{a=1}^k u_a \bg_a + \bg\right\ra\\
&=(-1)^{\tr(\beta\sum_{a=1}^k \tau_a u_a^7)}\;\overline{|\bu\ra}\\
&=\overline{\otimes_{i=1}^k U_7^{\beta\tau_i}}\;\overline{|\bu\ra},
\end{align*}
giving the result.
\end{proof}

\label{sec:U7_in_F8_protocols}
Throughout, when we say $\CCZ$, we mean the regular $3$-qubit $\CCZ$ gate unless otherwise specified. We emphasize that the $U_7$ gate in $\mathbb{F}_8$ is equivalent to a single such $\CCZ$ gate on the three qubits of the $\mathbb{F}_8$ qudit; see Section~\ref{sec:U7_gate}.
\begin{construction}[$8\CCZ\to 2\CCZ$ at $d=2$]
\label{con:8CCZ_to_2CCZ}
We construct an $8CCZ\to2CCZ\;(d=2)$ protocol derived from the $\F_8$ Reed-Solomon code. We choose the two rows $\{x^2,x\}$ to form the $X$ logical generators, and the $1$ row to form the $X$ stabilizer. Explicitly, we write $\F_8$ as $\mathbb{F}_2(\alpha)$, where $\alpha$ is an element satisfying $\alpha^3 = \alpha+1$, the matrix defining the qudit code is
\begin{equation}
\begin{matrix}
\ev(x^2)\\ \ev(x)\\ \ev(1)
\end{matrix}\quad
\begin{pmatrix}
0 & \alpha^2 & \alpha^4 & \alpha^6 & \alpha^8 & \alpha^3 & \alpha^5 & 1\\
0 & \alpha & \alpha^2 & \alpha^3 & \alpha^4 & \alpha^5 & \alpha^6 & 1\\\hline
1 & 1 & 1 & 1 & 1 & 1 & 1 & 1\\
\end{pmatrix} 
\label{eq:RS_8_to_2}
\end{equation}
These three rows $\{x^2,x,1\}$ satisfy Eq.~\eqref{eq:U7_orthogonal} with $k=2$. One can verify this by hand: notice that $x^7=1,\;\forall x\in\F_8^\times$ and hence $\sum_{x\in\F_8}(x^2)^7=\sum_{x\in\F_8}(x^1)^7=1$ and $\sum_{i=1}^8 (1)^7=8\cdot 1=0$. One can also verify that $\sum_{i=1}^8(g_{a,i})^4(g_{b,i})^2(g_{c,i})^1=0$ for $a,b,c\in\{0,1,2\}$ not all equal using $\sum_{i=1}^7\alpha^{bi}=0$ for $\gcd(7,b)=1$ (in this case $\alpha^b$ is also a primitive element of this field). For example,
$\sum_{i=1}^8 (g_{1,i})^6 (g_{2,i})^1=\sum_{i=1}^7 (\alpha^{2i})^{6} (\alpha^i)^1=\sum_{i=1}^7 \alpha^{13i}=0$. We generalize this construction to multi-control-$Z$ distillation in Con.~\ref{con:multi-control-Z-d2}, and an algebraic proof can be found there.
\end{construction}
Let us show concretely how the above matrix can be binarized to obtain a distillation protocol for the qubit $\CCZ$ gate. A self-dual normal basis for $\mathbb{F}_8$ is
\begin{equation}
\label{eq:F8_self_dual_normal_basis}
\alpha_0=\alpha^2+\alpha+1,\quad \alpha_1=\alpha + 1=\alpha_0^2,\quad \alpha_2=\alpha^2+1=\alpha_0^4,
\end{equation}
where $\alpha^3 = \alpha + 1$.
Each $\GF(8)$ qudit entry $\gamma$ in the generator matrix can be expanded into a $3 \times 3$ binary matrix as:
\begin{equation}
\begin{pmatrix}
\tr(\gamma\alpha_0\alpha_0) & \tr(\gamma\alpha_0\alpha_1) & \tr(\gamma\alpha_0\alpha_2)\\
\tr(\gamma\alpha_1\alpha_0) & \tr(\gamma\alpha_1\alpha_1) & \tr(\gamma\alpha_1\alpha_2)\\
\tr(\gamma\alpha_2\alpha_0) & \tr(\gamma\alpha_2\alpha_1) & \tr(\gamma\alpha_2\alpha_2)
\end{pmatrix}.
\label{eq:qudit_to_qubit_expansion}
\end{equation}
The resulting expanded qubit code of 
$\llbracket8,2,2\rrbracket_8$ is as follows. As always, the upper rows (in this case six) are $X$ logicals, and the latter rows (in this case three) are $X$ stabilizers.
\begin{equation}
\left(\begin{smallmatrix}
000&111&110&011&001&101&010&100\\
000&100&111&110&011&001&101&010\\
000&101&010&100&111&110&011&001\\[2pt]
000&001&111&101&110&010&011&100\\
000&011&100&001&111&101&110&010\\
000&111&101&110&010&011&100&001\\[2pt]\hline\\
100&100&100&100&100&100&100&100\\
010&010&010&010&010&010&010&010\\
001&001&001&001&001&001&001&001
\end{smallmatrix}\right)
\label{eq:RS_F8_binary}
\end{equation}
Applying $\CCZ$ to every adjacent triplet of qubits (columns) yield two logical $\CCZ$ gates, supported on rows $1,2,3$ and rows $4,5,6$ respectively.

\begin{remark}
There is an $8\CCZ\to 2\CCZ\;(d=2)$ protocol~\cite[Fig.~16]{chamberland2022building} constructed from three different binary codes, each removing one of the three logical rows of the $[[8,3,2]]$ Reed-Muller code. We find this protocol to be inequivalent to our Eq.~\eqref{eq:RS_F8_binary}. In fact, when concatenating with the $8T\to \CCZ\;(d=2)$ first-level factory, we find their resulting protocol to be equivalent to the $64T\to 2\CCZ\;(d=4)$ protocol from~\cite{Haah2018} (the equivalency can be proved similarly to our App.~\ref{sec:decreasing_monomial_code}). Their error suppression $2944p^4$ is different from what is obtained by concatenating our code with the $8T \to \CCZ$ code, which is $3136p^4$. Notice that the error coefficient is worse in this case, but need not be generally so.
\end{remark}

\begin{construction}[$9\CCZ\to 1\CCZ$ at $d=3$]
\label{con:9-to-1-CCZ}
The $9\to1\;(d=3)$ protocol is obtained by padding the row $x$ with a one and rows $x^2,1$ with a zero, see below. Padded $x^2$ is treated as logical, and $x,1$ are stabilizers.
\begin{equation}
\left(\begin{array}{cccccccccc}
0 & 0 & \alpha^2 & \alpha^4 & \alpha^6 & \alpha^8 & \alpha^3 & \alpha^5 & 1\\\hline
1 & 0 & \alpha & \alpha^2 & \alpha^3 & \alpha^4 & \alpha^5 & \alpha^6 & 1\\
0 & 1 & 1 & 1 & 1 & 1 & 1 & 1 & 1\\   
\end{array}\right)
\label{eq:projective_RS_equation}
\end{equation}
\end{construction}

\begin{remark}
    Another way to obtain a protocol with these parameters is via a projective Reed-Solomon code. The nine points of $\bbP_{\F_8}$ are $(0:1)$ and $\{(1:\gamma)\;|\;\gamma\in \F_8^\times\}$. One can construct a matrix as the evaluation of the homogeneous polynomials $XY$ (logical), $Y^2$, $X^2$ on these nine $\bbP_{\F_8}$ points\footnote{As we mentioned in Sec.~\ref{sec:affine_projective_space}, for a homogeneous function $F$ of degree $d$, one has $f^h(\lambda x_0,\dots,\lambda x_n)=\lambda^d\cdot f^h(x_0,\dots,x_n)$. So scaling by $\lambda\in\F_8^\times$ of a projective point leads to a scalar multiplication of the corresponding column by $\lambda$. However, since $\lambda^7=1$, the scaling has no effect in Eq.~\ref{eq:U7_ortho_wrt_metric}.}:
\begin{equation}
\begin{matrix}
\ev(XY)\\ \ev(Y^2)\\ \ev(X^2)
\end{matrix}\quad
\begin{pmatrix}
0 & 0 & \alpha & \alpha^2 & \alpha^3 & \alpha^4 & \alpha^5 & \alpha^6 & 1\\\hline
1 & 0 & \alpha^2 & \alpha^4 & \alpha^6 & \alpha^8 & \alpha^3 & \alpha^5 & 1\\
0 & 1 & 1 & 1 & 1 & 1 & 1 & 1 & 1\\   
\end{pmatrix}
\label{eq:projective_RS_equation_homo_eq_deg_2}
\end{equation}
\end{remark}

\begin{construction}[$14\CCZ\to 4\CCZ$ at $d=2$]
We can double the above $\F_8$ RS construction as follows to obtain a $14\to4\;(d=2)$ protocol
\begin{equation}
\left(\begin{array}{ccccccc|ccccccc}
\alpha^2 & \alpha^4 & \alpha^6 & \alpha^8 & \alpha^3 & \alpha^5 & 1 & 0&0&0&0&0&0&0\\
\alpha & \alpha^2 & \alpha^3 & \alpha^4 & \alpha^5 & \alpha^6 & 1 & 0&0&0&0&0&0&0\\
0&0&0&0&0&0&0 & \alpha^2 & \alpha^4 & \alpha^6 & \alpha^8 & \alpha^3 & \alpha^5 & 1\\
0&0&0&0&0&0&0 & \alpha & \alpha^2 & \alpha^3 & \alpha^4 & \alpha^5 & \alpha^6 & 1\\\hline
1 & 1 & 1 & 1 & 1 & 1 & 1 & 1 & 1 & 1 & 1 & 1 & 1 & 1\\
\end{array} \right)
\label{eq:doubled_RS_equation}
\end{equation}
\end{construction}

\begin{construction}[$16\CCZ\to 2\CCZ$ at $d=3$]\label{con:16to2CCZ}
A $16\to2\;(d=3)$ protocol can be obtained by doubling the projective RS code
\begin{equation}
\left(\begin{array}{cccccccc|cccccccc}
0 & \alpha & \alpha^2 & \alpha^3 & \alpha^4 & \alpha^5 & \alpha^6 & 1 & 0&0&0&0&0&0&0 & 0\\
0 & 0&0&0&0&0&0&0 & \alpha & \alpha^2 & \alpha^3 & \alpha^4 & \alpha^5 & \alpha^6 & 1 & 0\\\hline
1 & \alpha^2 & \alpha^4 & \alpha^6 & \alpha^8 & \alpha^3 & \alpha^5 & 1 & 0&0&0&0&0&0&0 & 0\\
0 & 0&0&0&0&0&0&0 & \alpha^2 & \alpha^4 & \alpha^6 & \alpha^8 & \alpha^3 & \alpha^5 & 1 & 1\\
0 & 1 & 1 & 1 & 1 & 1 & 1 & 1 & 1 & 1 & 1 & 1 & 1 & 1 & 1 & 0\\
\end{array} \right)
\label{eq:doubled_padded_RS_equation}
\end{equation}
\end{construction}

\begin{construction}[$(6k+2)\CCZ\to 2k\CCZ$ at $d=2$]\label{con:6k+2CCZto2kCCZ}
Inspired by the $(3k+8)T\to kT$ protocol from \cite{bravyi2015doubled} and the $(6k+2)T\to k$CCZ protocol from \cite{campbell2017unified} (especially the latter), here we give a $(6k+2)CCZ\to 2k CCZ$ protocol.
\begin{equation}
\left(\begin{array}{cccccc}
L & M & 0 & 0 & \cdots & 0\\
L & S & M & 0 & \cdots & 0\\
L & S & S & M & \cdots & 0\\
\vdots & \vdots & \vdots & \vdots & \ddots & \vdots\\
L & S & S & S & \dots & M\\\hline
1\;1 & \mathbf{1}_6 & \mathbf{1}_6 & \mathbf{1}_6 & \cdots & \mathbf{1}_6
\end{array}
\right)
\label{eq:asymp_one_third}
\end{equation}
where $L=\begin{pmatrix}0&1\\ 0& 1\end{pmatrix}$, $M=\begin{pmatrix}\alpha^2&\alpha^4&\alpha^6&\alpha^8&\alpha^{10}&\alpha^{12}\\\alpha&\alpha^2&\alpha^3&\alpha^4&\alpha^5&\alpha^6\end{pmatrix}$, $S=\begin{pmatrix}1&1&1&1&1&1\\1&1&1&1&1&1\end{pmatrix}$ and $\mathbf{1}_6=(1\;1\;1\;1\;1\;1)$.
\end{construction}

\begin{construction}[$21\CCZ\to 3\CCZ$ at $d=3$]
\label{con:21CCZ-to-3CCZ}
Considering gluing three projective lines $X=0$, $Y=0$, $Z=0$ in $\bbP^2_{\F_8}$ along a triangle. That is, we consider $21$ evaluation points
$$\{(0:1:\gamma)\;|\;\gamma\in \F_8^\times\} \cup \{(\gamma:0:1)\;|\;\gamma\in \F_8^\times\} \cup \{(1:\gamma:0)\;|\;\gamma\in\F_8^\times\}\;.$$
and we consider constructing a stabilizer matrix whose columns correspond to these evaluation points. The matrix is constructed as follows, with three logical $X$ rows and three stabilizer $X$ rows (and $21$ columns):
\begin{equation}
\begin{array}{c|ccc}
      & (0:1:\gamma) & (\gamma:0:1) & (1:\gamma:0)  \\
\hline
\ev(YZ^2) & \gamma^2 & 0 & 0 \\
\ev(ZX^2) & 0 & \gamma^2 & 0 \\
\ev(XY^2) & 0 & 0 & \gamma^2 \\\hline
\ev(X) & 0 & \gamma & 1 \\
\ev(Y) & 1 & 0 & \gamma \\
\ev(Z) & \gamma & 1 & 0.
\end{array}
\end{equation}
One checks that this satisfies Eq.~\eqref{eq:U7_orthogonal}. The distance of this code is three; indeed, one checks that it is at least three by observing that no two stabilizer columns are equal up to scalar multiplication.\footnote{Note that strictly speaking this verifies only that the $Z$-distance is at least three. However, the $Z$-distance is the relevant distance for distillation because one can twirl away $X$ errors~\cite{bravyi2012magic} (one may also check that the $X$-distance is greater than the $Z$-distance any way).}
\end{construction}

\begin{construction}[$64\CCZ\to 8\CCZ$ at $d=4$]
\label{con:64CCZ-to-8CCZ}
Here we construct a protocol using affine monomial codes (see the end of Sec.~\ref{sec:ring_field}) defined over the quotient ring $$\F_8[x,y]/(x^8-x,\;y^8-y).$$
There are $8^2=64$ points on the affine grid; denote them as $(x,y)\in \F_8^2$.
All the $X$-logicals and stabilizers we choose for the distillation matrix are in the form of evaluation of monomials $x^a y^b$, $0\le a,b\le 6$ over the affine grid $\F_8^2$.

To satisfy Eq.~\eqref{eq:U7_orthogonal}, we choose the $X$ stabilizers as $1,x,y,x^2,y^2$ (whose dual code is the hyperbolic code of designed distance $4$)\footnote{One can see that $xy$ cannot be an $X$ stabilizer (in the metric of all-one), because raised to its seventh-power, all coordinates of $x^7y^7$ sum up to one.}. We are able to find eight compatible logicals by solving for the maximum independent set in a certain hypergraph, constructed as follows.

We start by searching for all exponents\footnote{Notice here we can restrict ourselves to both $a,b$ being positive. This is because, we want $\sum_{(x,y)\in\F_8^2} (x^a y^b)^7=1$. The left-hand side can be rewritten into $\left(\sum_{x\in \F_8} x^{7a}\right)\cdot \left(\sum_{y\in\F_8} y^{7b}\right)$ and is nonzero only when both $a,b>0$.} $(a,b),\; 1\le a,b\le 6$ such that $x^a y^b$ together with the five chosen stabilizers, satisfy Eq.~\eqref{eq:U7_orthogonal} with $k=1$. There are $29$ such exponents $(a,b)$. Think of each of these exponents, as well as the five stabilizer exponents, as vertices. We then add a hyperedge to every triplet and pair of vertices that would cause a violation of Eq.~\ref{eq:U7_ortho_wrt_metric}. For example, for a triplet of exponents $(a_1,b_1),(a_2,b_2),(a_3,b_3)$, we add a hyperedge if $\sum_{(x,y)\in\F_8^2} (x^{a_1} y^{b_1})^1 (x^{a_2} y^{b_2})^2 (x^{a_3} y^{b_3})^4\neq 0$. These (hyper)edges specify which simultaneous choice of exponents are forbidden. Therefore, finding a maximum number of logicals is equivalent to finding a maximum independent set within this forbidden hypergraph. While finding a maximum independent set in a hypergraph is NP-hard in general, it is tractable in this small instance via a depth-first search. As a result, we find three maximum independent sets all of size $8$:
\begin{align}
\{(1,2),\;(1,5),\;(2,5),\;(2,6),\;(4,2),\;(4,6),\;(5,1),\;(5,5)\},\\
\{(1,3),\;(2,2),\;(2,4),\;(3,1),\;(3,4),\;(4,2),\;(4,3),\;(4,4)\},\label{eq:64CCZ_symmetric_choice}\\
\{(1,5),\;(2,1),\;(2,4),\;(5,1),\;(5,2),\;(5,5),\;(6,2),\;(6,4)\}.
\end{align}
\begin{remark}
    The logical exponent set in Eq.~\eqref{eq:64CCZ_symmetric_choice} is symmetric with respect to $x$ and $y$, and in this case, we can use the footprint bound (see App.~\ref{sec:decreasing_monomial_code}) to show that the $X$-distance of the protocol is lower bounded by $16$.
\end{remark}
\end{construction}

\subsection{Norm gate protocols}

Besides $\tr(\cdot)$, the norm function $\Nm(\cdot)$ also maps $\GF(2^s)$ elements to $\{0,1\}$; it takes a particularly simple form in $\GF(2^s)$: $\Nm(\gamma):=\prod_{i=0}^{s-1}\gamma^{2^i}=\gamma^{2^s-1}$. Note that $\Nm(\gamma)=0$ if and only if $\gamma=0$ and else $\Nm(\gamma)=1$, which follows from the fact that the multiplicative group of a finite field is cyclic.
We define the norm gate as follows,
\begin{equation}
|\gamma\ra\mapsto (-1)^{\Nm(\gamma)}|\gamma\ra=(-1)^{\gamma^{2^s-1}}|\gamma\ra.
\end{equation}
That is to say, $|\gamma\ra$ gains a phase of $(-1)^0=1$ when $\gamma=0$ and $-1$ when $\gamma\neq 0$. One can see that the norm gate, as an $s$-qubit gate, is equivalent to the $\text{C}^{s-1}\text{Z}$ gate, up to Paulis, regardless of the way the qudit is decomposed into qubits. Indeed, to implement the norm gate on the qubit level (up to a global phase), expand $\gamma=\sum_{i=1}^s b_i\alpha_i$. One first applies qubit-wise $X$ gates on the $s$ qubits, such that $b_i\mapsto b_i\oplus 1$, then a $\text{C}^{s-1}\text{Z}$ gate, and finally the $X$ gates again. Then, only in the case of $b_1=\cdots=b_s=0$ does the intermediate state gain a phase of $-1$.

Associating $u_1,\cdots,u_m$ with each row from $G$, then the computational basis state
$|\sum_{i=1}^m u_i\bg_i\ra$ corresponding to a codeword of the classical code gains a phase of $(-1)^{p(u_1,\cdots,u_m)}$, where the phase polynomial reads
\begin{align}
p(u_1,\cdots,u_m) &= \sum_{j=1}^n \left(\sum_{i=1}^mu_i g_{i,j}\right)^{2^s-1}.
\end{align}
To expand this, we use the multinomial expansion\footnote{In the following equation, the sum is over tuples of non-negative integers $(e_i)_{i=1}^m$ that sum to $2^s-1$.}
\begin{equation}
(a_1+a_2\cdots+a_m)^{2^s-1}=\sum\limits_{e_1+\cdots +e_m=2^s-1}\binom{2^s-1}{e_1,e_2,\cdots,e_m} a_1^{e_1}\cdots a_m^{e_m}
\end{equation} and Lemma~\ref{lemma:lucas_theorem}, which states that the multinomial coefficient
\begin{equation}
\binom{2^s-1}{e_1,e_2,\cdots,e_m}=\begin{cases}\text{odd} & \text{no carrying when adding }e_1+\cdots +e_m=2^s-1 \text{ in base }2\\\text{even} & \text{otherwise}\end{cases}.\footnote{Recall that saying that there is no carrying when adding $e_1 + \ldots e_m = 2^s-1$ in base $2$ is equivalent to saying that there is no overlap in the support of the binary expansion of the $e_i$. Note that the binary expansion of $2^s-1$ is simply $11\ldots 1$.}
\end{equation}
We therefore have
\begin{align}
p(u_1, \ldots, u_m)&=\sum_{\substack{e_1+\ldots + e_m = 2^s-1\\\supp(\bin(e_i))\text{ pairwise disjoint}\\ \sqcup_{i=1}^{m}\supp(\bin(e_i))=\{1,\cdots,s\}}} \left(\prod_{i=1}^m u_i^{e_i}\right)\left[\sum_{j=1}^n \prod _{i=1}^m g_{i,j}^{e_i}\right],\label{eq:norm_gate_phase_poly}
\end{align}
where $\bin(\cdot)$ denotes the binary expansion of the integer $e_i$, and $\supp(\bin(\cdot))$ denotes the position of the ones in the binary expansion.

For the logical action to also be a norm gate on each of the $k$ logical qudits, the phase polynomial should simplify to $\sum_{i=1}^k u_i^{2^s-1}$.

One may check that this is ensured by the following condition, which we call \emph{twisted $s$-orthogonality}
\begin{equation}
\label{eq:norm_gate_condition}
\sum_{i=1}^n (g_{a_1, i})^{2^{0}}(g_{a_2, i})^{2^{1}}\cdots (g_{a_s, i})^{2^{s-1}}=\begin{cases} 1 & \text{if }1\le a_1=a_2=\cdots=a_s\le k\\0 & \text{otherwise}\end{cases}
\end{equation}
Note that the $U_7$ gate is the norm gate in $\mathbb{F}_8$, and indeed we can see that this condition specializes to the $U_7$ gate's transversality condition (see Eq.~\ref{eq:U7_orthogonal}) when we set $s=3$.

\begin{construction}[$2^s \text{C}^{s-1}\text{Z}\to 2\text{C}^{s-1}\text{Z}$ at distance $2$]
\label{con:multi-control-Z-d2}
Take the Reed-Solomon code over $\F_{2^s}$ and use $x^2,x^1$ as $X$ logicals, and $1$ as the $X$ stabilizer. Let us show that these three rows satisfy the the above norm-gate orthogonality condition Eq.~\eqref{eq:norm_gate_condition}.

Call $\alpha$ a primitive element of this field. First note that 
\begin{equation}
\sum_{i=1}^{2^s-1}\alpha^{ij}=\begin{cases}1& \text{if }j\equiv 0\mod{2^s-1}\\0 & \text{otherwise}\end{cases}
\end{equation}
This is to say, the entry-wise sum of the evaluation vector $\ev(x^j)\in\F_{2^s}^{2^s}$ is one if and only if $j\equiv 0\mod{2^s-1}$.
Therefore, the logical rows $x^2$ and $x^1$ each satisfy $\sum_{i=1}^n (g_{a,i})^{2^0}(g_{a,i})^{2^1}\cdots(g_{a,i})^{2^{s-1}}=\sum_{i=1}^n (g_{a,i})^{2^s-1}=1$ for $a=1,2$. 

For the other conditions in Eq.~\eqref{eq:norm_gate_condition}, suppose $\ev(x^2)$ is raised entry-wise to the $e_1$-th power, $\ev(x^1)$ to the $e_2$-th power, and $\ev(1)$ to the $e_3$-th power, where $e_1+e_2+e_3=2^s-1$, and $\supp(\bin(e_1))\sqcup\supp(\bin(e_2))\sqcup\supp(\bin(e_3))=\{1,\dots,s\}$. Multiplying these and summing, we are simply considering the sum of the evaluation vector $\ev(x^{2e_1+e_2})$. As above, we know that this is zero unless $2^s-1 \mid 2e_1+e_2$. Notice that $2e_1+e_2=e_1+(e_1+e_2)\le2(2^s-1)$ and the equality holds if and only if $e_1=2^s-1$ and $e_2=e_3=0$, which has already been discussed. Also $2e_1+e_2=0$ if and only if $e_3=2^s-1$ and $e_1=e_2=0$, which again has already been discussed. The only remaining possibility to consider is $2e_1+e_2 = 2^s-1$. One can actually see that the already-discussed case of $(e_1, e_2, e_3) = (0,2^s-1,0)$ is the only possibility here. Indeed, aside from that case, $e_1$ must be non-zero. Let $i$ be the position corresponding to the smallest power of $2$ that is supported in the binary expansion of $e_1$. Since $\supp(e_1)\cap\supp(e_2) = \emptyset$, and since $2e_1$ has the binary expansion of $e_1$ shifted by one, the binary expansion of $2e_1+e_2$ is not supported at position $i$, meaning that $2e_1+e_2 = 2^s-1 = 11\ldots 1_2$ is impossible.

\end{construction}
\begin{construction}[$(2^s+1) \text{C}^{s-1}\text{Z}\to 1\text{C}^{s-1}\text{Z}$ at distance $3$]\label{con:d3multicontrol}
One can pad the above Reed-Solomon to its projective version just as in Construction~\ref{con:9-to-1-CCZ}. That is, the $\ev(x^2)$ row is padded with a $0$ and treated as a logical, the $\ev(x)$ row is padded with a $1$ and treated as a stabilizer, and the $\ev(1)$ row is padded with a $0$ and treated as a stabilizer. The satisfaction of Equation~\ref{eq:norm_gate_condition} and the distance follows from the same argument there.
\end{construction}

\subsection{\texorpdfstring{Qubit $\CS$ protocols from \texorpdfstring{$\mathbb{F}_4$}{}-linear codes}{Qubit CS protocols from GF(4)-linear codes}}
\label{sec:CS_distillation_protocols}

Here, we only consider codes over $\F_4=\{0,1,\omega,\omega^2\}$. As a reminder, we have $\omega^2=\omega+1$ and $\omega^3=1$, $\tr(1)=0=\tr(0)$, and $\tr(\omega)=\tr(\omega^2)=1$. We can expand an $\F_4$ qudit using two qubits in the self-dual basis $\{\omega,\omega^2\}$, and consider applying a $\CS$ gate between these two qubits.
\begin{equation}
\gamma=b_1\omega+b_2\omega^2,\quad\CS|\gamma\ra:=\CS|b_1\ra|b_2\ra=i^{b_1b_2}|\gamma\ra.
\end{equation}
One can extract $b_1,b_2$ as $b_1=\tr(\gamma\omega)$ and $b_2=\tr(\gamma\omega^2)$ using the self-duality of the basis. Then, $b_1b_2=(\gamma\omega+\gamma^2\omega^2)(\gamma\omega^2+\gamma^2\omega^4)=\gamma^3+\tr(\gamma)$. Now, note that in this expression, $\gamma^3, \tr(\gamma) \in \{0,1\}$ are summed as elements of $\mathbb{F}_4$. Since they are exponentiated with base $i$, we need to convert the addition to $\mathbb{Z}_4$ addition. With $\oplus$ denoting the addition in $\mathbb{F}_4$, and $+$ denoting the addition in $\mathbb{Z}_4$, one checks that $\gamma^3 \oplus \tr(\gamma) = \gamma^3-\tr(\gamma)\pmod 4$.
Therefore,
\begin{equation}
\CS:|\gamma\ra\mapsto i^{\gamma^3 - \tr(\gamma)}|\gamma\ra,
\end{equation}
with the subtraction performed in $\mathbb{Z}_4$.
The condition for physical $\CS$ to implement logical $\CS$ is
\begin{equation}
\label{eq:CS-to-CS}
\sum_{i=1}^n g_{a,i}\;g_{b,i}\;g_{c,i}=\begin{cases}1 & \text{if }1\le a=b=c\le k\\ 0 & \text{otherwise}\end{cases},
\end{equation}
in terms of the matrix $G$ over $\F_4$ that defines the protocol.
The proof is in App.~\ref{sec:CS_transversality_condition}; we show two ways of deriving the condition. The first way in App.~\ref{sec:embedded_code} converts the calculation to one of generalized triorthogonality in a modified binary code and is most easy to understand. The second way in App.~\ref{sec:CS_direct_calculation} is much more involved but also more insightful, since we directly compute the phase polynomial over $\Z_4$ (instead of $\F_4$).

Notice that this condition is the same as for the $\F_4$-qudit CCZ, i.e., three-orthogonality. Therefore, the following protocols presented for $\CS$-to-$\CS$ distillation can also be used for $\F_4$-qudit CCZ distillation. 

We now present our first construction for $\CS$-to-$\CS$ distillation using codes over $\mathbb{F}_4$.
\begin{construction}[$4\CS\to 1\CS$ at $d=2$]
\label{con:4-to-1}
The Reed-Solomon code over $\F_4$ naturally satisfies the above three-orthogonality condition.
\begin{equation}
\label{eq:4-to-1-RS}
\begin{pmatrix}
0&1&\omega&\omega^2\\\hline
1&1&1&1
\end{pmatrix}
\end{equation}
\end{construction}
It turns out that any of our protocols that takes $\CS$ as input can be turned into a protocol that takes $T$ (and $T^\dagger$) as input, as follows. Given a qubit code, imagine we are about to apply physical $CS$ gates on disjoint pairs of physical qubits, such as for our binarized codes over $\mathbb{F}_4$. There is a way to turn these two-qubit gates into transversal $T/T^\dagger$ gates on an embedded code~\cite{webster2023transversal}, constructed as follows. For each such pair, add one ancilla qubit initialized in $|0\ra$, and perform $\CNOT$ gates from the two qubits onto this ancilla (see the circuit below). Then, when acting with $T\otimes T\otimes T^\dagger$ on this embedded code, it gives a phase of $x_1 + x_2 + 7(x_1\oplus x_2) = 2x_1x_2\mod 8$. This means that a $\CS$ gate has been performed on the two non-ancillary qubits.
\begin{center}
\begin{quantikz}[row sep={0.8cm,between origins}, column sep={0.7cm,between origins}]
\lstick{$x_1$} & \ctrl{2} & \qw       & \gate{T}      & \qw       & \ctrl{2} & \qw \\
\lstick{$x_2$} & \qw      & \ctrl{1}  & \gate{T}      & \ctrl{1}  & \qw      & \qw \\
\lstick{$\ket{0}$} & \targ{}   & \targ{}    & \gate{T^\dagger} & \targ{}    & \targ{}   & \qw
\end{quantikz}
\;=\;
\begin{quantikz}[row sep={0.7cm,between origins}, column sep={0.7cm,between origins}]
\lstick{$x_1$} & \ctrl{1}  & \qw \\
\lstick{$x_2$} & \gate{S}  & \qw
\end{quantikz}
\end{center}
Using this idea, it is easy to determine the $X$ logical/stabilizer matrix of the embedded code, simply by stabilizer propagation. Note that the logical action of physical $T/T^\dagger$ on the embedded code directly translates back to the logical action on the original code. It should be clear that the distance of the embedded code is lower bounded by that of the original code, because the weight after the stabilizer propagation is non-decreasing.

Let us demonstrate the idea using our Reed-Solomon example:
\begin{equation}
\label{eq:gadgetize_CS_to_3T}
\begin{pmatrix}
0&1&\omega&\omega^2\\\hline
1&1&1&1
\end{pmatrix} 
\xrightarrow{\text{binarize}}
\begin{pmatrix}
0 0 & 1 0 & 0 1 & 1 1\\
0 0 & 0 1 & 1 1 & 1 0\\\hline
1 0 & 1 0 & 1 0 & 1 0\\
0 1 & 0 1 & 0 1 & 0 1
\end{pmatrix}
\xrightarrow{\text{gadgetize}}
\begin{pmatrix}
0 0 0 & 1 0 1 & 0 1 1 & 1 1 0\\
0 0 0 & 0 1 1 & 1 1 0 & 1 0 1\\\hline
1 0 1 & 1 0 1 & 1 0 1 & 1 0 1\\
0 1 1 & 0 1 1 & 0 1 1 & 0 1 1
\end{pmatrix}
\end{equation}
One can use~\eqref{eq:phase_poly_expansion} to verify that $TTT^\dagger TTT^\dagger TTT^\dagger TTT^\dagger$ on this code, i.e., transversal $T$ with metric 
$$\bGamma=(1,1,-1,1,1,-1,1,1,-1,1,1,-1),$$
indeed gives logical $\CS$ (up to Clifford corrections).

Doubling the above $4\to 1$ construction gives a $6\text{CS}\to 2\text{CS}\;(d=2)$ protocol.
$$\begin{pmatrix}
1&\omega&\omega^2&0&0&0\\
0&0&0&1&\omega&\omega^2\\\hline
1&1&1&1&1&1
\end{pmatrix}$$
Adding the all-one stabilizer row to the second logical row, one can see that this protocol is the $k=2$ special case of the following construction.

\begin{construction}[$(2k+2)\CS\to k\CS$ at $d=2$]
\label{con:2k+2-to-k-CS}
Again inspired by \cite[Sec.~IV.E]{campbell2017unified}, one can design a family of distance-two protocols of asymptotic overhead $2$.
\begin{equation}
\left(\begin{array}{cccccc}
0\;1 & \omega\;\omega^2 & 0\;0 & 0\;0 & \cdots & 0\;0\\
0\;1 & 1\;1 & \omega\;\omega^2 & 0\;0 & \cdots & 0\;0\\
0\;1 & 1\;1 & 1\;1 & \omega\;\omega^2 & \cdots & 0\;0\\
\vdots & \vdots & \vdots & \vdots & \ddots & \vdots\\
0\;1 & 1\;1 & 1\;1 & 1\;1 & \dots & \omega\;\omega^2\\\hline
1\;1 & 1\;1 & 1\;1 & 1\;1 & \cdots & 1\;1
\end{array}
\right)
\label{eq:2k+2CS-to-kCS}
\end{equation}

\end{construction}

Similar to \cite[Lemma~7]{nezami2022classification}, we show in App.~\ref{sec:optimality_d=2_CS-to-CS} that $2$ is asymptotically optimal for $d=2$ $\CS$-to-$\CS$ protocols derived from $\F_4$-linear triorthogonal codes.

One can concatenate these schemes with themselves to obtain higher distance protocols. For example, concatenating the $4\to 1$ scheme with itself leads to a $16\CS\to 1\CS$ protocol at $d=4$, and more generally, one can have a $4^t\CS\to 1\CS$ at distance $d=2^t$. On the other hand, the situation is more subtle when $k>1$ protocols are involved in the concatenation. For example, one might naively think that first distilling $4\CS$ states from the $10\text{CS}\to 4\text{CS}$ protocol, then further distilling them using the $4\to 1$ protocol would yield a $10\to 1$ protocol at distance $4$. This is not true in general ($d=2$ in this case, in fact), because the errors on the outputs of the $k>1$ protocols are correlated. Protocols with $k>1$ can, however, be used with care in concatenation. For example, a $24 \to 2;(d=4)$ protocol is possible by beginning with four $6 \to 2$ protocols, and then using two copies of the $4 \to 1$ protocol, where each run of the $6\to 2$ protocol is routed to the two second-level protocols. We now list some possible schemes: some conatenated, and some not.

\begin{construction}[$16\CS\to 1\CS$ at $d=4$, more generally, $4^t\CS\to 1\CS$ at $d=2^t$]\label{con:12cs1csd3}
One can either concatenate the $4\to1\;(d=2)$ protocol (construction~\ref{con:4-to-1}) with itself, or use the two-variable Reed-Muller code over $\F_4$ to obtain a distance-four $16\CS\to 1\CS$ protocol. For the latter construction, take $xy$ as the $X$ logical, and $y^3,y^2,y,x,1$ as $X$ stabilizers. The two methods give equivalent protocols. In fact, it turns out that this is a general connection. The $4^t \to 1\;(d=2^t)$ protocol obtained by iteratively concatenating $4 \to 1$ with itself is always equivalent to a protocol constructed using a decreasing monomial code; we provide the exact construction and proof in App.~\ref{sec:decreasing_monomial_code}.
\end{construction}
\begin{construction}[$12\CS\to 1\CS$ at $d=3$]
\label{con:12CS-to-1CS}
We first provide the matrix and then explain how we obtained it.
\begin{equation}
\label{eq:12-to-1-CS-d3}
\begin{pmatrix}
0 & 0 & 1 & 0 & 0 & 1 & 0 & \omega^2 & 1 & \omega & \omega^2 & \omega\\\hline
0 & 0 & 0 & 0 & 1 & 1 & 1 & 1 & 1 & 1 & 1 & 1\\
0 & 1 & 1 & 1 & \omega & \omega^2 & 0 & 0 & 0 & 1 & 1 & 1\\
1 & 1 & \omega & \omega^2 & 0 & 0 & 1 & \omega & \omega^2 & 1 & \omega & \omega^2\\
\end{pmatrix}
\end{equation}
The $X$ stabilizer space is obtained as follows. One first writes down the $21$ points of the projective space $\bbP^2_{\F_4}$ as the columns. That is, we consider the $21$ points of $\bbP^2_{\F_4}$, written canonically as 
\begin{equation}
  \{(1:y:z):y,z\in\F_4\}
  \mathbin{\cup}
  \{(0:1:z):z\in\F_4\}
  \mathbin{\cup}
  \{(0:0:1)\},
  \label{eq:canonical-P2-F4}
\end{equation}
and consider a matrix in $\F_4^{3 \times 21}$ whose rows are $\ev(X)$, $\ev(Y)$ and $\ev(Z)$ on these points. Of course, given that points in the projective space are only defined up to scalar multiplication by non-zero elements of $\F_4$, we have only defined this matrix up to scalar multiplication of any column by a non-zero element of $\F_4$. It turns out that any choice works, however.\footnote{Any choice works because Equation~\ref{eq:CS-to-CS} remains true even after re-scaling a column by a non-zero element of $\F_4$ (since for any $\eta \in \F_4^\times$, we have $\eta^3 = 1$). Note that the condition for $U_7$-orthogonality in $\F_8$ is special for the same reason.} At this point, we have a $3 \times 21$ matrix over $\F_4$, which is a parity-check matrix for a code of distance $3$, since no two distinct points $(x:y:z)$ and $(x':y':z')$ in $\bbP^2_{\F_4}$ can be related by a scalar, by definition.

We next consider the nine points in $\bbP^2_{\F_4}$ lying on the Hermitian curve over $\F_4$ in Hermitian space, that is, the $9$ points $(x:y:z)$ satisfying $Y^2Z+YZ^2 = X^3$ (we introduce this curve in detail later in Sec.~\ref{sec:intro_to_AG_codes}). Removing the corresponding $9$ columns from the matrix yields a $3 \times 12$ matrix over $\F_4$ whose rows are $3$-orthogonal.\footnote{When we say the rows of a matrix are $3$-orthogonal, we mean that those rows satisfy Eq.~\ref{eq:CS-to-CS} with $k=0$. It turns out that the original $3 \times 21$ matrix is $3$-orthogonal, and the $3 \times 9$ sub-matrix formed by the $9$ points on the Hermitian curve is $3$-orthogonal. In turn, the $3 \times 12$ matrix that we have formed by removing those $9$ points is $3$-orthogonal.} 

Given these stabilizer rows, we wish to add as many logical rows as possible, that is, we wish to complete the code to a maximal triorthogonal space~\cite{Haah2018}. To do this, we can consider a matrix whose rows are the vectors $(g_{i,1}g_{j,1}, \ldots, g_{i,n}g_{j,n})$ for every pair of (not necessarily distinct) indices $(i,j)$ labelling two stabilizer rows. The kernel of this matrix contains vectors $(s_1, \ldots, s_n)$ that can be added as logical rows, as long as they satisfy $\sum_{i=1}^n s_i^3 \neq 0$. We find that the resulting kernel has dimension $6$, including the three-dimensional $X$-stabilizer space itself. It also turns out that one cannot add more than one logical row in this case, because adding one logical row adds in more constraints, which reduces the kernel to just the three-dimensional $X$-stabilizer space, that is, we have found a maximal triorthogonal space. One choice of the logical (which satisfies $\sum_{i=1}^ns_i^3 = 1$) is chosen and displayed in the resulting matrix. The protocol has distance $3$ since the $X$-stabilizer rows are a parity-check matrix for a code of distance $3$, as discussed.

\end{construction}

One may wonder what would have happened if we removed a different set of points from the original $21$. For all $n < 12$, we enumerate all the possible sets of $n$ points in the projective space and repeat the procedure. In each case, we are not able to add a logical row, and so we cannot find an $n \to 1\;(d=3)$ $\CS$-to-$\CS$ protocol with these methods, although this does not rule out that one exists (we constrained ourselves to three-dimensional stabilizer spaces).

The kernel method above can be applied iteratively in order to add multiple logical rows: randomly select one codeword from the appropriate kernel space whose triple self-intersection is one, build a new set of linear equations and solve for the kernel. We repeat until there is no valid logical row.

The following two constructions are found by mixing the iterative kernel search together with geometric intuitions.

\begin{construction}[$16\CS\to 2\CS$ at $d=3$]\label{con:16cs2csd3}
We begin by considering the regular $\F_4$ Reed-Muller codes in two variables; the $16$ points of the affine plane $\A^2_{\F_4}$ are taken as evaluation points. The evaluations of $1,x,y$ naturally serve as stabilizers, and $xy$ is a logical row. However, no other monomial rows can be added as a logical row. It turns out, however, that it is possible to add a further row by considering linear combinations of monomials.
Our kernel search, which searches over the $\F_4$-\emph{linear combinations} over the monomials (spanned by $\{x^2,y^2,x^3,y^3\}$ together with the existing rows in this case), gives many possible solutions, including the example shown, which is $\ev(y^3+\omega x^3)$. It further turns out that, no matter which extra logical row is picked, the new kernel reduces to the three-dimensional stabilizer space, i.e., we cannot add a third logical row.
\begin{equation}
\begin{matrix}
\ev(y^3+\omega x^3)\\
\ev(xy)\\
\ev(x)\\
\ev(y)\\
\ev(1)
\end{matrix}\quad
\begin{pmatrix}
0 & 1 & 1 & 1 & \omega & \omega^2 & \omega^2 & \omega^2 & \omega & \omega^2 & \omega^2 & \omega^2 & \omega & \omega^2 & \omega^2 & \omega^2\\
0 & 0 & 0 & 0 & 0 & \omega^2 & 1 & \omega & 0 & 1 & \omega & \omega^2 & 0 & \omega & \omega^2 & 1\\\hline
0 & 0 & 0 & 0 & \omega & \omega & \omega & \omega & \omega^2 & \omega^2 & \omega^2 & \omega^2 & 1 & 1 & 1 & 1\\
0 & \omega & \omega^2 & 1 & 0 & \omega & \omega^2 & 1 & 0 & \omega & \omega^2 & 1 & 0 & \omega & \omega^2 & 1\\
1 & 1 & 1 & 1 & 1 & 1 & 1 & 1 & 1 & 1 & 1 & 1 & 1 & 1 & 1 & 1
\end{pmatrix}
\label{eq:16CS-2CS-affine}
\end{equation}

\end{construction}

\begin{construction}[$21\CS\to 4\CS$ at $d=3$]
\label{con:21CS-to-4CS}
We begin by writing down all the $21$ points of $\bbP^2_{\F_4}$ in~\eqref{eq:canonical-P2-F4} as columns to form a set of $X$ stabilizers. The rows can be viewed as $\ev(X)$, $\ev(Y)$ and $\ev(Z)$.\footnote{As always, because we are evaluating these polynomials over the points in $\mathbb{P}_{\mathbb{F}_4}^2$, the matrix will only be well-defined up to multiplications of each column by elements in $\mathbb{F}_4^\times$, because doing so does not affect Eq.~\eqref{eq:CS-to-CS}, since $\eta^3 = 1$ for all $\eta \in \mathbb{F}_4^\times$.}
The kernel has dimension $15$ over $\F_4$, and can be interpreted as the $15$ degree-four (quartic) monomial evaluations.
Among those, $\ev(X^4)=\ev(X),\;\ev(Y^4)=\ev(Y),\;\ev(Z^4)=\ev(Z)$ are the stabilizers. 
Since the projective plane can be viewed as three affine charts glued together, the previous construction can be seen as being evaluated on the $Z=1$ chart; this naturally leads us to choose $\ev(XYZ^2)$ and $\ev(Z(Y^3+\omega X^3))$ as the two logical rows to match that construction. The new kernel is spanned by $\{XZ^3, XY^3, YX^3, YZ^3\}$ besides the existing rows.
We find through brute-force iterative search that we can add at most two more logical rows; the most symmetric solution is presented below:
\begin{equation*}
\begin{matrix}
\ev(X(Z^3+\omega Y^3))\\
\ev(Y(X^3+\omega Z^3))\\
\ev(Z(Y^3+\omega X^3))\\
\ev(XYZ^2)\\
\ev(X)\\
\ev(Y)\\
\ev(Z)
\end{matrix}\;\;
\begingroup
\setlength{\arraycolsep}{3.5pt}
\begin{pmatrix}
0 & 1 & 1 & 1 & \omega & \omega^2 & \omega^2 & \omega^2 & \omega & \omega^2 & \omega^2 & \omega^2 & \omega & \omega^2 & \omega^2 & \omega^2 & 0 & 0 & 0 & 0 & 0 \\
0 & 0 & 0 & 0 & \omega & 1 & 1 & 1 & \omega^2 & \omega & \omega & \omega & 1 & \omega^2 & \omega^2 & \omega^2 & 0 & \omega & \omega & \omega & 0 \\
0 & \omega^2 & 1 & \omega & 0 & 1 & \omega & \omega^2 & 0 & 1 & \omega & \omega^2 & 0 & 1 & \omega & \omega^2 & 0 & \omega & \omega^2 & 1 & 0 \\
0 & 0 & 0 & 0 & 0 & 1 & \omega^2 & \omega & 0 & \omega & 1 & \omega^2 & 0 & \omega^2 & \omega & 1 & 0 & 0 & 0 & 0 & 0 \\ \hline
1 & 1 & 1 & 1 & 1 & 1 & 1 & 1 & 1 & 1 & 1 & 1 & 1 & 1 & 1 & 1 & 0 & 0 & 0 & 0 & 0 \\
0 & 0 & 0 & 0 & \omega & \omega & \omega & \omega & \omega^2 & \omega^2 & \omega^2 & \omega^2 & 1 & 1 & 1 & 1 & 1 & 1 & 1 & 1 & 0 \\
0 & \omega & \omega^2 & 1 & 0 & \omega & \omega^2 & 1 & 0 & \omega & \omega^2 & 1 & 0 & \omega & \omega^2 & 1 & 0 & \omega & \omega^2 & 1 & 1 \\
\end{pmatrix}
\endgroup
\end{equation*}

\end{construction}

\begin{construction}[$16\CS\to 1\TOF\#$ at $d=3$]\label{con:16cs1TOF}
Similar to the generalized triorthogonality for binary codes, one can relax the intersection of logical rows in three-orthogonality for $\CS$ protocols. The resulting logical action is usually a hypergraph magic gate. In the present instance, we use the two-variable RM code over $\F_4$, taking $xy,x^2,y^2$ as $X$ logicals, and $x,y,1$ as $X$ stabilizers. The logical action of applying physical $\CS$ gates transversally can be computed using the techniques described in App.~\ref{sec:embedded_code}. On the binary level, it is a $\CS$ on qubits 1,2 together with an $\F_4$-qudit CCZ on qubits 1 to 6 (cf. Fig.~\ref{fig:F4_qudit_CCZ}(a)). One can then measure the first qubit in $Z$ basis, and upon Clifford correction, the remaining circuit on qubits 2 to 6 is Clifford-equivalent to a $\TOF\#$ gate. Another way to see this is to set the first logical qubit to the zero state before applying the physical $\CS$ gates, such that all $\CCZ$ gates and $\CS$ gate involving this logical qubit have trivial effect.

Below we show the binary matrix describing the protocol for qubits. Applying the $\CS$ gate to adjacent pairs of physical qubits yields the logical gate $\CCZ_{013}\CCZ_{014}\CCZ_{023}$, which one checks is Clifford equivalent to $\TOF\#\sim_c\CCZ_{0,1,4}\CCZ_{0,2,3}$ by acting with $\CNOT_{43}$.
\begin{equation}
\setcounter{MaxMatrixCols}{32}
\setlength{\arraycolsep}{4pt}
\left(\begin{smallmatrix}
    0&0&0&0&0&0&0&0&0&0&1&1&0&1&1&0&0&0&0&1&1&0&1&1&0&0&1&0&1&1&0&1\\
    0&0&0&1&1&1&1&0&0&0&0&1&1&1&1&0&0&0&0&1&1&1&1&0&0&0&0&1&1&1&1&0\\
    0&0&1&1&1&0&0&1&0&0&1&1&1&0&0&1&0&0&1&1&1&0&0&1&0&0&1&1&1&0&0&1\\
    0&0&0&0&0&0&0&0&0&1&0&1&0&1&0&1&1&1&1&1&1&1&1&1&1&0&1&0&1&0&1&0\\
    0&0&0&0&0&0&0&0&1&1&1&1&1&1&1&1&1&0&1&0&1&0&1&0&0&1&0&1&0&1&0&1\\[1pt]\hline\\
    0&0&1&1&0&1&1&0&0&0&1&1&0&1&1&0&0&0&1&1&0&1&1&0&0&0&1&1&0&1&1&0\\
    0&0&1&0&1&1&0&1&0&0&1&0&1&1&0&1&0&0&1&0&1&1&0&1&0&0&1&0&1&1&0&1\\
    0&0&0&0&0&0&0&0&1&1&1&1&1&1&1&1&0&1&0&1&0&1&0&1&1&0&1&0&1&0&1&0\\
    0&0&0&0&0&0&0&0&1&0&1&0&1&0&1&0&1&1&1&1&1&1&1&1&0&1&0&1&0&1&0&1\\
    1&0&1&0&1&0&1&0&1&0&1&0&1&0&1&0&1&0&1&0&1&0&1&0&1&0&1&0&1&0&1&0\\
    0&1&0&1&0&1&0&1&0&1&0&1&0&1&0&1&0&1&0&1&0&1&0&1&0&1&0&1&0&1&0&1
\end{smallmatrix}\right)
\end{equation}
\end{construction}

\begin{construction}[$8\CS\to 1\TOF\#$ at $d=2$]\label{con:8cstoTOF}
One can also obtain an $8\CS\to 1\TOF\#$ ($d=2$) protocol from the Hermitian curve $y^2+y=x^3$ over $\F_4$ by choosing $y,x^2,x$ as $X$ logicals and $1$ as $X$ stabilizer. The logical intersection pattern (on the binary level) is the same as the above RM construction. Hence, one can use the same method to obtain $\TOF\#$. Again, we write down the binary matrix giving the protocol at the qubit level, where one checks that the application of the $\CS$ gate to adjacent pairs of physical qubits yields the same logical gate as in the previous construction.
\begin{equation}
\left(\begin{smallmatrix}
        0&0&0&1&1&0&1&1&1&0&1&1&1&0&1&1\\
        0&0&0&0&0&1&0&1&1&1&1&1&1&0&1&0\\
        0&0&0&0&1&1&1&1&1&0&1&0&0&1&0&1\\
        0&0&0&0&1&1&1&1&0&1&0&1&1&0&1&0\\
        0&0&0&0&1&0&1&0&1&1&1&1&0&1&0&1\\[1pt]\hline\\
        1&0&1&0&1&0&1&0&1&0&1&0&1&0&1&0\\
        0&1&0&1&0&1&0&1&0&1&0&1&0&1&0&1
    \end{smallmatrix}\right)
\end{equation}
\end{construction}

\begin{construction}[$8\CCS\to 1\CCS$ at $d=2$]\label{con:8ccsto1ccsd2}
The Reed-Solomon code over $\F_8$ can also be used for $\CCS$-to-$\CCS$ distillation.\footnote{We have already seen that this code admits transversal $\CCZ$, as a subcode of~\eqref{eq:RS_8_to_2}.}
\begin{equation}
\begin{pmatrix}
0 & \alpha & \alpha^2 & \alpha^3 & \alpha^4 & \alpha^5 & \alpha^6 & 1\\\hline
1 & 1 & 1 & 1 & 1 & 1 & 1 & 1\\
\end{pmatrix} 
\label{eq:RS_8_to_1_CCS}
\end{equation}
That is, this matrix can be binarized (using the self-dual normal basis for $\mathbb{F}_8$, see Eq.~\eqref{eq:F8_self_dual_normal_basis}), and the action of $\CCS$ on the $8$ sets of $3$ physical qubits yields a logical $\CCS$ gate on the three logical qubits.

To see the $\CCS$ transversality, imagine attaching the following embedded code gadget (a $k=3$ classical simplex code) to each triplet of qubits expanding an $\F_8$ qudit.
\begin{equation}
    |x_1\ra|x_2\ra|x_3\ra\underbrace{|0\ra^{\otimes 4}}_{\text{gadget qubits}}\mapsto |x_1\ra |x_2\ra|x_3\ra \underbrace{|x_1\oplus x_2\ra |x_1\oplus x_3\ra | x_2\oplus x_3\ra |x_1\oplus x_2\oplus x_3\ra}_{\text{gadget qubits}}.
\end{equation}
Applying $\sqrt{T} \sqrt{T} \sqrt{T} \sqrt{T}^\dagger \sqrt{T}^\dagger  \sqrt{T}^\dagger \sqrt{T}$ leads to a $\CCS$ gate on $|x_1\ra|x_2\ra|x_3\ra$, i.e., this state gains a phase of $i^{x_1x_2x_3}$ because of the following identity.
\begin{equation}
    4x_1x_2x_3=(x_1\oplus x_2\oplus x_3)-(x_1\oplus x_2)-(x_1\oplus x_3)-(x_2\oplus x_3)+x_1+x_2+x_3\;.
\end{equation}
One may use the self-dual normal basis for $\F_8$ in Eq.~\eqref{eq:F8_self_dual_normal_basis} to binarize the qudit code in Equation~\eqref{eq:RS_8_to_1_CCS} before gadgetizing it with the simplex code.
Then, on the resulting binarized and gadgetized $[[56,3,2]]$ qubit code, one can use the framework of Sec.~\ref{sec:generalized_triorthogonality} (but with $\text{mod }16$ cutoff rather than $\text{mod }8$ to calculate the induced logical action, noting that the logical action of $\CCS$ on the binarized code is the same as the logical action of $\sqrt{T}\sqrt{T}\sqrt{T}\sqrt{T}^\dagger\sqrt{T}^\dagger\sqrt{T}^\dagger\sqrt{T}$ on the binarized and gadgetized code. It can be calculated that this logical action is $\CCS_{123}\CS^\dagger_{12}\CS^\dagger_{13}\CS^\dagger_{23}S_1 S_2S_3$ (with no action between any stabilizer row and logical row). This logical action is Pauli equivalent to a $\CCS_{123}^\dagger$ (up to a global phase):
$$(X_1X_2X_3)\CCS_{123}^\dagger(X_1X_2X_3)=-i \CCS_{123}\CS^\dagger_{12}\CS^\dagger_{13}\CS^\dagger_{23}S_1 S_2S_3\;.$$

\end{construction}
\section{Protocols from Algebraic Geometry Codes}\label{sec:AG_codes}

\subsection{Introduction to algebraic curves}
\label{sec:intro_to_curves}

Here, we intend to provide an intuitive but non-rigorous introduction to algebraic curves and how they are used to define error-correcting codes. We leave the more rigorous treatment to App.~\ref{sec:curves_formal_definitions}, where one will find the bulk of the definitions. Our presentation in this section aims to provide some motivation for those abstract definitions later.

Let $K=\F_{2^s}$ throughout this section, and denote by $\bar{K}=\bar{\F}_{2^s}=\bar{\F}_2$ the algebraic closure of $K$.\footnote{The algebraic closure of a field \(K\) is \(K\) extended by the roots of all polynomials with coefficients in \(K\). That is, for any polynomial \(p \in K[T]\), there is a factorization \(p(T) = a(T-\alpha_1)\dots (T-\alpha_{\deg p})\) where each \(\alpha_i \in \bar{K}\). By considering the definition of $\mathbb{F}_{2^s}$, one can see that the algebraic closure of $\mathbb{F}_{2^s}$ is identical to that of $\mathbb{F}_2$ for all $s \geq 1$; that is, $\bar{\mathbb{F}}_{2^s} = \bar{\mathbb{F}}_2$.}
Informally, a curve is an algebraic set with one degree of freedom. In this paper, the curves are usually presented by one equation $f(x,y)=0$ in two affine variables; they are so-called \emph{plane curves}. 
All the curves we consider are presented in Table~\ref{table:maximal_curves_CCZ_only}. 
Later, when we say a curve $\chi$ is `defined by' an equation $f$ (or multiple equations when there are more variables), we mean $\chi$ `is the zero set of' that equation(s).
Points on the curve $\chi$ with all their coordinates\footnote{For \((x_1, x_2, \dots, x_m) \in \chi\), we refer to the \(i\)-th entry in the tuple as the \(i\)-th coordinate. When clear from context, we will simply refer to each coordinate by a unique variable e.g. \((x,y,z)\).} in $K$ are called \emph{$K$-rational points}, or simply \emph{rational points} when the field $K$ is clear from context.

To give an explicit example, consider the affine plane curve
\begin{equation}
\cH: y^2+y=x^3\quad\text{over }\F_4=\{0,1,\omega,\omega^2\}~.
\end{equation}
This is the $q=2$ Hermitian curve (in the form $y^q+y=x^{q+1}$ over $\F_{q^2}$), which also happens to be an elliptic curve\footnote{For us, an elliptic curve will be one with `genus' $g=1$; we will define the genus shortly.} in this case. It has eight $\F_4$-rational points:
\begin{align}
\label{eq:F4_Hermitian_points}
x=0:&\; (0,0), (0,1),\\
x=1:&\; (1,\omega), (1,\omega^2),\\
x=\omega:&\;(\omega,\omega), (\omega,\omega^2),\\
x=\omega^2:&\;(\omega^2,\omega), (\omega^2, \omega^2).
\end{align}

These eight $\F_4$-rational points $(\alpha,\beta)$ in the affine plane $\A^2_{\bar{K}}=\{(x,y)\;|\;x,y\in\bar{K}\}$ all satisfy $\beta^2+\beta=\alpha^3$; we will later refer to them as $P_{\alpha,\beta}$. Later, we shorthand $\A^2_{\bar{K}}$ as $\A^2$ for simplicity. There is one additional $\F_4$-rational point at infinity, denoted $Q_\infty$, obtained after projective closure, which we discuss momentarily.
This distinguished point $Q_\infty$ will be the point used in the construction of a one-point code.

One can embed the affine space $\A^2$ into the projective space $\bbP^2=\{(x:y:z)\neq(0:0:0)\}/(x:y:z)\sim (\lambda x:\lambda y:\lambda z)$ by identifying the point $(x,y)\in\A^2$ with the equivalence class of $(x:y:1)\in\bbP^2$. 
If an affine plane curve is defined by $f(x,y)=0$, then its projective closure in $\bbP^2$ is defined by the homogenization
\begin{equation}
f^h(x,y,z)=z^\ell\cdot f(x/z,y/z),\quad \ell=\deg f.
\end{equation}
The above Hermitian curve example $\cH$ has projective closure
\begin{equation}
\label{eq:F4_Hermitian_curve}
\cH': Y^2 Z + Y Z^2 = X^3~.\footnote{Note that capitalized variables are used to emphasize the fact that they are variables in the projective space, whereas lower-case variables denote that they take values in the affine space.}
\end{equation}
Projective points $(x:y:z)$ with $z=0$ are \emph{points at infinity}. 
In this example, there is exactly one such $\F_4$-rational point, namely
\begin{equation}
Q_\infty=(0:1:0).
\end{equation}
Together with this unique point at infinity, this gives nine $\F_4$-rational points on $\cH'$ in total.

\subsubsection{Genus of a smooth plane curve}
Let $\chi\subseteq \bbP^2$ be a projective plane curve defined by a homogeneous polynomial $f^h(X,Y,Z)\in K[X,Y,Z]$.
A point $P=(a:b:c)$ on $\chi$ is \emph{singular} if all three partial derivatives\footnote{The derivative of a polynomial \(\sum_i a_i x^i\) is defined to be the formal derivative \(\sum_i i\, a_i x^{i-1}\), where $i\,a_i$ denotes the summation of $i$ copies of $a_i$. The partial derivative in \(x_i\) of a multivariate polynomial in \((x_1, \dots, x_n)\) is computed by treating it as a univariate polynomial in \(x_i\).} of \(f^h\) vanish at $P$, i.e., $f^h_X(P)=f^h_Y(P)=f^h_Z(P)=0$. Otherwise $P$ is called \emph{smooth} or \emph{nonsingular}.
The curve is \emph{smooth} or \emph{nonsingular} if all of its points over $\bar{K}$ are smooth. Otherwise, it is \emph{singular}. For an affine curve defined by $f(x,y)=0$, the same condition says that a point $(a,b)$ is \emph{smooth} if $f_x(a,b)\neq 0$ or $f_y(a,b)\neq 0$. As an example, one can show by explicit calculation that the Hermitian curve example in Eq.~\ref{eq:F4_Hermitian_curve} is smooth.\footnote{Given $F = Y^2Z + YZ^2 + X^3$, we have $f^h_X = X^2, \; f^h_Y = Z^2, \; f^h_Z = Y^2$, noting that $2 = 0$ in extension fields of $\mathbb{F}_2$, and these functions cannot all be zero in the projective plane.} In particular, it is smooth at the point at infinity $Q_{\infty}$. Its affine version is also smooth.

In fact, almost all of the curves we consider, listed in Table~\ref{table:maximal_curves_CCZ_only}, and their projective versions, are smooth. There is one exception: the $y^2+y=x^5$ curve. Despite the affine curve being smooth ($f_y=1$), its projective closure $Y^2Z^3 + Y Z^4=X^5$ is singular at the point at infinity $(0:1:0)$. While smooth projective curves are easiest to make statement about, we will be able to treat this curve separately by noting that it has the form $y^h + \mu y = f(x)$ with $h = p^e > 1$, and satisfies the assumptions of Thm.~\ref{thm:stichtenoth_artin_schreier}; see~\cite[Prop. 6.4.1]{Stichtenoth}.\footnote{One may alternatively treat this curve by considering its smooth version, that is, its smooth projective model~\cite{wikipedia_resolution_of_singularities}\cite[Ch.~7]{fulton2008algebraic}. This is equivalent to considering the function field of the curve; we defer the discussion on this to App.~\ref{sec:curves_formal_definitions}.}

The genus (plural: genera), which we denote by $g$, is an important property of a given curve, and is a non-negative integer. Forgoing a formal definition for now, the genus characterizes the geometric complexity of a curve. For example, the projective line $\mathbb{P}^1$ has genus $g=0$, whereas elliptic curves have genus $g=1$. It is a deep fact that, as well as characterizing its topology, the genus also controls how many rational points a curve may have. For example, the Hasse-Weil bound~\cite[Sec. V]{Stichtenoth} states that for a smooth, irreducible\footnote{An irreducible plane curve has a defining polynomial that is irreducible, i.e., cannot be factored into two or more non-constant polynomials.} projective curve, one has
\begin{equation*}
    |\#\chi(\mathbb{F}_q) - (q+1)| \leq 2g\sqrt{q},
\end{equation*}
where $\#\chi(\mathbb{F}_q)$ is the number of $\mathbb{F}_q$-rational points on $\chi$.

The genus of a smooth projective plane curve defined by a degree $d$ polynomial is
\begin{equation}
g=\frac{1}{2}(d-1)(d-2),
\end{equation}
and a singular plane curve has a genus smaller than this value (see the \href{https://en.wikipedia.org/wiki/Pl\%C3\%BCcker_formula}{Plücker formula}).
In particular, for projective curves of the form $y^h+\mu y=f(x)$ with $h=p^e>1$ satisfying the assumptions in Thm.~\ref{thm:stichtenoth_artin_schreier}~\cite[Prop.~6.4.1]{Stichtenoth}, the genus equals
\begin{equation}
\frac{1}{2}(h-1)(\deg f-1).
\end{equation}

\subsubsection{Curves over finite fields with many points}
\label{sec:curves_with_many_points}
We aim to construct distillation protocols for $\mathbb{F}_q$-qudit CCZ states using curves over $\mathbb{F}_q$. We want to have curves with many rational points over $\mathbb{F}_q$, since later we will want to puncture away points (see Prop.~\ref{prop:puncturing-distillation-protocol}) to generate logical qudits while retaining good parameters. In this direction, let us consider curves with many rational points.

A curve of genus $g$ is called $\mathbb{F}_q$-maximal, or simply maximal, if it attains the Hasse-Weil upper bound:
\begin{equation}
\#\chi(\F_q)= q+1+2g\sqrt{q}.
\end{equation}
For genus $g=0$, the projective line $\mathbb{P}^1$ is maximal. If $g>0$, maximality may only be attained if $q$ is square. A tightening of the Hasse-Weil upper bound exists, called the Serre bound~\cite[Thm.~5.3.1]{Stichtenoth}:
\begin{equation}
\#\chi(\F_q)\le q+1+g\cdot\lfloor 2\sqrt{q}\rfloor
\end{equation}
A curve attaining this bound is called a \emph{defect-zero} curve. In particular, a maximal curve is a defect-zero curve. We will only work with defect-zero curves, and we will now move towards a classification of defect-zero curves of small genus over the field sizes of interest to us.
This will lead us to the curves in Table~\ref{table:maximal_curves_CCZ_only}.
Note that the field sizes of interest to us are small powers of \(2\), $q = 4,8,16,32,64$.

We start by constraining the possible genera for maximal curves over $\mathbb{F}_4$, $\mathbb{F}_{16}$ and $\mathbb{F}_{64}$. Indeed, the following facts are true for maximal curves over $\mathbb{F}_{q^s}$.
\begin{enumerate}
    \item Any smooth irreducible curve of genus $g>q(q-1)/2$ cannot be maximal~\cite{Ihara1982SomeRO};\footnote{We need not worry about non-smooth or reducible curves. In the first case, it is possible to smooth the curve~\cite{wikipedia_resolution_of_singularities}\cite[Ch.~7]{fulton2008algebraic}, and in the second case the curve may be decomposed into irreducible pieces.}
    \item There is a unique maximal curve over $\mathbb{F}_{q^2}$ of genus $g=q(q-1)/2$ up to isomorphism, namely the Hermitian curve $y^q + y = x^{q+1}$~\cite{Rck1994ACO};
    \item If $g<q(q-1)/2$ then $g\le \lfloor (q-1)^2/4\rfloor$~\cite{fuhrmann1996genus};
    \item For $q=2^s$, the second-largest genus is $q(q-2)/4$, which is achieved by the curve $\sum_{i=1}^{s} y^{q/2^i}=x^{q+1}$~\cite{max_curve_char_2}.
\end{enumerate}
It then follows that the possible genera for maximal curves in $\F_4$ are $\{0,1\}$, and in $\F_{16}$ are $\{0,1,2,6\}$. There are defect-zero curves over these fields with all these genera~\cite{van2000tables}, and they are listed in Table~\ref{table:maximal_curves_CCZ_only}.\footnote{It is unknown whether there are many than one up to isomorphism, outside of the genus $g=0$ case (the projective line), and the case of the Hermitian curves.}
For maximal curves over $\F_{64}$, we summarize the known constructions for certain genera in Table~\ref{table:F64_maximal_curves}. However, for constructing distillation protocols for qubits, we restrict ourselves to genus $0$ and $1$, in order to keep the resulting qubit distillation protocols reasonably small. We list the corresponding curves for $g=0$ and $1$ in Table~\ref{table:maximal_curves_CCZ_only}.

Moving onto $\F_8$ and $\F_{32}$, we are not aware of defect-zero curves other than the elliptic curves (those with $g=1$) and the Klein quartic curve over $\F_8$~\cite{van2000tables} (which has genus $g=3$).
The defect-zero elliptic curves that we use for these small alphabets ($s\le 6$), as summarized in our Table~\ref{table:maximal_curves_CCZ_only}, come from~\cite{elliptic_char_2} and they happen to have binary coefficients in their defining equations.

As it will be helpful later, we comment that for the elliptic curves of each type, the number of rational points, including the point at infinity, is equal to $2^s-\omega_1^s-\omega_2^s+1$, where the complex numbers $\omega_1, \omega_2$ are called Frobenius eigenvalues, as displayed in Table~\ref{Table:elliptic_curves_char2}.
\begin{table}[ht]
\centering
\begin{tabular}{|c|c|c|}
\hline
Type & Equation           & $\omega_1,\;\omega_2$   \\\hline
I    & $y^2+y=x^3+x^2+1$  & $1\pm i$                \\
II   & $y^2+xy=x^3+x^2+x$ & $\frac{1}{2}(1\pm i\sqrt{7})$  \\
III  & $y^2+y=x^3$        & $\pm i\sqrt{2}$         \\
IV   & $y^2+xy=x^3+x$     & $\frac{1}{2}(-1\pm i\sqrt{7})$ \\
V    & $y^2+y=x^3+x^2$    & $-1\pm i$               \\\hline
\end{tabular}
\caption{The five types of nonsingular elliptic curves over $\F_2$ \cite{elliptic_char_2}, and the corresponding Frobenius eigenvalues. The number of $\F_{2^s}$-rational points on the curve is $2^s-\omega_1^s-\omega_2^s+1$ (including the point at infinity).}
\label{Table:elliptic_curves_char2}
\end{table}

\subsection{Introduction to one-point codes}
\label{sec:intro_to_AG_codes}
Here, we introduce the most elementary way of constructing a (classical) algebraic geometry code from an algebraic curve: the one-point code construction.
Let us provide motivation for this construction by considering how algebraic geometry codes extend Reed-Solomon codes; this intuition was how Goppa first discovered algebraic geometry codes~\cite{goppa1983algebraico}. 

Typically, we think of a (generalized) Reed-Solomon code as defined by a set of points $\boldsymbol{\alpha} \subseteq \mathbb{F}_q$, where $|\boldsymbol{\alpha}| = n \leq q$, and some integer $k \leq n$. The Reed-Solomon code is then formed by taking every univariate polynomial over $\mathbb{F}_q$ of degree less than $k$, and evaluating this polynomial at the points $\boldsymbol{\alpha}$. One codeword is the evaluation of one polynomial, and it turns out that we obtain a code of length $n$ and dimension $k$ over $\mathbb{F}_q$. Note that for a ``full-length'' Reed-Solomon code, one takes $n=q$ (and $\boldsymbol{\alpha} = \mathbb{F}_q$).

An initially strange, but ultimately valid way to talk about polynomials of degree less than $k$ is to talk about the polynomials which ``blow up'' fewer than $k$ times at infinity (the polynomial $x+a$ blows up once at infinity, the polynomial $x^2+ax+b$ blow up twice at infinity, etc.). Slightly more formally, the polynomials of degree less than $k$ are those polynomials whose \textit{pole order} at infinity is less than $k$. The object that we are secretly considering is the projective line over $\mathbb{F}_q$ (which has $q+1$ $\mathbb{F}_q$-rational points, namely $\mathbb{F}_q \cup \{Q_\infty\}$), and we are considering the set of polynomials defined on this curve with pole order less than $k$ at infinity. We then form a code by evaluating these polynomials at some choice of (non-infinite points) $\boldsymbol{\alpha} \subseteq \mathbb{F}_q$ on the projective line.

With this intuition, passing from Reed-Solomon codes to more general algebraic geometry codes is natural. Indeed, the one-point code construction over an algebraic curve calls for us to pick a point $Q$ on the curve and construct codewords by considering polynomials with pole order below some prescribed value at $Q$. A codeword is then exactly the evaluation of one of these polynomials at some prescribed points on the curve. While the Reed-Solomon case of $g=0$ is very transparent, passing to higher genera $g$ allows us to construct longer codes over the same field $\mathbb{F}_q$, because raising the genus $g$ yields curves with more points over $\mathbb{F}_q$.

There are two immediate questions. First, which point $Q$ to pick, and second, how can one compute the pole order of a polynomial at $Q$?
For the first question, we take $Q$ to be $(0:1:0)$, which can be verified to be a point on the projective closures of all the curves we present in Table~\ref{table:maximal_curves_CCZ_only}. Since this is a point at infinity, we write $Q_\infty$ instead of $Q$ for this special point, like we did for the Reed-Solomon/projective line example above.
For the second question, given an irreducible smooth projective curve, the pole order of a function at $Q_\infty$ can be computed using \emph{intersection theory}.\footnote{The one curve we consider whose projective closure is not smooth is $y^2+y = x^5$, as mentioned. However, in this case, one can use the smooth projective model of the curve~\cite{wikipedia_resolution_of_singularities}\cite[Ch.~7]{fulton2008algebraic}, or use the places of the function field. We defer this to App.~\ref{sec:curves_formal_definitions}.}
We will elucidate the general procedure behind this calculation through examples, starting with an example of the Hermitian curve $\cH: \; y^2+y=x^3$ over $\F_4$. Note that all content between now and Thm.~\ref{thm:stichtenoth_artin_schreier} requires the projective curve to be smooth.

Suppose $\chi$ is an affine plane curve defined by an irreducible equation $f(x,y)=0$. Its \emph{affine coordinate ring} is $K[\chi]=K[x,y]/(f)$ (see Sec.~\ref{sec:ring_field} for the preliminary material on ring theory).
The \emph{function field} of $\chi$ is the fraction field $K(\chi)=\text{Frac}(K[\chi])$.\footnote{Note that $K[\chi]$ is an integral domain, following from the fact that $f$ is irreducible.}
Elements of $K(\chi)$ are called \emph{rational functions} on $\chi$.
One may view them as quotients of polynomial functions, subject to the relation $f(x,y)=0$, where the denominator is not identically zero.
For example, on the curve $y^2+y=x^3$ and its function field, one can replace any expression containing $y^2$ by $y+x^3$, thus reducing the degree of $y$ to $\le 1$.

On the projective closure of the curve, the rational functions are quotients of homogeneous polynomials of the same degree, again subject to the equation defining the curve, and the denominator not being identically zero.
A simple example of a rational function is $X/Z$, where $(X:Y:Z)$ are homogeneous coordinates on $\bbP^2$.

The projective plane is three affine charts\footnote{The notion of a chart will be familiar to those have worked with manifolds, where some patch of an unfamiliar topology like projective space may be mapped using familiar coordinates: in this case affine coordinates.} glued together; we define our affine curves in the $Z=1$ affine chart.
Thus, on the projective closure of the curve, the affine coordinates $x$ and $y$ can be viewed as the rational functions
\begin{equation}
    x=\frac{X}{Z},\quad y=\frac{Y}{Z}.
\end{equation}

Take the projective closure of $\cH$, which is $\cH':\; Y^2Z+YZ^2=X^3$. We seek the pole orders of $x$ and $y$ at $Q_\infty$ on the curve $\mathcal{H}'$. We can calculate these using intersections of $\mathcal{H}'$ with lines in $\bbP^2$.
Consider the line $\cL:\;X=0$. Substituting $X=0$ into $\cH'$ gives $YZ(Y+Z)=0$, so the intersection points of $\mathcal{H}'$ and $\mathcal{L}$ are $(0:0:1)=P_{0,0}$, $(0:1:1)=P_{0,1}$, $(0:1:0)=Q_\infty$, each occurring with multiplicity $1$. This situation is denoted
\begin{equation}
    \cH'\cdot\cL=P_{0,0}+P_{0,1}+Q_\infty,
\end{equation}
What we have on the right-hand side is an example of \textit{divisor notation}. 
For our purposes, a \emph{divisor} is simply a bookkeeping device telling us at which points some function $f$ has zeros and poles. Slightly more formally, the divisor of a function $f$, denoted $(f)$, is an integer linear combination of points/places,\footnote{For the purposes of this section, the words point and place may be used interchangeably, since we only consider points/places of degree one, that is, rational places. For the more general notions, see App.~\ref{sec:curves_formal_definitions}.} where the coefficient of a place is the multiplicity of its zero at the place (if positive) or minus the order of its pole (if negative). In this case, we use $\mathcal{H}' \cdot \mathcal{L}$ to denote the divisor of the function $YZ(Y+Z)$, which has three zeros, each of multiplicity one, and no poles.

Intuitively, it is clear that if functions $f,g$ have zeros of multiplicities $o_1,o_2$ at a point $P$, then the multiplication of the two functions $fg$ has a zero of multiplicity $o_1+o_2$ at $P$. A similar statement holds for poles and their orders, leading us to the statement $(fg)=(f)+(g)$.

Next, consider the line at infinity $\mathcal{N}:\;Z=0$. Substituting $Z=0$ into $\cH'$ gives $X^3=0$, so $\mathcal{N}$ meets $\cH'$ only at $Q_{\infty}$, with multiplicity $3$. Now we have
\begin{equation}
    \cH'\cdot\mathcal{N}=3Q_{\infty}.
\end{equation}
Since $x=X/Z$, its divisor is obtained by subtracting the intersection divisor with $\mathcal{N}:\;Z=0$ from that with $\mathcal{L}:\;X=0$:
\begin{equation}
(x)=\left(\frac{X}{Z}\right)= (X) - (Z) =\cH'\cdot\mathcal{L}-\cH'\cdot\mathcal{N} =P_{0,0}+P_{0,1}-2Q_{\infty}.
\end{equation}
We have found that $x$ has a pole of order $2$ at $Q_\infty$.

One can do a similar calculation for $(y)$ by considering the intersection of $\cH'$ with $\mathcal{M}:\;Y=0$, which is $\cH'\cdot \mathcal{M}=3P_{0,0}$. Hence,
\begin{equation}
(y)=\left(\frac{Y}{Z}\right)=3P_{0,0}-3Q_\infty,
\end{equation}
and $y$ has a pole of order $3$ at $Q_\infty$.

This motivates the notion of the valuation of a function at a given place. Considering, for example, the place $Q_\infty$, for a nonzero rational function $f$, define $v_{Q_\infty}(f)$ to be its order of vanishing at $Q_\infty$, with poles recorded as negative values. In our example,
\begin{equation}
v_{Q_{\infty}}(x)=-2,\quad v_{Q_\infty}(y)=-3.
\end{equation}
We also have $v_{Q_\infty}(x^i y^j)=-(2i+3j)$.

More generally, for any point $P$ on the curve and any nonzero rational function, the integer $v_P(f)$ measures the order of vanishing or pole of $f$ at $P$. The divisor of $f$ is 
\begin{equation}
(f)=\sum_P v_P(f)P.
\end{equation}

Given a divisor $D$, write $D\ge 0$ if all its integer coefficients are nonnegative, and write $\supp(D)$ for the set of points with a nonzero coefficient in $D$.
For our purpose, the degree $\deg(D)$ is the sum of the coefficients.\footnote{In general, if $D=\sum_P n_PP$, then $\deg (D)=\sum_P n_P \cdot \deg(P)$, but for our purpose, all points appearing in a divisor are places of degree one, i.e., we only consider places $P$ with $\deg(P) = 1$. The formal definition of the degree of a place can be found in App.~\ref{sec:curves_formal_definitions}.}

Given a divisor $G$, its Riemann-Roch space is the set of functions with divisor $\ge -G$.
\begin{equation}
\cL(G)=\{f\in K(\chi)^{\times}: (f)+G\ge 0\}\cup \{0\}
\end{equation}
One can see that the Riemann-Roch is a vector space over the field $K$, and its dimension over \(K\) is denoted $\ell(G):=\dim\cL(G)$.
For the one-point codes we are particularly interested in, we take $G=rQ_\infty$, so the Riemann-Roch space is
\begin{equation}
\cL(rQ_\infty)=\{f\in K(\chi)^\times:\;v_{Q_\infty}(f)\ge -r,\; v_P(f)\ge 0,\; \forall P\neq Q_\infty\}\cup\{0\}~.
\end{equation}
This consists of all rational functions having pole order at most $r$ at $Q_\infty$ and \emph{no poles} (that is, the functions are regular) away from $Q_\infty$. Stepping back for a moment, in our elementary Reed-Solomon example, the underlying projective curve was the projective line, and the Riemann-Roch space used to define the Reed-Solomon code of dimension $k$ was the vector space of polynomials whose pole order at the point at infinity, $Q_\infty$, was less than $k$. In other words, we had $G = (k-1)Q_\infty$. In that case, the Riemann-Roch space was a $k$-dimensional vector space, as was the resulting code.

Just as in the Reed-Solomon case, the dimension of the Riemann-Roch space turns out to control the dimension of the corresponding code, and so we need tools to calculate the dimension of the Riemann-Roch space.\footnote{The dimension of the code will equal the dimension of the Riemann-Roch space when the evaluation map is injective, which we always take to be the case. The evaluation map is injective when the degree of the divisor defining the Riemann-Roch space is less than the number of points at which we evaluate the functions, that is, $\deg(G) < n$. This is the natural analogue of the familiar statement in the Reed-Solomon case, where we have that a polynomial of degree $<n$ is uniquely specified by its evaluations at $n$ points.}
In full generality, this can be hard to do precisely, since for a given pole order $r$, there might not exist a rational function that has a pole of order $r$ at $Q_\infty$, so the dimension may be smaller than one would otherwise expect.
This is the situation that $\mathcal{L}(rQ_\infty) = \mathcal{L}((r-1)Q_\infty)$ for some positive integer $r$, in which case $r$ is called a \textit{gap}; proving the non-existence of these gaps is beyond the scope of this introduction. For our purposes, we just utilize the following corollary of Riemann-Roch theorem~\cite[Cor.~2.58]{AG_codes_handbook}.

Let $D$ be a divisor on a smooth projective curve of genus $g$ and let $\deg(D)>2g-2$. Then, the dimension of the corresponding Riemann-Roch space is
\begin{equation}
\label{eq:Riemann-Roch_space_dim}
\ell(D)=\deg(D)-g+1~.
\end{equation}
This formula may be understood through the lens of the Weierstrass gap theorem which states that the number of gaps is equal to the genus $g$ of the curve, and all the integers $r\ge 2g$ are non-gaps. For example, consider our above $\mathcal{H}'$ case, which had genus $g=1$. There must be $g=1$ gaps, but all the integers greater than or equal to $2$ are non-gaps, meaning that the only gap is $1$ itself. This gives us our conclusion on the dimension of the Riemann-Roch space in this case:
\begin{equation*}
    l(rQ_\infty) = \begin{cases}
        1 &\text{, if } r = 0\\
        r &\text{, if } r > 0
    \end{cases}.
\end{equation*}

We now have the tools to calculate the dimension of certain Riemann-Roch spaces arising from one-point code constructions over curves given in Table~\ref{table:maximal_curves_CCZ_only}. Moreover, using the intersection theory tools we have developed, we can establish the pole orders of $x$ and $y$ at $Q_\infty$, and thus argue for the existence of monomial bases of the Riemann-Roch spaces of the form $\{x^iy^j\}$, for certain combinations of $(i,j)$. For all curves in our Table~\ref{table:maximal_curves_CCZ_only} except $y^2+y=x^5$, one can apply the above intersection theory approach to establish the pole order of $x$ and $y$ at $Q_\infty$. Indeed, to do this, we first take all the eligible monomials that could be in the Riemann-Roch space by considering their pole orders at $Q_\infty$. In some cases (see the example in the next paragraph for the Klein quartic) we need to further filter out the ones that have poles at some point other than $Q_\infty$.\footnote{These are not allowed in the Riemann-Roch space $\mathcal{L}(rQ_\infty)$, because the definition of this Riemann-Roch space only allows its members to have poles at $Q_\infty$.} To form the monomial basis for the Riemann-Roch space, we also need to reduce the monomials by the equation defining the curve (e.g. replace $y^2$ by $y+x^3$ in the $\cH'$ example above, thus decreasing the degree of $y$ below $2$), thus guaranteeing linear independence between the monomials. To show that the monomials span the Riemann-Roch space, take $r=2g$ and do dimension counting; the number of gaps must equal the genus of the curve.

To further illustrate the above procedure, consider the Klein quartic curve $x^3 y + y^3 + x=0$ over $\F_8$~\cite[Example~2.34, 2.75, 2.76]{AG_codes_handbook}.
Let $\chi: Y^3Z+X^3Y+Z^3X=0$ be its projective closure.
One can see $P=(1:0:0)$, $Q_\infty=(0:1:0)$, $R=(0:0:1)$ are on $\chi$; among them, $P$ and $Q_\infty$ are points at infinity because $z=0$. Let $\mathcal{L}$ be the line with equation $X=0$. Substituting $X=0$ into $\chi$, one is left with $Y^3 Z=0$, and therefore $\mathcal{L}$ intersects $\chi$ in points $R$ and $Q_\infty$; the intersection divisor reads $\chi\cdot\mathcal{L}=3R+Q_\infty$.
Similarly, one can consider the intersection of $\mathcal{M}:Y=0$ or $\mathcal{N}:Z=0$ with $\chi$, which give $\chi\cdot\mathcal{M}=3P+R$ and $\chi\cdot\mathcal{N}=3Q_\infty+P$. 
Therefore, $(X/Z)=3R-P-2Q_\infty$, $(Y/Z)=R+2P-3Q_\infty$, and $(x^iy^j)=-(2i+3j)Q_\infty+(2j-i)P+(3i+j)R$. If we consider Riemann-Roch spaces defined by the place $Q_\infty$, i.e., a divisor $G=rQ_\infty$, 
the functions in $\cL(rQ_\infty)$ should satisfy $(f)+rQ_\infty\ge 0$. For some $x^iy^j$ to be in the Riemann-Roch space, the coefficient of $P$ in $(x^iy^j)$ must be non-negative, meaning $i\le 2j$. 
Now we know that $\{x^i y^j\;|\; 0\le 2i+3j\le r,\; 0\le i\le 2j\}$ are in $\cL(rQ_\infty)$, take $r=2g=6$, we have four functions $1,y,xy,y^2$ with poles in $Q_\infty$ of order $0,3,5,6$, respectively. Hence
they are independent and since $\ell(6Q_\infty)=6-g+1=4$, they form a basis of $\cL(6Q_\infty)$. 
If $i>3$ in $x^iy^j$, since $i\le 2j$, we can replace $x^3y$ by $y^3+x$ and decrease $i$. Therefore, we should take $0\le i<3$ for basis.
For any $r\ge 6$, it is clear that we can find $i,j$ such that $2i+3j=r$ and $0\le i<3,\;0\le i\le 2j$. Therefore, $\{x^i y^j\;|\; 0\le 2i+3j\le r,\;0\le i<3,\; 0\le i\le 2j,\}$ form a basis of $\cL(rQ_\infty)$.

One can repeat the above procedure to the two elliptic curves with a cross term $xy$ in our Table~\ref{table:maximal_curves_CCZ_only}, and establishes that $\{x^i y^j\;|\; j\in\{0,1\},\; 2i+3j\le r\}$ forms a basis of $\cL(rQ_\infty)$. For the rest of the curves in our Table~\ref{table:maximal_curves_CCZ_only}, including $y^2+y=x^5$ (whose projective closure is not smooth), the following theorem (we only care about $p=2$) establishes a monomial basis for them. This result is useful for actually constructing codes, because we can write down a basis for the Riemann-Roch space corresponding to the code, and therefore a basis for the code as the evaluations of those basis functions.

\begin{theorem}\cite[Prop.~6.4.1]{Stichtenoth} 
\label{thm:stichtenoth_artin_schreier}
Let $K$ be a field of characteristic $p>0$. Consider a function field $F=K(x,y)$ defined by $y^h+\mu y=f(x)\in K[x]$, where $h=p^e>1$, and $\mu\in K^\times$. Assume that $\deg f=:m>0$ is coprime to $p$, and that all roots of $y^h+\mu y=0$ are in $K$. 
Then, there is a unique place at infinity, $Q_{\infty}$, of degree one.\footnote{For the purposes of this section, it suffices to only think about rational places, which have degree one.}
Moreover, the elements $x^iy^j$ with $0\le i$, $0\le j\le h-1$,\footnote{It is natural to restrict the degree of $y$ to less than $h$, because $y^h$ can be replaced with $y+f(x)$.} $h i+mj\le r$ form a basis of the space $\cL(rQ_\infty)$ over $K$.
The genus is $g=(h-1)(m-1)/2$.
\end{theorem}

We may now consider applying Theorem~\ref{thm:stichtenoth_artin_schreier} to the curves we consider. In the $m>h$ case, $Q_\infty=(0:1:0)$ because the projective closure has degree $m$. Indeed, writing $f(x) = \sum_{i=0}^ma_ix^i$ for $a_m \neq 0$, the projective closure is $Y^hZ^{m-h} + \mu Y Z^{m-1} = a_mX^m + \sum_{i=0}^{m-1}a_iX^iZ^{m-i}$, which is clearly solved by $(0:1:0)$. %
The case $h=m$ is ruled out by the fact that $m$ and $p$ are coprime. The final case to consider is that of $h>m$. In this case, $(1:0:0)$ is the unique point at infinity.
\\[3pt]

Now, the one-point codes we are interested in take the following form. Given a smooth projective curve over $K$, or a curve satisfying the conditions of Thm.~\ref{thm:stichtenoth_artin_schreier}, take $G=rQ_\infty$ and let $D=P_1+P_2+\cdots+P_n$ be the sum of all $K$-rational points on the affine curve ($\supp\;G\;\cap\;\supp\;D=\emptyset$).
The algebraic geometry (AG) code $C_\cL(D,G)$ associated with the divisors $D$ and $G$ is defined to be the evaluations of elements of the Riemann-Roch space
\begin{equation}
C_\cL(D,G):=\{(f(P_1),\dots,f(P_n))\;|\; f\in\cL(G)\}\subseteq\F_q^n.
\end{equation}
Let us give an explicit example again using $\cH'$. Here $D$ is the sum of the eight affine $\F_4$-rational points listed in Eq.~\eqref{eq:F4_Hermitian_points}. The code $C_\cL(D, 3Q_\infty)$ consists of the evaluation of $\{1,x,y\}$:
\begin{equation}
\begin{array}{c|cccccccc}
      & P_{0,0} & P_{0,1} & P_{1,\omega} & P_{1,\omega^2} & P_{\omega,\omega} & P_{\omega,\omega^2} & P_{\omega^2,\omega} & P_{\omega^2,\omega^2} \\
\hline
\mathrm{ev}(y) & 0 & 1 & \omega & \omega^2 & \omega & \omega^2 & \omega & \omega^2 \\
\mathrm{ev}(x) & 0 & 0 & 1 & 1 & \omega & \omega & \omega^2 & \omega^2 \\
\mathrm{ev}(1) & 1 & 1 & 1 & 1 & 1 & 1 & 1 & 1
\end{array}
\end{equation}

The dimension of a code $C_\cL(D,G)$ is determined as follows.
The evaluation map $\ev_D:\cL(G)\to K^n$ has kernel\footnote{A function $f\in\cL(G)$ lies in $\ker(\ev_D)$ when $f(P_i)=0$ for every $i$, equivalently $v_{P_i}(f)\ge 1$ for every $i$. Since $\supp(G)\cap \supp(D)=\emptyset$, this is saying $(f)+G-D\ge 0$. Thus $\ker(\ev_D)=\cL(G-D)$.} $\cL(G-D)$. Hence,
\begin{equation}
\label{eq:ev_code_dimension}
k=\dim C_\cL(D,G)=\ell(G)-\ell(G-D).
\end{equation}
It turns out in general that if $\deg(A)<0$ for any divisor $A$, we have $\ell(A) = 0$, and therefore if $\deg(G)<n$%
, then $\deg(G-D)<0$, so $\ell(G-D)=0$, from which it follows that the evaluation map is injective. Therefore, $k=\ell(G)$.

If, moreover $\deg(G)>2g-2$,
$\ell(G) = \deg(G)-g+1$ and thus
\begin{equation}
k=\deg(G)-g+1,\quad (2g-2<\deg(G)<n).
\end{equation}
For the one-point code $G=rQ_\infty$, we have
\begin{equation}\label{eq:one-point-code-dim}
k=r-g+1,\quad (2g-2<r<n).
\end{equation}
We state without proof the following distance lower bounds:
\begin{enumerate}
\item \cite[Cor.~2.2.3]{Stichtenoth} $C_\cL(D,G)$ has distance $d\ge n-\deg(G)$~.
\item \cite[Thm.~2.2.7,~2.2.8]{Stichtenoth} The dual code $C_\cL(D,G)^\perp$ has distance $d^\perp \ge \deg(G)-(2g-2)$.
\end{enumerate}

Since the divisor of the product of two functions is the sum of the individual divisors, \((fg) = (f) + (g)\), one can see that one-point codes satisfy a multiplication property.
That is, for two codewords of a one-point code $(c_1,\cdots,c_n)\in C_\cL(D,rQ_\infty)$ and $(c_1',\cdots,c_n')\in C_\cL(D,r'Q_\infty)$, their entrywise-product is also contained in a one-point code.
\begin{equation}\label{eq:one-point-multiplication}
(c_1c_1',\cdots,c_nc_n')\in C_\cL(D,(r+r')Q_\infty).
\end{equation}
This is convenient for our purpose of constructing e.g., triorthogonal spaces. Puncturing such a space allows us to obtain qudit CCZ distillation protocols, as we will see shortly.

\subsection{Qudit distillation protocols from punctured one-point codes}
\label{sec:qudit-CCZ-one-point-codes}

We now show how to construct distillation protocols from one-point codes by puncturing.
Recall the triorthogonality conditions from Sec.~\ref{sec:protocols}.
Our procedure will take a space satisfying one of these triorthogonality conditions in some trivial way (i.e. \(k=0\)) and puncturing to obtain a non-trivial basis.
A map $F:\F_q^n\to \F_q$ is called \emph{multi-additive} if if is additive in each argument separately: $\forall j$,
$$F(x_1,\ldots,x_j+x_j',\dots,x_r)=F(x_1,\dots,x_j,\ldots,x_r)+F(x_1,\dots,x_j',\ldots,x_r).$$
Let $p=\text{char}\;\F_q$. 
If $s,t,u$ are powers of $p$, then
\begin{equation}
\lambda(x,y,z) = x^s y^t z^u\;,
\end{equation}
is multi-additive, equivalently $\F_p$-trilinear, because $(x_1+x_2)^{p^i}=x_1^{p^i}+x_2^{p^i}$ for non-negative integers $i$ and $x_1,x_2\in\F_q$. 
It need not be $\F_q$-linear in any argument. In the following, we will restrict \(\lambda\) to always be a map of this form.

For a metric $\bGamma\in\F_q^n$, define a map $\Lambda_{\bGamma} \colon \F_q^n \times \F_q^n \times \F_q^n \to \F_q$,
\begin{equation}
\label{eq:Lambda_ortho_wrt_Gamma}
\Lambda_{\bGamma}(\bg_1,\bg_2,\bg_3)=\sum_{i=1}^n \Gamma_i \lambda(g_{1,i},g_{2,i},g_{3,i})\;.
\end{equation}
We call an $\F_q$-linear code $C\subseteq \F_q^n$ $\Lambda_{\bGamma}$-orthogonal if 
\begin{align}\label{eq:generalized-triortho}
     \quad \Lambda_{\bGamma}(\bg_a, \bg_b, \bg_c) = 0\;,\quad\forall\;\bg_a,\bg_b,\bg_c \in C.
\end{align}
If $(g_a)_{a=1}^m$ is an $\F_q$-basis of $C$, it suffices to check all $m^3$ basis triples, since
$$\Lambda_{\bGamma}\!\left(\sum_a\alpha_ag_a,\sum_b\beta_bg_b,\sum_c\gamma_cg_c\right)=\sum_{a,b,c}\alpha_a^s\beta_b^t\gamma_c^u\Lambda_{\bGamma}(g_a,g_b,g_c).$$
Thus the condition depends on the space $C$, rather than on a choice of generator matrix. The qudit-CCZ and the twisted-triorthogonal conditions for $U_7$ of Sec.~\ref{sec:protocols} are instances of this construction.

It is typically easier to first construct a triorthogonal space and then puncture it to obtain a basis satisfying the triorthogonality conditions from Sec.~\ref{sec:protocols}.

\begin{proposition}[Punctured triorthogonal spaces]
\label{prop:puncturing-distillation-protocol}
    Let $C\subseteq \F_q^n$ be an $[n,m]_{q}$ $\F_q$ linear code that is $\Lambda_{\bGamma}$-triorthogonal. 
    Let $I=\{i_1,\dots,i_k\}\subseteq [n]$, put $J=[n]\backslash I$, and assume:
    \setlist{nolistsep}
    \begin{enumerate}
        \item the restriction of $C$ to the columns $I$, denoted $C|_I$, is the full space $C|_I=\F_q^k$; and
        \item puncturing on $I$ in injective on $C$, equivalently no nonzero codeword of $C$ is supported fully in $I$.
    \end{enumerate}
    After reordering the columns $[I]$ first (permute the coordinates of $\bGamma$ accordingly) and changing the row basis, a generator matrix of $C$ has the form $\tilde{G}=\left(\begin{array}{c|c}I_k & G_1\\\hline 0 & G_0\end{array}\right)$.
    The punctured matrix $G:=\begin{pmatrix}G_1\\\hline G_0\end{pmatrix}\in \F_q^{m \times (n-k)}$ has rank $m$.
    Define an $[[n-k,k]]_q$ CSS code by taking the $X$-stabilizers to be the $\text{rowspan}(G_0)$, and the $X$-logical representatives to be $G_1$ (the $Z$-stabilizers are $\text{rowspan}(G^\perp)$). 
    Call $(\bg_a)_{a=1}^m$ the rows of $G$, then they satisfy the following relation
    \begin{align}
    \label{eq:Lambda_after_puncture}
        \Lambda_{\bGamma|_J}(\bg_a,\bg_b,\bg_c) = \begin{cases}
            -\Gamma_a\;, & 1\le a = b = c \le k, \\
            0\;, & \text{otherwise}.
        \end{cases}
    \end{align}
    Let $d(C)$ and $d(C^\perp)$ be the minimum distance of $C$ and $C^\perp$, respectively.
    Then the qudit $X$ and $Z$ distances of the CSS code satisfy
    \begin{equation}
        d_X\ge d(C)-k\;,\quad d_Z\ge d(C^\perp)-k\;.
    \end{equation}
    A sufficient condition for assumption 2 is $k < d(C)$.
\end{proposition}
\begin{proof}
    Assumption 1 gives the systematic form; $\text{rank}(G)=m$ follows from assumption 2, and hence $\text{rank}(G_0)=m-k$.

    For $a\le k$, row $a$ of $\tilde{G}$ is $(\be_a\;|\;\bg_a)$ where $\be_a$ is the elementary vector that is one in the $a$-th coordinate and zero otherwise, while rows with $a>k$ are zero on $I$. Eq.~\eqref{eq:Lambda_after_puncture} can be established by extending \(\Lambda_{\bGamma|_J}\) acting on rows of \(G\) to \(\Lambda_{\bGamma}\) acting on rows of \(\tilde{G}\) and utilizing $\Lambda_{\bGamma}$-triorthogonality of the space.

    $d_X$ is lower bounded by the distance of the code spanned by the rows of \(G\), $d(\text{rowspan}(G))$, and puncturing $k$ coordinates lowers the weight of a codeword by at most $k$, establishing $d_X\ge d(C)-k$. For $d_Z$, by Prop.~\ref{prop:matrix_to_code}, $d_Z\ge d((\text{rowspan}(G_0))^\perp)$. Since the dual code distance can be interpreted as the least number of columns that are linearly dependent, and the columns of $I_k$ are linearly independent among themselves, one has $d_Z\ge d((\text{rowspan}(G_0))^\perp)\ge d(C^\perp)-k$.

    One can see that a sufficient condition for assumption 2 to hold is $k < d(C)$, since a codeword of a distance \(d\) code is uniquely recoverable from erasure of \(d-1\) coordinates.
\end{proof}
\begin{remark}
\label{remark:zeros_in_metric}
If $C$ is $\Lambda$-orthogonal w.r.t. some $\bGamma$ that contains some zero coordinates. Let $Z=\{i\;|\;\Gamma_i=0\}$ and $z=|Z|$. Then $C$ punctured at $Z$, denoted as $\tilde{C}=\text{Punc}_Z(C)$. 
$\tilde{C}$ is $\Lambda_{\text{Punc}_Z(\bGamma)}$-triorthogonal.
If puncturing is injective, equivalently no nonzero codeword of $C$ is supported fully in $Z$, then $\dim\tilde{C}=\dim C$. A sufficient condition for this assumption is $z<d(C)$; this will hold in our one-point codes below.
We also have $d(\tilde{C})\ge d(C)-z$, $d(\tilde{C}^\perp)\ge d(C^\perp)$. 

After the puncturing of size $k$ from $\tilde{C}$ in the above proposition, we have a $[[n-z-k,k]]_q$ CSS code with $d_X\ge d(C)-z-k$, $d_Z\ge d(C^\perp)-k$. Since the lower bound of $Z$ distance is unaffected, and $Z$ distance is the bottleneck (for all our protocols), puncturing $Z$ allows us to pay $z$ less magic input resources to distillation without compromising the distance. Therefore, after constructing $\bGamma$, we always delete the zero coordinates.
\end{remark}    

We will now show that the one-point codes are $\Lambda_{\bGamma}$-triorthogonal spaces for some $\bGamma$ that is not all-zero.
\begin{lemma}[One-point codes are \(\lambda\)-triorthogonal spaces]\label{lemma:one-point-triortho-space}
    Let $\chi$ be a smooth projective curve of genus $g$ over $\F_q$ ($\text{char}\;\F_q=p$). Let $D=P_1+\dots+P_n$ be a sum of distinct $\F_q$-rational points on $\chi$, and let $Q\notin \supp(D)$ be rational. Assume $n>g$ and put $H=(n+g-2)Q$. There exists a nonzero $\bGamma\in C_\cL(D,H)^\perp$ with $\wt(\bGamma)\ge n-g$. In particular, $\bGamma$ has at most $g$ zero entries.
    For $\Lambda_{\bGamma}(\bg_a,\bg_b,\bg_c)=\sum_{i=1}^n \Gamma_i (g_{a,i})^s (g_{b,i})^t (g_{c,i})^u$ where $s,t,u$ are powers of $p=\text{char}\;\F_q$. If $(s+u+t)r\le n+g-2$, then $C_\cL(D,rQ)$ is a $\Lambda_{\bGamma}$-triorthogonal space.
\end{lemma}
\begin{proof}
    Since $\deg H=n+g-2\ge 2g-2$, the Riemann-Roch theorem, Eq.~\eqref{eq:Riemann-Roch_space_dim}, gives $\ell(H)=\deg H-g+1=n-1$. Hence $\dim C_\cL(D,H)\le n-1$ by Eq.~\eqref{eq:ev_code_dimension}, so its dual code contains a nonzero codeword $\bGamma$. By~\cite[Thm.~2.2.7,~2.2.8]{Stichtenoth}, the dual code $C_\cL(D,H)^\perp$ has distance $\ge \deg(H)-(2g-2)=n-g$, and thus $\wt(\bGamma)\ge n-g$.

    By the multiplication property of one-point codes Eq.~\eqref{eq:one-point-multiplication}, and the containment relationship: $C_\cL(D,r'Q)\subseteq C_\cL(D,rQ)$ if $r'\le r$. We have $$((g_{a,1})^s(g_{b,1})^t(g_{c,1})^u,\dots,(g_{a,n})^s(g_{b,n})^t(g_{c,n})^u))\in C_\cL(D,(s+t+u)rQ)\subseteq C_\cL(D,H),$$
    for all $\bg_a,\bg_b,\bg_c\in C_\cL(D,rQ)$. It follows that $\bGamma \in C_\cL(D,H)^\perp\subseteq C_\cL(D,(s+t+u)rQ)^\perp$, and therefore, $\Lambda_{\bGamma}(\bg_a,\bg_b,\bg_c)=0$.
\end{proof}
We are now ready to use the defect-zero curves in Table~\ref{table:maximal_curves_CCZ_only} to construct qudit-CCZ distillation protocols. 
For all these curves, define $D$ as sum of distinct $\F_q$-rational points disjoint from $Q=(0:1:0)$.
The explicit calculation of $\bGamma$ for them can be found in App.~\ref{sec:calculate_metric}.
As a summary of the calculations there, $\bGamma$ can be taken to be all-one for all curves in Tab.~\ref{table:maximal_curves_CCZ_only} except for the Klein quartic and elliptic curves over $\F_{2^s}$, $s>2$. For the latter, one will see later that some coordinates of $\bGamma$ can be zero. 
By remark~\ref{remark:zeros_in_metric}, we can remove those places from $D$.

\begin{corollary}[punctured one-point CSS code parameters]
\label{cor:punctured_code_parameters}
Let $\chi$ be a smooth projective curve of genus $g$ over $\F_q$.
Let $D=P_1+\dots+P_n$ be a sum of distinct $\F_q$-rational points on $\chi$, and let $Q\notin \supp(D)$ be rational.
Suppose $C=C_L(D,rQ)$ is $\Lambda_{\bGamma}$-triorthogonal, where $2g-2<r<n$. Let $z$ be the number of zero entries of $\bGamma$. If
$$k\le r-g+1,\quad z+k<n-r,$$
then one can delete the $z$ zero coordinates and choose $k$ remaining coordinates to puncture so as to obtain a CSS code $[[n-z-k,\;k,\;(d_X,d_Z)]]_q$ with
$$d_X\ge n-z-r-k,\qquad d_Z\ge r-2g+2-k.$$

\end{corollary}

\begin{proof}
The Riemann-Roch theorem and previous distance lower bounds give
$$\dim C=r-g+1,\quad d(C)\ge n-r,\quad d(C^\perp)\ge r-2g+2.$$
The inequality $z+k<n-r$ guarantees puncturing to be injective.
After deleting the zero coordinates of $\bGamma$, the code still has dimension $r-g+1$, and one can further chooses $k$ columns to puncture. Prop.~\ref{prop:puncturing-distillation-protocol} and remark~\ref{remark:zeros_in_metric} give the result.
\end{proof}
\input{maximal_curve_CCZ_only}

For the purpose of distilling qudit-CCZ gates, take $\lambda(x,y,z)=xyz$ in Eq.~\eqref{eq:Lambda_ortho_wrt_Gamma}. 
Lemma~\ref{lemma:one-point-triortho-space} applies when $3r\le n+g-2$. Thus one is free to choose $r_*=\left\lfloor \frac{n+g-2}{3} \right\rfloor$.
Assume $n>5g$, which is the case for defect-zero curves over binary fields. Then $r_*$ satisfies $r_*>2g-2$, and Cor~\ref{cor:punctured_code_parameters} applies.
For any admissible $k$, we get a $(n-z-k)\to k$ qudit-CCZ distillation protocol with distance bounded below by $r_*-2g+2-k$ (the $d_X$ lower bound is no smaller than this).

In fact, when working with the curves in Table~\ref{table:maximal_curves_CCZ_only}, we find that sometimes one can take a bigger $r_*$, especially when $g$ becomes larger. In Lemma~\ref{lemma:one-point-triortho-space}, one can take $H$ such that $\dim C_\cL(D,H)^\perp=1$ and take $\bGamma$ to be the nonzero codeword of $C_\cL(D,H)^\perp$. We do not know how to write down a closed formula for $H$, because this depends on where the gaps are. One must have $\deg H\le n+2g-1$ (see the next subsection for the equality case), however. This is because $C_\cL(D,(n+2g-1)Q_\infty)=\F_q^n$: one has $\dim C_\cL(D,G)=\ell(G)-\ell(G-D)$ and take $G=(n+2g-1)Q_\infty$ in this case. Since $\deg(G)>\deg(G-D)=n+2g-1-n>2g-2$, by Riemann-Roch, one has $\ell(G)=(n+2g-1)-g+1=n+g$ and $\ell(G-D)=(2g-1)-g+1=g$.

\begin{remark}
\label{remark:metric_for_U7}
One can also use one-point codes to distill $U_7$ gates. A naive way to do so is to take $\lambda(x,y,z)=x^4y^2z$ in Eq.~\eqref{eq:Lambda_ortho_wrt_Gamma}, and this interpreting the $U_7$ orthogonality condition Eq.~\eqref{eq:U7_ortho_wrt_metric} as seven-orthogonality (instead of twisted orthogonality). As we commented in \ref{sec:U7_protocols}, this is highly inefficient, because one actually never needs to test the overlap of more than three rows. This suboptimality is the reason why we choose to not include a table of $U_7$ gate distillation protocols parameters.\\
Nevertheless, we do want to comment on the fact that metric can make the parameters of this naive seven-orthogonality construction better.
For example, by considering a generalized RS code over $\F_{32}$ ($U_7$ in this field is Clifford equivalent to $\TOF\#$), the naively punctured $31\TOF\#\to1\TOF\# \;(d=5)$ protocol can be made $29$TOF\#$\to1$TOF\# ($d=5$) using metric. One can also use metrics to obtain $15\to1\;(d=3)$ and $8\to1\;(d=2)$ TOF\#$\to$TOF\# protocols.
\end{remark}

\subsection{Distillation without puncturing}\label{sec:distillation_without_puncturing}
We now give some exceptional distillation protocols that do not appear to be puncturings of some $\Lambda$-triorthogonal spaces.
We constrain ourselves to one-point code with monomial basis in this subsection and leave a precise understanding of the most general phenomenon to further work.

We call a monomial evaluation row a \emph{natural logical row}, if it has non-zero self-intersection and has zero $\Lambda_{\bGamma}$ intersection with the other rows. 
An example would be $\ev(x)$ in the $[[4,1,2]]_4$ Reed-Solomon code, where $\ev(1)$ serves as the stabilizer row. More generally, let $q=2^s$ with $s$ even, then $(2^2-1)$ divides $(2^s-1)=q$. Let $a=\frac{q-1}{3}$, then $\ev(x^a)$ naturally has non-zero triple self-intersection: $\sum_{\alpha\in\F_q}(\alpha^{(q-1)/3})^3=\sum_{\alpha\in\F_q}\alpha^{q-1}=\sum_{\alpha\in\F_q^\times}1=1$. Take the $X$ stabilizers to be generated by $\ev(1),\ev(x),\dots,\ev(x^{a-1})$. This gives a $[[q,1,(d_X=q-a,\;d_Z=a+1)]]_q$ quantum Reed-Solomon code with transversal qudit-CCZ. In fact, when $s$ is even, one can construct a $[[4,1,(3,2)]]_q$ generalized Reed-Solomon code with transversal qudit-CCZ: choose four distinct points $\alpha_1,\dots,\alpha_4\in\F_q$, and define $\Gamma_i=\left(\prod_{j\neq i}(\alpha_i-\alpha_j)\right)^{-1}$; use $(\alpha_i)_{1\le i\le 4}$, which is $\ev(x)$, as the logical row and all-ones as the stabilizer row, together with the metric $(\Gamma_i)_{1\le i\le 4}$.

Besides Reed-Solomon, other AG codes can demonstrate the same behavior, although our understanding of them is case-by-case.

The one-point code over the Hermitian curve $y^4+y=x^5$ over $\F_{16}$ is of genus $6$ and has $64$ affine $\F_{16}$-rational points. One also has $v_{Q_\infty}(x)=-4$ and $v_{Q_\infty}(y)=-5$, thus $\{x^iy^j\;|\; 0\le i, 0\le j\le 3,\;4i+5j\le r\}$ form a basis of $\cL(rQ_\infty)$, cf. Thm.~\ref{thm:stichtenoth_artin_schreier}. 
For this one-point code, we have $\dim C_\cL(D,74Q_\infty)^\perp=1$ and the only codeword (up to scalar multiplication) in $C_\cL(D,74Q_\infty)^\perp$ is all-one (cf. App.~\ref{sec:metric_artin_schreier}); we also have $\dim C_\cL(D,75 D_\infty)=0$ by our comments before Remark~\ref{remark:metric_for_U7} because $75=n+2g-1$.
Thus, we can take $C_\cL(D, 24Q_\infty)$ to be the stabilizer space and $\bGamma=\mathbf{1}$. 
We find that $\ev(y^5)$ of pole order $25$ at $Q_\infty$ serves as a natural logical row: $\sum_{(x,y)\in D} (y^5)^3=1$. 
\begin{remark}
\label{remark:one_point_code_with_natural_logical_row}
In general, for qudit-CCZ distillation, if we use $C_\cL(D,r_0 Q)$ as $X$ stabilizers, and a basis $\{\bg_1,\dots,\bg_h\}$ of $C_\cL(D,r_1 Q)\backslash C_\cL(D,r_0 Q)$ serving as natural logical rows, i.e.
$\Lambda_{\bGamma}(\bg_a,\bg_b,\bg_c)\neq 0$ if and only if $1\le a=b=c\le h$.
Then, after deleting the $z$ zero coordinates of $\bGamma$ and puncturing $k-h$ columns (with the injectivity assumption), one obtains an $(n-z-(k-h))\to k$ qudit-CCZ distillation protocol with 
$$d_X\ge n-r_1-z-(k-h),\qquad d_Z\ge r_0-2g+2-(k-h).$$
A sufficient condition for such a non-zero $\bGamma$ to exist is $\dim C_\cL(D,(2r_0+r_1)Q)^\perp\ge 1$.
\end{remark}

For $U_7$ on Reed-Solomon codes, one can get a natural logical row when $s$ is divisible by $3$, since then $(2^3-1)$ divides $(2^s-1)$. Let $a=\frac{q-1}{7}$, then $\ev(x^a)$ naturally serves as a logical row and choosing the span of $\ev(1),\ev(x),\dots,\ev(x^{a-1})$ as stabilizers gives a $[[q,1,(q-a,a+1)]]_q$ quantum Reed-Solomon code with transversal $U_7$ gate.
Our $\F_8$ Reed-Solomon construction in~\ref{con:8CCZ_to_2CCZ} is an example, where $\ev(x^1)$ is a natural logical row. However, it turns out that in this case, $\ev(x^2)$ can also be added to the logical subspace. This hints at the fact that, naively following the degree (interpreting the twisted-three orthogonality as seven orthogonality) gives only suboptimal protocols for $U_7$ gate, and we believe this is the case for distilling gates whose orthogonality conditions involve twisting.

For this reason, we developed the forbidden hypergraph method in Con.~\ref{con:64CCZ-to-8CCZ}. We impose the Ansatz of rows being evaluations of monomials for this numerical method to be tractable. 
The next construction based on Klein-quartic $Y^3Z+X^3Y+Z^3X=0$ over $\F_8$ is another demonstration of this method:
\begin{construction}[$24\CCZ\to 4\CCZ$ at $d=3$]
\label{con:24CCZ-to-4CCZ-d3}
For the $X$ stabilizers, we write down its columns as the $24$ points in $\bbP^2_{\F_8}$ that lie on the Klein quartic curve $Y^3 Z+ X^3 Y + Z^3 X=0$, which is the evaluation of $X$, $Y$, $Z$. Again, since the columns are not related to each other by a scalar, the dual code distance is at least three.
For the $X$-logicals, we choose $X^1 Y^2 Z^5$, $X^2 Y^5 Z^1$, $X^5 Y^1 Z^2$ and $X^5 Y^5 Z^5$. They are obtained similar to Con.~\ref{con:64CCZ-to-8CCZ}, that is to solve for a maximum independent set in a certain forbidden hypergraph. However, we put a restriction on the total degree of the monomials, that is $a+b+c$ in $X^a Y^b Z^c$ to be $a+b+c\equiv 1\!\!\mod{7}$. This is because, normalization is required when mixing the evaluation of homogeneous polynomials (in projective space; a monomial is automatically a homogeneous polynomial) of different degree. This is because $(X:Y:Z)$ and $(\alpha X : \alpha Y : \alpha Z)$ define the same projective point. To normalize a homogeneous polynomial of total degree $d$ at each point, one can divide the evaluation by $\alpha^d$, where $\alpha$ is the first coordinate of the projective point that is nonzero\footnote{This is not the only way of doing normalization. Say in our case, since $X+Y+Z\neq 0$ on any of the points lying on the Klein quartic, we can choose to normalize the evaluation of $X^a Y^b Z^c$ by $(X+Y+Z)^d$ where $d=a+b+c$.}. As we commented before, the twisted three-orthogonality is special for $\F_8$ in the sense that scaling a column of the distillation matrix by $\gamma\in\F_8^\times$ does not change the intersection pattern. Therefore, normalizing by the same scalar at a projective point implies that the total degree of the homogeneous polynomials should be the same modulo seven.

After committing to $X,Y,Z$ serving as the stabilizers, we search over all monomials $X^a Y^b Z^c$ where $0\leq a,b,c\le 6$ and $a+b+c\equiv 1\!\!\mod{7}$ that each individually form a valid logical with stabilizers $X,Y,Z$. There are $16$ such candidates. Then we build a forbidden hypergraph just like in Con.~\ref{con:64CCZ-to-8CCZ} and then solve for the maximum independent set, the unique solution of which is the exponent set $$(a,b,c)=\{(1,2,5),\;(2,5,1),\;(5,1,2),\;(5,5,5)\}.$$
\end{construction}

\subsection{Protocols based on trace codes}
\label{sec:trace_codes}
Let $K=\F_{q^m}$ and $k=\F_q$. For a $K$-linear code $C\subseteq K^n$, its \emph{subfield subcode} and \emph{trace code} are
\begin{equation}
C|_{\F_q}:=C\cap k^n,\quad \tr(C)_{K/k}:=\{ (\tr_{K/k}(c_i))_{i=1}^n\;|\; (c_1,\dots,c_n) \in C\}\subseteq k^n
\end{equation}
Delsarte's theorem relates the two by duality
\begin{equation}
    (C|_k)^\perp=\tr_{K/k}(C^\perp),\quad \text{hence } \tr_{K/k}(C)^\perp = C^\perp|_k \;.
\end{equation}
A consequence is that, if $\dim_K C=r$, then $r\le \dim_k \tr_{K/k}(C)\le \min\{n, mr\}$. We refer the readers to \cite[Ch.~9]{Stichtenoth} for proofs.

Our purpose here is to obtain distillation protocols with $T$ inputs, thus we will only consider the case of tracing down to $k=\F_2$. 

Let $q=2^s$. For $\bc\in\mathbb F_q^n$, put
$T(\bc)=\{\tr_{\F_q/\F_2}(\lambda \bc);|\;\lambda\in\F_q\}\subseteq\F_2^n$, i.e., trace is applied coordinatewise. 
If $(\alpha_i)_{i=1}^s$ is an $\F_2$-basis of $\F_q$, then
$$T(\bc)=\text{span}_{\F_2}\{\tr_{\F_q/\F_2}(\alpha_i \bc);|\; 1\le i\le s\}.$$
Consequently, if $\bc_1,\dots,\bc_r$ generate an $\F_q$-linear code $C$, then
$$\tr(C):=\tr(C)_{\F_q/\F_2}=\text{span}_{\F_2}\{\tr_{\F_q/\F_2}(\alpha_i \bc_a)\;|\;1\le i\le s,\; 1\le a\le r\}.$$
These $sr$ rows need not be independent.
As an example, for $\bc=(c_1,\dots,c_n)$, the two subspaces $T((c_1^{2^t},\dots,c_n^{2^t}))$ and $T(\bc)$ are equivalent, because $\tr(\lambda c_j^{2^t})=\tr(\lambda^{2^{s-t}}c_j)$, and $\lambda\mapsto \lambda^{2^{s-t}}$ is an automorphism on $\F_q$. Thus the (evaluation of) functions $f,f^2,f^4,\dots$ generate the same subspace after taking trace.

The reason we consider trace code is because a cyclic code can be obtained as trace code of Reed-Solomon code~\cite[Prop.~9.2.4]{Stichtenoth}, and the punctured Reed-Muller code is a cyclic code~\cite[Ch.~13, Th.~11]{theoryEC} \footnote{The binary $\RM(r,m)^*$ is the trace code of $\ev_{\F_q^\times} \text{span}_{\F_q} \{x^a\;|\;\wt(a)\le r\}$ where $q=2^m$, which is strictly speaking not a Reed-Solomon code, but a monomial evalution code on $\F_q$.}. Motivated by the $64T\to 2\CCZ$ at $2944p^4$ protocol constructed using Reed-Muller code~\cite{Haah2018}, we next give a $64T\to 2\CCZ$ protocol at $2720p^4$ based on the trace code of Hermitian code.

\begin{example}[Trace code of the Hermitian code]
Let $D$ be the set of affine $\F_{q^2}$-rational points on the Hermitian curve $y^q+y=x^{q+1}$, then $D$ contains $q^3$ points.
For $x^i y^j$ with $0\leq i\le q^2-1$ and $0\le j\le q-1$, one has
\begin{equation}
\sum_{P\in D}\ev_P(x^i y^j)=\begin{cases}1, & \text{if }i=q^2-1,\;j=q-1,\\ 0, & \text{otherwise.}\end{cases}
\end{equation}
This is because for a fixed $x\in \F_{q^2}$, the $q$ solutions of $y^q+y=x^{q+1}$ form a affine translate of subfield $y_0+\F_q$. For $0\le j\le q-1$, $\sum_{u\in\F_q}(y_0+u)^j=\begin{cases}1, & j=q-1,\\ 0, & 0\le j<q-1.\end{cases}$. 
The sum over $D$ therefore vanishes unless $j=q-1$. In that case it becomes $ \sum_{x\in\F_{q^2}}x^i = \begin{cases} 0,&0\le i<q^2-1,\\1,&i=q^2-1.\end{cases}$.

\end{example}
\begin{construction}[$64T\to 2\CCZ$ at $2720 p^4$]\label{con:64tto2ccz}
Here we give a $64T\to 2CCZ$ protocol, for which the probability of any error is $2720p^4$. This is based on the trace code of the Hermitian one-point code derived from $y^4+y=x^5$ over $\F_{16}$. This improves over the error rate of the known $64T \to 2CCZ$ protocol, which is $2944p^4$ (although both protocols have a per-state average error probability equal to $2368p^4$). 

From Th.~\ref{thm:stichtenoth_artin_schreier}, one can see that the monomial basis for $\cL(12Q_\infty)$ is $\{1, x, y, x^2, xy, y^2, x^3\}$. Moreover, recall that for trace codes, if row $x$ and $y$ are chosen, then $x^2$ and $y^2$ are redundant. Consider the trace code of $\{1,x,y,xy,x^3\}$, the only odd triple intersections occur between trace codes of row $xy$ and $x^3$; treat them as the logical operators and $\{1,x,y\}$ as stabilizers.

Only $xy$ and $x^3$ intersects non-trivially because among all the monomials $x^i y^j$ where $0\le i\le 15, 0\le j\le 3$, only $\sum_{P\in D} \ev_P(x^{15}y^3)=1$ and $(xy)^1 (xy)^2 (x^3)^4=x^{15} y^3$. 
The exact intersection between their trace subspaces can be calculated as (for $\lambda,\mu,\nu\in\F_{16}$):
$$\sum_{P\in D} \tr(\lambda xy(P)) \tr(\mu xy(P)) \tr(\nu x^3(P))=
\tr\!\left((\lambda\mu^2+\lambda^2\mu)\nu^4\right).$$

Using the self-dual basis $(\alpha^3, \alpha^{12}, \alpha^7, \alpha^{13})$ for $\F_{16}$ defined by $\alpha^4=\alpha+1$ and expanding each qudit row $\bc$ into four qubit rows as $\tr(\alpha_i \bc)$, we obtain the odd triple-intersections on the eight logical generators derived from $xy,x^3$ is as follows. 
\begin{center}
\begin{quantikz}[row sep={0.4cm, between origins}, column sep={0.4cm, between origins}]
\lstick{$x_1$} & \ctrl{5} & \ctrl{6} & \ctrl{7} & \ctrl{6} & \ctrl{7} &  &  &  &  &  &  &  &  &  &  &  & \\
\lstick{$x_2$} &  &  &  &  &  & \ctrl{4} & \ctrl{6} & \ctrl{5} & \ctrl{6} & \ctrl{6} &  &  &  &  &  &  & \\
&  &  &  &  &  &  &  &  &  &  & \ctrl{4} & \ctrl{5} & \ctrl{5} &  &  &  & \\
&  &  &  &  &  &  &  &  &  &  &  &  &  & \ctrl{3} & \ctrl{3} & \ctrl{4} & \\
\lstick{$x_3$} & \control{} & \control{} & \control{} &  &  & \control{} & \control{} &  &  &  & \control{} & \control{} &  & \control{} &  &  & \\
\lstick{$x_4$} & \control{} &  &  & \control{} &  & \control{} &  & \control{} & \control{} &  &  &  & \control{} &  & \control{} & \control{} & \\
\lstick{$x_5$} &  & \control{} &  & \control{} & \control{} &  &  & \control{} &  & \control{} & \control{} &  &  & \control{} & \control{} &  & \\
\lstick{$x_6$} &  &  & \control{} &  & \control{} &  & \control{} &  & \control{} & \control{} &  & \control{} & \control{} &  &  & \control{} & \\
\end{quantikz}
\end{center}
Next, we purge two logicals (third and fourth wire), and the CCZ circuit on the rest six logicals is Clifford equivalent to two decoupled CCZs; in terms of phase polynomial, it differs by only degree $\le 2$ terms to $x_1(x_3+x_4+x_6)(x_5+x_4+x_6)+x_2(x_4+x_3+x_5)(x_6+x_3+x_5)$. This can be seen by observing $x_3x_4+x_3x_5+x_3x_6+x_4x_5+x_5x_6=(x_3+x_4+x_6)(x_5+x_4+x_6)+x_4+x_6$, and similarly for the CCZs involving $x_2$. 
\end{construction}

\begin{construction}[$64T\to \F_4$-qudit CCZ at $d=4$]\label{con:64ttof4qudit}
Set the first two wires to $|0\ra$ and so that all the CCZs associated with them are removed. What is left, i.e., the last six CCZ gates on the last six logicals, form an $\F_4$ qudit CCZ gate (cf.~\ref{fig:F4_qudit_CCZ}).
\end{construction}

\phantomsection
\section*{Acknowledgments}
\addcontentsline{toc}{section}{Acknowledgments}

The authors thank Jonathan Moussa for pointing out to us the 8CCZ-to-2CCZ distance two protocol from~\cite{chamberland2022building}.
The authors acknowledge inspiring discussions with Shraddha Singh.

The authors acknowledge the use of large language models, in particular ChatGPT 5.6 Sol, and Claude Opus 5, to generate code for the calculation of error rates, and spacetime footprints, of concatenated protocols. The code was checked and tested by the authors, and the authors take responsibility for all content.

A.G. acknowledges funding from the Swiss State Secretariat for Education, Research and Innovation (SERI) under contract No. 20QU-1\_225224. A.G. is affiliated with the Institute of Theoretical Physics of ETH Zurich.
C.A.P. is currently a Simons-CIQC postdoctoral fellow at the Simons Institute for the Theory of Computing, supported by NSF QLCI Grant 2016245. A.W. acknowledges support from the MIT Department of Physics, from the MIT-IBM Watson AI Lab, and from NSG grant PHY-2325080. A.W. is affiliated with the Center for Theoretical Physics — a Leinweber Institute
Massachusetts Institute of Technology, Cambridge, MA. This pre-print is assigned number MIT-CTP/6086.

\printbibliography

\appendixtitleon
\appendixtitletocon
\begin{appendices}
\input{appendix}
\end{appendices}

\end{document}

%% file: macrosetup.tex
\newcommand{\A}{\mathbb{A}}

\newcommand{\N}{\mathbb{N}}
\newcommand{\F}{\mathbb{F}}

\newcommand{\bbP}{\mathbb{P}}
\newcommand{\Q}{\mathbb{Q}}

\newcommand{\Z}{\mathbb{Z}}

\newcommand{\cG}{\mathcal{G}}
\newcommand{\cL}{\mathscr{L}}
\newcommand{\cH}{\mathcal{H}}

\newcommand{\cO}{\mathcal{O}}

\newcommand{\ra}{\rangle}
\newcommand{\la}{\langle}

\newcommand{\ba}{\mathbf{a}}

\newcommand{\bc}{\mathbf{c}}
\newcommand{\be}{\mathbf{e}}
\newcommand{\bg}{\mathbf{g}}
\newcommand{\bs}{\mathbf{s}}
\newcommand{\bu}{\mathbf{u}}

\newcommand{\bx}{\mathbf{x}}

\newcommand{\bl}{\boldsymbol\ell}
\newcommand{\bGamma}{\boldsymbol{\Gamma}}

\newcommand{\bin}{\text{bin}}
\newcommand{\wt}{\text{wt}}

\newcommand{\diag}{\text{diag}}

\newcommand{\GL}{\text{GL}}
\newcommand{\GF}{\text{GF}}

\newcommand{\RM}{\text{RM}}

\newcommand{\CNOT}{\text{CNOT}}

\newcommand{\CCZ}{\text{CCZ}}
\newcommand{\CCS}{\text{CCS}}
\newcommand{\CS}{\text{CS}}
\newcommand{\CZ}{\text{CZ}}
\newcommand{\res}{\text{res}}

\newcommand{\supp}{\text{supp}}

\newcommand{\tr}{\text{tr}}
\newcommand{\Nm}{\text{Nm}}
\newcommand{\TOF}{\text{TOF}}
\newcommand{\rk}{\text{rank}}
\newcommand{\ev}{\mathrm{ev}}
\newcommand{\LM}{\text{LM}}

\renewcommand{\epsilon}{\ensuremath\varepsilon}

%% file: maximal_curve_CCZ_only.tex
\begin{table}[htbp]
\centering
\begingroup
\small
\renewcommand{\arraystretch}{1.5}
\setlength{\tabcolsep}{4pt}

\setcellgapes{2.5pt}
\makegapedcells

\begin{tabularx}{\linewidth}{|c|c|c|c|P{0.24\linewidth}|Y|}
\hline
\multirow{1}{*}{Field}
& \multirow{1}{*}{$g$}
& \multirow{1}{*}{$\#\chi(\F_q)$}
& \multirow{1}{*}{$\Delta n$}
& \multirow{1}{*}{Curve}
& \multicolumn{1}{|c|}{Protocols $(k,d)$} \\
\hline

\multirow{2}{*}{$\F_4$}
  & $0$ & $5$ & $0$ & projective line
      & $(k,3-k)$
  \\ \cline{2-6}
  & $1$ & $9$ & $0$ & $y^2+y=x^3$
      & $(k,3-k)$
  \\ \hline

\multirow{3}{*}{$\F_8$}
  & $0$ & $9$ & $1$ & projective line
      & $(k,4-k)$
  \\ \cline{2-6}
  & $1$ & $14$ & $2$ & $y^2+xy=x^3+x^2+x$
      & $(1,3),\; (3,2)$
  \\ \cline{2-6}
  & $3$ & $24$ & $1$ & $y^3+x^3y=x$
      & $(2,4),\; (4,3),\; (6,2)$
  \\ \hline

\multirow{4}{*}{$\F_{16}$}
  & $0$ & $17$ & $0$ & projective line
      & $(k,7-k)$
  \\ \cline{2-6}
  & $1$ & $25$ & $1$ & $y^2+y=x^3+x^2+1$
      & \makecell[c]{$(k,8-k),\;1\le k\le 5$,\\ $(7,2)$}
  \\ \cline{2-6}
  & $2$ & $33$ & $1$ & $y^2+y=x^5$
      & \makecell[c]{$(k,9-k),\;1\le k\le 5$,\\ $(6,4),\; (7,3),\; (9,2)$}
  \\ \cline{2-6}
  & $6$ & $65$ & $0$ & $y^4+y=x^5$
      & \makecell[c]{$(k,15-k),\;1\le k\le 7,\;\; (k,16-k),\;k\in\{8,9\}$,\\ $(k,17-k),\;10\le k\le 12,\;\; (k,18-k),\;k\in\{13,14\}$,\\ $(15,4),\;(16,3),\; (18,2)$}
  \\ \hline

\multirow{2}{*}{$\F_{32}$}
  & $0$ & $33$ & $1$ & projective line
      & $(k,12-k)$
  \\ \cline{2-6}
  & $1$ & $44$ & $2$ & $y^2+xy=x^3+x$
      & \makecell[c]{$(k,14-k),\;1\le k\le 11$,\\ $(13,2)$}
  \\ \hline

\multirow{2}{*}{$\F_{64}$}
  & $0$ & $65$ & $0$ & projective line
      & $(k,23-k)$
  \\ \cline{2-6}
  & $1$ & $81$ & $0$ & $y^2+y=x^3$
      & \makecell[c]{$(k,27-k),\;1\le k\le25,$\\ $(26,2)$}
  \\ \hline
\end{tabularx}

\endgroup

\caption{The defect-zero curves (cf. Sec.~\ref{sec:curves_with_many_points}) we use to define qudit CCZ distillation protocols, and the parameters of the resulting qudit CCZ distillation protocols. $g$ denotes the genus of the curve, $\#\chi(\F_q)$ the number of rational points on the curve, see Sec.~\ref{sec:intro_to_curves} for an introduction. 
The discussion on the transformation from a curve to a qudit CCZ distillation protocol is contained in Sections~\ref{sec:intro_to_AG_codes} to~\ref{sec:distillation_without_puncturing}.
For the protocols, a tuple $(k,d)$ implies an $(\#\chi(\F_q)-\Delta n-k) \to k$ distillation protocol, with distance $d$, taking qudit CCZ states as input and output. 
$\Delta n=1+z-h$, where $z$ is the number of zero in metric (cf. Cor.~\ref{cor:punctured_code_parameters}), and $h$ is the number of natural logical row (cf. remark~\ref{remark:one_point_code_with_natural_logical_row}). For the distances, lower bounds come from Cor.~\ref{cor:punctured_code_parameters}, and upper bounds are found by QDistRnd \cite{QDistRnd_paper, QDistRnd_repo}. For each $k$, we randomly puncture $k$ columns a hundred times and pick the largest obtained distance. When the upper bound is larger than the lower bound, we brute-force certify if the upper bound holds. 
}
\label{table:maximal_curves_CCZ_only}
\end{table}

%% file: appendix.tex
\section[Orthogonality condition for the CS gate]{Orthogonality condition for the $\CS$ gate}
\label{sec:CS_transversality_condition}
In this appendix, we focus on codes over $\mathbb{F}_4$, deriving necessary and sufficient conditions for a transversal $\CS$ to be a valid logical gate on their binarized versions. The corresponding protocols were covered in Section~\ref{sec:CS_distillation_protocols}. To be clear, the physical gate will be the $\CS$ gate applied on the pairs of qubits forming each $\mathbb{F}_4$-qudit.
\subsection{Method 1: Embedded code}
\label{sec:embedded_code}
After binarization, an entry $\gamma \in \mathbb{F}_4$ appearing in the matrix defining the $\mathbb{F}_4$ code, appears in the binarized $X$ logical/stabilizer matrix as $M(\gamma):=\begin{pmatrix}\tr(\gamma \omega^2) & \tr(\gamma) \\ \tr(\gamma) & \tr(\gamma \omega) \end{pmatrix}$. Here, recall that we are representing $\mathbb{F}_4$ as $\{0,1,\omega, \omega^2\}$, where $\omega^2 = \omega + 1$, and $\{\omega, \omega^2\}$ forms a self-dual basis (note that $\omega^3 = 1$). Further, after the embedded code gadget described in Section~\ref{sec:CS_distillation_protocols}, its contribution becomes
\begin{equation}\begin{pmatrix}\tr(\gamma \omega^2) & \tr(\gamma) & \tr(\gamma\omega^2)+\tr(\gamma) \\ \tr(\gamma) & \tr(\gamma \omega) & \tr(\gamma)+\tr(\gamma\omega) \end{pmatrix}=\begin{pmatrix}\tr(\gamma \omega^2) & \tr(\gamma) & \tr(\gamma\omega) \\ \tr(\gamma) & \tr(\gamma \omega) & \tr(\gamma \omega^2) \end{pmatrix}.\end{equation}
To be explicit, the new $X$ logical / stabilizer matrix becomes
\begingroup
\small
\setlength{\arraycolsep}{2pt}
\begin{equation*}
\begin{pmatrix}
\bg_1 \\ \cdots \\\hline\cdots \\ \bg_m
\end{pmatrix}=
\begin{pmatrix}
g_{11} & \cdots & g_{1n} \\
& \cdots & \\\hline
& \cdots & \\
g_{m1} & \cdots & g_{mn}
\end{pmatrix}
\rightarrow
\begin{pmatrix}
\tr(g_{11} \omega^2) & \tr(g_{11}) & \tr(g_{11} \omega) & \cdots & \tr(g_{1n}  \omega^2) & \tr(g_{1n}) & \tr(g_{1n} \omega)\\
\tr(g_{11}) & \tr(g_{11} \omega) & \tr(g_{11} \omega^2) & \cdots & \tr(g_{1n}) & \tr(g_{1n} \omega) & \tr(g_{1n} \omega^2)\\
\cdots & \cdots & \cdots & \cdots & \cdots & \cdots \\\hline
\cdots & \cdots & \cdots & \cdots & \cdots & \cdots \\
\tr(g_{m1} \omega^2) & \tr(g_{m1}) & \tr(g_{m1}\omega) & \cdots & \tr(g_{mn} \omega^2) & \tr(g_{mn}) & \tr(g_{mn}\omega)\\
\tr(g_{m1}) & \tr(g_{m1} \omega) & \tr(g_{m1}\omega^2) & \cdots & \tr(g_{mn}) & \tr(g_{mn} \omega) & \tr(g_{mn}\omega^2)
\end{pmatrix}
\quad
\begin{matrix}
x_1 \\ y_1 \\ \cdots \\ \cdots \\ x_m \\ y_m
\end{matrix},
\end{equation*}
\endgroup
\normalsize
where on the right-hand side, we associate the binary variables $(x_1,y_1,\cdots,x_m, y_m)$ to each row for computing the triorthogonal condition.

First, note that the weight of each row is even, i.e., 
\begin{equation*}
\sum_{i=1}^n \left(\tr(g_{ai} \omega^2) + \tr(g_{ai}) + \tr(g_{ai} \omega)\right) \mod{2}=\sum_{i=1}^n \tr\left(g_{ai} (\omega^2+1+\omega)\right)\mod 2=0,
\end{equation*}
where we use the fact that $\Z_2$ addition is the same as $\F_2$ addition.

Next, let us determine the double intersections. 
\begin{enumerate}[itemsep=1mm]
\item The coefficient of $x_a y_a$ modulo two is
\begin{align*}
&\sum_{i=1}^n \left(\tr(g_{ai} \omega^2)\tr(g_{ai}) + \tr(g_{ai}) \tr(g_{ai} \omega) + \tr(g_{ai} \omega) \tr( g_{ai} \omega^2) \right) 
= \sum_{i=1}^n\; \big{[}(g_{ai} \omega^2 + (g_{ai})^2 \omega)(g_{ai}+ (g_{ai})^2)\\
&+ (g_{ai} + (g_{ai})^2) (g_{ai} \omega + (g_{ai})^2 \omega^2) + 
(g_{ai} \omega + (g_{ai})^2 \omega^2) (g_{ai} \omega^2 + (g_{ai})^2 \omega)\big{]} 
= \sum_{i=1}^n (g_{ai})^3
\end{align*}
\item The coefficient of $x_a y_b$ modulo two ($a\neq b$) is $\sum_{i=1}^n \tr(g_{ai} (g_{bi})^2 \omega^2)$.
\end{enumerate}
For the triple intersections ($a\neq b$ or $a,b,c$ pairwise distinct):
\begin{enumerate}[itemsep=1mm]
\item The coefficient of $x_a y_a x_b$ modulo two is $\sum_{i=1}^n \tr(g_{ai} (g_{bi})^2 \omega^2)$.
\item The coefficient of $x_a y_a y_b$ modulo two is $\sum_{i=1}^n \tr(g_{ai} (g_{bi})^2 \omega)$.
\item The coefficient of $x_a x_b x_c$ modulo two is $\sum_{i=1}^n \tr(g_{ai} g_{bi} g_{ci})$.
\item The coefficient of $x_a x_b y_c$ modulo two is $\sum_{i=1}^n \tr(g_{ai} g_{bi} g_{ci} \omega)$.
\item The coefficient of $x_a y_b y_c$ modulo two is $\sum_{i=1}^n \tr(g_{ai} g_{bi} g_{ci} \omega^2)$.
\item The coefficient of $y_a y_b y_c$ modulo two is $\sum_{i=1}^n \tr(g_{ai} g_{bi} g_{ci})$.
\end{enumerate}
We know that for $1\le a\neq b \le m$ and $a,b$ not both in $\{1, \ldots, k\}$, we need the coefficient of $x_a y_b$, $x_a y_a x_b$, $x_a y_a y_b$ to be even. This enforces $\sum_{i=1}^n g_{ai} (g_{bi})^2 = 0$. Similarly, for $1\le a,b,c\le m$ pairwise-distinct and $a,b,c$ not all in $\{1, \ldots, k\}$, the coefficient of $x_a x_b x_c,\dots,y_a y_b y_c$ needs to be even, so $\sum_{i=1}^n g_{ai} g_{bi} g_{ci} =0$.

The necessary and sufficient condition for the non-Clifford action induced from physical transversal $\CS$ to be only on the logical space is thus also three-orthogonality, i.e., the same as for qudit CCZ.

If we further require the logical action to be disjoint $\CS$ (up to Clifford corrections), we can more compactly write the condition as
\begin{equation}
\sum_{i=1}^n g_{ai}\; g_{bi}\; g_{ci} = \begin{cases}1 & \text{if }1\le a = b = c\le k\\ 0 & \text{otherwise}\end{cases}.
\end{equation}

\subsection{Method 2: Direct calculation}
\label{sec:CS_direct_calculation}

We have established that $\CS|u\ra=i^{u^3-\tr(u)}|u\ra$, $u\in\GF(4)$ in the main text. Note that both $u^3$ (norm of $u$) and $\tr(u)$ only take on values in $\{0,1\}$. The part that gives the phase $i^{-\tr(u)}$ is Clifford, and corresponds to $\CZ_{12} S^\dagger_1 S^\dagger_2$. Since we allow Clifford correction between stabilizer and logical rows, we can ignore the $i^{-\tr(u)}$ part and only focus on applying transversal $i^{u^3}$. 

Write the $X$ logical and stabilizer rows as $\bg_a$, $a\in[m]$, and the first $k$ are logical rows. A classical codeword can be written as $\sum_{a=1}^m u_a\bg_a$ with $u_a\in\GF(4)$, and the phase it gained after applying transversal $i^{u^3}$ is $i^{p(u_1,\cdots,u_m)}$. The phase polynomial reads
\begin{equation}
\label{eq:GF(4)_phase_poly}
p(u_1,\cdots,u_m)=\sum_{i=1}^n \left(\sum_{a=1}^m u_a g_{ai}\right)^3.
\end{equation}
The inner summation, i.e., $\sum_{a=1}^m$ is over $\GF(4)$, so that anything plus itself is zero. The outer summation, i.e., $\sum_{i=1}^n$ is over $\Z$; this is well defined since each individual term $\in\{0,1\}$. To avoid confusion, we will use $+$ for $\Z$ addition, and $\hat{+}$ for $\GF(4)$ addition (or $\oplus$ if the two summands are $\{0,1\}$). Unless stated otherwise, the summation symbol $\sum$ will refer to regular addition over $\mathbb{Z}$.

Let us first prove a lemma:
\begin{lemma}
For $a_1,\dots, a_m \in \GF(4)$, then
\begin{equation*}
(a_1\hat{+}\cdots\hat{+}a_m)^3\equiv  \sum_{i=1}^m a_i^3 + 2\sum_{1\le i<j\le m} a_i^3 a_j^3 + \sum_{1\le i<j\le m}\tr(a_i^2 a_j) + 2\sum_{1\le i<j<k\le m}\tr(a_i a_j a_k)\mod{4}.
\end{equation*}
\end{lemma}
\begin{proof}
We prove this lemma by induction, and we will be using the following identity to switch between addition in $\GF(4)$ and $\Z_4$ (Eq.~\eqref{eq:xor_to_sum} modulo $4$):
\begin{equation}
\label{eq:GF(4)_Z4_summation}
\bigoplus_{i=1}^{m'} b_i \equiv \sum_{i=1}^{m'} b_i + 2\sum_{1\le i<j\le m'}b_i b_j \mod{4},
\end{equation}
where $b_1,\dots, b_{m'}\in\{0,1\}$.

We show the base case for $m=2$:
\begin{align*}
(a_1\hat{+}a_2)^3 &= a_1^3\hat{+}a_2^3\hat{+}\tr(a_1^2 a_2)\\
&\equiv a_1^3 + a_2^3 +\tr(a_1^2 a_2) + 2a_1^3 a_2^3 + 2 a_1^3 \tr(a_1^2 a_2) + 2a_2^3\tr(a_1^2 a_2) \mod{4}\\
&\equiv a_1^3+a_2^3 + 2 a_1^3 a_2^3 + \tr(a_1^2 a_2) \mod{4},
\end{align*}
where in the last step, using $a_i^4=a_i$, we have $a_1^3\tr(a_1^2 a_2)=a_1^3(a_1^2 a_2 + a_1 a_2^2)=\tr(a_1^2 a_2)$; similarly, $a_2^3 \tr(a_1^2 a_2)=\tr(a_1^2 a_2)$.

Assuming the formula is true for $m$, let us prove it for $m+1$. Call $A_m := a_1\hat{+}\cdots\hat{+}a_m$ and $A_{m+1}:= A_m \hat{+} a_{m+1}$.
\begin{equation*}
(\underbrace{a_1\hat{+}\cdots\hat{+}a_m}_{A_m} \hat{+} a_{m+1})^3 \equiv A_m^3 + a_{m+1}^3 + 2 A_m^3 a_{m+1}^3 +\tr(A_m a_{m+1}^2)\mod{4},
\end{equation*}
where we used $\tr(a^2 b)=\tr(a b^2)$. Now, what we want to show is
\begin{align*}
A_{m+1}^3 &\equiv  \sum_{i=1}^{m+1} a_i^3 + 2\sum_{1\le i<j\le m+1} a_i^3 a_j^3 + \sum_{1\le i<j\le m+1}\tr(a_i^2 a_j) + 2\sum_{1\le i<j<k\le m+1}\tr(a_i a_j a_k)\mod{4}
\end{align*}
but this is equivalent to
\begin{align*}
    A_{m+1}^3&\equiv A_m^3 + a_{m+1}^3 + 2\sum_{i=1}^m a_i^3 a_{m+1}^3 + \sum_{i=1}^m \tr(a_i^2 a_{m+1}) + 2\sum_{1\le i<j\le m}\tr(a_i a_j a_{m+1}) \mod{4}
\end{align*}
It thus suffices to show that
\begin{equation*}
2 A_m^3 a_{m+1}^3 +\tr(A_m a_{m+1}^2)  \equiv 2\sum_{i=1}^m a_i^3 a_{m+1}^3 + \sum_{i=1}^m \tr(a_i^2 a_{m+1}) + 2\sum_{1\le i<j\le m}\tr(a_i a_j a_{m+1}) \mod{4}
\end{equation*}
Indeed,
\begin{align*}
\text{LHS} &\equiv 2\left(\sum_{i=1}^m a_i^3\right) a^3_{m+1} + 2\left(\sum_{1\le i<j\le m} \tr(a_i^2 a_j) \right) a^3_{m+1} + (\tr(a_1 a_{m+1}^2)\hat{+}\cdots\hat{+}\tr(a_{m} a_{m+1}^2)) \mod{4}\\
&\equiv 2\sum_{i=1}^m a_i^3 a_{m+1}^3 + 2\sum_{1\le i<j\le m}\tr(a_i^2 a_j)a^3_{m+1} + \sum_{i=1}^m \tr(a_i a_{m+1}^2) + 2\sum_{1\le i<j\le m}\tr(a_i a_{m+1}^2)\tr(a_j a_{m+1}^2) \mod{4} 
\end{align*}
We then use
\begin{align*}
\tr(a_i a_{m+1}^2)\tr(a_j a_{m+1}^2) &=(a_i a_{m+1}^2\hat{+} a_i^2 a_{m+1})(a_j a_{m+1}^2\hat{+} a_j^2 a_{m+1})\\
&= a_ia_ja_{m+1} \hat{+} a_ia_j^2 a_{m+1}^3 \hat{+} a_i^2 a_j a_{m+1}^3 \hat{+} a_i^2 a_j^2 a_{m+1}^2\\
&= \tr(a_i a_j a_{m+1})\hat{+}\tr(a_i a_j^2)a_{m+1}^3\\
&\equiv \tr(a_i a_j a_{m+1})+\tr(a_i a_j^2) a_{m+1}^3 +2 \tr(a_ia_ja_{m+1})\tr(a_i a_j^2)a_{m+1}^3 \mod{4}
\end{align*}
Substituting this in, we may conclude:
\begin{align*}
\text{LHS} &\equiv 2\sum_{i=1}^m a_i^3 a_{m+1}^3 + 2\sum_{1\le i<j\le m}\tr(a_i^2 a_j)a^3_{m+1} + \sum_{i=1}^m \tr(a_i a_{m+1}^2)\\\nonumber
&\quad + 2\sum_{1\le i<j\le m}\tr(a_i a_j a_{m+1}) + 2\sum_{1\le i<j\le m} \tr(a_i^2 a_j)a_{m+1}^3 \mod{4} \\
&\equiv 2\sum_{i=1}^m a_i^3 a_{m+1}^3 + \sum_{i=1}^m \tr(a_i a_{m+1}^2) + 2\sum_{1\le i<j\le m}\tr(a_i a_j a_{m+1})  \mod{4},
\end{align*}
as required.
\end{proof}
With this lemma, we can continue calculating the phase polynomial:
\begin{align}
p(u_1,\cdots,u_m) &= \sum_{i=1}^n \left(\sum_{a=1}^m u_a g_{ai}\right)^3\\
&\equiv \sum_{i=1}^n \Bigg[\sum_{a=1}^m u_a^3 (g_{ai})^3 + 2\sum_{1\le a<b\le m} u_a^3 u_b^3 (g_{ai})^3 (g_{bi})^3 + \sum_{1\le a<b\le m}\tr(u_a^2 u_b (g_{ai})^2 g_{bi})\\
&\quad\quad + 2\sum_{1\le a<b<c\le m} \tr(u_a u_b u_c g_{ai} g_{bi} g_{ci})\Bigg] \mod{4}\\
&\equiv \sum_{a=1}^m u_a^3\left[\sum_{i=1}^n(g_{ai})^3\right] + 2\sum_{1\le a<b\le m} u_a^3 u_b^3 \left[\sum_{i=1}^n(g_{ai})^3 (g_{bi})^3\right]\label{eq:phase_poly_GF(4)}\\\nonumber
&\quad\quad \sum_{1\le a<b\le m}\sum_{i=1}^n\tr(u_a^2 u_b (g_{ai})^2 g_{bi}) + 2\sum_{1\le a<b<c\le m}\sum_{i=1}^n \tr(u_a u_b u_c g_{ai} g_{bi} g_{ci}) \mod{4}
\end{align}
Notice that all additions here were over the integers. Next, to derive the triple intersection conditions over $\GF(4)$, we need to turn the summation into $\GF(4)$ addition given fixed $a,b,c$. We can again use Eq.~\eqref{eq:GF(4)_Z4_summation} for this purpose.
\begin{align*}
&2\sum_{1\le a<b<c\le m}\sum_{i=1}^n \tr(u_a u_b u_c\; g_{ai} g_{bi} g_{ci})\\
\equiv\quad & 2\sum_{1\le a<b<c\le m}\left[\tr(u_a u_b u_c\; g_{a1} g_{b1} g_{c1})\hat{+}\cdots\hat{+}\tr(u_a u_b u_c\; g_{an} g_{bn} g_{cn})\right]\mod{4}\\
\equiv\quad & 2\sum_{1\le a<b<c\le m}\tr(u_a u_b u_c [g_{a1} g_{b1} g_{c1}\hat{+}\cdots\hat{+}g_{an} g_{bn} g_{cn}]) \mod{4}
\end{align*}
$i$ to the power of $2\,\tr(u_a u_b u_c\gamma)$ is non-Clifford if $\gamma\neq 0$. In that case, binarizing $u_a,u_b,u_c$ leads to $\CCZ$ gates across the three pairs of logical qubits.

The term with $u_a^3$ is a $\CS$ gate on the pair of qubits constituting the qudit, so instead of $\sum_{i=1}^n (g_{ai})^3\mod{4}$, we only need $\sum_{i=1}^n (g_{ai})^3\mod{2}=\hat\sum_{i=1}^n (g_{ai})^3$ to tell us the non-Clifford logical action.
The middle two terms in Eq.~\eqref{eq:phase_poly_GF(4)} that involves $u_a,u_b$ require more work:
\begin{equation*}
2\sum_{1\le a<b\le m} u_a^3 u_b^3 \left[\sum_{i=1}^n(g_{ai})^3 (g_{bi})^3\right] \equiv 2\sum_{1\le a<b\le m} u_a^3 u_b^3 \left[(g_{a1})^3(g_{b1})^3 \hat{+}\cdots\hat{+} (g_{an})^3(g_{bn})^3\right] \mod{4}
\end{equation*}
\begin{align*}
\sum_{1\le a<b\le m} \sum_{i=1}^n\tr(u_a^2 u_b (g_{ai})^2 g_{bi}) &\equiv \sum_{1\le a<b\le m} \Bigg[\tr(u_a^2 u_b [(g_{a1})^2 g_{b1}\hat{+}\cdots\hat{+}(g_{an})^2 g_{bn}])\\
&+ 2\sum_{1\le i<j\le n} \tr(u_a^2 u_b (g_{ai})^2 g_{bi})\tr(u_a^2 u_b (g_{aj})^2 g_{bj})\Bigg]\mod{4}\\
&\equiv \sum_{1\le a<b\le m} \Bigg[ \tr(u_a^2 u_b [(g_{a1})^2 g_{b1}\hat{+}\cdots\hat{+}(g_{an})^2 g_{bn}])\\
&+ 2 \sum_{1\le i<j\le n} \left[\tr(u_a^2 u_b g_{ai} g_{aj} (g_{bi} g_{bj})^2) + \tr(g_{bi} g_{aj} (g_{ai} g_{bj})^2) u_a^3 u_b^3\right] \Bigg]\mod{4},
\end{align*}
where we used Eq.~\eqref{eq:GF(4)_Z4_summation} in the first step, and in the second we used
\begin{align*}
&2\sum_{1\le i<j\le n} \tr(u_a^2 u_b (g_{ai})^2 g_{bi})\tr(u_a^2 u_b (g_{aj})^2 g_{bj})\\
\equiv\; & 2 \sum_{1\le i<j\le n} \left[\tr(u_a^2 u_b g_{ai} g_{aj} (g_{bi} g_{bj})^2) \hat{+} \tr(g_{bi} g_{aj} (g_{ai} g_{bj})^2) u_a^3 u_b^3\right] \mod{4}\\
\equiv\; & 2 \sum_{1\le i<j\le n} \left[\tr(u_a^2 u_b g_{ai} g_{aj} (g_{bi} g_{bj})^2) + \tr(g_{bi} g_{aj} (g_{ai} g_{bj})^2) u_a^3 u_b^3\right] \mod{4}
\end{align*}
Combining terms containing $u_a^3 u_b^3$ we get
\begin{align*}
&2\left[ u_a^3 u_b^3  \hat{\sum_{1\le i\le n}}(g_{ai})^3 (g_{bi})^3 + \hat{\sum_{1\le i<j\le n}} \tr(g_{bi} g_{aj} (g_{ai} g_{bj})^2) u_a^3 u_b^3\right]\\
=\;&2\left[ u_a^3 u_b^3  \hat{\sum_{1\le i\le n}}(g_{ai})^3 (g_{bi})^3 \hat{+} \hat{\sum_{1\le i<j\le n}} \tr(g_{bi} g_{aj} (g_{ai} g_{bj})^2) u_a^3 u_b^3\right]\\
=\;&2 u_a^3 u_b^3 \left[\hat{\sum_{1\le i\le n}}(g_{ai})^2 g_{bi}\right]^3
\end{align*}
Merging terms containing $u_a^2 u_b$ leads to
\begin{align*}
\sum_{1\le a<b\le m} \tr(u_a^2 u_b \left[\hat{\sum_{1\le i\le n}} (g_{ai})^2 g_{bi} \right]) + 2 \tr(u_a^2 u_b \left[\hat{\sum_{1\le i<j\le n}}  g_{ai} g_{aj} (g_{bi} g_{bj})^2\right]) \mod{4}
\end{align*}
The second term is Clifford, as will be seen next when we expand $\tr(u_a^2 u_b\gamma)\mod{4}$. 

Call $\gamma:=\hat\sum_{1\le i\le n} (g_{ai})^2 g_{bi}$; we show
\begin{equation}
\label{eq:ua_ub_coeff}
2u_a^3 u_b^3 \gamma^3 + \tr(u_a^2 u_b \gamma) \mod{4}
\end{equation}
is non-Clifford (in the third-level) if $\gamma \neq 0$.
Let $u_a=a_1\omega+a_2\omega^2$, $u_b=b_1\omega+b_2\omega^2$ where $a_1,a_2,b_1,b_2\in\{0,1\}$. 
\begin{align*}
2u^3_a u^3_b \gamma^3   \equiv\; & 2 (a_1+a_2-a_1a_2)(b_1+b_2-b_1b_2)\gamma^3  \mod{4}\\
\equiv\; & 2(a_1+a_2)(b_1+b_2)\gamma^3 - 2(a_1b_1b_2+a_2b_1b_2+a_1a_2b_1+a_1a_2b_2)\gamma^3 + 2a_1a_2b_1b_2\gamma^3 \mod{4}
\end{align*}
\begin{align*}
&\tr(u_a^2 u_b \gamma)\mod{4}\\
\equiv\; & (a_1\omega^2+a_2\omega)(b_1\omega+b_2\omega^2)\gamma \hat{+} (a_1\omega+a_2\omega^2)(b_1\omega^2+b_2\omega)\gamma^2 \mod{4}\\
\equiv\; & a_1b_1\tr(\gamma)\hat{+}a_1b_2\tr(\omega\gamma)\hat{+} a_2b_1\tr(\omega^2\gamma)\hat{+}a_2b_2\tr(\gamma) \mod{4}\\
\equiv\; & a_1b_1\tr(\gamma) + a_1b_2\tr(\omega\gamma) + a_2b_1\tr(\omega^2\gamma) + a_2b_2\tr(\gamma) + 2a_1a_2b_1b_2\underbrace{(\tr(\gamma)\tr(\gamma)+\tr(\omega\gamma)\tr(\omega^2\gamma))}_{\gamma^3}\\\nonumber
& + 2a_1b_1b_2\tr(\gamma)\tr(\omega\gamma) + 2a_1a_2b_1\tr(\gamma)\tr(\omega^2\gamma)+2a_1a_2b_2\tr(\omega\gamma)\tr(\gamma)+2a_2b_1b_2\tr(\omega^2\gamma)\tr(\gamma)\mod{4}
\end{align*}
Therefore, the potential $\CCZ$ part of Eq.~\eqref{eq:ua_ub_coeff} is
\begin{multline*}
2a_1b_1b_2[\tr(\gamma)\tr(\omega\gamma)-\gamma^3] + 2a_1a_2b_1[\tr(\gamma)\tr(\omega^2\gamma)-\gamma^3]+\\2a_1a_2b_2[\tr(\omega\gamma)\tr(\gamma)-\gamma^3]+2a_2b_1b_2[\tr(\omega^2\gamma)\tr(\gamma)-\gamma^3]\mod{4}
\end{multline*}
For this part to disappear, we must have $\tr(\gamma)\tr(\omega\gamma)=\gamma^3=\tr(\gamma)\tr(\omega^2\gamma)$, which holds only when $\gamma=0$.

\subsection{\texorpdfstring{Concatenated $[[4,1,2]]_4$ and decreasing monomial code}{Concatenated 4-qudit code and decreasing monomial code}}

\label{sec:decreasing_monomial_code}
We show that the $m$-fold concatenation of the $[[4,1,2]]_4$ code (Eq.~\eqref{eq:4-to-1-RS}) can be interpreted as a monomial code, where the evaluation is over the affine space $\F_4^m=\{(x_1,x_2,\dots,x_m)\;|\;x_1,\dots,x_m\in\F_4\}$. The $X$ logical of the code is (the evaluation of) $x_1x_2\dots x_m$, and the $Z$ logical is $x_1^2x_2^2\dots x_m^2$. The $X$ (resp. $Z$) stabilizers are (the evaluation of) all the monomials that are below the $X$ (resp. $Z$) logical $x_1x_2\dots x_m$ (resp. $x_1^2x_2^2\cdots x_m^2$) in the \href{https://en.wikipedia.org/wiki/Monomial_order}{lexicographic order} $x_m\succ x_{m-1}\succ\dots\succ x_1$: write $x^{\boldsymbol{\alpha}}:=x_1^{\alpha_1}\cdots x_m^{\alpha_m}$, where $\mathbf{\alpha}=(\alpha_1,\dots,\alpha_m)$ and $\alpha_i\in\Z_{\ge 0}$, then $x^{\boldsymbol{\alpha}}\succ x^{\boldsymbol{\beta}}$ if the right-most non-zero entry in the difference $\boldsymbol{\alpha}-\boldsymbol{\beta}\in \Z^m$ is positive\footnote{Here the order is defined for monomials in $\F_4[x_1,\dots,x_m]$ rather than $\F_4[x_1,\dots,x_m]/\la x_1^4-x_1,\dots,x_m^4-x_m\ra$. That is to say, the exponent vectors live in $\Z^m_{\ge 0}$. This is a well-defined monomial order, i.e., if monomials $u\prec v$ and $w$ is any other monomial, then $uw\prec vw$.}. This is to say, the concatenated $[[4,1,2]]$ code can be interpreted as a decreasing monomial code.

One can easily verify the commutation relationship by noticing that the overlap of the two evaluation vectors $\ev(x^{\boldsymbol{\alpha}})$ and $\ev(x^{\boldsymbol{\beta}})$, where $\alpha_i,\beta_i\in\{0,1,2,3\}$ is $1$ only when $\alpha_i+\beta_i\in\{3,6\}$ for all $i\in[m]$, and zero otherwise. 
This can be seen by expanding the overlap as $\sum_{P\in\F_4^m} x^{\boldsymbol{\alpha}}(P)x^{\boldsymbol{\beta}}(P)=\prod_{i=1}^m \left(\sum_{t\in\F_4} t^{\alpha_i+\beta_i}\right)$. One can also use the above overlap calculation to verify that the choice of $X$ logical and stabilizers indeed form a triorthogonal space of $k=1$.

We show the above correspondence with decreasing monomial codes by induction. The base case $[[4,1,2]]$ is easy to verify -- the $X$ logical is $\ev(x)$ and the $X$ stabilizer is $\ev(1)$. To facilitate understanding, we show the $m=2$ case by explicitly writing down the $X$-logical and $X$-stabilizers of the $[[16,1,4]]_4$ code (see below). Call the two variables $x_1,x_2$, and evaluation is over the affine space $\F_4^2$. With the concatenation approach, the logical row is $x_1x_2$. There shall be five $X$ stabilizer rows, one is $x_2$, and the other four come from the $\ev(1)=(1,1,1,1)$ stabilizer of the $[[4,1,2]]_4$ code, one can write them as $\delta_0(x), \delta_1(x), \delta_\omega(x), \delta_{\omega^2}(x)$, where the delta function $\delta_\gamma(x)=\begin{cases}1 & \text{if }x=\gamma\\ 0 & \text{if }x\neq \gamma\end{cases}$ can be expressed as a polynomial: $\delta_\gamma(x)=1-(x-\gamma)^3$, $\gamma\in\F_4$. 

\begin{table}[ht]
\small
\setlength{\tabcolsep}{0.5pt}
\begin{tabular}{c|cccccccccccccccc}
&\( (0,0) \)&\( (0,1) \)&\( (0,\omega) \)&\( (0,\omega^2) \)&\( (1,0) \)&\( (1,1) \)&\( (1,\omega) \)&\( (1,\omega^2) \)&\( (\omega,0) \)&\( (\omega,1) \)&\( (\omega,\omega) \)&\( (\omega,\omega^2) \)&\( (\omega^2,0) \)&\( (\omega^2,1) \)&\( (\omega^2,\omega) \)&\( (\omega^2,\omega^2)\)\\\hline
\(\ev(x_1x_2) \)&\( 0\)&\(0\)&\(0\)&\(0\)&\(0\)&\(1\)&\(\omega\)&\(\omega^2\)&\(0\)&\(\omega\)&\(\omega^2\)&\(1\)&\(0\)&\(\omega^2\)&\(1\)&\(\omega\)\\\hline
\(\delta_{0}(x_1) \)&\(1\)&\(1\)&\(1\)&\(1\)&\(0\)&\(0\)&\(0\)&\(0\)&\(0\)&\(0\)&\(0\)&\(0\)&\(0\)&\(0\)&\(0\)&\(0\)\\
\(\delta_{1}(x_1) \)&\(0\)&\(0\)&\(0\)&\(0\)&\(1\)&\(1\)&\(1\)&\(1\)&\(0\)&\(0\)&\(0\)&\(0\)&\(0\)&\(0\)&\(0\)&\(0\)\\
\(\delta_{\omega}(x_1) \)&\(0\)&\(0\)&\(0\)&\(0\)&\(0\)&\(0\)&\(0\)&\(0\)&\(1\)&\(1\)&\(1\)&\(1\)&\(0\)&\(0\)&\(0\)&\(0\)\\
\(\delta_{\omega^2}(x_1) \)&\(0\)&\(0\)&\(0\)&\(0\)&\(0\)&\(0\)&\(0\)&\(0\)&\(0\)&\(0\)&\(0\)&\(0\)&\(1\)&\(1\)&\(1\)&\(1\)\\
\(\ev(x_2) \)&\(0\)&\(1\)&\(\omega\)&\(\omega^2\)&\(0\)&\(1\)&\(\omega\)&\(\omega^2\)&\(0\)&\(1\)&\(\omega\)&\(\omega^2\)&\(0\)&\(1\)&\(\omega\)&\(\omega^2\)
\end{tabular}
\end{table}

One observes that $\text{span}_{\F_4}\{\delta_0(x), \delta_1(x), \delta_\omega(x), \delta_{\omega^2}(x)\}=\text{span}_{\F_4}\{1-x^3, 1-(x-1)^3, 1-(x-\omega)^3, 1-(x-\omega^2)^3\}=\text{span}_{\F_4}\{1,x,x^2,x^3\}$ from $\delta_\gamma(x)=1-(x-\gamma)^3$. Therefore, one can equivalently write the $X$ stabilizers as (the evaluation of) $1,x_1,x_1^2,x_1^3, x_2$, which are indeed all the monomials below the $X$ logical $x_1x_2$.

Now suppose we iteratively concatenate $[[4,1,2]]_4$ with itself $m$ times; the last step is to encode each of the $\GF(4)$-qudits of the $[[4^{m-1},1,2^{m-1}]]_4$ code into the $[[4,1,2]]_4$ code. Therefore, the new $X$ logical is obtained by replacing each of its entry $\gamma\in\GF(4)$ with the vector $\gamma\cdot(0,1,\omega,\omega^2)$, where $(0,1,\omega,\omega^2)=\ev(x)$ is the $X$ logical of the $[[4,1,2]]_4$ code. 
Therefore, the $X$ logical $\ev(x_1x_2\cdots x_{m-1})$ gets encoded into $\ev(x_1x_2\cdots x_m)$, now being the $X$ logical of the $[[4^{m},1,2^m]]_4$ code. 

We can similarly apply induction to the $X$ stabilizers, which we denote $S_m$ for the $[[4^{m},1,2^m]]_4$ code. We need to show that $S_m$ is formed by $\text{Fun}(x_1,\dots,x_{m-1})$ which consists of all functions $\F_4^{m-1}\to F_4$ on the variables $x_1,\dots,x_{m-1}$, and $x_m\cdot S_{m-1}$. Here, $S_{m-1}$ consists of (the evaluation of) all monomials $\prec x_1x_2\ldots x_{m-1}$, which is the set of $X$ stabilizers of the $[[4^{m-1},1,2^{m-1}]]_4$ code, by the inductive hypothesis. Note that $\text{Fun}(x_1,\dots,x_{m-1})\cup (x_m\cdot S_{m-1})$ indeed form \emph{all} the monomials $\prec x_1x_2\cdots x_{m}$. 

There are two equivalent generating sets of $\text{Fun}(x_1,\dots,x_{m-1})$: the first is by all the monomials generated by $x_1, \ldots, x_{m-1}$, the second is the delta functions $\{\delta_P(x_1,\dots,x_{m-1})\;|\;P\in\F_4^{m-1}\}$. The delta functions precisely provide the $4^{m-1}$ local stabilizers of the form $(0,\cdots,0,1,1,1,1,0\cdots,0)$ (weight four) coming from encoding each of the $\GF(4)$-qudits of the $[[4^{m-1},1,2^{m-1}]]_4$ code into the $[[4,1,2]]_4$ code. The other part of the contribution to $S_m$, namely $x_m\cdot S_{m-1}$, arises from encoding each of the entry of an $X$ stabilizer of the $[[4^{m-1},1,2^{m-1}]]_4$ using the logical row $(0,1,\omega,\omega^2)$ of $[[4,1,2]]_4$.

Although the distances of the code are natural to see from the concatenation, namely $[[4^m,1,d_X=3^m\;/\; d_Z=2^m]]$, we can alternatively prove the distances via the footprint bound \cite{footprint_bound,geil2008evaluation}, which says the following. Fix any monomial order '$\prec$' (which should be valid with respect to multiplication, i.e., if monomials $u\prec v$ and $w$ is any other monomial, then $uw\prec vw$), we can define the \emph{leading monomial} $\LM(f)$ for a given polynomial $f\in \F_q[x_1,\dots,x_m]$ (take $q=4$ for our case) as the largest monomial appearing in $f$ with respect to '$\prec$'. For an ideal $J$ of polynomials, define its \emph{footprint} $\Delta_\prec(J)$ to be the set of those monomials which are not the leading monomial of any polynomial in $J$. Equivalently, $\Delta_\prec(J)$ contains the monomials that are not contained in the ideal $\la \LM(J)\ra$ (here we use $\la \cdot\ra$ instead of $(\cdot)$ to denote the ideal generated by the polynomials surrounded by the brackets), where $\LM(J)=\{\LM(f)\;|\;f\in J\}$. The footprint bound says that if $\Delta_\prec(J)$ is finite, then $|\Delta_\prec(J)|\ge |V(J)|$, where $V(J):=\{x\in \F_q^m\;|\;f(x)=0,\;\forall f\in J\}$. 

Knowing the number of zeros of $f$ on the affine space $\F_q^m$ allows us to determine the weight of $\ev(f)$, since the weight counts the number of coordinates where the evaluation is nonzero. Call $V(f)$ the set of zeros of $f$, then $\wt(\ev(f))=q^m-|V(f)|=q^m-|V(\la f\ra)|$ since $V(f)=V(\la f\ra)$.
Moreover, let $I:=\la x_1^q-x_1,\dots,x_m^q-x_m\ra$; its zero locus\footnote{Given an ideal $I$ of a polynomial ring, its zero locus $V(I)$ are exactly the points in the affine space $\mathbb{F}_q^m$ on which every polynomial in $I$ vanishes.} $V(I)$ is the full affine space $\F_q^m$, and therefore, $|V(\la f\ra)|=|V(\la f, x_1^q-x_1,\dots,x_m^q-x_m\ra)|$.

Now we are ready to apply the footprint bound to determine the $X$ distance, which is the lowest weight among all evaluation vectors $\{\ev(x_1x_2\cdots x_m+s)\;|\;s\in S_m\}$. The leading monomial of all these polynomials $f=x_1x_2\cdots x_m+s$ is $\LM(f)=x_1x_2\cdots x_m$. We have
\begin{equation}
\label{eq:monomial_weight}
\wt(\ev(f))=4^m-|V(\la f,x_1^4-x_1,\dots,x_m^4-x_m\ra)|\ge 4^m-|\Delta_\prec(\la f,x_1^4-x_1,\dots,x_m^4-x_m\ra)|,
\end{equation}
where the footprint bound is applied in the second step.

Recall that the footprint $\Delta_\prec(\la f, x_1^4-x_1,\dots,x_m^4-x_m\ra)$ contains those monomials that are not in $$\la \LM(\la f, x_1^4-x_1,\dots,x_m^4-x_m\ra)\ra \supseteq \la \LM(f),\LM(x_1^4-x_1),\dots,\LM(x_m^4-x_m)\ra=
\la x_1x_2\cdots x_m, x_1^4,\dots,x_m^4 \ra.$$ Therefore, the footprint is contained in th set of monomials $x^{\boldsymbol{\beta}}$ with $0\le \beta_i<4$ that are not divisible by $x_1 x_2\dots x_m$, of which there are $4^m-3^m$. Eq.~\eqref{eq:monomial_weight} then implies $\wt(\ev(f))\ge 3^m$, and consequently $d_X\ge 3^m$. One can similarly show that $d_Z\ge 2^m$. Moreover, the $X$ logical (resp. $Z$ logical) $x_1x_2\cdots x_m$ (resp. $\prod_{i=1}^m (1+x_i+x_i^2)$) achieve these bounds.

\textbf{Clifford aspects of the binarized \boldmath $[[4^m,1,2^m]]_4$ code.}
The binarized code has parameters $[[2\cdot 4^m, 2, d_X=4\cdot3^{m-1}\;/\;d_Z=2^m]]_2$, which are the same parameters as two patches of surface code. We will now show how this code admits implementations of the full Clifford group in full qudit distance, which equals the qubit distance in this case. Note that we already know for this code that the logical $\CS$ between the two logical qubits may be implemented using qudit-transversal physical $\CS$, thereby retaining the qudit distance.

As for the Cliffords, we first note that individual $S$ gates can be done with full distance on the two logical qubits, using the identity $X_2 \CS_{12} X_2 \CS_{12}=S_1$.
Next, we show that all $\CNOT$-type circuits between the two logical qubits can also be obtained in a qudit-transversal way (together with qudit permutations). It is observed in~\cite{koh2026entangling} that, for an $\mathbb{F}_4$-qudit code, for any $\gamma\in \mathbb{F}_4$,  $|\gamma\ra_L\to |\alpha \gamma\ra_L$ can be achieved through scaling each physical $\mathbb{F}_4$-qudit by $\alpha\in\F_4^\times$. To complete the set of $\CNOT$-type circuits, one only needs to find an implementation of the logical Frobenius transform $|\gamma\ra_L\mapsto |\gamma^2\ra_L$. Note that in the self-dual basis that we use for binarization, $\{\omega, \omega^2\}$, the Frobenius transform on one $\mathbb{F}_4$-qudit just corresponds to swapping the two qubits. For our $[[4^m,1,2^m]]_4$ code, we show that the logical Frobenius transform can be achieved through physical Frobenius transforms (swapping the two physical qubits making up each physical qudit), and permuting the physical $\GF(4)$-qudits, as we now show through induction. In other words, the two logical qubits can be swapped by simply re-labeling the physical qubits.

For $[[4,1,2]]_4$, whose $X$-stabilizer is $(1,1,1,1)$ and $X$-logical is $\bl_X=(0,1,\omega,\omega^2)$. After individual physical Frobenius transforms, one swaps the third and fourth Galois qudit. The $X$ stabilizer is preserved, and the $X$ logical also transforms as desired: $$\gamma \bl_X=(0,\gamma,\gamma\omega,\gamma\omega^2)\mapsto (0,\gamma^2,\gamma^2\omega^2,\gamma^2\omega)\mapsto (0,\gamma^2,\gamma^2\omega,\gamma^2\omega^2)=\gamma^2\bl_X.$$

Now suppose we know how to do the logical Frobenius transform on $[[4^{m-1},1,2^{m-1}]]_4$ using physical Frobenius and permutation. Then for the $[[4^m,1,2^m]]_4$ code, think of it as encoding each physical Galois-qudit of $[[4,1,2]]_4$ using a $[[4^{m-1},1,2^{m-1}]]_4$ code. We first do the Frobenius transform on each of the four patches of $[[4^{m-1},1,2^{m-1}]]_4$ using the induction hypothesis, then swap the third and fourth patch.

With the full CNOT-type circuits and individual $S$ achieved with full qudit distance, we can now use the Hadamard teleportation gadget in \cite[Fig.~13]{koh2026entangling} to get individual Hadamard. Therefore, the full Clifford group on the two logical qubits can be done in full qudit distance, which equals the qubit distance in this case.

\subsection[Optimality of asymptotic overhead two in d=2 CS-to-CS distillation]{Optimality of asymptotic overhead two in $d=2$ $\CS$-to-$\CS$ distillation}

\label{sec:optimality_d=2_CS-to-CS}

We show that a $\F_4$-linear code satisfying the triorthogonal constraints for $\CS$-to-$\CS$ distillation (Eq.~\eqref{eq:CS-to-CS}) with $d\geq 2$ must satisfy $n\ge 2k$. Moreover, if there is an $X$ stabilizer of full support (weight $n$), we can strengthen the bound to $n\ge 2k+2$; in this case our Con.~\ref{con:2k+2-to-k-CS} is asymptotically optimal.

We only need to use the two-logical-one-stabilizer zero-overlap condition to show $n\ge 2k$. Interestingly, our proof is different from \cite{nezami2022classification} for $T$-to-$T$ distillation; in fact, our proof can serve as an alternative proof for their statement.

The convenience of the two-logical-one-stabilizer zero-overlap condition lies in the independence of the choice of basis for the logical and stabilizer space, respectively. This is because the overlap expression $T(\bg_a, \bg_b, \bg_c)=\sum_{i=1}^n{g_{ai}} g_{bi} g_{ci}$ is a trilinear form.

We prove the following lemma by induction.
\begin{lemma}\label{lem:dim_bound_triple_overlap}
Consider two $\F_q$-linear spaces $L,S\subset \F_q^n$, where $S$ is non-zero. $S$ has support at every point, that is, given a basis for $S$, there is no coordinate $i \in [n]$ at which every basis element is zero.\footnote{This property may also be referred to by saying that $S$ ``has no zero column''. In this description, we imagine a generator matrix for $S$, that is, a matrix whose rows form a basis for $S$. This property means that the generator matrix for $S$ has no zero column.} If $L$ and $S$ satisfy $T(\bl_1,\bl_2,\bs)=0$ for any $\bl_1,\bl_2\in L, \bs\in S$, then $n\ge 2\cdot \dim L$.
\end{lemma}

\begin{proof}
To facilitate understanding, we identify $L$ and $S$ with their generator matrices, that is, matrices whose rows form an $\mathbb{F}_4$-basis for their respective spaces (the choice of basis does not matter). We call $k=\rk(L)$.

We proceed via induction on $n$. For $n=1$, the only option for $S$ to be non-zero is $S=(\gamma)$ for $\gamma\in \F_4^\times$ and hence $L=(0)$, so $k=0$. For $n=2$, if $k=2$, then we can write $L=\begin{pmatrix}1&0\\0&1\end{pmatrix}$, and no nonzero solution of $S$ satisfies the overlap zero constraint.
Note that $k=1$ for $n=2$ is possible,  though: consider $L=S=(1\;1)$.

In the general case, assume the induction hypothesis holds for all $n'$ smaller than $n$, which is the number of columns of $L$ and $S$.
Take an arbitrary nonzero row $\bs$ from $S$, and let $P=\supp(\bs)$. Turn $L$ into a reduced row-echelon form on columns $P$ through row operations on $L$ (row operations on $L$ preserve $T(\bl_1,\bl_2,\bs)=0$). This partitions $L$ into row space $L_0$, which vanishes on $P$, and $L_1$, which contains the pivots.

We now consider the following lemma.
\begin{lemma}
Suppose we have $\bs\in\F_q^m$ and $s_i\neq 0,\;\forall i\in [m]$ and a matrix $A$ with entries in $\F_q$ having $m$ columns, such that for any two rows $\ba_i,\ba_j$ of $A$, we have $T(\ba_i,\ba_j,\bs)=\sum_{l=1}^m a_{il}a_{jl}s_{l}=0$. Then $\rk(A)\le \lfloor\frac{m}{2}\rfloor$.
\end{lemma}
\begin{proof}
The vanishing overlap condition can be written equivalently as $A D A^T=0$, where $D=\diag(s_1,\dots,s_m)$. Therefore, $A (AD)^T=0$, which tells us that $\text{im}((AD)^T) \subseteq \ker(A)$. This implies that $\rk(AD)=\rk((AD)^T)\le \dim(\ker(A))=m-\rk(A)$. Since $D$ is invertible, we have $\rk(AD)=\rk(A)$. It follows that $\rk(A)\le m-\rk(A)$.
\end{proof}
Let us apply this lemma with $m$ being $|P|$, $A$ being $L_1[P]$, and $\bs$ being $\bs[P]$, where appending $[P]$ refers to the restriction of a matrix/row-vector to the columns indexed by a set $P$. We get $\rk(L_1[P]) = \rk(L_1) \leq |P|/2$.

If $P=[n]$, then the result follows. Otherwise, put $\bar{P} = [n]\setminus P$, we would next like to apply the inductive hypothesis of Lemma~\ref{lem:dim_bound_triple_overlap} to $L_0[\bar{P}]$ and $S[\bar{P}]$, which we may do since they have $n-|P|<n$ columns. In order to do this, let us verify the assumptions of the inductive hypothesis. Since $S$ has no zero columns, so does $S[\bar{P}]$. We need to verify the assumptions before applying the induction hypothesis to $L_0[\bar{P}]$ and $S[\bar{P}]$ (taking the columns $\bar{P}=[n]\backslash P$ from the matrices $L_0,S$). They both have $n-|P|<n$ columns. Since $S$ contains no all-zero column, so does $S[\bar{P}]$, and also we have that $S[\bar{P}]$ is a non-zero space. Recall that $L_0$ vanishes on $P$, and thus $T(\bl_1[\bar{P}],\bl_2[\bar{P}],\bs'[\bar{P}])=T(\bl_1,\bl_2,\bs')=0$ for any $\bl_1,\bl_2\in L_0,\;\bs'\in S$. The induction hypothesis gives $\rk(L_0[\bar{P}]) = \rk(L_0)\le (n-|P|)/2$, where $\rk(L_0[\bar{P}]) = \rk(L_0)$ follows from the fact that $L_0[P]$ is the zero matrix.

Since $L_0,L_1$ are obtained via row operations from the full-rank matrix $L$, we have $\rk(L)=\rk(L_0)+\rk(L_1)\le n/2$.
\end{proof}

To apply Lemma~\ref{lem:dim_bound_triple_overlap}, we consider $L=L_X$ and $S=S_X$, the $\mathbb{F}_4$-linear spaces spanned by the rows of the logical part of the matrix (the upper $k$ rows), and the stabilizer part (the remaining rows), where all these rows together form a matrix for a code with $k = \dim L_X$ logical qudits.\footnote{It is interesting that Lemma~\ref{lem:dim_bound_triple_overlap} applies in a stronger setting, however, since it does \emph{not} assume that $L\cap S=\{0\}$, or impose any triorthogonal constraint on $S$ alone.} Note also that $L_X,S_X$ forming the $X$-logicals and $X$-stabilizers of a quantum code of distance $\ge 2$ implies that no column of $S_X$ is all-zero, that is, $S_X$ has full support.\footnote{Strictly speaking, distance $\geq 2$ is possible with $S_X$ having a zero column, as long as $L_X$ has the same column being zero. However, in this case, the all zero column may simply be removed from both spaces, which leaves all parameters unchanged, except the length of the code is lower, and we may apply our result to this new code, thus lower-bounding the length of the original code.} From here, we may conclude $n \geq 2k$.

In the present $\mathbb{F}_4$ case, we can further show that, if $\exists \,\bs\in S_X$ with full support, then $n\ge 2k+2$ (no induction is needed). Note first that we can rescale $\bs$ to all-one, because scaling a column by $\gamma\in\F_4^\times$ preserves the triorthogonality condition (since $\gamma^3=1$). Then, considering the matrix $A \coloneq \begin{pmatrix}L_X\\\mathbf{1}_n\end{pmatrix}$, the presence of the all-ones row gives us $AA^\top = 0 \Rightarrow \text{im}(A^T) \subseteq \ker(A) \Rightarrow \rk(A^T) = \rk(A) \leq n-\rk(A) \Rightarrow 2\rk(A) \leq n$. The conclusion then follows from the fact that $\rk(A) = k+1$, where $k = \rk(L_X)$.

\section{Algebraic curves: definitions and applications}\label{sec:algebraic_curves}
\subsection{Advanced preliminaries / Formal Algebraic Geometry Code definitions}
\label{sec:curves_formal_definitions}

We now make some of the notions in Sec.~\ref{sec:AG_codes} more precise and introduce basic notions such as function fields, extension and ramification, differentials and residue theorem, so that we can understand the dual codes better.
It is recommended to read Sec.~\ref{sec:intro_to_curves} and~\ref{sec:intro_to_AG_codes} before reading this section.

A particularly special type of ideal is called a maximal ideal. An ideal $I$ of $R$ is called a \emph{maximal ideal} if $I\neq R$ and the only ideals
containing $I$ are $I$ and $R$. It is possible to create a field by quotienting any ring by a maximal ideal; in fact, $R/I$ is a field if and only if $I$ is a maximal ideal of $R$. Another important type of ideal is a prime ideal. An ideal $I$ of $R$ is called a prime ideal if $I \neq R$ and whenever $ab \in I$, we have $a \in I$ or $b \in I$. A maximal ideal is a prime ideal.

Let us now explain the notion of \emph{localization} of a ring. Let $R$ be a ring. Consider a multiplicatively closed set $S \subset R$, that is, a subset $S$ for which $1 \in S$ and $a,b \in S \Rightarrow ab \in S$. One defines an equivalence relation on $R \times S$:
\begin{equation}
(a,s)\sim(a',s')\Leftrightarrow \exists u\in S \text{ s.t. }u(as'-a's)=0\;.
\end{equation}
We denote the equivalence class of a pair $(a,s)\in R\times S$ by $\frac{a}{s}$.
The set of all equivalence classes
\begin{equation}
S^{-1}R=\left\{\frac{a}{s}\;|\;a\in R,\;s\in S\right\}
\end{equation}
is called the \emph{localization} of $R$ \emph{at the multiplicatively closed set} $S$. It is a ring together with the addition and multiplication 
\begin{equation}
\frac{a}{s}+\frac{a'}{s'}:=\frac{as'+a's}{ss'}\;,\quad \frac{a}{s}\cdot \frac{a'}{s'}:=\frac{aa'}{ss'}\;.
\end{equation}
A particularly important example is localization of $R$ at a prime\footnote{Since a maximal ideal is prime, one can define localization at a maximal ideal. This plays an important role in the function fields, because a place (introduced later) is a maximal ideal of certain valuation ring. See Eq.~\eqref{eq:localize_A_at_m_P} for an example of localization at a maximal ideal.} ideal $P$ of $R$.
Let $S = R\setminus P$, one can verify that $S$ is multiplicatively closed, since $a\notin P$ and $b\notin P$ implies $ab\notin P$. One usually denotes the localization $S^{-1}R$ by $R_P$. This localization $R_P$ has exactly one maximal ideal $\{\frac{a}{s}\;|\;a\in P,\;s\notin P\}$. A ring is called \emph{local} if it has exactly one maximal ideal. Therefore, the localization $R_P$ is a local ring.\\[3pt]

Every function field $F$ over $K$ can be regarded as a finite field extension of a rational functional field $K(x)$. Throughout, think of $K$ as a finite field $\F_q$ and is therefore perfect. 
Assume that $K$ is the full constant field of $F$. For a plane model $f(x,y)=0$, one has $F=K(x,y)$ subject to this relation.

\textbf{Local parameters and valuations. }
Let $P$ be a smooth point. A \emph{local parameter} at $P$ is a function $t$ that vanishes \emph{to first order} at $P$. Equivalently, $t$ generates the maximal ideal $\mathfrak{m}_P$ of the local ring $\cO_P$. In particular, every nonzero rational function $f\in K(\chi)^\times$ can be written locally as $f=t^m u$ where $m\in Z$ and $u$ is a unit in $\cO_P$, meaning that $u(P)\neq 0$. Define $v_P(f)=m$. Depending on $v_P(f)$ greater or smaller than $0$, we say $f$ has a zero or a pole at $P$, and $|v_P(f)|$ is the multiplicity. $v_P(f)=0$ means $f(P)\neq 0$ and is finite.

The valuation satisfies
\begin{equation}
\label{eq:valuation_rules}
v_P(fg)=v_P(f)+v_P(g),\quad v_P(f+g)\ge \min\{v_P(f),v_P(g)\}.
\end{equation}
A corollary is the strict triangle inequality~\cite[Lemma~1.1.11]{Stichtenoth}:
\begin{equation}
\label{eq:strict_triangle_ineq}
v_P(f+g)=\min\{v_P(f),\;v_P(g)\}\text{ if }v_P(f)\neq v_P(g).
\end{equation}

For a smooth affine or projective plane curve, one can write down the local parameters as follows.

Given a smooth affine plane curve $f(x,y)=0$ and a point $P=(a,b)$. One of $x-a$ or $y-b$ is usually a local parameter, but not always both. Since $f_xdx+f_ydy=0$, we have $f_x(a,b)(x-a)+f_y(a,b)(y-b)+\text{higher terms}=0$.
If $f_y(a,b)\neq 0$, then locally $y$ is a function of $x$, and $t=x-a$ is a local parameter.
Similarly, if $f_x(a,b)\neq 0$ then $t=y-b$ is a local parameter. $P$ being smooth implies that at least one of $f_x(a,b), f_y(a,b)$ is nonzero.

As an example, consider the point $P=(0,0)$ on $y^2+y=x^3$. Here $f_x(0,0)=0$ and $f_y(0,0)=1$. So $x$ is a local parameter, not $y$. Indeed, near $P$, $y=x^3+$higher-order terms. 

For this curve, let us also show how to apply the valuation rules to re-derive $v_{Q_\infty}(x)=-2$ and $v_{Q_{\infty}}(y)=-3$ where $Q_\infty=(0:1:0)$ (cf. Sec~\ref{sec:intro_to_AG_codes}). Since the $y$ coordinate of $Q_\infty$ is nonzero, we work in the $Y=1$ affine chart. Define $u=X/Y$ and $v=Z/Y$, then the curve equation becomes $v+v^2=u^3$. Since $\partial (v+v^2+u^3)/\partial v=1$, $u$ is a local parameter at $Q_\infty$ and $v_{Q_\infty}(u)=1$. From $v(1+v)=u^3$, since $1+v$ is a unit, i.e., $v_{Q_\infty}(1+v)=0$, we get $v_{Q_\infty}(v)=3$. Therefore, $x=X/Z=u/v$ and $y=Y/Z=1/v$ has valuation $v_{Q_\infty}(x)=1-3=-2$ and $v_{Q_\infty}(y)=-3$, respectively.

For a smooth projective curve, suppose $P=(a:b:c)\subset \bbP^2$. One of $a,b,c$ is nonzero, suppose $c\neq 0$ and we work in the $Z=1$ affine chart: $x=X/Z$, $y=Y/Z$, $P=(a/c, b/c)$. 
Then the candidate affine local parameters are: $x-\frac{a}{c}=\frac{cX-aZ}{cZ}$ and $y-\frac{b}{c}=\frac{cY-bZ}{cZ}$. Which one to use again depends on the partial derivative.

\textbf{Places and degree. }
Let $F/K$ be the function field of a curve over the constant field $K$.

In the function-field literature, a \emph{place} of $F/K$ is often defined valuation-theoretically. 
Namely, a \emph{valuation ring} of $F/K$ is a proper subring $\cO\subset F$ such that $K\subseteq \cO$ and, for every $0\neq z\in F$, either $z\in\cO$ or $z^{-1}\in\cO$. 

For a nonsingular curve, the valuation rings that occur are discrete valuation rings. Such a ring has a unique maximal ideal and is a local ring. One often calls this maximal ideal the \emph{place}.

Equivalently, a place $P$ gives a discrete valuation $v_P:F^\times \to \Z$ such that $$\cO_P=\{f\in F\;|\; v_P(f)\ge 0\}$$
is the corresponding valuation ring, and $$P:=\mathfrak{m}_P=\{f\in F\;|\; v_P(f)>0\}$$ is its maximal ideal. 
The elements of $\cO_P$ are the functions that are regular, i.e., have no pole, at $P$, while the elements of $\mathfrak{m}_P$ are the functions vanishing at $P$.

The \emph{residue field} of $P$ is $F_P:=\cO_P/\mathfrak{m}_P$ and the \emph{degree} of the place is $\deg(P)=[F_P:K]$.

The $K$-rational points of a nonsingular curve are exactly the places of degree one. The above valuation-theoretic definition agrees with this more geometric picture as follows.

Suppose the affine curve $\chi$ is defined by $f(x,y)=0$, and let $A=K[x,y]/(f)$ be its affine coordinate ring.
If the point $(\alpha,\beta) \in \mathbb{A}^2$ is on the curve, i.e., $f(\alpha,\beta)=0, \;\alpha,\beta\in K$, then its corresponding maximal ideal of $A$ is $\mathfrak{m}_P=(x-\alpha,y-\beta)$, the ideal generated by the polynomials \(x-\alpha\) and \(y-\beta\). Indeed, $\mathfrak{m}_P$ is maximal in $A$ because $A/\mathfrak{m}_P\cong K$ is a field.
The corresponding valuation ring is the local ring (localization of $A$ at the maximal ideal $\mathfrak{m}_P$)
\begin{equation}
\label{eq:localize_A_at_m_P}
\cO_P=A_{\mathfrak{m}_P}\cong K[x,y]_{(x-\alpha,y-\beta)}/(f)=\left\{\frac{g(x,y)}{h(x,y)}\;|\;g,h\in K[x,y]/(f),\;h(\alpha,\beta)\neq0\right\}.
\end{equation}
Its maximal ideal is $\mathfrak{m}_P\cO_P=(x-\alpha,y-\beta)A_{\mathfrak{m}_P}$.
The residue map is evaluation at $P$: $\frac{g(x,y)}{h(x,y)}\mapsto \frac{g(\alpha,\beta)}{h(\alpha,\beta)}$.

For the places at infinity. Take $Q_\infty=(0:1:0)$ on $y^2+y=x^5$ as an example. To write its local ring in ideal form, use the affine chart $Y\neq 0$. Put $u=\frac{X}{Y}$, $v=\frac{Z}{Y}$. In terms of the original affine functions , $x=\frac{X}{Z}=\frac{u}{v}$, $y=\frac{Y}{Z}=\frac{1}{v}$. The projective equation becomes $v^3+v^4=u^5$.
The point $Q_\infty$ corresponds to $u=0,\;v=0$, but the curve equation is still singular at $(u,v)=(0,0)$.
For the first blow-up, use $v=uz$, thus $u^5+u^3z^3+u^4z^4=0\Rightarrow u^2+z^3+uz^4=0$. $Q_\infty$ corresponds to $(u,z)=(0,0)$, and the curve equation is still singular here.
For the second blow-up, use $u=zw$, and $z^2w^2+z^3+z^5w=0\Rightarrow w^2+z+z^3w=0$. $Q_\infty$ corresponds to $(z,w)=(0,0)$, and since $\frac{\partial (w^2+z+z^3w)}{\partial z}=1$, the curve equation is smooth here and $w$ is a local parameter. In particular, $v_{Q_\infty}(w)=1$.
We have $z=\frac{w^2}{1+wz^2}$. Since $1+wz^2$ is a unit, we have $v_{Q_\infty}(z)=2$.
Tracing back the previous changes of variables, we have $z=\frac{v}{u}=\frac{1/y}{x/y}=\frac{1}{x}$ and $w=\frac{u}{z}=\frac{x/y}{1/x}=\frac{x^2}{y}$.
The valuation is $\cO_{Q_\infty}=K[z,w]_{(z,w)}/(w^2+z+wz^3)$, and the maximal ideal is $(z,w)\cO_{Q_\infty}=(\frac{1}{x},\frac{x^2}{y})\cO_{Q_\infty}$.

Here we also give an example of places of higher degree, since they are useful in multiplication in the finite field, which is the magic we are distilling. Again, we use the Hermitian curve example, $y^2+y=x^3$ over $\F_4$. This curve has no degree-two places: all of its $\F_{4^2}$-rational points are already defined over $\F_4$. We therefore give an example of a degree-three place. Such a place can be written as an ideal $Q=(p(x),y+q(x))$, where $p(x)\in \F_4[x]$ is irreducible of degree $3$, and $q(x)\in\F_4[x]$ satisfies $q(x)^2+q(x)\equiv x^3 \mod p(x)$ since $y=q(x)$ in this ideal. Then $A/Q\cong \F_4[x]/(p(x))$ and thus $\deg(Q)=3$. An explicit example is $Q=(x^3+x^2+1,y+x^2+x)$. Geometrically, if $\alpha$ is a root of $x^3+x^2+1$ in $\F_{4^3}$, then this degree-three place represents the Frobenius orbit (the Frobenius map here is $z\mapsto z^4$) $$(\alpha,\alpha^2+\alpha),\;(\alpha^4,\alpha^8+\alpha^4),\;(\alpha^{16}, \alpha^{32}+\alpha^{16}).$$

Readers interested only in the reduction of large field multiplication into smaller ones may proceed directly to App.~\ref{sec:binary_field_multiplication}. The remainder of this subsection is needed for the metric $\bGamma$ calculation (cf. Lemma~\ref{cor:punctured_code_parameters}) in App.~\ref{sec:calculate_metric}.

\textbf{Function field extension and ramification~\cite[Sec.~3]{Stichtenoth}. } 
Every function field over $K$ (a perfect field, e.g., a finite field) can be regarded as a finite field extension $F$ of a rational function field $K(x)$.
Over the field extension $K(x,y)/K(x)$, think of the defining equation of the curve as a polynomial of $y$ with coefficients in $K(x)$. For example, for the elliptic curve $y^2+xy=x^3+x^2+x$, rewrite the equation as $f(y)=y^2+xy-(x^3+x^2+x)$. Informally, a point above (\emph{lies over}) a given $x=\alpha$ is ramified when $f(y)$ has a double root in $y$. In our example, the point $P=(0,0)$ above $x=0$ is ramified, with a \emph{ramification index} $e(P\;|\;x=0)=2$, because $f(y)=y^2=0$ has a double root. 
More generally, let $P$ be a place of $F$ above a place $Q$ of $K(x)$. The ramification index $e(P|Q)$ is defined by $v_P(h)=e(P|Q)\cdot v_Q(h)$, $h\in K(x)^\times$; $Q$ is unramified if $e(P|Q)=1$, and ramified if $e(P|Q)=1$.
In this example, the point at infinity $Q_\infty=(0:1:0)$ is also ramified and lies over $P_\infty$ (the pole of $x$ in $K(x)$). One has $v_{P_\infty}(x)=-1$, since $x$ has a simple pole at infinity. Using intersection theory, we can obtain $v_{Q_\infty}(x)=-2$, which is a pole of order $2$. Therefore, the ramification index is $e(Q_\infty\;|\;P_\infty)=2$.

\textbf{Differentials and residues. }
We only state the definitions and important theorems here. Examples of these notions appear in Sec.~\ref{sec:calculate_metric}, where we calculate the metric for elliptic curves and Klein quartic.

Here we restrict ourselves to a degree-one place $P$; this is enough for our metric calculation later. A \emph{Weil differential} is an object that can locally be written in the form 
\begin{equation}
\omega=\left(\sum_{i\ge m} a_i t^i\right)dt,
\end{equation}
where $t$ is a local parameter at $P$.

Define $v_P(\omega)=m$, and the \emph{residue} of $\omega$ at $P$ is the coefficient of $t^{-1}dt$, i.e., $\res_P(\omega)=a_{-1}$. This definition is independent of the choice of local parameter.

The \emph{residue theorem} says that for every Weil differential $\omega$, $\res_P(\omega)=0$ for almost all places $P$, and
\begin{equation}
\sum_{P}\res_P(\omega)=0.
\end{equation}

The divisor of a nonzero differential is $(\omega)=\sum_P v_P(\omega)P$.
For a divisor $A$, define
\begin{equation}
\Omega(A)=\{\omega\;|\;(\omega)\ge A\}\cup\{0\}.
\end{equation}

\textbf{Differential AG codes and duality}~\cite[Thm.~2.2.8-Cor.~2.2.11]{Stichtenoth}.
Let $D=P_1+\dots+P_n$ be a sum of distinct $K$-rational places, and let $G$ be a divisor whose support is disjoint from $\supp(D)$. The evaluation AG code is
\begin{equation}
C_{\cL}(D,G)=\{(f(P_1),\dots,f(P_n))\;|\;f\in\cL(G)\}.
\end{equation}
The differential AG code is
\begin{equation}
C_\Omega(D,G)=\{(\res_{P_1}(\omega),\dots,\res_{P_n}(\omega))\;|\;\omega\in\Omega(G-D)\}.
\end{equation}
The residue theorem implies $C_\Omega(D,G)\subset C_\cL(D,G)^\perp$ as follows. If $f\in\cL(G)$ and $\omega\in \Omega(G-D)$, recall
\begin{equation}
\cL(G)=\{f\in K(\chi)^\times\;|\;(f)+G\ge 0\}\cup\{0\},\quad \Omega(G-D)=\{\omega\;|\;(\omega)\ge G-D\}\cup\{0\},
\end{equation}
then $(f\omega)\ge -D$, i.e., $f\omega$ has possible simple poles only at the points $P_i$. 

Since $\res_{P_i}(f\omega)=f(P_i)\res_{P_i}(\omega)$, we have
\begin{equation}
\sum_{i=1}^n f(P_i)\res_{P_i}(\omega)=\sum_P \res_P(f\omega)=0.
\end{equation}
In fact, dimension counting using the Riemann-Roch theorem gives the equality
\begin{equation}
\label{eq:dual_code}
C_\Omega(D,G)=C_\cL(D,G)^\perp.
\end{equation}
Now suppose there exists a differential $\eta$ such that
\begin{equation}
v_{P_i}(\eta)=-1\quad\text{and}\quad\Gamma_i:=\res_{P_i}(\eta)\neq0,\;\forall i.
\end{equation}
Then every differential in $\Omega(G-D)$ can be written as $h\eta$ for some $h\in\cL(D-G+(\eta))$. Therefore,
\begin{equation}
C_\cL(D,G)^\perp=C_\Omega(D,G)=\bGamma\cdot C_\cL(D,D-G+(\eta)),
\end{equation}
where $\bGamma\cdot(a_1,\cdots,a_n)=(\Gamma_1 a_1,\dots,\Gamma_n a_n)$.

\input{F64_maximal_curves}

\subsection{Calculation of metric}
\label{sec:calculate_metric}

\subsubsection{Klein quartic codes}
Recall that the Klein quartic cruve $x^3y+y^3+x=0$ over $\F_8$ has two points at infinity: $Q_\infty=(0:1:0)$ and $P=(1:0:0)$. Call $R=(0:0:1)$, and recall from the calculation below Thm.~\ref{thm:stichtenoth_artin_schreier} that 
\begin{equation}
\label{eq:Klein_quartic_x_y}
(x)=3R-P-2Q_\infty\;,\quad (y)=R+2P-3Q_\infty\;.
\end{equation}

We construct the one-point codes by taking $G=rQ_\infty$ in $C_\cL(D,G)$. Let $D_0$ be the sum of the $21$ affine rational places excluding $P_{0,0}=R=(0:0:1)$. Denote these $21$ points by $P_{\alpha,\beta}$ for $\alpha\in\F_8^\times,\beta\in \F_8$ that satisfy $\alpha^3\beta+\beta^3+\alpha=0$. We explore the two possibilities that $D=D_0+R$ and $D=D_0+R+P$. 

Random puncturing (a hundred trials) gives the following distillation protocol parameters
\begin{align}
D=D_0 + R\; :&\quad 20\to 1\;(d=4),\quad 18\to 3\;(d=3),\quad 16\to 5\;(d=2)\;;\\
D=D_0 + R + P \;:&\quad 21\to 2\;(d=4),\quad 19\to 4\;(d=3),\quad 17\to6\;(d=2)\;. 
\end{align}

\textbf{Metric for \boldmath $C_\cL(D_0+R,rQ_\infty)$.}
The metric is the single codeword (up to multiplication by a scalar) of $C_\cL(D,23Q_\infty)^\perp=C_\Omega(D,23Q_\infty-D)=(\res_{P_{0,0}}(\omega),\dots)$, where $\omega$ satisfies $$(\omega)\ge 23Q_\infty-D=23Q_\infty-D_0-R.$$
Let $h(x)=\prod_{\alpha\in A}(x-\alpha)$ where $A$ is the set of $x$-coordinates that actually occur among the affine rational points. One can verify that $A=\F_8$, i.e., $\forall \alpha\in \F_8$, $\exists\beta\in\F_8$ such that $\alpha^3\beta+\beta^3+\alpha=0$ by taking $\beta=\alpha^5\eta$ for $\eta$ that satisfies $\eta^3+\eta+1=0$. Therefore, $$h(x)=x^8+x.$$ 
We seek $\omega$ in the form $\frac{x^i y^j\; dx}{h(x)}$ for integer $i,j$. Eq.~\eqref{eq:Klein_quartic_x_y} implies that 
\begin{equation}
(x^iy^j)=-(2i+3j)Q_\infty + (2j-i)P+(3i+j)R,
\end{equation}
and we still need to calculate $(h(x))$ and $(dx)$.

At each affine point $P_{\alpha,\beta}$ with $\alpha\neq 0$, the function $x-\alpha$ is a local parameter because $F_y=y^2+x^3\neq 0$ (assume $F_y(P_{\alpha,\beta})=0$, then the curve equation gives $\alpha=0$, contradiction). Writing $t=x-\alpha$, we have $$h(x)=x^8+x=(\alpha+t)^8+\alpha+t=t^8+t=t(1+t^7),$$
so $h(x)$ has a simple zero at each point of $D_0$.

At $R$, we have $v_R(x)=3$ from Eq.~\eqref{eq:Klein_quartic_x_y}, so $h(x)=x^8+x=x(1+x^7)$ and $1+x^7$ is a unit at $R$. Hence $v_R(h(x))=v_R(x)=3$.

At the two points at infinity, $x$ has a pole. Recall from Eq.~\eqref{eq:valuation_rules} that $$v(f+g)=\min\{v(f),\;v(g)\}\text{ whenever }v(f)\neq v(g).$$
Since $v_P(x)=-1$, we have $v_P(x^8+x)=v_P(x^8)=8v_P(x)=-8$. Simiarly, since $v_{Q_\infty}(x)=-2$, $v_{Q_\infty}(x^8+x)=-16$. Thus 
\begin{equation}
(h(x))=D_0+3R-8P-16Q_\infty\;.
\end{equation}

Now we compute $(dx)$. Differentiating the curve equation $F(x,y)=x^3y+y^3+x=0$ gives
\begin{equation}
\label{eq:dx_dy}
F_x\;dx+F_y\;dy=0\;\Rightarrow\; F_x\;dx=F_y\;dy,\quad F_x=x^2y+1,\quad F_y=y^2+x^3.
\end{equation}

At every affine point other than $R$, one have $F_y\neq 0$, so $x-\alpha$ is a local parameter, since $dx=d(x-\alpha)$, $dx$ has valuation $0$.

At $P=(1:0:0)$, since $v_P(x)=-1$, $\frac{1}{x}$ is a local parameter. We have $du=\frac{dx}{x^2}$ and hence $dx=x^2\;du$. Therefore $v_P(dx)=v_P(x^2)=-2$.

At $R$, since $v_R(y)=1$, $y$ is a local parameter. $du=dy=\frac{F_x}{F_y}dx=\frac{x^2y+1}{y^2+x^3}dx$.
We have $v_R(y^2+x^3)=\min\{v_R(y^2),\;v_R(x^3)\}=2$, and $v_R(x^2y+1)=\min\{2v_R(x)+v_R(y),\;0\}=0$. 
Thus $v_R(dx)=v_R(\frac{y^2+x^3}{x^2y+1})=v_R(y^2+x^3)-v_R(x^2y+1)=2$. In fact, one can also use $u=\frac{x}{y^2}$ as a local parameter at $R$, then $du=\frac{y^2\;dx-2xy\;dy}{y^4}=\frac{dx}{y^2}$. Hence $dx=y^2\;du$ and $v_R(dx)=v_R(y^2)=2$.

The subtle point is $Q_\infty=(0:1:0)$. Here $v_{Q_\infty}(x)=-2$ and $v_{Q_\infty}(y)=-3$ and so $u=\frac{x}{y}$ is a local parameter. $du=\frac{x\;dy-y\;dx}{y^2}=\frac{x^3}{y(y^2+x^3)}dx$. 
$v_{Q_\infty}(y^2+x^3)$ requires care to compute -- since $v_{Q_\infty}(y^2)=v_{Q_\infty}(x^3)=-6$, cancellation may occur. 
Dividing the affine equation $x^3y+y^3+x=0$ by $y$ yields $x^3+y^2+\frac{x}{y}=0$.
Therefore, $v_{Q_\infty}(y^2+x^3)=v_{Q_\infty}(x/y)=1$.
Consequently, $v_{Q_\infty}(dx)=v_{Q_\infty}(y)+v_{Q_\infty}(y^2+x^3)-3v_{Q_\infty}(x)=4$. Thus
\begin{equation}
\label{eq:dx}
(dx)=2R-2P+4Q_\infty\;.
\end{equation}

Hence 
\begin{equation}
\label{eq:omega_divisor}
(\omega)=\left(\frac{x^i y^j\;dx}{h(x)}\right)=(3i+j-1)R+(2j-i+6)P+(20-2i-3j)Q_\infty-D_0.
\end{equation}
Requiring $(\omega)\ge 23Q_\infty-D_0-R,$ yields a unique solution $i=1, j=-2$.

Let us now calculate the residue of $\omega=\frac{xdx}{y^2 (x^8+x)}=\frac{dx}{y^2 (x^7+1)}$. 

At $P_{\alpha,\beta}$ where $\alpha\neq 0$, use the local parameter $t=x-\alpha$, then $\omega=\frac{1}{t}\cdot \frac{x}{y^2(t^7+1)} \cdot dt$, so
\begin{equation}
\res_{P_{\alpha,\beta}}=\frac{\alpha}{\beta^2},\;\alpha\neq 0.
\end{equation}

At $R=(0:0:1)$, recall that $u=\frac{x}{y^2}$ is a local parameter and $du=\frac{dx}{y^2}$. Therefore
$$\omega=\frac{xdx}{y^2 (x^8+x)}=\frac{du}{x^7+1}.$$
Since $v_R(x)=3>0$, $x^7+1$ is a unit.
The Laurent expansion of $\omega$ in the local parameter $u$ has no $u^{-1}du$ term. Therefore
\begin{equation}
\res_R(\omega)=0.
\end{equation}

\textbf{Metric for \boldmath $C_\cL(D_0+R+P,rQ_\infty)$.}
Let us first comment on the evaluation of the basis functions
\begin{equation}
\cL(rQ_\infty)=\{x^iy^j\;|\;i,j\in\Z_{\ge 0},\;2i+3j\le r,\;i\le 2j\}
\end{equation}
at the infinity $P$. The condition $i\le 2j$ forces that $v_P(x^iy^j)=2j-i\ge 0$. For those $i,j$ such that $v_P(x^iy^j)>0$, the evaluation is zero since the order of vanishing is positive. One needs some extra computation for $i=2j$ though. For example, to evaluate $x^2y$ at $P$: notice that dividing the affine equation $x^3y+y^3+x=0$ by $x$ gives $x^2y+\frac{y^3}{x}+1=0$. $v_{P}(y^3/x)=7>0$ and thus $\ev_P(y^3/x)=0$. Therefore, 
\begin{equation}
\label{eq:eval_x^2y_at_P}
\ev_P(x^2y)=\ev_P(y^3/x+1)=\ev_P(1)=1.
\end{equation}

The metric is the single codeword of $C_\cL(D_0+R+P,25Q_\infty)^\perp$. We still seek $\omega$ in the form of $\frac{x^iy^j\;dx}{h(x)}$, and that Eq.~\eqref{eq:omega_divisor} still holds.
Imposing $(\omega)\ge 25Q_\infty-D_0-R-P$
yields a unique solution $i=1,\;j=-3$. The residue of $\omega=\frac{x\;dx}{y^3 (x^8+x)}$ at each place is
\begin{equation}
\label{eq:metric_Klein_quartic_23}
\res_{P_{\alpha,\beta}}(\omega)=\frac{\alpha}{\beta^3} \text{ for } \alpha\neq 0,\quad \res_{P}(\omega)=1,\quad \res_R(\omega)=1.
\end{equation}
We only show $\res_P(\omega)=1$ and leave the rest to the readers. We calculate in Eq.~\eqref{eq:eval_x^2y_at_P} that $\ev_P(x^2y)=1$, therefore, $\ev_P(x^6y^3)=1$. At $P$, $u=\frac{1}{x}$ is a local parameter.
$$\omega=\frac{x\;dx}{y^3(x^8+x)}=\frac{u^6}{y^3(1+u^7)}\cdot\frac{du}{u}\;,$$
and $\ev_P(1+u^7)=1+\ev_P(u^7)=1$ since $v_P(u^7)=7>0$. Also, $\frac{u^6}{y^3}=\frac{1}{x^6y^3}=1$. Hence $\res_P(\omega)=\ev_P\left(\frac{u^6}{y^3(1+u^7)}\right)=1$.

\subsubsection{Elliptic curve codes}
\label{sec:elliptic_curve}
Elliptic curves have genus $g=1$, by Riemann-Roch, the dimension of $\cL(r\Q_\infty)$ is $\ell(rQ_\infty)=r$ for $r>0$. Take $r=n-1$, then the metric is the single codeword in $C_\cL(D,(n-1)Q_\infty)^\perp=C_\Omega(D,(n-1)Q_\infty)=\{(\res_{P_1}\omega,\dots,\res_{P_n}\omega)\}$ for a differential $\omega$ whose divisor satisfies $(\omega)\ge (n-1)Q_\infty-D$. An explicit choice is $\omega=\frac{dx}{h(x)}$, where $h(x)=\prod\limits_{\alpha\in A}(x-\alpha)$, and $A$ is the set of $x$-coordinates that occur among the $K$-rational affine points of the curve.
We distinguish between two cases from Table.~\ref{Table:elliptic_curves_char2}:

\textbf{Type I, III, V.} The curve has no cross term $xy$ and is of the form $y^2+y=x^3+a_2 x^2+a_6$.
Here $(dx)=0$ and each $\alpha\in A$ gives two affine points of the curve (if $y=\beta$ is a solution, then $y=\beta+1$ is another), so $(h(x))=D-nQ_\infty$.
thus $(\omega)\ge (n-1)Q_\infty-D$.
At each place $P_i=(\alpha,\beta)\in\bbP_F$ where $\alpha\in A$, $(x-\alpha)$ is a local paramter and $\res_{P_i}\omega=1/h'(\alpha)=\prod\limits_{\alpha'\in A,\; \alpha'\neq\alpha}(\alpha-\alpha')^{-1}$.

\textbf{Type II, IV}. The curve has a cross term and is of the form $y^2+xy=x^3+a_2x^2+x$. Let $H(x)=\prod_{\alpha\in A\backslash \{0\}}(x-\alpha)$, then $h(x)=xH(x)$. One calculates that $(\frac{dx}{x})=0$ (in fact, $(dx)=(x)=2P-2Q$ where $Q=(0:1:0)$ and $P=(0:0:1)$ is the unique affine point above $x=0$. So $(\omega)=(n-1)Q_\infty - D + P\ge (n-1)Q_\infty-D$.
At each place $P_i=(\alpha,\beta)\in\bbP_F$ where $\alpha\in A\backslash\{0\}$, $\res_{P_i}\omega=1/h'(\alpha)=\prod\limits_{\alpha'\in A,\; \alpha'\neq\alpha}(\alpha-\alpha')^{-1}$.
At $x=0$ where $F_y=0$ but $F_x=1$, $y$ is a local parameter instead of $x$.
Differentiating the curve equation gives $x\;dy=(x^2+y+1)\;dx$, so $\frac{dx}{h(x)}=\frac{dx}{x}\cdot\frac{1}{H(x)}=\frac{dy}{y+1}\cdot \frac{1}{H(x)}$.
Since $H(0)=\prod_{\alpha\in A\backslash\{0\}}(-\alpha)\neq 0$ and $y+1=1$ at $P$, $\omega$ is regular at $P$ and hence $\res_P(\omega)=0$.
For the three-orthogonality purpose, we can remove $P$ from $D$ since the metric here is $0$; this reduces the code length by one.

\subsubsection[One-point codes]{One-point codes from \cite[Prop.~6.4.1]{Stichtenoth}}
\label{sec:metric_artin_schreier}

Consider a function field $F=K(x,y)$ with $y^h+\mu y=f(x)\in K[x]$ over $K=\F_{2^s}$ for $\mu\in K^\times$, where $h=p^e>1$ and $\deg f=:0>0$ is coprime to $p$, and that all roots of $y^h+\mu y=0$ are in $K$. 

\cite[Prop.~6.4.1]{Stichtenoth} shows that $P_\infty=(x)$ is the only place of $K(x)$ that ramifies in $F/K(x)$, and $Q_\infty$ is unique place of $F/K$ above $P_\infty$. Moreover, $Q_\infty|P_\infty$ is totally ramified with $e(Q_\infty|P_\infty)=h$. Since $x$ has a simple pole at $P_\infty$, we thus have $v_{Q_\infty}(x)=-h$.
Since $v_{Q_\infty}(y^h+\mu y)=v_{Q_\infty}(y^h)=h,\; v_{Q_\infty}(y)=v_{Q_\infty}(f(x))=m\cdot v_{Q_\infty}(x)$, we have $v_{Q_\infty}(y)=-m$. Finally, $(dx)=(2g-2)Q_\infty$.

Let $A\subseteq K$ be the set of $x$-coordinates that occur among the $K$-rational affine points of the curve, and $h(x)=\prod_{\alpha\in A}(x-\alpha)$. Since all roots of $y^h+\mu y=0$ lie in $K$, if there is one solution for $y$, then there are exactly $h$ solutions in $K$, and $(h(x))=D-nQ_\infty$, where $D=\sum_{\alpha}\sum_{\substack{\beta\in K\\ \beta^q+\mu\beta=f(\alpha)}}P_{\alpha,\beta}$ is the sum of all affine $K$-rational places and there are $n=h|A|$ many.
Put $\omega=\frac{dx}{h(x)}$, then its divisor $$(\omega)=(2g-2)Q_\infty - (D-nQ_\infty)=(n+2g-2)Q_\infty - D.$$
Since $F_y=\mu\neq 0$, $x-\alpha$ is a local parameter at $P_{\alpha,\beta}$, and $\res_{P_{\alpha,\beta}}(\omega)=\frac{1}{h'(\alpha)}$.

In the special case that $A=K$, one can show that the metric is all-one because $h(x)=x^{|K|}-x$ and $h'(x)=|K|x^{|K|-1}-1=1$.
This is the case for, e.g., the Hermitian curve $y^q+y=x^{q+1}$ over $K=\F_{q^2}$, since the defining equation is effectively $\tr_{\F_{q^2}/\F_q}(y)=\text{Nm}_{\F_{q^2}/\F_q}(x)$ and both sides can take all values in $\F_q$.

\subsection{Multiplication in binary extension fields}
\label{sec:binary_field_multiplication}

The general recipe to do multiplication in $\F_{2^n}$ using a curve of genus $g$ is to choose $Q$ to be a degree-$n$ place on the curve and $D$ to be a divisor of degree $n+g-1$ that has non-overlapping support with $Q$.
Here, we restrict ourselves to $D$ being a single place of degree $n+g-1$.
The place $Q$ is chosen so that the evaluation at $Q$, $\ev_Q:\cL(D)\to \cO_Q/Q\cong\F_{2^n}$ is bijective. 
Let $\{h_1,\dots,h_n\}$ be a basis of $\cL(D)$.
One can then use $\xi_i=\ev_Q(h_i)\in\F_{2^n},\; i=1,\dots,n$ as a basis of the field $\F_{2^n}$. 
Suppose we want to multiply $a,b\in \F_{2^n}$ with the basis expansion
$a=a_1\xi_1+\cdots+a_n\xi_n$ and $b=b_1\xi_1+\cdots+b_n\xi_n$.
We identify them with the functions $$f_a=a_1h_1+\cdots+a_n h_n\in\cL(D),\quad f_b=b_1h_1+\cdots+b_nh_n\in\cL(D).$$ 
The product $f_af_b\in\cL(2D)$ can be expanded as \footnote{Had we picked arbitrary functions $h_1,\dots,h_n$, their pairwise products $h_i h_j$ may span a space of dimension as large as roughly $n(n+1)/2$. Here, by choosing $h_1,\dots,h_n$ as a basis of the Riemann-Roch space $\cL(D)$, the pairwise product space is much smaller because $\cL(D)\cdot \cL(D)\subseteq \cL(2D)$.}
\begin{equation}
\label{eq:f_af_b_expansion}
f_af_b=\sum_{j=1}^{2n+g-1}c_j h_j.
\end{equation}
One can recover the product $ab\in \F_{2^n}$ by evaluating at $Q$, i.e., $ab=\ev_Q(f_a f_b)=\sum_{j=1}^{2n+g-1} c_j \ev_Q(h_j)$.
The only CCZ operations lie in determining $c_j\in\F_2$ from the unknowns $a_i,b_i\in \F_2$; all other operations, such as multiplying by constants $\ev_Q(h_j)$ and adding things up, interpolation, are Clifford only.

To determine $c_j$, we first evaluate on $N$ auxiliary places $P_1,\dots,P_N$ on the curve and then interpolate back. These places are picked so that $\deg(P_1)+\dots+\deg(P_N)=2n+g-1$, and the following map $\varphi$ is injective.
\begin{align}
\varphi:\cL(2D)&\to \prod_{i=1}^N{\F_{2}^{\deg(P_i)}}\cong \F_2^{2n+g-1},\\
f&\mapsto (\ev_{P_1}(f),\dots,\ev_{P_N}(f))^\top.
\end{align}
Then clearly 
\begin{equation}
\varphi(f_af_b)=(\ev_{P_1}(f_a)\cdot\ev_{P_1}(f_b),\dots,\ev_{P_N}(f_a)\cdot\ev_{P_N}(f_b))^\top.
\end{equation}
In other words, at individual places, say $P_i$, we are doing multiplication in $\F_{2^{\deg(P_i)}}$. 
Therefore, if $\mu_2(n)$ represents the number of CCZ gates required for multiplication in $\F_{2^n}$, then this evaluate-then-interpolate procedure requires $\mu_2(\deg(P_1))+\dots+\mu_2(\deg(P_N))$ CCZs.

Note that $\ev_{P_i}(f)\in \F_{2^{\deg(P_i)}}$, but we can binarize this into $\deg(P_i)$ bits.
With this binarization applied to every place, we can write $\varphi(f_af_b):=(m_1,\dots,m_{2n+g-1})^\top$ for $m_i\in\F_2$.

$\varphi$ is an $\F_2$-linear map because evaluating the sum of two functions at a place is equivalent to evaluating each and then adding up the results.
Therefore, given $f_af_b=\sum_{j=1}^{2n+g-1}c_j h_j$, we have
\begin{equation}
\label{eq:phi_f_af_b}
\varphi(f_af_b)=\sum_{j=1}^{2n+g-1} c_j\varphi(h_j)=G (c_1,\dots,c_{2n+g-1})^\top,
\text{ where }G:=\begin{pmatrix}\vert & \quad &\vert\\
\varphi(h_1) & \cdots & \varphi(h_{2n+g-1})\\
\vert & \quad & \vert
\end{pmatrix}.
\end{equation}
One can recover $(c_1,\dots,c_{2n+g-1})^\top$ in Eq.~\eqref{eq:f_af_b_expansion} as $G^{-1}(m_1,\dots,m_{2n+g-1})^\top$.

One needs to verify that $G$ is invertible. We will do so explicitly in the end for our following example of
$\F_{2^5}$ multiplication using the elliptic curve $y^2+y=x^3$ over $\F_2$.
This curve has three $\F_2$-rational places (degree one): 
$$P_{00}=(x,y),\quad P_{01}=(x,y+1),\quad P_\infty,$$
and three degree-two places:
$$R_1=(x+1),\quad R_2=(x^2+x+1,\;y+x),\quad R_3=(x^2+x+1,\;y+x+1).$$
We choose the degree-$5$ place $Q=(x^5+x^4+x^2+x+1,\; y+x^4+x^2+x)$.
Furthermore, we choose a degree-$(n+g-1)=5$ place $D=(f(x),\;y+q(x))$ where,
\begin{equation}
f(x)=x^5+x^3+1,\quad q(x)=x^3+x^2+x.
\end{equation}
One can verify that $f(x)$ is irreducible over $\F_2$ and indeed $y^2+y=q(x)^2+q(x)\equiv x^3\mod{f(x)}$.

Instead of the (simplest) procedure described above, where each auxiliary place is used only once, we follow \cite{CENK2010} which suggests using some places multiple times. In our example, we use the place $P_{01}$ twice and all other places only once. This gives a $14$ CCZ implementation of $\F_{2^5}$ multiplication (to be explained later), which is not optimal since $13$ is known to be possible (and is the minimal for unitary synthesis)~\cite{chudnovsky1988algebraic}. Nonetheless, we find it suitable for pedagogical purposes.

We can write down a basis of $\cL(D)$ as follows. The functions are of the form $\frac{A(x)y+B(x)}{f(x)}$. Note that we do not need to include $y^2$ terms because they can be reduced via $y^2=y+x^3$. The division by $f(x)$ is because, by definition, $\cL(D)$ contains functions with poles only on $D$. Moreover, $\cL(D)$ should not contain functions with poles on $D'=(f(x),\;y+\bar{q}(x))$ where 
\begin{equation*}
\bar{q}(x):=q(x)+1=x^3+x^2+x+1.
\end{equation*}
$D'$ is the other degree-five place, in fact, the conjugate place of $D$, because $\bar{q}(x)^2+\bar{q}(x)=q(x)^2+q(x)\equiv x^3\mod{f(x)}$.
For the function $\frac{A(x)y+B(x)}{f(x)}$ to vanish on $D'$, we thus require
$$A(x)\bar{q}(x)+B(x)\equiv 0\mod{f(x)}.$$
As for a concrete basis of $\cL(D)$, we can take $A(x)$ to be $1,x^1,x^2,x^3$ and let the corresponding $B(x)$ be $A(x)\bar{q}(x)\mod{f(x)}$. To be more explicit:
\begin{align*}
h_1=1, \quad h_2=\frac{y+x^3+x^2+x+1}{f(x)},&\quad h_3=\frac{xy+x^4+x^3+x^2+x}{f(x)},\\\quad h_4=\frac{x^2y+x^4+x^2+1}{f(x)},&\quad h_5=\frac{x^3y+x+1}{f(x)}.
\end{align*}
Next, let us find $h_6,\dots,h_{10}$, such that together with $h_1,\dots,h_5$, they form a basis of $\cL(2D)$.
Similarly, the functions take the form $\frac{A(x)y+B(x)}{f(x)^2}$, and again, they need to vanish on $D'$.
\begin{equation*}
A(x)\bar{Q}(x)+B(x)\equiv 0\mod{f(x)^2},
\end{equation*}
where $y\equiv \bar{Q}(x)\mod{f(x)^2}$ needs to satisfy
\begin{align*}
\label{eq:Q_bar_mod_f}
\bar{Q}(x)\equiv \bar{q}(x)&=y\mod{f(x)},\\
y^2+y\equiv \bar{Q}(x)^2+\bar{Q}(x)&\equiv x^3\mod{f(x)^2}.
\end{align*}
The first equality implies $\bar{Q}(x)=\bar{q}(x)+f(x)s(x)$ for some $s(x)$. Substituting this form into the second equality and solving for $s(x)$ yields $s(x)\equiv x\mod{f(x)}$. Therefore,
$$\bar{Q}(x)\equiv x^6+x^4+x^3+x^2+1\mod{f(x)^2}.$$

A convenient choice of basis is again by taking $A(x)$ to be $x^4,\dots,x^8$ and set $B(x)\equiv A(x)\bar{Q}(x)\mod{f(x)^2}$; concretely:
\begin{align*}
h_6=\frac{x^4y+x^8+x^7+x^4+1}{f(x)^2}&,\quad h_7=\frac{x^5y+x^9+x^8+x^5+x}{f(x)^2},\\
h_8=\frac{x^6y+x^9+x^2+1}{f(x)^2},\quad h_9=&\frac{x^7y+x^6+x^3+x+1}{f(x)^2},\quad h_{10}=\frac{x^8y+x^7+x^4+x^2+x}{f(x)^2}.
\end{align*}

Now, proceeding to the evaluation-then-interpolation part.
We need to determine the matrix $G$ from Eq.~\eqref{eq:phi_f_af_b}.
Since we are going to use the place $P_{01}$ twice, we need to slightly modify the $\F_2$-linear map $\varphi$.
\begin{align}\varphi:\cL(2D)\;\to &\;(\F_{2^{\deg(P_{01})}})^{\times 2}\times \F_{2^{\deg(P_{00})}}\times \F_{2^{\deg(P_\infty)}}\times \F_{2^{\deg(R_1)}}\times \F_{2^{\deg(R_2)}}\times \F_{2^{\deg(R_3)}}\nonumber\\
=&\;\F_2^2\times \F_2\times \F_2\times \F_4\times \F_4\times \F_4\cong \F_2^{10}\\
f\;\mapsto &\; (\varphi_{P_{01}}(f),\ev_{P_{00}}(f),\ev_{P_\infty}(f),\ev_{R_1}(f),\ev_{R_2}(f),\ev_{R_3}(f))^\top,
\end{align}
where $\varphi_{P_{01}}(f):=(\alpha_0,\alpha_1)$ which are the first two coefficients in expanding $f$ in a local parameter $t$, i.e., $f=\alpha_0+\alpha_1t+\alpha_2t^2+\cdots$, $\alpha_i\in \F_{2^{\deg(P_{01})}}$, at the place $P_{01}$.

A valid choice of local parameter at $P_{01}=(x,y+1)$ is $t=x$ because $F_y=1\neq 0$. Moreover, to multiply $\varphi_{P_{01}}(f_a)=(\alpha_0,\alpha_1)$ and $\varphi_{P_{01}}(f_b)=(\beta_0,\beta_1)$, locally,
\begin{equation*}
f_a=\alpha_0+\alpha_1x+O(x^2),\quad f_b=\beta_0+\beta_1x+O(x^2).
\end{equation*}
So $f_af_b=\alpha_0\beta_0+(\alpha_0\beta_1+\alpha_1\beta_0)x+O(x^2)$.
Calculating $\varphi_{P_{01}}=(\alpha_0\beta_0,\alpha_0\beta_1+\alpha_1\beta_0)$ requires three CCZs, instead of four; this is an elementary example of Karatsuba multiplication.

For the degree-two places, multiplication in $\F_4=\{0,1,\omega,\omega^2\}$ (where $\omega^2=\omega+1$) takes also three CCZs only by Karatsuba: 
\begin{align*}
(\alpha_0+\alpha_1\omega)(\beta_0+\beta_1\omega)&=(\alpha_0\beta_0+\alpha_1\beta_1)+(\alpha_1\beta_0+\alpha_0\beta_1+\alpha_1\beta_1)\omega\\
\gamma_1=\alpha_0\beta_0,\; \gamma_2&=\alpha_1\beta_1,\;\gamma_3=(\alpha_0+\alpha_1)(\beta_0+\beta_1),
\end{align*}
and $\varphi_{R_i}(f_af_b)=(\gamma_1+\gamma_2, \gamma_1+\gamma_3)$.

Multiplication in $\F_2$ takes one CCZ, so in total, $3+1+1+3\times 3=14$ CCZs are used.

Finally, one can verify that the $G$ matrix (cf. Eq.~\ref{eq:phi_f_af_b}) below is invertible.\\[3pt]

\begin{minipage}[c]{0.45\textwidth}
    \centering
    {\tiny
    \begin{tabular}{c|c|c|c|c|c|c}
        \hline
        $j$ & $\varphi_{P_{01}}(h_j)$ & $h_j(P_{00})$ & $h_j(P_\infty)$ & $h_j(R_1)$ & $h_j(R_2)$ & $h_j(R_3)$ \\ \hline
        $1$ & $(1,0)$ & 1 & 1 & 1 & 1 & 1 \\ 
        $2$ & $(0,1)$ & 1 & 0 & $\omega$ & 1 & $\omega^2$\\
        $3$ & $(0,0)$ & 0 & 0 & $\omega$ & $\omega$ & 1\\
        $4$ & $(1,0)$ & 1 & 0 & $\omega^2$ & $\omega$ & $\omega^2$\\
        $5$ & $(1,1)$ & 1 & 0 & $\omega$ & $\omega$ & 0\\
        $6$ & $(1,0)$ & 1 & 0 & $\omega$ & $\omega^2$ & $\omega$\\
        $7$ & $(0,1)$ & 0 & 0 & $\omega$ & 1 & $\omega^2$\\
        $8$ & $(1,0)$ & 1 & 0 & $\omega^2$ & $\omega^2$ & 0\\
        $9$ & $(1,1)$ & 1 & 0 & $\omega$ & 0 & 1\\
        $10$ & $(0,1)$ & 0 & 0 & $\omega$ & 0 & $\omega$\\ \hline
    \end{tabular}
    }
\end{minipage}
\begin{minipage}[c]{0.45\textwidth}
\begin{equation*}
G=\left(
\begin{smallmatrix}
1 & 0 & 0 & 1 & 1 & 1 & 0 & 1 & 1 & 0\\
0 & 1 & 0 & 0 & 1 & 0 & 1 & 0 & 1 & 1\\
1 & 1 & 0 & 1 & 1 & 1 & 0 & 1 & 1 & 0\\
1 & 0 & 0 & 0 & 0 & 0 & 0 & 0 & 0 & 0\\
1 & 0 & 0 & 1 & 0 & 0 & 0 & 1 & 0 & 0\\
0 & 1 & 1 & 1 & 1 & 1 & 1 & 1 & 1 & 1\\
1 & 1 & 0 & 0 & 0 & 1 & 1 & 1 & 0 & 0\\
0 & 0 & 1 & 1 & 1 & 1 & 0 & 1 & 0 & 0\\
1 & 1 & 1 & 1 & 0 & 0 & 1 & 0 & 1 & 0\\
0 & 1 & 0 & 1 & 0 & 1 & 1 & 0 & 0 & 1\\
\end{smallmatrix}\right)
\end{equation*}
\end{minipage}

\section{Building Distillation Circuits}\label{sec:building_distillation_circuits}

The majority of this manuscript is dedicated to leveraging the Galois qudit formalism to find novel distillation protocols. Since many of these protocols are quite small, we discuss concrete circuit realizations within the same formalism in order to aid practical implementation. In particular, we will discuss how qubit distillation protocols may be realized on fewer qubits that one naively expects using the transformation of~\cite{GOSC}, and we will show how the same transformation applies to our Galois qudit protocols.

\subsection{Distillation protocols and Pauli-based computation (PBC)}
Specializing what was discussed in Section~\ref{subsec:diagonal_gates_ch}, a third-level diagonal phase gate on $k$ qubits can be represented by a degree at most three phase polynomial modulo $8$. To be more concrete, $\CCZ_{1,2,3},\;\CS_{1,2},\;T_1$ correspond to $4x_1x_3x_3,\;2x_1x_2$, and $x_1\!\!\mod{8}$, respectively, and similarly for $\CZ,\;S,\;Z$.
Such a phase polynomial can be rewritten as a weighted sum of linear Boolean functions
\begin{equation}
\label{eq:phase_poly_sum_XOR}
P(\bx)=\sum_{\bu\in \F_2^k\backslash\{0\}} a_\bu (u_1 x_1\oplus u_2 x_3 \oplus \cdots \oplus u_k x_k), 
\end{equation}
where the coefficients $a_\bu$ are integers modulo $8$. Here, note that $\oplus$ denotes the summation of integers modulo $2$, whereas $\sum$ denotes regular integer addition unless otherwise stated.

As an example, the $\CCZ$ gate 
\begin{equation*}
4x_1x_2x_3\equiv x_1+x_2+x_3+7(x_1\oplus x_2)+7(x_1\oplus x_3)+7(x_2\oplus x_3)+(x_1\oplus x_2 \oplus x_3),
\end{equation*}
has coefficients $(a_{001},a_{010},a_{011},\cdots,a_{111})=(1,1,7,1,7,7,1)$, where $+$ denotes regular integer addition unless otherwise stated.

In this form, an odd coefficient $a_\bu$ corresponds to a $T$ gate (up to $S$ and $Z$) on the CNOT-transformed bit $(u_1 x_1\oplus u_2 x_3 \oplus \cdots \oplus u_k x_k)=\bigoplus_{i\in\supp(\bu)}x_i$ \footnote{\cite{Litinski_2019} calls this a $\pi/8$ Pauli-product rotation $\exp\left[i\pi \left(\bigoplus_{i\in\supp(\bu)}x_i)/8\right)\right]$.}, and we conclude from this identity that the $\CCZ$ gate may be unitarily synthesized from $7$ $T$ gates, and Cliffords.

In certain cases, the number of odd coefficients can be reduced by adding some phase polynomial identities, and such identities\footnote{A phase polynomial identity in the form of Eq.~\eqref{eq:phase_poly_sum_XOR} corresponds to a Reed-Muller codeword where the $\bu$ coordinate is set to $(a_\bu\!\!\mod{2})$. \cite[Fig.~2]{Litinski_2019} provides an example of phase polynomial identity.} are codewords of Reed-Muller codes $\RM(k-4,k)$ punctured at the all-zero coordinate~\cite{T_count_RM}. Minimizing the ancilla-free $T$ count for such phase polynomial gates is therefore equivalent to decoding punctured Reed-Muller codes: a notoriously hard problem computationally.

Nevertheless, given Eq.~\ref{eq:phase_poly_sum_XOR} for a third-level diagonal phase gate, \cite{campbell2017unified} observes that its non-Clifford part is equivalently represented by a so-called binary gate-synthesis matrix $A$ which has $k$ rows: the vector $\bu$ occurs as a column precisely when $a_\bu$ is odd. $A$ can be imagined as the logical part of a binary $X$-generator matrix; see Definition~\ref{def:X_gen_matrix}. Associating binary variables $x_1, \ldots, x_k$ to the rows of $A$, one may perform generalized triorthogonality calculations --- see Sec.~\ref{sec:generalized_triorthogonality} --- and go on to construct distillation protocols. In the case that the $k$-qubit third-level diagonal phase gate is $\CCZ$-only, one only needs at most two qubits to convert this gate synthesis matrix into a $d=2$ distillation protocol.
\begin{lemma}~\cite{campbell2017unified}
Let $A\in\F_2^{k\times t}$ be a full-rank gate-synthesis matrix for a $\CCZ$-only unitary $U$. Define
\begin{equation}
    \delta=\begin{cases}
        0, & t\text{ even and }\mathbf{1}_t^\top\notin\text{rowspan}(A),\\
        2, & t\text{ even and }\mathbf{1}_t^\top\in\text{rowspan}(A),\\
        1, & t\text{ odd},
    \end{cases}\quad\quad\quad\quad
    G=\begin{pmatrix}A \;\;\mathbf{0}_{k\times \delta}\\[1pt]\hline
    \rule{0pt}{13pt}
    \mathbf{1}_{t+\delta}^\top\end{pmatrix}.
\end{equation}
Then $G$ defines a $[[t+\delta,k,d\ge 2]]$ code on which transversal $T$ realizes $U$ up to Clifford corrections.
\end{lemma}
The following is an example of the $\delta=1$ case.
\begin{example}[$12T\to \TOF\#$ at $d=2$~\cite{campbell2017unified}]
Consider a $\TOF\#$ gate $\CCZ_{123}\CCZ_{145}$ whose phase polynomial is 
\begin{multline*}
4x_1x_2x_3+4x_1x_4x_5\equiv x_1+x_5+(x_2\oplus x_4) + (x_3\oplus x_4) - (x_1\oplus x_5) - (x_4\oplus x_5) - (x_1\oplus x_2\oplus x_4) \\- (x_1\oplus x_3 \oplus x_4) - (x_2\oplus x_3\oplus x_4)+(x_1\oplus x_2\oplus x_3 \oplus x_4) + (x_1\oplus x_4\oplus x_5) \pmod 8.
\end{multline*}
Notice that, in the above, the right-hand side contains only terms with odd integer coefficients, namely $1$ and $7$. This leads us to the following generalized triorthogonal matrix, but note that in the general case when passing to the matrix, one ignores the even integer coefficient terms.
\begin{equation}
G=\begin{pmatrix}
1& 0& 0& 0& 1& 0& 1& 1& 0& 1& 1& 0\\
0& 0& 1& 0& 0& 0& 1& 0& 1& 1& 0& 0\\
0& 0& 0& 1& 0& 0& 0& 1& 1& 1& 0& 0\\
0& 0& 1& 1& 0& 1& 1& 1& 1& 1& 1& 0\\
0& 1& 0& 0& 1& 1& 0& 0& 0& 0& 1& 0\\\hline
1&1&1&1&1&1&1&1&1&1&1&1
\end{pmatrix}
\quad
\begin{matrix}
x_1\\
x_2\\
x_3\\
x_4\\
x_5\\
\quad
\end{matrix}
\end{equation}
For example, the third column of the logical part $(0,1,0,1,0)^T$ corresponds to $(x_2\oplus x_4)$ in the phase polynomial above.
Sometimes, one can avoid Clifford corrections by applying $T^{\Gamma_1}\otimes\cdots\otimes T^{\Gamma_{12}}$ to the physical qubits. In this example, solving a linear equation in $\bGamma$ over $\Z_8$~\cite[App.~H.2]{koh2026entangling} shows that $\bGamma=(1,1,1,1,7,7,7,7,7,1,1,7)$ does the job. 
\end{example}

\subsubsection{Galois qudit gates as measurements}
We will now proceed to the discussion of running our distillation protocols as circuits based on Pauli-based computation (PBC). We will show how all of the most important tricks and techniques for running and compressing qubit magic state distillation protocols apply naturally to the protocols based on codes over binary extension fields.

Many quantum error-corrected architectures leverage Pauli-based computation \cite{Bravyi_2016} as a natural gate-set for logical operations 
\cite{GOSC, he2025extractorsqldpcarchitecturesefficient, yoder2025tourgrossmodularquantum, Xu2025}. Since for any $n$-qudit vector of field elements $u \in \mathbb{F}_{2^s}^n$ there is a natural definition of $X(u)$ and $Z(u)$ as $nm$-qubit Pauli matrices, generalization of this Galois qudit formalism to Pauli based computation initially looks trivial. However, especially when considering magic state distillation protocol implementation, we often care about measurements of the qubit operators corresponding to $X(\alpha u)$, for all $\alpha \in \mathbb{F}_{2^s}$, as a collection. The result of all these measurements is a map $\mathcal{M}: \mathbb{F}_{2^s} \to \mathbb{F}_2$, where $\mathcal{M}(\alpha)$ is the measurement outcome of $X(\alpha u)$. This map is clearly linear, motivating the following definition~\cite{wills2026review}.

\begin{definition} For some $u_x, u_z \in \mathbb{F}_{2^s}^n$, we call the $n$-qudit \textbf{measurement generated by } $i^{\tr(u_x \cdot u_z)}  X(u_x)Z(u_z)$ the measurement of all qubit operators corresponding to $i^{\tr(u_x \cdot u_z)}  X(\alpha u_x)Z(\alpha u_z)$ for all $\alpha \in \mathbb{F}_{2^s}$. We say the measurement \textbf{has outcome $\gamma \in \mathbb{F}_{2^s}$} if the outcome of the qubit measurement $i^{\text{tr}(u_x \cdot u_z)}  X(\alpha u_x)Z(\alpha u_z)$ is equal to $\text{tr}(\alpha \gamma)$. Given an $n$-qudit state $\ket{\psi}$, we say that $\ket{\psi}$ has eigenvalue $\gamma \in \mathbb{F}_{2^s}$ under $i^{\tr(u_x \cdot u_z)}  X(u_x)Z(u_z)$ if the outcome of the measurement generated by $i^{\tr(u_x \cdot u_z)}  X(u_x)Z(u_z)$ on $\ket{\psi}$ has the outcome $\gamma$ with certainty.
\end{definition}
While \cite{Bravyi_2016} considered only magic state preparation and Pauli measurement, future works \cite{GOSC, yoder2025tourgrossmodularquantum} have extended Pauli-based computation to also include gates of the form $\exp(i \phi Z) = \sum_{a \in \{0,1\}} \exp(i \phi (-1)^a) \ket{a}\bra{a}$, which naturally generalizes to other Pauli bases and multiple qubits. As an example, we now work towards expressing the gate $U_7^\beta = \sum_{\eta \in \mathbb{F}_{2^s}} \exp( i \pi \text{tr}(\beta \eta^7)) \ket{\eta}\bra{\eta}$ in a similar manner. The same procedure will allow us to execute Galois qudit gates in the PBC framework, which will go on to be used for running our distillation protocols based on binary extension fields in a small amount of space (on a smaller number of logical Galois qudits than one naively expects).

First, consider the following method for realizing $\ket{\psi} \to \exp\left(i \phi \frac{1+P}{2}\right)\ket{\psi}$, for some $n$-qubit Pauli $P$:
\begin{enumerate}
\item Initialize a qubit ancilla in the $\ket{+}$ state.
\item Apply the joint measurement $Z \otimes P$ to $\ket{+}\otimes \ket{\psi}$ to get outcome $\alpha \in \{0,1\}$.
\item Apply the map $\ket{a} \to e^{i \phi (a\oplus \alpha)}\ket{a}$ to the ancilla.
\item Measure the ancilla in the $X$ basis to obtain outcome $\gamma \in \{0,1\}$.
\item Apply $P$ to the $\ket{\psi}$ register when $\gamma = 1$.
\end{enumerate}

Here it was natural to consider $\exp\left(i \phi \frac{1+P}{2}\right)$ rather than $\exp\left(i \phi P\right)$ since the eigenvalues of $\frac{1+P}{2}$ live in $\mathbb{F}_2$, corresponding to the computational basis states of a qubit. This motivates the notation $\hat P = \frac{1+P}{2}$ as a Hermitian operator on $n$ qubits with eigenvalues in $\mathbb{F}_2$. We generalize to $n$-qudit operators $\hat P$ for $P = i^{\text{tr}(u_x \cdot u_z)}  X(u_x)Z(u_z)$ with eigenvalues in $\mathbb{F}_{2^s}$, by implementing an operation $ \exp( i \pi \text{tr}( \hat P^7))$:

\begin{enumerate}
\item Initialize a qudit ancilla in the $\ket{+}$ state.
\item Apply the measurement generated by $Z \otimes P$ to $\ket{+}\otimes \ket{\psi}$ to get outcome $\alpha \in F_{2^s}$.
\item Apply the map $\ket{\eta} \to \exp( i \pi \text{tr}((\eta - \alpha)^7)) \ket{\eta}$ to the ancilla.
\item Measure the ancilla in the $X$ basis to obtain outcome $\gamma \in F_{2^s}$.
\item If $P = i^{\text{tr}(u_x \cdot u_z)}  X(u_x)Z(u_z)$,  apply $i^{\text{tr}(\gamma^2 u_x \cdot u_z)}  X(\gamma u_x)Z(\gamma u_z)$ to the $\ket{\psi}$ register.
\end{enumerate}

As a specialization, we see that the gate $U_7^{\beta}$ can be written as $\exp(i \pi \text{tr}(\beta \hat Z(1)^7))$, and is generalized to a larger family of operators $U_7^{\beta}[g] := \exp(i \pi \text{tr}( \beta \hat Z(g)^7))
$ for $g \in \mathbb{F}_{2^s}^n$, which is the set of all conjugations of $U_7^{\beta}$ by multi-qudit CNOT circuits. In particular:
\begin{align}
U_7^{\beta}[g] = \sum_{u \in \mathbb{F}_{2^s}^n}  \exp\left(i \pi \text{tr}\left( \beta \left( \sum_{i=1}^n g_a u_a \right)^7\right)\right) \ket{u}\bra{u}.
\end{align}

\subsubsection{Realization of Distillation Protocols on Fewer Qudits}

Suppose a matrix $G \in \mathbb{F}_{q}^{m\times n}$ (say, a matrix satisfying some triorthogonality property) is an $X$-generator matrix for a quantum error-correcting code (see Definition~\ref{def:X_gen_matrix}) with $k$ logical qudits and a transversal $U^\beta_7$ gate (for example). Traditionally, a magic state distillation protocol based on the code takes the following form: 
\begin{enumerate}
\item Initialize $n$ qudits into a logical $\ket{+}^{\otimes k}$ state within the codespace, 
\item Apply $U^\beta_7$ on all $n$ qudits,
\item Measure the $X$ stabilizers on the code and either post-select or decode,
\item Unencode from the code to obtain the final distilled $\left( U^\beta_7 \ket{+}\right)^{\otimes k}$.
\end{enumerate}

However, for the case of $\mathbb{F}_2$, it is known that there exists a quantum circuit supported on only $m$ qubits (where $m$ is the number of rows in the $X$-generator matrix) that realizes the protocol \cite{Litinski_2019}. In this section we show that the same idea works for qudits, and given a matrix $G \in \mathbb{F}_q^{m \times n}$, the corresponding distillation protocol may be realized on $m$ Galois qudits.

\begin{proposition} Suppose $G \in \mathbb{F}_{q}^{m\times n}$ is an $X$-generator matrix for a code with $k$ logical qudits. Consider the following quantum circuit:
\begin{enumerate}
\item Initialize $m$ qudits in the state $\ket{+}^{\otimes m}$.
\item For $i = 1$ through $n$, apply faulty twirled $U_7^\beta[ g_i ]$ gates,\footnote{The notion of twirling for Galois qudits is again entirely analogous to the qubit case, see for example~\cite{bravyi2012magic}, and allows us to only have to worry about $Z$ errors on the inputted magic states, and not $X$ errors.} where $g_i$ is the $i$'th column of $G$. 
\item Measure qudits $k+1$ through $m$ in the $X$ basis.
\end{enumerate}
 Then, the measurement results correspond to the $X$ stabilizers of the CSS code, and qudits 1 through $k$ contain the final distilled $\left( U^\beta_7 \ket{+}\right)^{\otimes k}$.
\end{proposition}

This result enables a realization of magic state distillation on fewer qubits, but perhaps more importantly enables a straightforward interpretation of $G$ as a quantum circuit. The first $k$ rows denote the output qudits, and the latter $m-k$ rows denote $X$-check qudits. Since the input magic states have been twirled, we can assume that only $Z$ errors are present, so only the $X$ checks are necessary. The columns of $G$ denote multi-qudit Pauli degrees of freedom into which $U_{7}^\beta$ are injected via $U_7^{\beta}[g_i]$.

\begin{proof} Recall that the error correction code defined by $G$ is given by CSS(X, $\mathcal{G}_0$; Z, $\mathcal{G}^\perp$), where $\mathcal{G}_0$ is the $\mathbb{F}_q$-span of the latter $m-k$ rows of $G$, and $\mathcal{G}$ is the $\mathbb{F}_q$-span of all rows of $G$. Consider the code defined \emph{just} by the $Z$ checks, giving $m -k$ additional logical qudits since the $X$ stabilizers are ignored. This code is spanned by vectors
\begin{align}
\ket{\bar u} = \left| \sum_{a=1}^m u_a  g^a \right\rangle
\end{align}
for $u \in \mathbb{F}_q^{m}$, since the $Z$ checks of the form $Z(v)$ for $v \in \mathcal{G}^\perp$ satisfy $\sum_i v_i \cdot u_a g^a_i = 0$, implying $Z(v)\ket{\bar u} = \ket{\bar u}$. Now, let $\mathcal{E} :\mathbb{C}^{\mathbb{F}_q^{m}} \to \mathbb{C}^{\mathbb{F}_q^{n}}$ be an encoding isometry acting as $\mathcal{E}\ket{u} = \ket{\bar u}$. (This isometry could be realized by padding with $n-m$ many $\ket{0}$ states and then running a CNOT circuit.) 

The traditional form of the distillation circuit is then expressed as follows. Recall that, for $u' \in \mathbb{F}_q^k$, a computational basis for the full CSS(X, $\mathcal{G}_0$; Z, $\mathcal{G}^\perp$) code can be written as
\begin{align}
\ket{\bar{u'}} = \frac{1}{\sqrt{|\mathcal{G_0}|}} \sum_{g \in \mathcal{G}_0} \left| \sum_{a=1}^k u'_a g^a + g \right\rangle,
\end{align}
so the state $\ket{+}^{\otimes k}$ encoded into CSS(X, $\mathcal{G}_0$; Z, $\mathcal{G}^\perp$) can be expressed as
\begin{align}
\frac{1}{\sqrt{q^k}}  \sum_{u' \in \mathbb{F}_q^k} \ket{\bar{u'}}   &= \frac{1}{\sqrt{q^k}}  \sum_{u' \in \mathbb{F}_q^k} \frac{1}{\sqrt{|\mathcal{G_0}|}}  \sum_{g \in \mathcal{G}_0} \left| \sum_{a=1}^k u'_a g^a + g \right\rangle \\
&= \frac{1}{\sqrt{q^m}}  \sum_{u \in \mathbb{F}_q^m} \left| \sum_{a=1}^m u_a  g^a \right\rangle\\
&= \frac{1}{\sqrt{q^m}}  \sum_{u \in \mathbb{F}_q^m} \ket{\bar u}\\
&= \mathcal{E} \frac{1}{\sqrt{q^m}}  \sum_{u \in \mathbb{F}_q^m} \ket{u}\\
&= \mathcal{E} \ket{+}^{m}.
\end{align}

Let $(U_7^\beta)_{[i]}$ denote a unitary on $n$ qudits, acting as $U_7^\beta$ on qudit $i$ and as identity elsewhere. Then, we derive:
\begin{align}
\mathcal{E}^\dagger (U_7^\beta)_{[i]} \mathcal{E} \ket{u} &= \mathcal{E}^\dagger (U_7^\beta)_{[i]}  \ket{\bar u} \\
  &=  \mathcal{E}^\dagger (U_7^\beta)_{[i]} \left| \sum_{a=1}^m u_a  g^a \right\rangle \\
  &= \mathcal{E}^\dagger \exp\left( i \pi \text{tr}\left( \beta \left[ \sum_{a=1}^k u_a g^a_i  \right]^7 \right) \right) |\bar u\rangle\\
  &=  \exp\left( i \pi \text{tr}\left( \beta \left[ \sum_{a=1}^k u_a g^a_i  \right]^7 \right) \right) | u\rangle. \\
  &= U_7^{\beta}[g_i]  | u\rangle.
\end{align}
Consequently:
\begin{align}
\mathcal{E}^\dagger \bigotimes_{i=1}^n U_7^\beta \mathcal{E} &= \prod_{i=1}^n \mathcal{E}^\dagger (U_7^\beta)_{[i]} \mathcal{E}\\
&= \prod_{i=1}^n U_7^{\beta}[g_i].
\end{align}

Putting these pieces together, we see how the state at the second step of the traditional circuit is given by:
\begin{align}
 \bigotimes_{i=1}^n U_7^\beta \cdot \mathcal{E} \ket{+}^{m} = \mathcal{E} \cdot \prod_{i=1}^n U_7^{\beta}[g_i] \cdot \ket{+}^m,
\end{align}
which is of course the state $\otimes_{a=1}^k U_7^\beta \ket{+}$ encoded into CSS(X, $\mathcal{G}_0$; Z, $\mathcal{G}^\perp$).

Finally, consider an $X$ stabilizer corresponding to an element $g \in \mathcal{G}_0$. There exists a $v \in \mathbb{F}_{q}^m$ with support only on entries $v_{k+1}$ through $v_m$ such that $g = \sum_{a=1}^m v_a g^a_i$. The stabilizer $X(g)$ has the action on an encoded state:
\begin{align}
\mathcal{E}^\dagger X(g) \mathcal{E} \ket{u} &= \mathcal{E}^\dagger X(g) \left| \sum_{a=1}^m u_a  g^a \right\rangle \\
&= \mathcal{E}^\dagger \left| \sum_{a=1}^m u_a  g^a + g \right\rangle \\
&= \mathcal{E}^\dagger \left| \sum_{a=1}^m (u+ v)_a  g^a \right\rangle \\
&= \ket{u + v} \\
&= X(v) \ket{i}.
\end{align}
As a result, for $a = k+1$ through $m$, the $X$ stabilizer corresponding to row $g^a$ is equivalent to Pauli $X$ on the $a$'th qubit in the code space that is the image of $\mathcal{E}$. We can conclude that (1) for an arbitrary state $\ket{\psi}$ on $m$ qudits, measurements of qubits $a = k+1$ through $m$ in the $X$ basis are equivalent to measuring $X(g^a)$ on $\mathcal{E} \ket{\psi}$, and (2) for an $m$ qudit state $\ket{\psi}$ such that $\mathcal{E}\ket{\psi}$ is in the codespace of CSS(X, $\mathcal{G}_0$; Z, $\mathcal{G}^\perp$), qubits $k+1$ through $m$ of $\ket{\psi}$ are all in the $\ket{+}$ state, so the logical information must be on qubits 1 through $k$.
\end{proof}

\begin{remark}
    Note that, for certain special distillation factories based on Galois qudits, the protocol may be run on even fewer than $m$ qudits by using the qudit analogue of the magic state distillation compression technique from~\cite{xu2026distillingmagicstatesbicycle}. To use this compression technique, one needs the check rows of the matrix $G$ (the latter $m-k$ rows) to be writeable, after some row operations, into a form where some check rows have disjoint support. Using row operations over $\mathbb{F}_q$, rather than $\mathbb{F}_2$, the same conclusions hold for our Galois qudit protocols. For example, considering Construction~\ref{con:16to2CCZ}, with matrix
    \begin{equation*}
\left(\begin{array}{cccccccc|cccccccc}
0 & \alpha & \alpha^2 & \alpha^3 & \alpha^4 & \alpha^5 & \alpha^6 & 1 & 0&0&0&0&0&0&0 & 0\\
0 & 0&0&0&0&0&0&0 & \alpha & \alpha^2 & \alpha^3 & \alpha^4 & \alpha^5 & \alpha^6 & 1 & 0\\\hline
1 & \alpha^2 & \alpha^4 & \alpha^6 & \alpha^8 & \alpha^3 & \alpha^5 & 1 & 0&0&0&0&0&0&0 & 0\\
0 & 0&0&0&0&0&0&0 & \alpha^2 & \alpha^4 & \alpha^6 & \alpha^8 & \alpha^3 & \alpha^5 & 1 & 1\\
0 & 1 & 1 & 1 & 1 & 1 & 1 & 1 & 1 & 1 & 1 & 1 & 1 & 1 & 1 & 0\\
\end{array} \right),
\end{equation*}
after applying the above methodology, the protocol may be run on $5$ qudits. However, because the first and second check rows are disjoint, one qudit used for the checks may be re-used, and in fact the protocol may be run on only $4$ qudits.
\end{remark}

\section{Time, Space, and Error Calculations}\label{sec:time_footprint_calculations}

In this appendix, we provide detail on how we benchmark our various concatenated schemes, both their performance metrics --- time rate, spatial footprint, volumetric rate --- and their per-state average errors. The corresponding results were shown in Section~\ref{subsec:footprint_results}.

\subsection{Time and Space Calculations}\label{subsec:spacetime_calcs}

To motivate our time and space calculations, let us begin by explaining our implementation approach for a given concatenated chain of protocols. We use, as a running example, the ${(64\T\to 2\CCZ}$ \cite{Haah2018}$)\cdot (8\CCZ \to 2\CCZ)$ concatenation. First, for the overall setup, note that magic state distillation protocols like $64\T \to 2\CCZ$ that have multiple outputs usually have correlations on the errors of their outputs. Because of this, in order to ensure that the concatenation has the full distance (in this case $8$), if a protocol with multiple outputs is used in a chain, we must use multiple of any protocols following it (this does not affect situations where the given protocol is the last protocol in the chain). In the present example, we must run two of the $8\CCZ \to 2\CCZ$ protocols. In a three-level situation, where one run of the first-level protocol produces $k_A$ outputs, and one run of the second-level protocol produces $k_B$ outputs, we must run the second-level protocol $k_A$ times, and the third-level protocol $k_Ak_B$ times.

In terms of the actual setup in the quantum computer, our ${(64\T\to 2\CCZ}$ \cite{Haah2018}$)\cdot (8\CCZ \to 2\CCZ)$ has some logical qubits reserved for running one copy of the $64\T\to2\CCZ$ protocol repeatedly, and logical qubits are reserved for two copies of the $8\CCZ\to2\CCZ$. The spatial footprint of the concatenated scheme is therefore taken to be the spatial footprint of $64\T\to2\CCZ$, plus twice the spatial footprint of $8\CCZ \to 2\CCZ$. We will come on to how the spatial footprint of an individual level, like $64\T \to 2\CCZ$ is calculated, momentarily. Let us note now that this methodology is used for all concatenated schemes. That is, there is one copy of the top-level factory that is run serially, feeding possibly multiple copies of lower-level factories. In particular, top-level states are injected serially. This is one implementation choice that we fix in this work for the sake of simplicity: others are possible.

Let us now consider the spatial footprint of an individual level, like $64\T\to2\CCZ$ or $8\CCZ\to2\CCZ$. For most protocols, this is simply the number of rows in the binary matrix describing the protocol (the number of qubit $X$ logicals and $X$ stabilizer generators). For protocols like $64\T\to2\CCZ$ that are naturally defined over qubits, it is well-known that the protocol may be run on this many logical qubits~\cite{GOSC}. In Appendix~\ref{sec:building_distillation_circuits}, we showed that the same ideas naturally transfer to our case of Galois qudit codes. That is, given a matrix $G \in \mathbb{F}_{2^s}^{m \times n}$ describing the distillation protocol, the protocol may be run on a number of logical Galois qudits equaling the number of rows, $m$, in the matrix. When run at the level of qubits, this means its spatial footprint is $s\cdot m$. As an example, the $8\CCZ\to2\CCZ$ protocol is described by a matrix over $\mathbb{F}_8$ with $3$ rows --- see Construction~\ref{con:8CCZ_to_2CCZ} --- and therefore the spatial footprint of this protocol is $9$.

There is one more consideration when calculating the spatial footprint: that of compression. Certain protocols may be run on even fewer logical qubits than the number of rows in their binary matrices~\cite{xu2026distillingmagicstatesbicycle}. For example, the matrix for the Haah-Hastings $64\T\to2\CCZ$ protocol has $17$ rows, but by writing its $X$ stabilizer generators in a particular way, and ordering its columns appropriately, it becomes clear that it may be run on only $10$.\footnote{It also turns out that our novel $64\T\to2\CCZ$ $(d=4)$ protocol --- see Construction~\ref{con:64tto2ccz} --- with the smaller any-error coefficient can be compressed from $15$ to $14$ qubits. Again, it turns out that this protocol is uncompetitive relative to the corresponding Haah-Hastings protocol for the main distillation tasks that we benchmark, and so we need not consider this.} We briefly discussed in Appendix~\ref{sec:building_distillation_circuits} how the compression technique may also be applied at the level of Galois qudits to the matrices of $\mathbb{F}_q$ describing Galois qudit codes, thus allowing for some compression. Improvements in the spatial footprint appear small, however, and so we forgo this for simplicity. The only compressed footprint that will appear in our competitive protocols will be the $64\T \to 2\CCZ$ Haah-Hastings protocol mentioned. Now, when this Haah-Hastings protocol appears at the first level of a concatenated scheme, we will always used its compressed footprint, that is, we always imagine the factory being run in its compressed form. Accordingly, for example, the spatial footprint of the $(64\T\to2\CCZ)\cdot(8\CCZ\to2\CCZ)$ factory is taken to be $10 + 2\times9 = 28$. On the other hand, if the Haah-Hastings protocol appears at any lower level of a concatenated scheme like, for example, $(14\T\to2\T)\cdot(64\T\to2\CCZ)$, we may or may not consider each $64\T\to2\CCZ$ to be compressed. We will consider multiple setups where the protocol is run in its compressed, and its uncompressed form; we will return to this subtlety when we elucidate the error calculations in Section~\ref{subsec:error_calcs}. As a final note on spatial footprints, note that the (non-distillation) synthesis conversions can be run on $3$ logical qubits in the case of $3\T\to\CS$~\cite{howard2017application}, and $4$ logical qubits in the case of $2\CS\to\CCZ$~\cite{beverland2020lower_bounds}.

We now turn to the time calculations. We aim to calculate the ``time rate'' of a concatenated scheme. As mentioned in Section~\ref{subsec:footprint_results}, the time rate is defined to be the number of average retained output states per unit time for a given scheme. We also recall that the unit of time is chosen to be one $\ket{\T}$ state injection, since our setup requires the injection of top-level states one-by-one, and we ignore the time taken to perform Cliffords, since these are architecture-specific considerations, and we wish to be as architecture-agnostic as possible. We also recall that the time taken to inject one $\ket{\CS}$ state is taken to be two units of time. Ignoring post-selection, the time-rate of a given post-selection scheme is very easily calculated, since all output states are retained. For example, one run of the $(64\T\to2\CCZ)\cdot(8\CCZ\to 2\CCZ)$ concatenation injects $64 \times 8 = 512$ $\ket{\T}$ states sequentially,\footnote{Recall that, because the output states of the top-level factory have correlated errors, we must run two of the $8\CCZ\to2\CCZ$ lower-level factory. The top-level factory must therefore be run $64$ times in total to fill out those factories inputs.} as is therefore imagined as taking $512$ units of time, when ignoring post-selection. One run of this concatenated scheme produces $2\times 2 = 4$ $\ket{\CCZ}$ states and so, ignoring post-selection, the time rate of this concatenated scheme is $4/512 = 1/128\approx 0.0078$. As another example that consider $\ket{\CS}$ state injection, consider a scheme $(16\CS\to2\CS)\cdot (6\CS\to2\CS)\cdot (2\CS \to 1\CCZ)$, where we recall that $2\CS\to1\CCZ$ is simply a synthesis scheme, not distillation. In this case, to overcome the problem of correlated errors in the top factory, two copies of the second level $6\CS \to 2\CS$ distillation scheme are run in parallel. Each of the two outputs from each of the two second-level factories are fed to separate $2\CS \to 1\CCZ$ protocols, meaning that four of those bottom-level protocols are run in parallel.\footnote{Strictly speaking, because this bottom-level protocol is just synthesis, not distillation, one does not have to worry about correlated errors, and so one could run only two of these in parallel in this case. However, for simplicity in our calculations, we treat these synthesis protocols in the same way as the distillation protocols: separating the outputs of former levels into multiple copies of them.} The spatial footprint of this protocol is thus taken to be $10 + 2\times 6 + 4\times 4 = 38$. For its time rate, one run of this protocol requires the top-level $16\CS\to2\CS$ protocol to be run $12$ times serially, so that the two middle $6\CS\to2\CS$ protocols may each be run twice, ultimately outputting $4$ distilled $\ket{\CCZ}$ states. Because the injection of one $\ket{\CS}$ state is deemed to take two units of time, the time-rate of this protocol, ignoring post-selection, is taken to be $\frac{4}{2*(16*12)} = \frac{1}{96} \approx 0.0104$. More generally, the time-rate of a concatenated scheme, ignoring post-selection, may be easily calculated by multiplying the number of outputs at each level, multiplying the number of inputs at each level, dividing the former by the latter, and further dividing by two if the injection is that of $\ket{\CS}$ states.

Taking post-selection into account makes little difference to the time rates at the error rates we care about ($p = 10^{-3}$ and $10^{-6}$), for competitive protocols. Nevertheless, let us consider how the time rate may be calculated to first order in $p$. Higher order terms may be found via similar reasoning; we find via extensive testing that these have no meaningful effect for our competitive protocols, and at our values of $p$. To find the first order correction to the time rate, we note that retries of the protocol may be necessary in a concatenated chain because of error detections in any level. However, when considering the first order corrections to the time rate, we need only consider the error detections in the first distillation protocols of the chain, since lower levels detect errors with probabilities $\Omega(p^2)$. Therefore, to obtain the first-order correction to the time rate, we may multiply the time rate ignoring post-selection by the success probability of the first distillation protocol in the chain. Going back to our example $(64\T\to2\CCZ)\cdot (8\CCZ\to2\CCZ)$, the success probability of one run of the $64\T\to2\CCZ$ protocol, to first order in $p$, is $1-64p$. As such, the time rate, to first order in $p$, of this concatenated scheme is $\frac{1}{128}(1-64p) = \frac{1}{128}-\frac{1}{2}p$.

We consider two further examples now for completeness. Suppose that we are injecting $\ket{\CS}$ states, and our top-level protocol is $4\CS \to 1\CS$. We consider a uniform error model on injected $\ket{\CS}$ states, that is, the two qubits of the injected $\ket{\CS}$ states suffer $Z$ errors independently, each with probability $p$. With this, one may show that the $4\CS \to 1\CS$ protocol has success probability $1-8p + O(p^2)$.\footnote{It is important to note here how we binarize the protocol. Throughout, we binarize our $\mathbb{F}_4$ codes in the self-dual normal basis $\{\omega, \omega^2\}$ --- see Section~\ref{sec:CS_distillation_protocols} --- and our $\mathbb{F}_8$ codes in the self-dual normal basis given in Equation~\eqref{eq:F8_self_dual_normal_basis}. Given a self-dual basis $(\alpha_i)_{i=0}^{s-1}$ for $\mathbb{F}_{2^s}/\mathbb{F}_2$, the matrix over $\mathbb{F}_{2^s}$ describing the distillation protocol may be binarized by replacing each entry of the matrix, $\gamma$, with the $s \times s$ binary matrix which has entries $(\tr(\gamma\alpha_i\alpha_j)))$, where $i$ and $j$ run from $0$ to $s-1$.} The first-order time rate of any concatenated scheme with $4\CS\to1\CS$ at the top level may be found by multiplying its no-post-selection time rate by $(1-8p)$. The final example we give is the situation in which the concatenated chain begins with a synthesis protocol; consider, for example, that the chain begins with $(3\T \to \CS)\cdot (4\CS\to1\CS)$. There is no post-selection in a synthesis protocol, and finding the first-order correction to the time-rate requires us to consider the error probability of the top \textit{distillation} protocol in the chain. Given our usual independent $Z$ errors in injected $\ket{\T}$ states, with probability $p$, the error distribution of $3\T\to\CS$ turns out to be $Z_1, Z_2$, and $Z_1Z_2$, each with probability $p$, to leading order in $p$. With such $\CS$ states, the success probability of $4\CS \to 1\CS$ turns out to be $1-12p$, to leading order in $p$. Thus, the first-order time rate of any concatenated chain beginning with $(3\T\to\CS)\cdot (4\CS\to1\CS)$ may be found by multiplying the no-post-selection time rate by $1-12p$.

We conclude this subsection by noting that our protocols' volumetric rate may be found simply by dividing its time rate by its spatial footprint. Returning to our running example of $(64\T\to2\CCZ)\cdot(8\CCZ\to2\CCZ)$, its volumetric rate, to first order in $p$, is $\frac{1}{128\cdot28} - \frac{1}{2\cdot 28}p = \frac{1}{3584} - \frac{1}{56}p$.

\subsection{Error Calculations}\label{subsec:error_calcs}

We now make clear how we calculate the error rate of our (concatenated) protocols, when they are successful. Since one run of our schemes usually gives multiple outputs, the natural error rate to consider when benchmarking these schemes in the context of a wider quantum computational architecture is the \textit{per-state average error, conditioned on acceptance}. The best way to explain this initially is to consider the Haah-Hastings $64\T \to 2\CCZ$ protocol~\cite{Haah2018}. The probability of any error in the outputs (conditioned on acceptance) is $2944p^4$, to leading order in $p$. Breaking this down, it turns out that the probability that the probability that the ``first'' output state, but not the ``second'', has an error, is $576p^4$, and it happens that the probability that the second output state has an error, but not the first, is the same: $576p^4$. The probability that both have an error is $1792p^4$. We therefore find that the probability that the first state has an error is $2368p^4$, which happens to also be the probability that the second state has an error. Therefore, the per-state average error, conditioned on acceptance, is $2368p^4$, to leading order in $p$. Note that, by multiplying the per-state average error by the number of outputs, one obtains an upper bound on the probability of any error ($2944p^4$ in this case), but this bound is typically very loose, as one can see here, because the errors on outputted states are usually highly correlated.

There is an additional important point to make on this front, in the context of how we run our concatenated schemes. We adopt what one might call a \textit{global post-selection policy} for our concatenated protocols, which means the following. Suppose, as is often the case, that the first distillation protocol in the chain produces $k_A > 1$ outputs. Then, to overcome the problem of correlated errors, we run $k_A$ copies of the second-level factory in parallel. If one copy of the second-level factory has $n_B$ inputs, we run the first-level factory $n_B$ times, such that the second-level factories may be completed. After completing the second-level factories, there can be a subtlety. Suppose that some of the second-level factories accept, but some reject. In this case, our stricter \textit{global post-selection policy} says that we should throw away the outputs of \textit{all} second-level factories. This is in contrast to a weaker, ``local'' post-selection policy, which would call for one to only reject the outputs of the second-level factories that rejected. Our choice of global post-selection has a significant effect on the error rates of accepted states (by multiple orders of magnitude in some cases).

Let us now gain some intuition for why global post-selection improves error rates over local post-selection with an example, namely, the $(64\T\to2\CCZ)\cdot(8\CCZ\to2\CCZ)$ concatenation, where two of the second-level factory are run in parallel. Recalling the matrix for the $8\CCZ \to 2\CCZ$ protocol, see Construction~\ref{con:8CCZ_to_2CCZ}, we see that its only check, at the level of Galois qudits, is an all-ones row. This means that it is a distance $2$ protocol, and at leading order, one run of the $8\CCZ\to2\CCZ$ protocol produces errorful outputs when $2$ of its $8$ inputs have an error, \textit{and they are the same error}. For example, suppose that $2$ of the $8$ inputs to one run of this protocol have errors $e_1, e_2 \in \mathbb{F}_2^3$: $e_1, e_2 \neq 0$. There will be an error on the outputs exactly when $e_1 \neq e_2$. Thus, when running the $(64\T\to2\CCZ)\cdot(8\CCZ\to2\CCZ)$ concatenated scheme, with our global post-selection policy, the leading-order way that this can produce accepted outputs with errors is that some $2$ out of the $8$ runs of the top-level $64\T\to2\CCZ$ protocol produce outputs with errors $(e_1\mid f_1) \in \mathbb{F}_2^3 \oplus \mathbb{F}_2^3$ and $(e_2 \mid f_2) \in \mathbb{F}_2^3 \oplus \mathbb{F}_2^3$, where $(e_1\mid f_1) \neq (0 \mid 0)$, and $(e_2 \mid f_2) \neq (0 \mid 0)$, and moreover $e_1 = e_2$, and $f_1 = f_2$. On the other hand, with local post-selection, accepting errorful outputs states only requires $e_1 = e_2$, or separately, $f_1 = f_2$. The result of this distinction is apparent in the combinatorics. Fixing $2$ out of the $8$ runs of the $64\T\to2\CCZ$ factory to fail, considering the $2944^2 = 8,667,136$ ways they can fail, one finds that $2,056,192$ of them result in accepted errorful outputs under local post-selection, where errors need only agree on one second-level factory, whereas only $385,024$ of them result in accepted errorful outputs under global post-selection, where errors must agree on both second-level factories.\footnote{A few caveats are in order for the exposition of this paragraph. The first is that we are only considering leading-order behaviour. The same intuitions hold for higher-order behaviour, but the calculations become difficult to neatly expose; we also find throughout the analysis of our competitive protocols that higher-order terms provide minimal corrections to the leading-order behaviour. The second caveat is that what we have considered in this last calculation is the any-error probability (because it is easier to explain), whereas what we actually want is the per-state average error probability, although the intuition for the latter is very much the same. The final caveat is that, very shortly, we will discuss how we will be performing certain state twirling/averaging in order to make our largest calculations tractable, and the resulting performance as strong as possible; what we have shown above does not consider this.}

The above simple example shows a case where global post-selection provides a moderate improvement to the error rate over local post-selection, although in other cases the improvement is much more substantial. We therefore deem this choice to be worthwhile, despite the drawback that making the calculations under the global post-selection policy will be much harder than under the local post-selection policy. Indeed, with the definitions given, it is possible in-principle to directly calculate the error rate of any of our concatenated schemes by, for example, forming one large binary matrix describing the protocol, and enumerating error paths. However, for our largest protocols, this becomes computationally intractable, especially if we want to obtain higher-order terms.

To make our error rate computations tractable, and also to obtain the strongest possible performance, we invoke several twirling/averaging procedures, to simplify the problem. In the context of magic state distillation, the most well-known form of twirling is the Clifford-twirling of diagonal magic states~\cite{MSD, bravyi2012magic}, so that we may treat our magic states as only having $Z$ errors, without $X$ errors. We only work with diagonal magic states in this paper, and we assume throughout that they are twirled in this way. There are further, less standard, twirling/averaging techniques employed, to make the calculations more tractable, and to ensure the strongest possible performance, as follows:
\begin{enumerate}
    \item\label{randomisation:perm_twirl} For every (noisy) multi-qubit magic state (like $\ket{\CS}$ and $\ket{\CCZ}$), at every stage, we randomly permute the qubits of the multi-qubit magic state ($2$ or $3$ in the case of $\ket{\CS}$ and $\ket{CCZ}$, respectively). This has the effect of uniformizing the probabilities of errors over their weights. For example, for a (noisy) $\ket{\CCZ}$ state, we do not need to keep track of the probabilities of $7$ $Z$ errors, but only the probabilities of $3$ different error weights. We may call this randomization a \textit{permutation twirl}.
    \item\label{randomisation:routing} Between different layers in a concatenation, we randomly route different magic states into different copies of the next-level factory. For example, in the $(20\T\to3\CCZ$~\cite{campbell2017unified}$)\cdot(8\CCZ\to2\CCZ)$ concatenation, since the first level has three outputs, we run three copies of the second level $8\CCZ\to2\CCZ$ factory. For each of the $8$ successful runs of the $20\T\to3\CCZ$, the $3$ outputs are randomly routed to the three copies of the second-level factory. This means we only have to keep track of per-state average error probabilities between concatenation levels. In the case of multi-qubit magic states being passed between levels, like $\ket{\CS}$ and $\ket{\CCZ}$, we must only keep track of per-state average error probabilities of each weight.
    \item\label{randomisation:column} For every copy of every distillation step in a concatenated scheme, that is not at the top level, we randomly permute the columns of the matrix describing that copy of the protocol. For example, in the case of a concatenation of $(14\T\to2\T)\cdot(64\T\to2\CCZ)$, there are two copies of the second-level factory, and we consider randomly permuting the columns of the two $64\T\to2\CCZ$ factories independently, averaging the resulting output error rates over this random choice. Note that, when such a protocol arises from an $\mathbb{F}_4$ or $\mathbb{F}_8$-code (meaning it is used for $\ket{\CS}$ or $\ket{\CCZ}$ distillation, respectively), we randomly permute the columns of the matrix over $\mathbb{F}_{2^s}$, before binarizing the matrix (rather than binarizing and then randomizing), since each column of the $\mathbb{F}_{2^s}$-matrix corresponds to one input multi-qubit magic state consumption, and it would not make sense to split the consumption of individual input states.
\end{enumerate}
One may wonder what it actually means, in terms of actually running a concatenated protocol, to make such randomizations. Formally, what we are doing here, is considering the above randomizations, and we are going to calculate the (per-state average) output error rate \textit{in expectation over the above random choices}. Having calculated a value (or at least an upper bound) on this expected value, it necessarily means that there is \textit{some} choice of these randomizations that leads to the (per-state average) output error rate being at most this value.\footnote{This is not a deep statement; it is simply the statement that if $X$ is a random variable, and $f$ is a function, then $\mathbb{E}_X\left[f(X)\right] \leq V \implies$ there is \textit{some} choice for $X$ giving $f(X)$ at most $V$; this is exactly the same idea used in~\cite{wills2026concatenating} to make the calculation of the error rate of a concatenated scheme more tractable.}

Of the above three randomizations, it will likely be clear that the first two make the error calculations more tractable, with the intuition being that they allow us to keep track of less information. The third randomization, point~\ref{randomisation:column}, does not always make the calculation more tractable, but can dramatically improve performance, as we describe now with an example.

\begin{example}\label{example:14T2T*64T2CCZ}
We consider the leading order term of the per-state average error probability of the $(14\T\to2\T$~\cite{bravyi2012magic}$)\cdot(64\T\to2\CCZ$~\cite{Haah2018}$)$ concatenation. First, it turns out that, at leading order, all accepted logical errors on the $14\T\to2\T$ protocol affect both outputs. Indeed, with probability $7p^2 + O(p^3)$, one run of the top-level protocol fails, and it leaves $Z$ errors on both output $\ket{\T}$ states. Because of this, and because $\ket{\T}$ states have only one qubit, the first two randomizations above have no effect. First, let us suppose that we did \textit{not} apply randomization~\ref{randomisation:column} above. We imagine that the two copies of the Haah-Hastings protocol are run with a common ordering of their columns. Any one copy of $64\T\to2\CCZ$ fails at fourth order. It turns out that out of $\left(\begin{smallmatrix}
    64\\4
\end{smallmatrix}\right) = 635376$ possible size-$4$ errors, $3248$ cause the protocol to accept. Of those, $304$ cause no logical error on the two output $\ket{\CCZ}$ states, $1152$ cause exactly one of them to have an error, and $1792$ cause both of them to be errorful. Thus, considering the full concatenate scheme, two of the four outputs have an error with probability $1152\cdot(7p^2)^4 = 2765952p^8$, whereas all four outputs have an error with probability $1792\cdot(7p^2)^4 = 4302592p^8$, at leading order. The leading order per-state average output of the whole concatenated scheme is therefore
\begin{equation}\label{eq:error_rate_worse}
    \frac{2765952}{2}p^8 + 4302592p^8 = 5685568p^8.
\end{equation}
Let us now instead suppose that the randomization of point~\ref{randomisation:column} above has been applied. Some $4$ out of the $64$ runs of the top-level factory fail with probability $M\cdot (7p^2)^4$, where $M \coloneq \left(\begin{smallmatrix}64\\4\end{smallmatrix}\right) = 635376$. For a leading-order failure, each copy of the Haah-Hastings protocol sees error sets of size $4$, independently distributed randomly amongst their inputs of size $4$. The number of such sets that cause one of the four outputs to be errorful is $2\times 304\times1152 = 700416$, the number that cause two of the four outputs to be errorful is $1152^2 + 2\times 304\times1792 = 2416640$, the number that cause three to be errorful is $2\times 1152\times 1792 = 4128768$, and $1792^2 = 3211264$ cause all four to be errorful. Given that four of the first-level factories fail, the probability of exactly one, two, three, and four bad outputs are thus $\frac{700416}{M^2}$, $\frac{2416640}{M^2}$, $\frac{4128768}{M^2}$, and $\frac{3211264}{M^2}$, respectively. The leading-order per-state average error probability is thus
\begin{equation}\label{eq:error_rate_better}
    M\cdot(7p^2)^4\cdot\left[\frac{1}{4}\frac{700416}{M^2} + \frac{2}{4}\frac{2416640}{M^2} + \frac{3}{4}\frac{4128768}{M^2} + \frac{4}{4}\frac{3211264}{M^2}\right] \approx 29064.25p^8.
\end{equation}
We see that by randomizing the columns of the second-level factories, we have improved the error rate by almost a factor of $200$.
\end{example}
There is a further consideration to be aware of, considering Example~\ref{example:14T2T*64T2CCZ} above. We recall from Section~\ref{subsec:spacetime_calcs} above that the $64\T\to2\CCZ$ protocol can be \textit{compressed} from a spatial footprint of $17$ logical qubits to $10$. However, doing so requires us to run the protocol in a very particular way, indeed, it requires a particular ordering of its columns. Two possibilities thus arise for this particular concatenated scheme; one can either consider a minimal space layout, where both copies of the second protocol are run in the compressed form, but since the two protocols share a common column ordering, the leading-order error rate will be as in Equation~\eqref{eq:error_rate_worse}. On the other hand, a different layout, with higher spatial footprint, can achieve leading-order error rate at most as in Equation~\eqref{eq:error_rate_better}. Note that the only thing that mattered in our calculation was the \textit{relative} ordering of the columns of the two protocols. We can, therefore, fix the ordering of the columns of one copy of this protocol to one allowing compression, and only randomize the other. Since the spatial footprint of $14\T\to2\T$ is $5$, the spatial footprint of the former layout is $5 + 2\times10 = 25$, whereas for the latter layout, it is $5+10+17 = 32$.

The impact of these layout considerations on our broader survey is relatively minor. Indeed, the only compressible protocol that we consider for our competitive schemes is this $64\T\to2\CCZ$ protocol~\cite{Haah2018}. In addition, these considerations only appear when such a protocol is not the top-level distillation scheme in a concatenated chain, since permuting the columns of a top-level factory has no effect. Given that we only consider concatenated schemes with distance $\leq 8$, the only concatenated chains where these considerations have an effect is the concatenation of a distance $2$ $(3k+8)\T\to k\T$ protocol~\cite{bravyi2012magic} with the $64\T\to2\CCZ$ factory. For each of these, we consider a low-footprint layout, where all second-level factories have a fixed column ordering, and all are compressed, and we consider a high-footprint layout, where one second-level factory is compressed, and the rest have randomized columns (one fact can always have fixed column ordering and compressed because it is only relative column ordering that matters), and so are assumed to each occupy $17$ logical qubits. Future work may consider achieving a ``best of both worlds'' situation, where well-chosen column orderings on each factory minimizes the error rate, while allowing for some compression.

Let us provide a further example now.
\begin{example}
    We consider the leading order error probability of the $(2\CS\to\CCZ)\cdot(9\CCZ\to1\CCZ)$ concatenation. Note that the first step is a known synthesis conversion~\cite{howard2017application,beverland2020lower_bounds}, whereas the latter step is our novel protocol, see Construction~\ref{con:9-to-1-CCZ}. Beginning with the synthesis step, with independent $Z$ errors on each of the $4$ qubits of the $2$ $\ket{\CS}$ states, one checks that, at leading order, the resulting $\ket{\CCZ}$ state has $Z$ errors on it with weight $1$, $2$ and $3$ with probability $w_1\coloneq \frac{5}{2}p$, $w_2\coloneq p$, and $w_3\coloneq\frac{p}{2}$, respectively. The random permutation of the columns of the $\mathbb{F}_8$ matrix describing $9\CCZ\to1\CCZ$, as in the randomization~\ref{randomisation:column}, has no effect in this case, because there is only one second-level factory. Next, since the $9\CCZ\to1\CCZ$ protocol is distance $3$, at leading order, it requires $3$ errorful inputs to fail. As always, the columns of this protocol are randomly permuted (meaning that the columns of the $\mathbb{F}_8$-matrix are randomly permuted before binarization). It is convenient to define a function $K(a,b,c)$, where $a,b,c \in \{1,2,3\}$, and $a \leq b \leq c$. $K(a,b,c)$ is the probability, given that exactly three of the nine input $\ket{\CCZ}$ states had errors, and those errors had weights $a,b,c$, what is the probability that the $9\CCZ\to1\CCZ$ protocol suffers an undetectable logical error? Formally, we have
    \begin{equation}
        K(a,b,c) = \frac{1}{9\cdot8\cdot7\left(\begin{smallmatrix}3\\a\end{smallmatrix}\right)\left(\begin{smallmatrix}3\\b\end{smallmatrix}\right)\left(\begin{smallmatrix}3\\c\end{smallmatrix}\right)}\sum_{\substack{i,j,k, = 1, \ldots, 9\\i,j,k\text{ pairwise distinct}}}\sum_{\substack{x,y,z \;\in\; \mathbb{F}_2^3\\|x|=a, \;|y|=b,\; |z|=c}}\mathbb{I}\left[\substack{E(x,y,z;i,j,k)\\\text{ is an undetectable logical error}}\right].
    \end{equation}
    Here, we write, for convenience, $E(x,y,z;i,j,k)$ for the input error to the binarized $9\CCZ\to1\CCZ$ protocol which has $x,y,z$ in the $i,j,k$'th positions, respectively, and $\mathbb{I}\left[\cdot\right]$ is the indicator function.\footnote{It is worth emphasizing that this function, and more generally the output error rate of the concatenated protocols involving the $\mathbb{F}_4$ and $\mathbb{F}_8$ codes, do depend in general on the way that the matrices are binarized. Throughout, we consider binarizing $\mathbb{F}_4$ matrices in the self-dual normal basis $\{\omega, \omega^2\}$; see Section~\ref{sec:CS_distillation_protocols}, and $\mathbb{F}_8$ matrices in the self-dual normal basis given in Equation~\eqref{eq:F8_self_dual_normal_basis}. The procedure for binarizing a matrix in this basis is shown above and in Equation~\eqref{eq:qudit_to_qubit_expansion}, and the procedure in $\mathbb{F}_4$ is entirely analogous.} The matrix is small, and we may easily calculate the values $K(a,b,c)$, as shown in Table~\ref{tab:9ccz-to-1ccz_conditional_failure_function}.
    \begin{table}[ht]
        \centering
        
        \begin{tabular}{@{}cc@{\qquad}cc@{}}
        \toprule
        \((a,b,c)\) & \(K(a,b,c)\) & \((a,b,c)\) & \(K(a,b,c)\)\\
        \midrule
        \((1,1,1)\) & \(1/63\)  & \((1,1,2)\) & \(17/567\)\\
        \((1,1,3)\) & \(1/189\) & \((1,2,2)\) & \(1/189\)\\
        \((1,2,3)\) & \(1/27\)  & \((1,3,3)\) & \(1/63\)\\
        \((2,2,2)\) & \(23/567\)& \((2,2,3)\) & \(1/189\)\\
        \((2,3,3)\) & \(1/63\)  & \((3,3,3)\) & \(1/21\)\\
        \bottomrule
        \end{tabular}
        \captionsetup{labelsep=none}
        \caption{}
        \label{tab:9ccz-to-1ccz_conditional_failure_function}
    \end{table}
    Next, we let $P_{abc}$ be the number of distinct permutations of the numbers $a\leq b\leq c$; for example $P_{123} = 6$, but $P_{112} = 3$. Then, we have the (per-state average) leading order error rate of the concatenated protocol:
    \begin{equation}
        \begin{pmatrix}
            9\\3
        \end{pmatrix}\sum_{1\leq a\leq b \leq c \leq 3}P_{abc}w_aw_bw_cK(a,b,c) = \frac{2888}{27}p^3.
    \end{equation}
\end{example}
\subsection{Tables of Schemes}\label{subsec:tables_footprints_errors}

We now tabulate the statistics relevant to some of our competitive protocols in Tables~\ref{tab:pareto-t-to-cs},~\ref{tab:pareto-t-to-ccz},~\ref{tab:pareto-cs-to-cs}, and~\ref{tab:pareto-cs-to-ccz}. To keep the tables of reasonable size, we only tabulate the statistics of protocols that appear in some Pareto frontier over a range that covers the $x$-axis for at least half an order of magnitude. For each concatenated scheme, we show the protocol chain, the leading order term for the (per-state average) output error, as well as the actual value at physical input error rates $10^{-3}$ and $10^{-6}$, the spatial footprint, and the time and volume rates. The values of output error rate are calculated using two-term expansions (leading order and second-leading order). Because of different amounts of post-selection, the time and volume rates do depend mildly on input physical error rate. We display values for input physical error rate $10^{-3}$, but the values at $10^{-6}$ are typically only different by a few percent.

\begin{table*}[ht]
\centering
\begingroup
\small
\scriptsize
\setlength{\tabcolsep}{2.5pt}
\renewcommand{\arraystretch}{1.14}
\resizebox{\textwidth}{!}{%
\begin{tabular}{@{}>{\raggedright\arraybackslash}p{0.34\textwidth}cccccc@{}}
\toprule
\multicolumn{1}{c}{\multirow[c]{2}{*}[-1ex]{Protocol chain}}
&
\multirow[c]{2}{*}[-0.9ex]{\makecell[c]{Leading-order\\output error}}
&
\multicolumn{2}{c}{Output error}
&
\multirow[c]{2}{*}[-0.9ex]{\makecell[c]{Spatial\\footprint}}
&
\multicolumn{2}{c}{Resource rate}
\\
\cmidrule(lr){3-4}
\cmidrule(lr){6-7}
&
&
\(p=10^{-3}\)
&
\(p=10^{-6}\)
&
&
\makecell[c]{Time}
&
\makecell[c]{Volume}
\\
\midrule
\((3\mathrm{T}\!\to\!\mathrm{CS})\mathbin{\ast}\allowbreak{}\boldsymbol{(72\mathrm{CS}\!\to\!35\mathrm{CS})}\) & \(2874p^{2}\) & \(3.08\!\times\!10^{-3}\) & \(2.87\!\times\!10^{-9}\) & 75 & \(1.27\!\times\!10^{-1}\) & \(1.69\!\times\!10^{-3}\) \\
\((3\mathrm{T}\!\to\!\mathrm{CS})\) & \(3p\) & \(3.00\!\times\!10^{-3}\) & \(3.00\!\times\!10^{-6}\) & 3 & \(3.33\!\times\!10^{-1}\) & \(1.11\!\times\!10^{-1}\) \\
\((3\mathrm{T}\!\to\!\mathrm{CS})\mathbin{\ast}\allowbreak{}\boldsymbol{(28\mathrm{CS}\!\to\!13\mathrm{CS})}\) & \(498p^{2}\) & \(5.11\!\times\!10^{-4}\) & \(4.98\!\times\!10^{-10}\) & 31 & \(1.42\!\times\!10^{-1}\) & \(4.57\!\times\!10^{-3}\) \\
\((12\mathrm{T}\!\to\!\mathrm{CS})\) & \(18p^{2}\) & \(1.81\!\times\!10^{-5}\) & \(1.80\!\times\!10^{-11}\) & 4 & \(8.23\!\times\!10^{-2}\) & \(2.06\!\times\!10^{-2}\) \\
\((3\mathrm{T}\!\to\!\mathrm{CS})\mathbin{\ast}\allowbreak{}\boldsymbol{(21\mathrm{CS}\!\to\!4\mathrm{CS})}\) & \(414p^{3}\) & \(4.22\!\times\!10^{-7}\) & \(4.14\!\times\!10^{-16}\) & 17 & \(5.95\!\times\!10^{-2}\) & \(3.50\!\times\!10^{-3}\) \\
\((3\mathrm{T}\!\to\!\mathrm{CS})\mathbin{\ast}\allowbreak{}\boldsymbol{(12\mathrm{CS}\!\to\!5\mathrm{CS})}\mathbin{\ast}\allowbreak{}\boldsymbol{(12\mathrm{CS}\!\to\!5\mathrm{CS})}\) & \(10^{3.18}p^{4}\) & \(1.55\!\times\!10^{-9}\) & \(1.52\!\times\!10^{-21}\) & 75 & \(5.58\!\times\!10^{-2}\) & \(7.44\!\times\!10^{-4}\) \\
\((12\mathrm{T}\!\to\!\mathrm{CS})\mathbin{\ast}\allowbreak{}\boldsymbol{(6\mathrm{CS}\!\to\!2\mathrm{CS})}\) & \(1296p^{4}\) & \(1.31\!\times\!10^{-9}\) & \(1.30\!\times\!10^{-21}\) & 10 & \(2.74\!\times\!10^{-2}\) & \(2.74\!\times\!10^{-3}\) \\
\((12\mathrm{T}\!\to\!\mathrm{CS})\mathbin{\ast}\allowbreak{}\boldsymbol{(21\mathrm{CS}\!\to\!4\mathrm{CS})}\) & \(10^{4.95}p^{6}\) & \(9.07\!\times\!10^{-14}\) & \(8.94\!\times\!10^{-32}\) & 18 & \(1.57\!\times\!10^{-2}\) & \(8.71\!\times\!10^{-4}\) \\
\((3\mathrm{T}\!\to\!\mathrm{CS})\mathbin{\ast}\allowbreak{}\boldsymbol{(21\mathrm{CS}\!\to\!4\mathrm{CS})}\mathbin{\ast}\allowbreak{}\boldsymbol{(14\mathrm{CS}\!\to\!6\mathrm{CS})}\) & \(10^{4.69}p^{6}\) & \(5.04\!\times\!10^{-14}\) & \(4.86\!\times\!10^{-32}\) & 73 & \(2.55\!\times\!10^{-2}\) & \(3.49\!\times\!10^{-4}\) \\
\((3\mathrm{T}\!\to\!\mathrm{CS})\mathbin{\ast}\allowbreak{}\boldsymbol{(10\mathrm{CS}\!\to\!4\mathrm{CS})}\mathbin{\ast}\allowbreak{}\boldsymbol{(21\mathrm{CS}\!\to\!4\mathrm{CS})}\) & \(10^{3.57}p^{6}\) & \(3.82\!\times\!10^{-15}\) & \(3.74\!\times\!10^{-33}\) & 69 & \(2.46\!\times\!10^{-2}\) & \(3.57\!\times\!10^{-4}\) \\
\((3\mathrm{T}\!\to\!\mathrm{CS})\mathbin{\ast}\allowbreak{}\boldsymbol{(12\mathrm{CS}\!\to\!5\mathrm{CS})}\mathbin{\ast}\allowbreak{}\boldsymbol{(16\mathrm{CS}\!\to\!2\mathrm{CS})}\) & \(10^{3.01}p^{6}\) & \(1.05\!\times\!10^{-15}\) & \(1.03\!\times\!10^{-33}\) & 65 & \(1.67\!\times\!10^{-2}\) & \(2.57\!\times\!10^{-4}\) \\
\((12\mathrm{T}\!\to\!\mathrm{CS})\mathbin{\ast}\allowbreak{}\boldsymbol{(4\mathrm{CS}\!\to\!\mathrm{CS})}\mathbin{\ast}\allowbreak{}\boldsymbol{(6\mathrm{CS}\!\to\!2\mathrm{CS})}\) & \(10^{6.23}p^{8}\) & \(1.71\!\times\!10^{-18}\) & \(1.68\!\times\!10^{-42}\) & 14 & \(6.86\!\times\!10^{-3}\) & \(4.90\!\times\!10^{-4}\) \\
\((18\mathrm{T}\!\to\!2\mathrm{CS})\mathbin{\ast}\allowbreak{}\boldsymbol{(6\mathrm{CS}\!\to\!2\mathrm{CS})}\mathbin{\ast}\allowbreak{}\boldsymbol{(14\mathrm{CS}\!\to\!6\mathrm{CS})}\) & \(10^{5.84}p^{8}\) & \(7.09\!\times\!10^{-19}\) & \(6.91\!\times\!10^{-43}\) & 74 & \(1.56\!\times\!10^{-2}\) & \(2.11\!\times\!10^{-4}\) \\
\((18\mathrm{T}\!\to\!2\mathrm{CS})\mathbin{\ast}\allowbreak{}\boldsymbol{(8\mathrm{CS}\!\to\!3\mathrm{CS})}\mathbin{\ast}\allowbreak{}\boldsymbol{(8\mathrm{CS}\!\to\!3\mathrm{CS})}\) & \(10^{5.14}p^{8}\) & \(1.41\!\times\!10^{-19}\) & \(1.37\!\times\!10^{-43}\) & 70 & \(1.53\!\times\!10^{-2}\) & \(2.19\!\times\!10^{-4}\) \\
\((20\mathrm{T}\!\to\!4\mathrm{T})\mathbin{\ast}\allowbreak{}(18\mathrm{T}\!\to\!2\mathrm{CS})\mathbin{\ast}\allowbreak{}\boldsymbol{(4\mathrm{CS}\!\to\!\mathrm{CS})}\) & \(10^{2.96}p^{8}\) & \(9.15\!\times\!10^{-22}\) & \(9.08\!\times\!10^{-46}\) & 63 & \(5.44\!\times\!10^{-3}\) & \(8.64\!\times\!10^{-5}\) \\
\bottomrule
\end{tabular}%
}
\endgroup
\caption{Statistics for a selection of competitive protocols taking injected $\ket{\T}$ states as input, and distilling $\ket{\CS}$ states as output. We display those protocols that appear in one of our Pareto frontiers over a range at least half an order of magnitude on the $x$-axis.}\label{tab:pareto-t-to-cs}
\end{table*}

\begin{table*}[ht]
\centering
\begingroup
\small
\scriptsize
\setlength{\tabcolsep}{2.5pt}
\renewcommand{\arraystretch}{1.14}
\resizebox{\textwidth}{!}{%
\begin{tabular}{@{}>{\raggedright\arraybackslash}p{0.34\textwidth}cccccc@{}}
\toprule
\multicolumn{1}{c}{\multirow[c]{2}{*}[-1ex]{Protocol chain}}
&
\multirow[c]{2}{*}[-0.9ex]{\makecell[c]{Leading-order\\output error}}
&
\multicolumn{2}{c}{Output error}
&
\multirow[c]{2}{*}[-0.9ex]{\makecell[c]{Spatial\\footprint}}
&
\multicolumn{2}{c}{Resource rate}
\\
\cmidrule(lr){3-4}
\cmidrule(lr){6-7}
&
&
\(p=10^{-3}\)
&
\(p=10^{-6}\)
&
&
\makecell[c]{Time}
&
\makecell[c]{Volume}
\\
\midrule
\((146\mathrm{T}\!\to\!24\mathrm{CCZ})\) & \(4030p^{2}\) & \(4.04\!\times\!10^{-3}\) & \(4.03\!\times\!10^{-9}\) & 73 & \(1.40\!\times\!10^{-1}\) & \(1.92\!\times\!10^{-3}\) \\
\((44\mathrm{T}\!\to\!7\mathrm{CCZ})\) & \(460p^{2}\) & \(4.61\!\times\!10^{-4}\) & \(4.60\!\times\!10^{-10}\) & 22 & \(1.52\!\times\!10^{-1}\) & \(6.91\!\times\!10^{-3}\) \\
\((8\mathrm{T}\!\to\!\mathrm{CCZ})\) & \(28p^{2}\) & \(2.81\!\times\!10^{-5}\) & \(2.80\!\times\!10^{-11}\) & 4 & \(1.24\!\times\!10^{-1}\) & \(3.10\!\times\!10^{-2}\) \\
\((20\mathrm{T}\!\to\!3\mathrm{CCZ})\mathbin{\ast}\allowbreak{}\boldsymbol{(20\mathrm{CCZ}\!\to\!6\mathrm{CCZ})}\) & \(10^{4.07}p^{4}\) & \(1.17\!\times\!10^{-8}\) & \(1.17\!\times\!10^{-20}\) & 73 & \(4.41\!\times\!10^{-2}\) & \(6.04\!\times\!10^{-4}\) \\
\((64\mathrm{T}\!\to\!2\mathrm{CCZ})\) & \(2368p^{4}\) & \(2.38\!\times\!10^{-9}\) & \(2.37\!\times\!10^{-21}\) & 10 & \(2.92\!\times\!10^{-2}\) & \(2.92\!\times\!10^{-3}\) \\
\((38\mathrm{T}\!\to\!6\mathrm{CCZ})\mathbin{\ast}\allowbreak{}\boldsymbol{(8\mathrm{CCZ}\!\to\!2\mathrm{CCZ})}\) & \(10^{3.24}p^{4}\) & \(1.75\!\times\!10^{-9}\) & \(1.74\!\times\!10^{-21}\) & 73 & \(3.80\!\times\!10^{-2}\) & \(5.20\!\times\!10^{-4}\) \\
\((8\mathrm{T}\!\to\!\mathrm{CCZ})\mathbin{\ast}\allowbreak{}\boldsymbol{(9\mathrm{CCZ}\!\to\!\mathrm{CCZ})}\) & \(10^{4.58}p^{6}\) & \(3.79\!\times\!10^{-14}\) & \(3.76\!\times\!10^{-32}\) & 13 & \(1.38\!\times\!10^{-2}\) & \(1.06\!\times\!10^{-3}\) \\
\((20\mathrm{T}\!\to\!3\mathrm{CCZ})\mathbin{\ast}\allowbreak{}\boldsymbol{(21\mathrm{CCZ}\!\to\!3\mathrm{CCZ})}\) & \(10^{4.08}p^{6}\) & \(1.20\!\times\!10^{-14}\) & \(1.19\!\times\!10^{-32}\) & 64 & \(2.10\!\times\!10^{-2}\) & \(3.28\!\times\!10^{-4}\) \\
\((8\mathrm{T}\!\to\!\mathrm{CCZ})\mathbin{\ast}\allowbreak{}\boldsymbol{(64\mathrm{CCZ}\!\to\!8\mathrm{CCZ})}\) & \(10^{7.36}p^{8}\) & \(2.29\!\times\!10^{-17}\) & \(2.27\!\times\!10^{-41}\) & 43 & \(1.55\!\times\!10^{-2}\) & \(3.60\!\times\!10^{-4}\) \\
\((64\mathrm{T}\!\to\!2\mathrm{CCZ})\mathbin{\ast}\allowbreak{}\boldsymbol{(8\mathrm{CCZ}\!\to\!2\mathrm{CCZ})}\) & \(10^{6.83}p^{8}\) & \(6.82\!\times\!10^{-18}\) & \(6.77\!\times\!10^{-42}\) & 28 & \(7.31\!\times\!10^{-3}\) & \(2.61\!\times\!10^{-4}\) \\
\((8\mathrm{T}\!\to\!\mathrm{CCZ})\mathbin{\ast}\allowbreak{}\boldsymbol{(14\mathrm{CCZ}\!\to\!4\mathrm{CCZ})}\mathbin{\ast}\allowbreak{}\boldsymbol{(8\mathrm{CCZ}\!\to\!2\mathrm{CCZ})}\) & \(10^{6.16}p^{8}\) & \(1.44\!\times\!10^{-18}\) & \(1.43\!\times\!10^{-42}\) & 55 & \(8.86\!\times\!10^{-3}\) & \(1.61\!\times\!10^{-4}\) \\
\((14\mathrm{T}\!\to\!2\mathrm{T})\mathbin{\ast}\allowbreak{}(64\mathrm{T}\!\to\!2\mathrm{CCZ})\) & \(10^{4.46}p^{8}\) & \(2.93\!\times\!10^{-20}\) & \(2.91\!\times\!10^{-44}\) & 32 & \(4.40\!\times\!10^{-3}\) & \(1.38\!\times\!10^{-4}\) \\
\((20\mathrm{T}\!\to\!4\mathrm{T})\mathbin{\ast}\allowbreak{}(64\mathrm{T}\!\to\!2\mathrm{CCZ})\) & \(10^{3.47}p^{8}\) & \(2.97\!\times\!10^{-21}\) & \(2.94\!\times\!10^{-45}\) & 68 & \(6.12\!\times\!10^{-3}\) & \(9.01\!\times\!10^{-5}\) \\
\bottomrule
\end{tabular}%
}
\endgroup
\caption{Statistics for a selection of competitive protocols taking injected $\ket{\T}$ states as input, and distilling $\ket{\CCZ}$ states as output. We display those protocols that appear in one of our Pareto frontiers over a range at least half an order of magnitude on the $x$-axis. The $((3k+8)\T\to k\T)*(64\T\to2\CCZ)$ entries use the higher-footprint randomized column ordering layout.}
\label{tab:pareto-t-to-ccz}
\end{table*}

\begin{table*}[ht]
\centering
\begingroup
\small
\scriptsize
\setlength{\tabcolsep}{2.5pt}
\renewcommand{\arraystretch}{1.14}
\resizebox{\textwidth}{!}{%
\begin{tabular}{@{}>{\raggedright\arraybackslash}p{0.34\textwidth}cccccc@{}}
\toprule
\multicolumn{1}{c}{\multirow[c]{2}{*}[-1ex]{Protocol chain}}
&
\multirow[c]{2}{*}[-0.9ex]{\makecell[c]{Leading-order\\output error}}
&
\multicolumn{2}{c}{Output error}
&
\multirow[c]{2}{*}[-0.9ex]{\makecell[c]{Spatial\\footprint}}
&
\multicolumn{2}{c}{Resource rate}
\\
\cmidrule(lr){3-4}
\cmidrule(lr){6-7}
&
&
\(p=10^{-3}\)
&
\(p=10^{-6}\)
&
&
\makecell[c]{Time}
&
\makecell[c]{Volume}
\\
\midrule
\(\boldsymbol{(74\mathrm{CS}\!\to\!36\mathrm{CS})}\) & \(10^{3.31}p^{2}\) & \(2.02\!\times\!10^{-3}\) & \(2.02\!\times\!10^{-9}\) & 74 & \(2.07\!\times\!10^{-1}\) & \(2.80\!\times\!10^{-3}\) \\
\(\boldsymbol{(34\mathrm{CS}\!\to\!16\mathrm{CS})}\) & \(472p^{2}\) & \(4.73\!\times\!10^{-4}\) & \(4.72\!\times\!10^{-10}\) & 34 & \(2.19\!\times\!10^{-1}\) & \(6.45\!\times\!10^{-3}\) \\
\(\boldsymbol{(4\mathrm{CS}\!\to\!\mathrm{CS})}\) & \(12p^{2}\) & \(1.20\!\times\!10^{-5}\) & \(1.20\!\times\!10^{-11}\) & 4 & \(1.24\!\times\!10^{-1}\) & \(3.10\!\times\!10^{-2}\) \\
\(\boldsymbol{(21\mathrm{CS}\!\to\!4\mathrm{CS})}\) & \(123p^{3}\) & \(1.25\!\times\!10^{-7}\) & \(1.23\!\times\!10^{-16}\) & 14 & \(9.12\!\times\!10^{-2}\) & \(6.52\!\times\!10^{-3}\) \\
\(\boldsymbol{(10\mathrm{CS}\!\to\!4\mathrm{CS})}\mathbin{\ast}\allowbreak{}\boldsymbol{(16\mathrm{CS}\!\to\!7\mathrm{CS})}\) & \(10^{3.21}p^{4}\) & \(1.64\!\times\!10^{-9}\) & \(1.63\!\times\!10^{-21}\) & 74 & \(8.58\!\times\!10^{-2}\) & \(1.16\!\times\!10^{-3}\) \\
\(\boldsymbol{(16\mathrm{CS}\!\to\!2\mathrm{CS})}\) & \(944p^{4}\) & \(9.48\!\times\!10^{-10}\) & \(9.44\!\times\!10^{-22}\) & 10 & \(6.05\!\times\!10^{-2}\) & \(6.05\!\times\!10^{-3}\) \\
\(\boldsymbol{(4\mathrm{CS}\!\to\!\mathrm{CS})}\mathbin{\ast}\allowbreak{}\boldsymbol{(21\mathrm{CS}\!\to\!4\mathrm{CS})}\) & \(10^{4.42}p^{6}\) & \(2.67\!\times\!10^{-14}\) & \(2.65\!\times\!10^{-32}\) & 18 & \(2.36\!\times\!10^{-2}\) & \(1.31\!\times\!10^{-3}\) \\
\(\boldsymbol{(21\mathrm{CS}\!\to\!4\mathrm{CS})}\mathbin{\ast}\allowbreak{}\boldsymbol{(14\mathrm{CS}\!\to\!6\mathrm{CS})}\) & \(10^{3.63}p^{6}\) & \(4.41\!\times\!10^{-15}\) & \(4.29\!\times\!10^{-33}\) & 70 & \(3.91\!\times\!10^{-2}\) & \(5.59\!\times\!10^{-4}\) \\
\(\boldsymbol{(16\mathrm{CS}\!\to\!2\mathrm{CS})}\mathbin{\ast}\allowbreak{}\boldsymbol{(32\mathrm{CS}\!\to\!15\mathrm{CS})}\) & \(10^{7.20}p^{8}\) & \(1.58\!\times\!10^{-17}\) & \(1.57\!\times\!10^{-41}\) & 74 & \(2.84\!\times\!10^{-2}\) & \(3.83\!\times\!10^{-4}\) \\
\(\boldsymbol{(16\mathrm{CS}\!\to\!2\mathrm{CS})}\mathbin{\ast}\allowbreak{}\boldsymbol{(6\mathrm{CS}\!\to\!2\mathrm{CS})}\) & \(10^{5.95}p^{8}\) & \(8.98\!\times\!10^{-19}\) & \(8.91\!\times\!10^{-43}\) & 22 & \(2.02\!\times\!10^{-2}\) & \(9.17\!\times\!10^{-4}\) \\
\(\boldsymbol{(6\mathrm{CS}\!\to\!2\mathrm{CS})}\mathbin{\ast}\allowbreak{}\boldsymbol{(6\mathrm{CS}\!\to\!2\mathrm{CS})}\mathbin{\ast}\allowbreak{}\boldsymbol{(14\mathrm{CS}\!\to\!6\mathrm{CS})}\) & \(10^{5.13}p^{8}\) & \(1.37\!\times\!10^{-19}\) & \(1.36\!\times\!10^{-43}\) & 74 & \(2.35\!\times\!10^{-2}\) & \(3.18\!\times\!10^{-4}\) \\
\(\boldsymbol{(6\mathrm{CS}\!\to\!2\mathrm{CS})}\mathbin{\ast}\allowbreak{}\boldsymbol{(8\mathrm{CS}\!\to\!3\mathrm{CS})}\mathbin{\ast}\allowbreak{}\boldsymbol{(8\mathrm{CS}\!\to\!3\mathrm{CS})}\) & \(10^{4.43}p^{8}\) & \(2.72\!\times\!10^{-20}\) & \(2.70\!\times\!10^{-44}\) & 70 & \(2.32\!\times\!10^{-2}\) & \(3.31\!\times\!10^{-4}\) \\
\bottomrule
\end{tabular}%
}
\endgroup
\caption{Statistics for a selection of competitive protocols taking injected $\ket{\CS}$ states as input, and distilling $\ket{\CS}$ states as output. We display those protocols that appear in one of our Pareto frontiers over a range at least half an order of magnitude on the $x$-axis.}
\label{tab:pareto-cs-to-cs}
\end{table*}

\begin{table*}[ht]
\centering
\begingroup
\small
\scriptsize
\setlength{\tabcolsep}{2.5pt}
\renewcommand{\arraystretch}{1.14}
\resizebox{\textwidth}{!}{%
\begin{tabular}{@{}>{\raggedright\arraybackslash}p{0.34\textwidth}cccccc@{}}
\toprule
\multicolumn{1}{c}{\multirow[c]{2}{*}[-1ex]{Protocol chain}}
&
\multirow[c]{2}{*}[-0.9ex]{\makecell[c]{Leading-order\\output error}}
&
\multicolumn{2}{c}{Output error}
&
\multirow[c]{2}{*}[-0.9ex]{\makecell[c]{Spatial\\footprint}}
&
\multicolumn{2}{c}{Resource rate}
\\
\cmidrule(lr){3-4}
\cmidrule(lr){6-7}
&
&
\(p=10^{-3}\)
&
\(p=10^{-6}\)
&
&
\makecell[c]{Time}
&
\makecell[c]{Volume}
\\
\midrule
\(\boldsymbol{(26\mathrm{CS}\!\to\!12\mathrm{CS})}\mathbin{\ast}\allowbreak{}(2\mathrm{CS}\!\to\!\mathrm{CCZ})\) & \(10^{2.76}p^{2}\) & \(5.82\!\times\!10^{-4}\) & \(5.81\!\times\!10^{-10}\) & 74 & \(1.09\!\times\!10^{-1}\) & \(1.48\!\times\!10^{-3}\) \\
\(\boldsymbol{(4\mathrm{CS}\!\to\!\mathrm{CS})}\mathbin{\ast}\allowbreak{}(2\mathrm{CS}\!\to\!\mathrm{CCZ})\) & \(24p^{2}\) & \(2.40\!\times\!10^{-5}\) & \(2.40\!\times\!10^{-11}\) & 8 & \(6.20\!\times\!10^{-2}\) & \(7.75\!\times\!10^{-3}\) \\
\(\boldsymbol{(21\mathrm{CS}\!\to\!4\mathrm{CS})}\mathbin{\ast}\allowbreak{}(2\mathrm{CS}\!\to\!\mathrm{CCZ})\) & \(246p^{3}\) & \(2.49\!\times\!10^{-7}\) & \(2.46\!\times\!10^{-16}\) & 30 & \(4.56\!\times\!10^{-2}\) & \(1.52\!\times\!10^{-3}\) \\
\((2\mathrm{CS}\!\to\!\mathrm{CCZ})\mathbin{\ast}\allowbreak{}\boldsymbol{(9\mathrm{CCZ}\!\to\!\mathrm{CCZ})}\) & \(10^{2.03}p^{3}\) & \(1.08\!\times\!10^{-7}\) & \(1.07\!\times\!10^{-16}\) & 13 & \(2.68\!\times\!10^{-2}\) & \(2.06\!\times\!10^{-3}\) \\
\(\boldsymbol{(16\mathrm{CS}\!\to\!2\mathrm{CS})}\mathbin{\ast}\allowbreak{}(2\mathrm{CS}\!\to\!\mathrm{CCZ})\) & \(1888p^{4}\) & \(1.90\!\times\!10^{-9}\) & \(1.89\!\times\!10^{-21}\) & 18 & \(3.02\!\times\!10^{-2}\) & \(1.68\!\times\!10^{-3}\) \\
\(\boldsymbol{(8\mathrm{CS}\!\to\!3\mathrm{CS})}\mathbin{\ast}\allowbreak{}\boldsymbol{(8\mathrm{CS}\!\to\!3\mathrm{CS})}\mathbin{\ast}\allowbreak{}(2\mathrm{CS}\!\to\!\mathrm{CCZ})\) & \(10^{3.15}p^{4}\) & \(1.41\!\times\!10^{-9}\) & \(1.41\!\times\!10^{-21}\) & 68 & \(3.46\!\times\!10^{-2}\) & \(5.09\!\times\!10^{-4}\) \\
\(\boldsymbol{(4\mathrm{CS}\!\to\!\mathrm{CS})}\mathbin{\ast}\allowbreak{}\boldsymbol{(4\mathrm{CS}\!\to\!\mathrm{CS})}\mathbin{\ast}\allowbreak{}(2\mathrm{CS}\!\to\!\mathrm{CCZ})\) & \(576p^{4}\) & \(5.78\!\times\!10^{-10}\) & \(5.76\!\times\!10^{-22}\) & 12 & \(1.55\!\times\!10^{-2}\) & \(1.29\!\times\!10^{-3}\) \\
\(\boldsymbol{(16\mathrm{CS}\!\to\!7\mathrm{CS})}\mathbin{\ast}\allowbreak{}\boldsymbol{(4\mathrm{CS}\!\to\!\mathrm{CS})}\mathbin{\ast}\allowbreak{}(2\mathrm{CS}\!\to\!\mathrm{CCZ})\) & \(10^{2.15}p^{4}\) & \(1.41\!\times\!10^{-10}\) & \(1.41\!\times\!10^{-22}\) & 72 & \(2.65\!\times\!10^{-2}\) & \(3.68\!\times\!10^{-4}\) \\
\(\boldsymbol{(4\mathrm{CS}\!\to\!\mathrm{CS})}\mathbin{\ast}\allowbreak{}(2\mathrm{CS}\!\to\!\mathrm{CCZ})\mathbin{\ast}\allowbreak{}\boldsymbol{(9\mathrm{CCZ}\!\to\!\mathrm{CCZ})}\) & \(10^{4.37}p^{6}\) & \(2.37\!\times\!10^{-14}\) & \(2.36\!\times\!10^{-32}\) & 17 & \(6.89\!\times\!10^{-3}\) & \(4.05\!\times\!10^{-4}\) \\
\(\boldsymbol{(6\mathrm{CS}\!\to\!2\mathrm{CS})}\mathbin{\ast}\allowbreak{}\boldsymbol{(21\mathrm{CS}\!\to\!4\mathrm{CS})}\mathbin{\ast}\allowbreak{}(2\mathrm{CS}\!\to\!\mathrm{CCZ})\) & \(10^{3.99}p^{6}\) & \(9.82\!\times\!10^{-15}\) & \(9.76\!\times\!10^{-33}\) & 66 & \(1.57\!\times\!10^{-2}\) & \(2.38\!\times\!10^{-4}\) \\
\(\boldsymbol{(21\mathrm{CS}\!\to\!4\mathrm{CS})}\mathbin{\ast}\allowbreak{}\boldsymbol{(6\mathrm{CS}\!\to\!2\mathrm{CS})}\mathbin{\ast}\allowbreak{}(2\mathrm{CS}\!\to\!\mathrm{CCZ})\) & \(10^{3.32}p^{6}\) & \(2.14\!\times\!10^{-15}\) & \(2.09\!\times\!10^{-33}\) & 70 & \(1.52\!\times\!10^{-2}\) & \(2.17\!\times\!10^{-4}\) \\
\(\boldsymbol{(16\mathrm{CS}\!\to\!2\mathrm{CS})}\mathbin{\ast}\allowbreak{}\boldsymbol{(12\mathrm{CS}\!\to\!5\mathrm{CS})}\mathbin{\ast}\allowbreak{}(2\mathrm{CS}\!\to\!\mathrm{CCZ})\) & \(10^{6.75}p^{8}\) & \(5.69\!\times\!10^{-18}\) & \(5.65\!\times\!10^{-42}\) & 74 & \(1.26\!\times\!10^{-2}\) & \(1.70\!\times\!10^{-4}\) \\
\(\boldsymbol{(16\mathrm{CS}\!\to\!2\mathrm{CS})}\mathbin{\ast}\allowbreak{}\boldsymbol{(4\mathrm{CS}\!\to\!\mathrm{CS})}\mathbin{\ast}\allowbreak{}(2\mathrm{CS}\!\to\!\mathrm{CCZ})\) & \(10^{5.95}p^{8}\) & \(8.98\!\times\!10^{-19}\) & \(8.91\!\times\!10^{-43}\) & 26 & \(7.56\!\times\!10^{-3}\) & \(2.91\!\times\!10^{-4}\) \\
\(\boldsymbol{(6\mathrm{CS}\!\to\!2\mathrm{CS})}\mathbin{\ast}\allowbreak{}\boldsymbol{(4\mathrm{CS}\!\to\!\mathrm{CS})}\mathbin{\ast}\allowbreak{}\boldsymbol{(4\mathrm{CS}\!\to\!\mathrm{CS})}\mathbin{\ast}\allowbreak{}(2\mathrm{CS}\!\to\!\mathrm{CCZ})\) & \(10^{4.92}p^{8}\) & \(8.36\!\times\!10^{-20}\) & \(8.29\!\times\!10^{-44}\) & 30 & \(5.15\!\times\!10^{-3}\) & \(1.72\!\times\!10^{-4}\) \\
\(\boldsymbol{(6\mathrm{CS}\!\to\!2\mathrm{CS})}\mathbin{\ast}\allowbreak{}\boldsymbol{(6\mathrm{CS}\!\to\!2\mathrm{CS})}\mathbin{\ast}\allowbreak{}\boldsymbol{(6\mathrm{CS}\!\to\!2\mathrm{CS})}\mathbin{\ast}\allowbreak{}(2\mathrm{CS}\!\to\!\mathrm{CCZ})\) & \(10^{4.82}p^{8}\) & \(6.69\!\times\!10^{-20}\) & \(6.64\!\times\!10^{-44}\) & 74 & \(9.15\!\times\!10^{-3}\) & \(1.24\!\times\!10^{-4}\) \\
\(\boldsymbol{(6\mathrm{CS}\!\to\!2\mathrm{CS})}\mathbin{\ast}\allowbreak{}\boldsymbol{(8\mathrm{CS}\!\to\!3\mathrm{CS})}\mathbin{\ast}\allowbreak{}\boldsymbol{(4\mathrm{CS}\!\to\!\mathrm{CS})}\mathbin{\ast}\allowbreak{}(2\mathrm{CS}\!\to\!\mathrm{CCZ})\) & \(10^{4.22}p^{8}\) & \(1.69\!\times\!10^{-20}\) & \(1.68\!\times\!10^{-44}\) & 70 & \(7.72\!\times\!10^{-3}\) & \(1.10\!\times\!10^{-4}\) \\
\bottomrule
\end{tabular}%
}
\endgroup
\caption{Statistics for a selection of competitive protocols taking injected $\ket{\CS}$ states as input, and distilling $\ket{\CCZ}$ states as output. We display those protocols that appear in one of our Pareto frontiers over a range at least half an order of magnitude on the $x$-axis.}
\label{tab:pareto-cs-to-ccz}
\end{table*}

%% file: F64_maximal_curves.tex
\begin{table}[htbp]
\centering
\begingroup
\small
\renewcommand{\arraystretch}{1.5}
\setlength{\tabcolsep}{4pt}

\setcellgapes{2.5pt}
\makegapedcells

\begin{tabularx}{\linewidth}{|c|Y|Y|}
\hline
$g$ & Curve(s) & Comments\\ \hline
$2$ & $y^2+y=x^5+x^3+\theta^5 x$~\cite{vandergeer1996} & - \\ \hline
$3$ & \makecell[c]{Two curves are known:\\$y^4+y^2+y=x^3$~\cite[Example~12]{garcia2001construction},\\  $y^4+y=x^3$.} & These two curves are not $\F_{64}$-isomorphic~\cite{Garcia_2009_unramified}.\\ \hline
$4$ & $y^2+y=x^9$~\cite[Method I]{vandergeer1996} & - \\ \hline
$6$ & \makecell[c]{$\begin{cases}
    y_1^2+y_1=\theta x^2 + \theta^4 x^3 + \theta^2 x^5\\
    y_2^2+y_2=\theta^2 x^2 + \theta x^3 + \theta^4 x^5
\end{cases}$\\\cite[Method II]{vandergeer1996}}
& \makecell[c]{A singular plane model\\$z^4+(\theta+1)z^2+\theta z=\theta^6 x^3 (x+1)^3$\\can be obtained by defining\\$z=y_1+\theta^6 y_2$.}\\ \hline
$7$ & $y^8+y=x^3$ & Follows from~\cite[Prop.~6.4.1]{Stichtenoth}.\\ \hline
$9$ & \makecell[c]{$c\cdot s(U,V,W)=0$ where \\
$s(U,V,W)=U^8W + W^8V+V^8U + (UVW)^3$,\\ 
$U=\beta X+Y+\beta^{65} Z$,\\ $V=\beta^{65} X + \beta Y + Z$, \\$W=X+\beta^{65} Y + \beta Z$.\\ \cite[Remark~2.2~\&~2.3]{cossidente2000curves}.} & \makecell[c]{$c=\beta^{27}$, $\beta\in\F_{512}$ satisfies $\beta^9+\beta+1=0$.\\ The multiplication by $c$ makes the coefficients \\in $c\cdot s(U,V,W)=0$ land in $\F_{64}$.}\\ \hline
$10$ & $y^9=x^6+x^3$ & \makecell[c]{This is a singular plane model of the\\ $q=2,\; n=3$ case of~\cite{Beelen_2018,GGS2010,GK2007}.}\\ \hline
$12$ & $\begin{cases}
    y_1^2 + y_1 = \theta^4 x^9\\
    y_2^2 + y_2 = \theta^2 x^9
\end{cases}$~\cite[Method I]{vandergeer1996}
& \makecell[c]{A singular plane model $z^4+z^2+z=x^9$\\ can be obtained by defining $z=\theta y_1 + \theta^2 y_2$.\\
This curve also appears as \\the (II.1) case of~\cite[Thm.~2.1]{cossidente2000curves}.}\\ \hline
$28$ & $y^{8}+y=x^{9}$ & Hermitian curve\\ \hline
\end{tabularx}
\endgroup

\caption{A summary of maximal curves over $\F_{64}$. $g$ is the genus and a maximal curve of that genus has $64+16g+1$ points. $\alpha$ is a primitive element of $\F_{64}$ that satisfies $\alpha^6+\alpha+1=0$, taking $\theta=\alpha^{27}$, then $\theta$ is a generator of the $\F_8$ subfield and satisfies $\theta^3+\theta+1=0$.}
\label{table:F64_maximal_curves}
\end{table}